\pdfoutput=1
\documentclass[a4paper,USenglish,thm-restate,cleveref,numberwithinsect]{lipics-v2021}
\nolinenumbers

\newcommand{\citet}{{\color{red}{add authors}}\cite} %

\newif\iflong
\longtrue %

\newif\ifdraft
\drafttrue %

\newcommand{\mytitleb}{%
Hyperproperty-Preserving Register Specifications }
\iflong
  \title{\mytitleb (Extended Version)}
\else
  \title{\mytitleb}
\fi

\author{Yoav {Ben Shimon}}{Tel Aviv University, Israel}{yoavbenshimon@gmail.com}{https://orcid.org/0000-0002-8893-8874}{}

\author{Ori Lahav}{Tel Aviv University, Israel}{orilahav@tau.ac.il}{https://orcid.org/0000-0003-4305-6998}{}

\author{Sharon Shoham}{Tel Aviv University, Israel}{sharon.shoham@gmail.com}{https://orcid.org/0000-0002-7226-3526}{}

\authorrunning{Y. Ben Shimon et al.}

\Copyright{Yoav {Ben Shimon}, Ori Lahav, Sharon Shoham}

\ccsdesc{Theory of computation~Concurrency}
\ccsdesc{Theory of computation~Program specifications}
\ccsdesc{Theory of computation~Distributed algorithms}
\ccsdesc{Theory of computation~Concurrent algorithms}

\keywords{Hyperproperties, Concurrent objects, Distributed objects, Linearizability, Strong linearizability, Simulation}

\funding{
This work is supported by
the European Research Council (ERC) under the European Union's Horizon 2020
research and innovation programme (grant agreement no.~[759102-SVIS] and~851811)
and the Israel Science Foundation (grant number~2117/23 and~814/22).}

\iflong\else
\relatedversion{An extended version of the paper is available at \url{???}}
\fi

\usepackage[T1]{fontenc}

\usepackage{multicol}
\usepackage{tikz}
\usepackage{mathpartir}
\usepackage{xspace}
\usepackage[inline]{enumitem}
\usepackage{xifthen}
\usepackage{pifont} %
\usepackage{subcaption}

\usepackage{framed}
\colorlet{shadecolor}{gray!10}

\usepackage{diagbox}

\let\oldparagraph\paragraph
\renewcommand{\paragraph}[1]{\oldparagraph{\bfseries #1}}

\usepackage[noend,ruled,vlined]{algorithm2e}

\SetKwIF{Await}{}{}{await}{do}{}{}{}
\SetKwFor{Case}{case}{then}{}%
\SetKw{Pick}{pick}
\SetKw{Broadcast}{broadcast}
\SetKwIF{Reply}{}{}{reply}{to}{}{}{}
\SetKwFunction{Read}{read}
\SetKwFunction{Write}{write}
\SetKwFunction{Query}{query}
\SetKwFunction{Update}{update}
\SetKwFunction{Barrier}{barrier}
\SetKwProg{METHOD}{Method}{}{}
\SetKwProg{FUNCTION}{Function}{}{}
\SetKwFunction{Coin}{coin}

\SetKwIF{When}{}{}{when}{received}{}{}{}
\SetKw{Wait}{wait until}
\SetKwIF{Atomic}{}{}{atomic}{}{}{}{}

\SetKwRepeat{Do}{do}{while}

\usepackage{hyperref}

\newcommand{\cmark}{\text{{\color{green!70!black}\ding{51}}}}%
\newcommand{\xmark}{\text{\color{red!70!black}\ding{55}}}%
\newcommand{\checkmarkt}[1]{%
	\edef\TVALUE{{#1}}%
	\expandafter\ifstrequal\TVALUE{yes}{\cmark}{}%
	\expandafter\ifstrequal\TVALUE{no}{\xmark}{}%
			\expandafter\ifstrequal\TVALUE{no*}{\xmark $^\star$}{}%
}

\AtEndPreamble{%
\newtheorem{notation}[theorem]{Notation}
\newtheorem*{notation*}{Notation}
}

\crefformat{section}{#2\S{}#1#3}
\Crefname{section}{Section}{Section}
\Crefformat{section}{Section #2#1#3}
\crefname{corollary}{\text{Corollary}}{\text{corollaries}}
\Crefname{corollary}{\text{Corollary}}{\text{Corollaries}}
\crefname{lemma}{\text{Lemma}}{\text{Lemmas}}
\Crefname{lemma}{\text{Lemma}}{\text{Lemmas}}
\crefname{proposition}{\text{Prop.}}{\text{Propositions}}
\Crefname{proposition}{\text{Proposition}}{\text{Propositions}}
\crefname{definition}{\text{Def.}}{\text{Definitions}}
\Crefname{definition}{\text{Definition}}{\text{Definitions}}
\crefname{notation}{\text{Notation}}{\text{Notations}}
\Crefname{notation}{\text{Notation}}{\text{Notations}}
\crefname{theorem}{\text{Thm.}}{\text{Theorems}}
\Crefname{theorem}{\text{Theorem}}{\text{Theorems}}
\crefname{figure}{\text{Fig.}}{\text{Figures}}
\Crefname{figure}{\text{Figure}}{\text{Figures}}
\crefname{example}{\text{Example}}{\text{Examples}}
\Crefname{example}{\text{Example}}{\text{Examples}}
\Crefname{algorithm}{\text{Algorithm}}{\text{Algorithms}}
\crefformat{appendix}{#2\S{}#1#3}

\newcommand{\condletter}{R}
\crefformat{enumi}{#2(\condletter#1)#3}

\newcommand{\para}[1]{\medskip \noindent \emph{#1.}}

\theoremstyle{definition}
\newtheorem{definition2}[theorem]{Definition}

\AtEndPreamble{%
\newcounter{claimcounter}
\numberwithin{claimcounter}{theorem}
\newenvironment{claim2}[1]{ \refstepcounter{claimcounter}\par\noindent\underline{Claim \theclaimcounter:}\space#1}{}

\crefname{claimcounter}{\text{Claim}}{\text{Claims}}
\Crefname{claimcounter}{\text{Claim}}{\text{Claims}}
}

\newlist{claimenum}{enumerate}{1}
\setlist[claimenum]{label=(\arabic*),ref=\theclaimcounter~(\arabic*)}
\crefalias{claimenumi}{claimcounter}

\newcommand{\commentout}[1]{}

\newcommand{\ie}{i.e.,\xspace}
\newcommand{\eg}{e.g.,\xspace}
\newcommand{\etal}{et~al.\@\xspace}
\newcommand{\wrt}{w.r.t.\@\xspace}
\newcommand{\aka}{a.k.a.\@\xspace}
\newcommand{\wlg}{w.l.o.g.,\@\xspace}
\newcommand{\quotes}[1]{``#1''}

\newcommand{\inarr}[1]{\begin{array}{@{}l@{}}#1\end{array}}
\newcommand{\inarrII}[2]{\begin{array}{@{}l@{\;}||@{\;}l@{}}\inarr{#1}&\inarr{#2}\end{array}}
\newcommand{\inarrIII}[3]{\begin{array}{@{}l@{~~}||@{~~}l@{~~}||@{~~}l@{}}\inarr{#1}&\inarr{#2}&\inarr{#3}\end{array}}

\newcommand{\set}[1]{\{{#1}\}}

\newcommand{\pfn}{\rightharpoonup}

\renewcommand{\st}{\; | \;} %
\newcommand{\N}{{\mathbb{N}}}
\newcommand{\dom}[1]{\textsf{dom}{({#1})}}

\newcommand{\tup}[1]{{\langle{#1}\rangle}}
\newcommand{\pair}[2]{{\langle{#1,#2}\rangle}}

\newcommand{\suq}{\subseteq}

\newcommand{\size}[1]{\left|{#1}\right|}

\newcommand{\setof}[1]{\mathsf{set}(#1)}

\newcommand{\maketil}[1]{{#1}\ldots{#1}}
\newcommand{\til}{\maketil{,}}

\DeclareRobustCommand
  \Compactcdots{\mathinner{\cdotp\mkern-2mu\cdotp\mkern-2mu\cdotp}}
\renewcommand*{\cdots}{\Compactcdots}

\newcommand{\rst}[1]{|_{#1}}

\newcommand{\defiff}{\mathrel{\stackrel{\mathsf{\triangle}}{\Leftrightarrow}}}
\newcommand{\defeq}{\triangleq}
\newcommand{\eqdef}{\defeq} %
\newcommand{\powerset}[1]{\mathcal{P}({#1})}

\newcommand{\range}[2]{[#1\mathbin{..}#2]}

\makeatletter
\newcommand{\raisemath}[1]{\mathpalette{\raisem@th{#1}}}
\newcommand{\raisem@th}[3]{\raisebox{#1}{$#2#3$}}
\makeatother

\newcommand{\astep}[1]{\xrightarrow{}_{#1}}
\newcommand{\asteplab}[2]{{}\mathrel{\raisebox{-0.8pt}{\ensuremath{\xrightarrow{#1}}}_{#2}}{}}

\newcommand{\bsteplab}[2]{{}\mathrel{\raisebox{-0.8pt}{\ensuremath{\xRightarrow{#1}}}_{#2}}{}}

\makeatletter
\newcommand{\xRightarrow}[2][]{\ext@arrow 0359\Rightarrowfill@{#1}{#2}}
\makeatother

\newcommand{\seq}{\mathbin{;}}

\newcommand{\squishlist}[1][$\bullet$]{%
 \begin{list}{#1}
  { \setlength{\itemsep}{0pt}
     \setlength{\parsep}{0pt}
     \setlength{\topsep}{1pt}
     \setlength{\partopsep}{0pt}
     \setlength{\leftmargin}{1.2em}
     \setlength{\labelwidth}{0.5em}
     \setlength{\labelsep}{0.4em} } }
\newcommand{\squishend}{
  \end{list}  }

\definecolor{DarkGreen}{rgb}{0.05, 0.45, 0.05}
\newcommand{\hide}[1]{}

\newcommand{\pc}{\mathit{pc}}

\makeatletter
\newcommand{\mylabel}[2]{#2\def\@currentlabel{#2}\label{#1}}
\makeatother

\makeatletter
\newcommand*{\eqassnh}{\mathrel{\rlap{%
                     \raisebox{0.3ex}{$\m@th\cdot$}}%
                     \raisebox{-0.3ex}{$\m@th\cdot$}}%
                     }

\makeatother

\newcommand{\cloc}[1]{\mathtt{
\ifthenelse{\equal{#1}{1}}{x}{
\ifthenelse{\equal{#1}{2}}{y}{
\ifthenelse{\equal{#1}{3}}{z}{
\ifthenelse{\equal{#1}{4}}{w}{
{#1}}}}}}}

\newcommand{\cspace}[1]{\mathtt{
\ifthenelse{\equal{#1}{main}}{X_\main}{
\ifthenelse{\equal{#1}{1}}{X}{
\ifthenelse{\equal{#1}{2}}{Y}{
\ifthenelse{\equal{#1}{3}}{Z}{
{#1}}}}}}}

\newcommand{\creg}[1]{\mathtt{
\ifthenelse{\equal{#1}{1}}{a}{
\ifthenelse{\equal{#1}{2}}{b}{
\ifthenelse{\equal{#1}{3}}{c}{
\ifthenelse{\equal{#1}{4}}{d}{
\ifthenelse{\equal{#1}{5}}{e}{
\ifthenelse{\equal{#1}{6}}{f}{
{#1}}}}}}}}}

\xspaceaddexceptions{1}

\newcommand{\client}[3][]{
  \ifthenelse{\equal{#1}{}}{
    {{#2}[{#3}]}
  }{
    {{{#2}[{#3}]}\shortparallel{#1}}
  }
}

\usepackage{listings}
\newcommand{\rel}{R}

\newcommand{\trans}{t}

\newcommand{\logtrans}{t_\ell}

\newcommand{\transsq}{s}

\newcommand{\prestate}{\mathsf{pre}}
\newcommand{\lab}{\mathsf{label}}
\newcommand{\poststate}{\mathsf{post}}

\newcommand{\last}{\mathsf{last}}
\newcommand{\lastop}{\mathsf{last_{op}}}

\newcommand{\lts}{A}
\newcommand{\lQ}{{\mathtt{Q}}}
\newcommand{\lzero}{{\mathtt{0}}}
\newcommand{\linit}{{\mathtt{q_\lzero}}}
\newcommand{\lSigma}{{\boldsymbol{\Sigma}}}
\newcommand{\lT}{{\mathtt{T}}}

\newcommand{\intcomp}[2]{{#1}\parallel{#2}}

\newcommand{\method}{\mathit{m}}

\newcommand{\val}{v}
\newcommand{\tid}{p}
\newcommand{\id}{k}

\newcommand{\Method}{\mathsf{M}}

\newcommand{\Val}{\mathsf{Val}}
\newcommand{\Tid}{\mathsf{Tid}}
\newcommand{\Id}{\mathsf{Id}}
\newcommand{\Opid}{\mathsf{OPid}}

\newcommand{\Progint}{\mathsf{PInt}}

\newcommand{\methodf}{{\mathtt{method}}}
\newcommand{\valf}{{\mathtt{val}}}
\newcommand{\tidf}{{\mathtt{tid}}}
\newcommand{\idf}{{\mathtt{id}}}
\newcommand{\opidf}{{\mathtt{opid}}}
\newcommand{\inputf}{{\mathtt{input}}}
\newcommand{\outputf}{{\mathtt{output}}}

\newcommand{\sq}[1]{\bar{#1}}

\newcommand{\invact}{\mathtt{inv}}
\newcommand{\resact}{\mathtt{res}}
\newcommand{\objact}{\mathtt{int}} %

\newcommand{\inv}{i}
\newcommand{\res}{r} %

\newcommand{\Inv}{\mathsf{I}}
\newcommand{\Res}{\mathsf{R}}
\newcommand{\Objint}{\mathsf{IInt}}

\newcommand{\act}{a}

\newcommand{\Act}{\mathsf{Act}}
\newcommand{\op}{o}

\newcommand{\ressq}{\sq\res}

\newcommand{\actsq}{\sq\act} %

\newcommand{\opset}{S}

\newcommand{\wop}{w} %
\newcommand{\rop}{r} %

\newcommand{\wopmin}{{\wop}_{\min}}
\newcommand{\ropmin}{{\rop}_{\min}}

\newcommand{\prefof}[2]{\mathtt{pref}_{#1}({#2})}
\newcommand{\prefinv}[1]{\mathtt{pref}(#1)}

\newcommand{\preleq}{\precsim}
\newcommand{\preeq}{\sim}

\newcommand{\partleq}{\preceq}
\newcommand{\partlneq}{\precneqq}

\newcommand{\subseq}{\preceq_\mathsf{S}}
\newcommand{\subsneq}{\precneqq_\mathsf{S}}
\newcommand{\pref}{\preceq_\mathsf{P}}
\newcommand{\prefneq}{\precneqq_\mathsf{P}}
\newcommand{\permeq}{\sim^{\mathsf{w*}}}
\newcommand{\permleq}{\precsim^{\mathsf{w*}}}
\newcommand{\emptyword}{\varepsilon}

\newcommand{\wsleq}{\precsim_{\mathsf{P}}^{\mathsf{ww}}}

\newcommand{\idrel}[1]{[#1]}
\newcommand{\writeid}{\idrel{\regwrite}}
\newcommand{\linrel}{R}
\newcommand{\linreltotal}{R_t}
\newcommand{\readrespondsbefore}{\mathtt{rfr}}

\newcommand{\cdotnew}{\oplus}

\newcommand{\suf}[2]{#2/#1} %

\newcommand{\impl}{I}
\newcommand{\obj}{O}
\newcommand{\prog}{P}
\newcommand{\scheduler}{S}

\newcommand{\srefine}{\leq_\mathsf{s}}

\newcommand{\mgc}{\ensuremath{\mathtt{MGC}}}
\newcommand{\mgcinit}{{\mathtt{\hat{q}}}}

\newcommand{\progrst}{p}

\newcommand{\abs}[1]{{#1}^{\#}}
\newcommand{\fsim}{R}
\newcommand{\leqtr}[1][]{\sqsubseteq^{\ifthenelse{\isempty{#1}}{}{#1}}}
\newcommand{\leqsim}[1][]{\sqsubseteq_{\mathsf{F}}^{\ifthenelse{\isempty{#1}}{}{#1}}}
\newcommand{\leqpp}[1][]{\sqsubseteq_{\mathsf{PP}}^{\ifthenelse{\isempty{#1}}{}{#1}}}

\newcommand{\guess}{\mathtt{guess}}

\newcommand{\progex}{e_p}
\newcommand{\tr}{\rho}

\newcommand{\traces}[1]{\mathsf{traces}({#1})}
\newcommand{\traceof}[1]{\mathsf{trace}({#1})}

\newcommand{\ex}{e}
\newcommand{\exmin}{e_\mathsf{m}}
\newcommand{\h}{h}
\newcommand{\seqh}{s}
\newcommand{\hs}[1]{\mathsf{h}({#1})}
\newcommand{\exec}[1]{\mathsf{E}({#1})}

\newcommand{\executions}[1]{\exec{#1}} %

\newcommand{\seqspec}{\mathit{Spec}}
\newcommand{\regspec}{\mathsf{Spec}_\Reg}

\newcommand{\obs}{\mathtt{obs}}
\newcommand{\obsbetween}{\mathtt{obs\text{-}between}}

\newcommand{\completed}{\mathsf{completed}}

\newcommand{\linleq}{\sqsubseteq}
\newcommand{\linmap}{L}
\newcommand{\linmaps}{\mathscr{L}}
\newcommand{\chain}{\mathscr{C}}
\newcommand{\boundlin}{\linmap_{\text{bound}}}

\newcommand{\ws}{\text{WS}}
\newcommand{\dec}{\text{D}}
\newcommand{\aws}{\dec}

\newcommand{\regread}{\mathtt{read}}
\newcommand{\regwrite}{\mathtt{write}}

\newcommand{\linpoint}{lin}
\newcommand{\rellinpoint}{f}
\newcommand{\verf}{ver}

\newcommand{\class}{\mathcal{I}}
\newcommand{\linclass}[1]{\class_{#1}}

\newcommand{\classrel}[2]{\class_{#1}(#2)}

\newcommand{\wslin}{\linclass{\mathsf{ws}}}
\newcommand{\dlin}{\linclass{\mathsf{d}}}
\newcommand{\awslin}{\dlin}

\newcommand{\relcomp}[2]{\impl_{#1}(#2)}

\newcommand{\Reg}{\mathsf{Reg}}
\newcommand{\strongreg}{\ensuremath{\mathtt{ATR}}\xspace}
\newcommand{\wsreg}{\ensuremath{\mathtt{WSR}}\xspace}
\newcommand{\dreg}{\ensuremath{\mathtt{DR}}\xspace}
\newcommand{\awsreg}{\dreg}

\newcommand{\ghost}[1]{{{#1}^{\mathsf{ex}}}}
\newcommand{\ltsg}{\ghost\lts}

\newcommand{\Pc}{\mathit{PC}}

\newcommand{\valset}{\mathcal{V}}
\newcommand{\valsetprev}{\mathcal{V}_\textup{prev}}
\newcommand{\reg}{X}
\newcommand{\tmp}{\mathit{tmp}}
\newcommand{\start}{s}
\newcommand{\retval}{\mathit{out}}

\newcommand{\replies}{\mathit{Replies}}
\newcommand{\broadcasts}{\mathit{Broadcasts}}

\newcommand{\tm}{t}
\newcommand{\Tmap}{\mathit{St}}
\newcommand{\Tmapid}{ID}

\newcommand{\Tmapval}{\mathit{V}} %
\newcommand{\tidset}{T}
\newcommand{\writeopset}{W}
\newcommand{\writetidset}{T_W}
\newcommand{\readopset}{R}
\newcommand{\readtidset}{T_R} %

\newcommand{\Tmaparg}{\mathit{A}}
\newcommand{\Tmappc}{\mathit{PC}}
\newcommand{\Tmapvalset}{\mathit{VS}}
\newcommand{\Tmapstart}{\mathit{S}}

\newcommand{\Fver}{f}
\newcommand{\tshist}{H}

\newcommand{\mbegin}[2]{{\mathtt{b}({#1},{#2})}}
\newcommand{\mpend}[1]{{\mathtt{p}(#1)}}
\newcommand{\mend}[2]{{\mathtt{e}({#1},{#2})}}

\newcommand{\main}{\mathtt{main}}
\newcommand{\wbegin}{\mathtt{w_b}}

\newcommand{\wwait}{\mathtt{w_w}}
\newcommand{\wrollback}{\mathtt{w_r}}

\newcommand{\wend}{\mathtt{w_e}}
\newcommand{\rbegin}{\mathtt{r_b}}

\newcommand{\rend}{\mathtt{r_e}}
\newcommand{\rlisten}{\mathtt{r_l}}

\newcommand{\rquery}{\mathtt{r_q}}
\newcommand{\rupdate}{\mathtt{r_u}}
\newcommand{\wquery}{\mathtt{w_q}}
\newcommand{\wupdate}{\mathtt{w_u}}

\newcommand{\ABD}{\ensuremath{\mathtt{ABD}}\xspace}
\newcommand{\TS}{\mathsf{TS}}
\newcommand{\Msg}{\mathsf{Msg}}

\newcommand{\msgmap}{M}
\newcommand{\msg}{m}

\newcommand{\Acks}{\mathsf{Acks}}

\newcommand{\ver}{\mathit{Ver}}
\newcommand{\ts}{\mathit{ts}}
\newcommand{\cnt}{C}
\newcommand{\queueturn}{t}

\newcommand{\quorum}{Q}
\newcommand{\lock}{L}
\newcommand{\lockLTS}{\ell}

\newcommand{\tsmax}{\ts_{\mathit{m}}}
\newcommand{\valmax}{\val_{\mathit{m}}}
\newcommand{\valin}{\val_\text{in}}
\newcommand{\valout}{\val_\text{out}}

\newcommand{\query}{\mathtt{query}}

\newcommand{\update}{\mathtt{update}}

\newcommand{\transrule}[1]{\textsc{#1}}

\newcommand{\internalleft}{\transrule{internal 1}}
\newcommand{\internalright}{\transrule{internal 2}}
\newcommand{\interface}{\transrule{interface}}

\newcommand{\invrule}{\transrule{inv}}
\newcommand{\loginv}{\transrule{log-inv}}
\newcommand{\resrule}{\transrule{res}}
\newcommand{\logres}{\transrule{log-res}}

\newcommand{\execset}{E}

\newcommand{\readinv}{\transrule{read-inv}}
\newcommand{\readinit}{\transrule{read-init}}
\newcommand{\readloop}{\transrule{read-loop}}
\newcommand{\readpick}{\transrule{read-pick}}
\newcommand{\readquery}{\transrule{read-query}}
\newcommand{\readupdate}{\transrule{read-update}}
\newcommand{\readend}{\transrule{read-end}}
\newcommand{\readres}{\transrule{read-res}}

\newcommand{\writeinv}{\transrule{write-inv}}
\newcommand{\writeint}{\transrule{write-int}}
\newcommand{\writeinit}{\transrule{write-init}}
\newcommand{\writetop}{\transrule{write-top}}
\newcommand{\writerollback}{\transrule{write-roll-back}}
\newcommand{\writerollforward}{\transrule{write-roll-forward}}
\newcommand{\writequery}{\transrule{write-query}}
\newcommand{\writeupdate}{\transrule{write-update}}
\newcommand{\writeend}{\transrule{write-end}}
\newcommand{\writeres}{\transrule{write-res}}

\newcommand{\serverqueryack}{\transrule{server-query-ack}}
\newcommand{\serverupdateack}{\transrule{server-update-ack}}

\usetikzlibrary{automata,positioning}
\tikzset{execution/.style={y=0.5cm,x=1.5cm,very thick,font=\sffamily}}
\tikzset{ltsstyle/.style = {scale = 1.2,draw, circle, initial text=}}
\tikzset{abovetext/.style={above=4pt,anchor=base}}
\tikzset{belowtext/.style={below=10pt,anchor=base}}

\usepackage{environ}
\NewEnviron{execution-graph}[2][]{
	\begin{figure}[H]
		\centering
		\begin{tikzpicture}[execution]
		\BODY
		\end{tikzpicture}
		\caption{#2}
		#1
	\end{figure}
}

\newcommand*\circled[1]{\tikz[baseline=(char.base)]{
            \node[shape=circle,draw,inner sep=0.7pt] (char) {\smaller \text{#1}};}}

\newcommand{\dlreg}{\ensuremath{\mathtt{DLR}}\xspace}
\newcommand{\tnsreg}{\ensuremath{\mathtt{TNSR}}\xspace}
\newcommand{\lreg}{a}
\newcommand{\lrega}{b}
\newcommand{\lregb}{c}
\newcommand{\exset}{T}

\newcommand{\exsetwr}{\exset_1}
\newcommand{\exsetww}{\exset_2}
\newcommand{\exsetwwb}{\exset_3}
\newcommand{\exsetwwr}{\exset_4}
\newcommand{\progwr}{\prog_1}
\newcommand{\progww}{\prog_2}
\newcommand{\progwwb}{\prog_3}
\newcommand{\progwwr}{\prog_4}

\newcommand{\winv}[1]{\mathtt{\stackrel{w\text{#1}}{}}}
\newcommand{\rinv}{{\mathtt{\stackrel{r}{}}}}
\newcommand{\rval}[1]{{\mathtt{\stackrel{ \text{#1}}{}}}}

\newcommand{\wint}[2]{\;|\!$\parbox[t][][t]{#2}{$\,\winv{#1}$\vspace*{-1pt}{\color{black}\hrule height 1pt}}$\!|\;}
\newcommand{\rint}[2]{\;|\!$\parbox[t][][t]{#2}{$\,\rinv\hfill\rval{#1}\,$\vspace*{-1pt}{\color{black}\hrule height 1pt}}$\!|\;}
\newcommand{\wintp}[2]{\;|\!$\parbox[t][][t]{#2}{$\,\winv{#1}$\vspace*{-1pt}{\color{black}\hrule height 1pt}}$\;}
\newcommand{\rintp}[2]{\;|\!$\parbox[t][][t]{#2}{$\,\rinv\,$\vspace*{-1pt}{\color{black}\hrule height 1pt}}$\;}

\newcommand{\sizeABD}{N}

\EventEditors{Dan Alistarh}
\EventNoEds{1}
\EventLongTitle{38th International Symposium on Distributed Computing (DISC 2024)}
\EventShortTitle{DISC 2024}
\EventAcronym{DISC}
\EventYear{2024}
\EventDate{October 28--November 1, 2024}
\EventLocation{Madrid, Spain}
\EventLogo{}
\SeriesVolume{319}
\ArticleNo{1}

\begin{document}
\iflong
\hideLIPIcs
\fi

\maketitle

\begin{abstract}
Reasoning about \emph{hyperproperties} of concurrent implementations,
such as the guarantees these implementations provide to randomized client programs,
has been a long-standing challenge.
Standard linearizability enables the use of \emph{atomic specifications} for reasoning about standard properties,
but not about hyperproperties.
A stronger correctness criterion, called \emph{strong linearizability}, enables such reasoning, but is rarely achievable,
leaving various useful implementations with no means for reasoning about their hyperproperties.
In this paper, we focus on registers and devise \emph{non-atomic specifications} that capture
a wide-range of well-studied register implementations and enable reasoning about their hyperproperties.
First, we consider the class of \emph{write strong-linearizable} implementations,
a recently proposed useful weakening of strong linearizability, which allows more implementations,
such as the well-studied single-writer \ABD distributed implementation.
We introduce a simple shared-memory register specification that
can be used for reasoning about hyperproperties
of programs that use write strongly-linearizable implementations.
Second, we introduce a new linearizability class, which we call \emph{decisive linearizability},
that is weaker than write strong-linearizability and includes multi-writer \ABD,
and develop a second shared-memory register specification
for reasoning about hyperproperties of programs that use register implementations of this class.
These results shed light on the hyperproperties guaranteed when
simulating shared memory in a crash-resilient message-passing system.

\end{abstract}

\section{Introduction}
\label{sec:intro}

Linearizability~\cite{HW:1990} is a widely accepted correctness criterion for concurrent and distributed implementations of objects,
allowing clients of an object to pretend that they use an atomic abstraction thereof,
whose behaviors are much easier to understand~\cite{Filipovic10}.
The observational refinement between a linearizable implementation and its atomic specification is, however, restricted to reasoning about reachability of `bad' states.
Dealing with more intricate properties, such as the ability of an adversary
to control the probability distribution of the results of an object's methods,
reveals that a linearizable implementation may manifest behaviors exceeding those permissible by the atomic specification~\cite{Golab11}.
In the terminology of~\cite{AE19}, linearizability ensures preservation of safety properties but fails to maintain \emph{hyperproperties},
which are properties of \emph{sets} of executions, rather than individual executions.
These properties allow one to express security guarantees, such as noninterference,
as well as probability distributions on program outcomes~\cite{ClarksonS10}.

The preservation of hyperproperties of concurrent implementations, \aka \emph{strong} observational refinement, necessitates a more strict connection between the implementation and its atomic specification,
known as \emph{strong linearizability}~\cite{Golab11},
 which is equivalent to (a certain form of) forward simulation between the implementation and the atomic specification~\cite{AE19,DongolSW22}.
Many implementations are, however, known to be \emph{non}-strongly linearizable,
leaving us with no means to reason about hyperproperties of programs that use
these implementations by assuming simpler abstractions.
In particular, the well-studied \ABD implementation,
which shows how shared memory can be simulated in a crash-tolerant message-passing system~\cite{AttiyaBD95},
 is not strongly linearizable, and, in fact,
 a strongly linearizable implementation with similar guarantees does not exist~\cite{AttiyaEW21}.

Focusing on registers and observing that strong linearizability is hardly achievable,
Hadzilacos \etal~\cite{HadzilacosHT21} recently proposed a weakening of strong linearizability,
called \emph{write strong-linearizability}, which captures more implementations.
This includes single-writer \ABD, and, in fact, as shown in \cite{HadzilacosHT21}, every linearizable implementation of a \emph{single-writer} register.
We are left, however, with substantial gaps: how should one reason about hyperproperties of programs that use write strongly-linearizable register implementations?
and what can be said about existing non-strongly-linearizable implementations of \emph{multi-writer} registers?

The current work aims to address these gaps.
Inspired by Attiya and Enea~\cite{AE19}, who propose to reason about hyperproperties of programs that use
non-strongly-linearizable implementations
by using simpler (albeit non-atomic) implementations related to them by strong observational refinement,
we present a simple (but necessarily not atomic) specification of a shared multi-writer register, which we call $\wsreg$
(for ``Write Strong Register'')
 that can be used for reasoning about hyperproperties of programs that use any write strongly-linearizable implementation
(including single-writer \ABD).
To do so, we prove that every write strongly-linearizable implementation has a forward simulation to $\wsreg$, and utilize the correspondence between forward simulation and strong observational refinement (preservation of hyperproperties).
Moreover, since write strong-linearizability is downward closed \wrt forward simulation and $\wsreg$ is write strongly-linearizable,
one can also prove write strong-linearizability for a given implementation by establishing a forward simulation to $\wsreg$, 
which may be more amenable to automatic/machine-assisted proofs than a direct proof.
Drawing an analogy to complexity theory, we refer to $\wsreg$ as a \emph{complete} implementation for the class of write strongly-linearizable register implementations:
$\wsreg$ is write strongly-linearizable and every write strongly-linearizable implementation has a forward simulation to $\wsreg$.

As for multi-writer registers, we present a second specification of a shared register
that is `complete' for a family of implementations that admit a weakening of write strong-linearizability, which we call \emph{decisive linearizability}.
In particular, we show that multi-writer \ABD~\cite{LynchS97} belongs to this family.
Thus, the complete implementation, which we call $\dreg$ (for ``Decisive Register''),
enables reasoning about hyperproperties of programs that use \ABD
via a simpler \emph{shared-memory specification}.
(We also use $\dreg$ to demonstrate that multi-writer \ABD is decisively linearizable by showing a forward simulation to $\dreg$.)
Intuitively speaking, unlike strong linearizability and write strong-linearizability,
decisive linearizability gradually commits on the relative order of operations in the sequential history,
rather than on the exact position of operations in that history.

\smallskip
\noindent\textbf{Outline.}
The rest of this paper is structured as follows.
In \cref{sec:overview} we make the introductory discussion more concrete by outlining several examples.
In \cref{sec:pre} we provide the necessary preliminaries for our formal development.
In \cref{sec:comp} we introduce and study the notion of a complete implementation for a given linearizability class.
In \cref{sec:ws} we present the complete implementation for write strong-linearizability.
In \cref{sec:decisive} we define decisive linearizability and present a complete implementation for this class.
We discuss related work and conclude in \cref{sec:related}.
In Appendix~\ref{sec:sketches} we provide proof sketches for several lemmas and theorems.
\iflong
In Appendices~\ref{app:pre}--\ref{app:abd} we provide full proofs. 
(See Appendix~\ref{app:pre} for more details on the structure of the appendix.)
\else
The full version of this paper \todo{cite arxiv} provides full proofs.
\fi

\newcommand{\figclients}{
\begin{figure}
\small
\captionsetup[subfigure]{justification=centering}
\captionsetup[subfigure]{labelformat=empty}
	\begin{subfigure}{0.33\textwidth}
		\begin{equation*}
			\inarrII{\Write{$1$};\\ \Write{$2$};\\ \lreg \gets\Coin{};}{\lrega \gets \Read{};}
		\end{equation*}
		\caption{Program $\progwr$\label{fig:client1}}
	\end{subfigure}
	\vline
	\hfill
	\begin{subfigure}{0.34\textwidth}
		\begin{equation*}
			\inarrII{\Write{$1$};\\ \lreg \gets\Coin{};\\ \Barrier{};}{\Write{$2$};\\ \Barrier{}; \\ {\lrega \gets \Read{};}}
		\end{equation*}
		\caption{Program $\progww$\label{fig:client2}}
	\end{subfigure}
	\vline
	\hfill
	\begin{subfigure}{0.3\textwidth}
	\begin{equation*}
			\inarrII{\Write{$1$};\\ \Barrier{};\\ \lreg \gets\Coin{};}{$\Write{$2$}$;\\ \Barrier{};\\ \lrega \gets \Read{}; }
	\end{equation*}
		\caption{Program $\progwwb$\label{fig:client3}}
	\end{subfigure}
	\hrule
	\smaller
	\begin{subfigure}{0.33\textwidth}
		\begin{equation*}
			\exsetwr\!=\!\left\{\inarr{\wint{1}{10pt} \wint{2}{10pt} \circled{1} \\ \; \rint{1}{40pt}}
			\;,
			\inarr{\wint{1}{10pt} \wint{2}{10pt} \circled{2} \\ \; \rint{2}{40pt}}\right\}
		\end{equation*}
	\end{subfigure}
	\vline
	\hfill
	\begin{subfigure}{0.34\textwidth}
		\begin{equation*}
			\exsetww\!=\!\left\{
			\inarr{\;\wint{1}{10pt} \circled{1} \\  \wint{2}{30pt} \!\! \rint{1}{10pt}}
			\;,
			\inarr{\;\wint{1}{10pt} \circled{2} \\  \wint{2}{30pt} \!\! \rint{2}{10pt}}
			\right\}
		\end{equation*}
	\end{subfigure}
	\vline
	\hfill
	\begin{subfigure}{0.3\textwidth}
		\begin{equation*}
			\exsetwwb\!=\!\left\{ \inarr{\; \wint{1}{10pt} \circled{1}  \\ \wint{2}{10pt} \quad\rint{1}{10pt}}
			\;,
			\inarr{\; \wint{1}{10pt} \circled{2}  \\ \wint{2}{10pt} \quad\rint{2}{10pt}}\right\}
		\end{equation*}

	\end{subfigure}
	\caption{Client programs (upper part) and corresponding trace sets (lower part) that,	
	if an adversary can generate them, violate the hyperproperty ``$\lreg=\lrega$ with probability $\frac{1}{2}$'' \label{fig:clients}}
\end{figure}}

\newcommand{\figregs}{
\begin{figure*}[t]
\footnotesize
\begin{minipage}{0.255\textwidth}
\begin{algorithm}[H]
	\NoCaptionOfAlgo
		\METHOD{\Read{}}{
			$\retval\gets\reg$\;
			\Return{$\retval$}\;
		}
	\vspace*{40.5pt}	
		\METHOD{\Write{$\val$}}{
			$\reg\gets\val$\;
			\Return\;
		}
	\caption*{Atomic ($\strongreg$)}
\end{algorithm}
\end{minipage}
\hfill
\begin{minipage}{0.335\textwidth}
\begin{algorithm}[H]
	\NoCaptionOfAlgo
		\METHOD{\Read{}}{
			$\retval_1\gets\reg$\;
			$\retval_2\gets\reg$\;
				\lIf{\normalfont{\texttt{*}}}{\Return{$\retval_1$}}\lElse{\Return{$\retval_2$}}
		}
	\vspace*{22pt}
		\METHOD{\Write{$\val$}}{
			$\reg\gets\val$\;
			\Return\;
		}
	\caption{Double load ($\dlreg$)}
\end{algorithm}
\end{minipage}
\hfill
\begin{minipage}{0.385\textwidth}
\begin{algorithm}[H]
	\NoCaptionOfAlgo
		\METHOD{\Read{}}{
			$\retval\gets\reg$\;
			\Return{$\retval$}\;
		}
	\smallskip	
		\METHOD{\Write{$\val$}}{
			$\lreg_1\gets\reg$\;
			$\lreg_2\gets\reg$\;
			\If{\normalfont{\texttt{*}}}{\lIf{$\lreg_1\neq\lreg_2$}{\Return{}}}
			$\reg\gets\val$\;
			\Return{}\;
		}
	\caption{Try-not-to-store ($\tnsreg$)}
\end{algorithm}
\end{minipage}

\begin{minipage}{1\textwidth}
\begin{algorithm}[H]
	\NoCaptionOfAlgo
	\textbf{Shared Variables:}
	A set $\broadcasts$ of query/update messages and a mapping $\replies$ from messages to their replies.
	
	\textbf{Local Variables:}
	Process $\tid$ stores the most recent value it observed, $\val_\tid$, and its timestamp, $\ts_\tid$.

    Timestamps are pairs $\ts=\tup{\tm,\tid}$ with $t\in\N$ ordered lexicographically (assuming an arbitrary order on process id's).
$\max \set{\tup{\val_1,\ts_1} \til \tup{\val_n,\ts_n}}$ retrieves the timestamped value $\tup{\val_i,\ts_i}$ with the maximum timestamp.

	\begin{multicols}{2}
		\METHOD{\Read{}}{
			$\tup{\val,\ts}\gets$\Query{}\;
			\Update{$\val,\ts$}\;
			\Return{$\val$}\;
		}
		
		\BlankLine

		\METHOD{\Write{$\val$}}{
			$\tup{\_,\tup{\tm,\_}}\gets$\Query{}\;
			\Update{$\val,\tup{\tm+1,\normalfont{\texttt{my\_process\_id}()}}$}\;
			\Return{}\;
		}
		\BlankLine		
		\FUNCTION{\Update{$\val,\ts$}}{
			\Broadcast{$\msg=\update(\val,\ts)$}\;
			\Wait{$\size{\replies(\msg)}>{\sizeABD/2}$}\;
			\Return{}\;
		}

		\BlankLine

		\FUNCTION{\Query{}}{
			\Broadcast{$\msg=\query$}\;
			\Wait{$\size{\replies(\msg)}>{\sizeABD/2}$}\;
			$\quorum\gets$\Pick{$\quorum\subseteq \replies(\msg) \;\text{s.t.}\; \size{\quorum}>{\sizeABD/2}$}\;
			\Return{$\max\quorum$}\;
		}
		
		\BlankLine
		
		Background activity by process $\tid$: \\
		\When{$\msg\in\broadcasts$}{
		\If{$\msg=\query$}{\lReply{$\tup{\val_\tid,\ts_\tid}$}{$\msg$}}
		\If{$\msg=\update(\val,\ts)$}{
			$\tup{\val_\tid,\ts_\tid}\gets\max\set{\tup{\val,\ts},\tup{\val_\tid,\ts_\tid}}$\;
			\lReply{``ack''}{$\msg$}
			}
		}
	\end{multicols}
		    \vspace{.4\baselineskip}%
	\caption{$\ABD$ implementation for $\sizeABD$ processes ($\ABD_\sizeABD$)}
\end{algorithm}
\end{minipage}

\caption{Four register implementations}
\label{fig:regs}
\end{figure*}}

\newcommand{\sumtable}{
\begin{figure}
\begin{minipage}{0.5\textwidth}
	\begin{tabular}{|c||*{4}{c|}}\hline
		&\makebox[3em]{$\strongreg$}&\makebox[3em]{$\dlreg$}&\makebox[3em]{$\tnsreg$}
		&\makebox[3em]{$\ABD_{\geq 3}$}\\\hline\hline
		$\progwr$ &$\cmark$&$\xmark$&$\cmark$&$\xmark$\\\hline
		$\progww$ &$\cmark$&$\cmark$&$\xmark$&$\xmark$\\\hline
		$\progwwb$ &$\cmark$&$\cmark$&$\cmark$&$\cmark$\\\hline
	\end{tabular}
\end{minipage}
\hfill
\begin{minipage}{0.47\textwidth}
\small
	For each program and register implementation, $\cmark$ indicates that the hyperproperty ``$\lreg=\lrega$ with probability $\frac{1}{2}$'' holds under any adversary, and $\xmark$ indicates the hyperproperty is refuted by some adversary.
\end{minipage}
		\caption{Summary of examples\label{fig:client-impl}}
\end{figure}}

\section{Motivating Examples: A Tale of Four Registers}
\label{sec:overview}

\figclients

This section demonstrates certain intricacies arising when examining hyperproperties
of client programs using (linearizable) implementations of concurrent registers.
The specifications we develop in the next sections are based on the observations
arising from these examples.
We keep the discussion informal, deferring the formal treatment to the next sections.

\Cref{fig:clients} (upper part) presents three programs,
in which two threads read and write from a shared register,
and invoke a method $\Coin{}$ that returns 1 or 2 uniformly at random.
Programs $\progww$ and $\progwwb$ also employ a synchronization method $\Barrier{}$
that ensures the threads wait for each other before executing the rest of the code.
Given that the underlying register implementation is linearizable, one can
analyze standard properties of a single (finite) trace
(\eg the final values of the variables) of these programs by considering
an \emph{atomic} register ($\strongreg$ in \cref{fig:regs})~\cite{Filipovic10}.
In technical terms, one says that every linearizable implementation \emph{observationally refines} the atomic register,
and that the atomic register provides a \emph{specification} (\aka \emph{reference implementation})
for any linearizable implementation.

However, as observed by Golab \etal~\cite{Golab11}, the atomic register
cannot be used for analyzing properties of \emph{sets} of program traces, \aka \emph{hyperproperties},
which cannot be deduced from a single program trace.
For investigating hyperproperties, one
considers the sets of program traces that can be generated by
an adversary that controls the scheduling and the steps of the implementation.
(By \emph{program trace} we mean the sequence of actions performed by the
client where the object's implementation internal actions are invisible.)
Specifically, we consider the standard \emph{strong} adversary
that sees the whole execution so far
and makes choices that depend on previous coin-toss results.

\figregs

For instance, for the programs above, we may aim to verify that
\emph{under any adversarial scheduling the probability that $\lreg=\lrega$
at the end of execution is exactly $\frac{1}{2}$,}
which indicates that the adversary cannot leak the coin-toss result from one thread to another.\footnote{
By adding conditional loops in the programs, one can correlate the probability that
$\lreg=\lrega$ with the probability that the program diverges, and thus concentrate on asking
whether an adversary can force non-termination, as considered in some previous work~\cite{HadzilacosHT21,AttiyaEW22}.}
With an atomic register, this property holds in all three programs.
For instance, in programs  $\progwr$ and $\progwwb$, 
if the adversary performs the atomic $\Read{}$ before the coin is tossed,
it cannot force a correlation between the coin and the read value; and
by the time the coin is tossed, there is only one possible value that can be read.

Next, we demonstrate that this does not mean that other
linearizable implementations guarantee this hyperproperty.
To this end, we depict below each program in \cref{fig:clients}
a set of traces that forces $\lreg=\lrega$ with probability $1$,
and to show that the hyperproperty of a program is violated for certain implementations,
we describe an adversary that generates this set.

We consider three linearizable register implementations,
in addition to the atomic register ($\strongreg$) discussed above,
presented in \Cref{fig:regs}: %
a ``double load'' implementation ($\dlreg$), a ``try-not-to-store'' implementation ($\tnsreg$),
and the well-studied \ABD implementation.
Like $\strongreg$, $\dlreg$ and $\tnsreg$ are shared-memory implementations,
using a single primitive (atomic) shared memory cell $\reg$ initialized to $0$
(all other variables are local).
We refer to the accesses to $\reg$ as loads/stores,
and to the methods of the register as reads/writes.
In contrast, \ABD is a register implementation in a crash-resilient message passing system,
originally proposed to demonstrate that such a system can emulate a shared memory~\cite{AttiyaBD95}.
We present the multi-writer version of \ABD from~\cite{LynchS97}.

\para{$\dlreg$}
This implementation loads twice and non-deterministically picks which value to return (using $\KwSty{if}~\normalfont{\texttt{*}}$).
Using $\dlreg$, in $\progwr$ the adversary can generate $\exsetwr$ by ensuring this particular interleaving of the two
threads, and moreover: execute the first load in the read method after $1$ is stored to $\reg$, so that $\retval_1=\text{1}$;
execute the second load after 2 is stored to $\reg$, so that $\retval_2=\text{2}$;
and resolve the non-deterministic choice only after the coin is tossed ensuring that
$\retval_1$ is returned if the coin result is 1, and $\retval_2$ is returned if the coin result is 2.
(Recall that the adversary controls object-implementation-internal steps, including non-deterministic choices.)
However, it is easy to see that for programs $\progww$ and $\progwwb$, the hyperproperty holds when $\dlreg$ is used.
Indeed, without a read concurrently executed with a write, $\dlreg$ behaves just like $\strongreg$.

\para{$\tnsreg$}
This implementation tries to avoid some stores by recognizing
that if the value is concurrently altered during a write operation,
then that operation does not have to actually store
as it may pretend it was overrun by the concurrent write.
With this implementation, the hyperproperty holds for $\progwr$.
Indeed, without two concurrently executed writes, $\tnsreg$ behaves just like $\strongreg$.
However, using $\tnsreg$, in $\progww$ the adversary can generate $\exsetww$ by ensuring
that the first load in $\Write{$2$}$ reads 0 (the initial value),
then execute $\Write{$1$}$ atomically and have the second load in $\Write{$2$}$ read 1.
Then, if the coin result is 1, the adversary makes $\Write{$2$}$ skip writing its value
(it can do so since the two loaded values are not equal).
Otherwise, if the coin result is 2, $\Write{$2$}$ stores its value.
Finally, it is easy to see that for $\progwwb$ the hyperproperty holds with $\tnsreg$.
Indeed, after both threads reach the barrier, only one value can be returned by the read method,
since at this point in the execution, both write methods are completed.

\para{$\ABD$}
This implementation uses \emph{timestamps} to order the written values
(breaking ties using some predetermined order on the process identifiers).
Each process maintains the most recent timestamped value it observed.
For reading, a process broadcasts a query,
waits for replies from a quorum (majority) of processes,
and returns the value with the largest timestamp,
but only after broadcasting this timestamped value
and receiving acknowledgments from a quorum of processes.
In turn, for writing value $v$ a process broadcasts a query,
waits for replies from a quorum of processes,
broadcasts $v$ with timestamp larger than all replies,
and waits for a quorum of acknowledgments.
Note that in \ABD, processes are also constantly active as ``servers'':
$(i)$ replying to queries with their current timestamped values,
and $(ii)$ acknowledging new written values after (possibly) updating their current timestamped values.

Using $\ABD$, the hyperproperty is violated for $\progwr$ and $\progww$.
For the violation we need to have at least three processes,
two of them running the code of the program, and the others
are used as servers that reply to messages and participate in quorums.
($\ABD_2$ is degenerate since a quorum must consist of all processes.)
Essentially,
$\ABD_{\geq 3}$ allows both the behaviors exposed by $\dlreg$
and the behaviors exposed by $\tnsreg$.
However, the actual adversaries for $\ABD_{\geq 3}$ are more complicated than the ones for $\dlreg$ and $\tnsreg$
due to the absence of a global centralized memory cell that values are stored in and loaded from.

We describe adversaries that generate $\exsetwr$ for $\progwr$  and $\exsetww$ for $\progww$:
\begin{itemize}
\item For $\progwr$, the adversary lets the reader invoke a query and lets the writer complete the execution of $\Write{$1$}$ by choosing a quorum of processes that acknowledge the new value.
Next, the adversary lets all the processes in the quorum %
reply to the query of the reader reporting value 1. Then, the adversary lets the writer execute $\Write{$2$}$, again obtaining a quorum of processes that are aware of the new value, where this time the adversary picks a quorum that includes at least one process that is not part of the previous quorum, and therefore has not yet replied to the reader's query (this is where at least three processes are needed). This process also replies to the reader's query, but with value 2. At this time, the query message of the reader has pending replies from a quorum in which all replies include value 1, and from one additional process that is already aware of the more recent value 2. However, the adversary postpones the delivery of the replies until after the coin toss, at which time it picks the replies to match the coin value:
if the coin value is 1, the replies from the homogeneous quorum where all replies include value 1 are delivered; otherwise the replies from the all but of one of the processes in the aforementioned quorum are delivered together with the reply of the additional process that includes value 2, thus forming a (heterogeneous) quorum whose most recent value is 2. Accordingly, the reader returns a value that is equal to the coin value.
\item %
For $\progww$, the adversary starts by invoking a query during $\Write{$2$}$ and making a quorum of processes send replies to the query (with the initial value) before $\Write{$1$}$ is initiated in the left process. The adversary then lets the left process execute up to the barrier, at which point at least one reply with the value 1 is sent to the right process's query by a process that is aware of the left process's update.
The adversary then performs the delivery of the replies to the query in $\Write{$2$}$ according to the coin value. %
If the coin result is 2, the adversary delivers a quorum of replies that includes the reply sent when the left process reached the barrier, causing the right process to be aware of the most recent timestamp of the left process, such that the right process updates the value 2 with a larger timestamp. On the other hand, if the coin result is 1, the adversary delivers only the replies sent before $\Write{$1$}$, whose timestamp is outdated, causing the right process to choose a timestamp for the new value 2 that is at a tie with the timestamp attached by the left process to the value 1. Assuming the id of the left process has precedence, the tie is resolved to its timestamp, making 1 appear to be the most recent value.
This determines the result of the subsequent read to be equal to the coin result.

\end{itemize}

Finally, the hyperproperty holds when $\ABD$ is used in $\progwwb$.
To see this, suppose, \wlg~that
the timestamp assigned to $2$ is larger than the one of $1$.
Then, after the two writes complete,
in every quorum there is at least one process that knows about the value $2$, and
a reader that queries after this point can only read $2$.

\sumtable

\medskip
\Cref{fig:client-impl} summarizes the above observations.
In particular, the hyperproperty holds in $\progwwb$ for all four implementations.
Nevertheless, as we show later in \cref{ex:cl3_lin}, it can be still violated by some linearizable implementations.

To capture differences between linearizable implementations,
such as the ones shown in the above examples,
\cite{AE19} introduced \emph{strong observational refinement} as a refinement relation
between an implementation and a specification
that preserves hyperproperties.
Then, while $\strongreg$, $\dlreg$, $\tnsreg$, and \ABD can be shown to be observationally equivalent
(\ie observationally refine each other), as we demonstrated above, they are not \emph{strongly} observational equivalent.
In particular, none of the relatively simple shared-memory implementations in \cref{fig:regs} can be used as a specification of \ABD when hyperproperties are considered,
as \ABD is not a strong observational refinement of any of them.
(This is unfortunate, since, as we have seen, reasoning about the sets of program traces generated when \ABD is used is much more involved than with the other implementations.)
We also note that 
each of the implementations admits a different linearizability criterion:
$\strongreg$ is strongly linearizable~\cite{Golab11}, $\dlreg$ is write strongly-linearizable~\cite{HadzilacosHT21},
while $\tnsreg$ and \ABD are neither.

In the rest of the paper we propose hyperproperty-preserving specifications
for classes of linearizable register implementations, including \ABD.
Such specifications can drastically simplify verification of hyperproperties of client programs using these implementations,
a task which is typically challenging, especially when complex implementations are considered,
since it requires reasoning about \emph{all} possible adversaries.

\section{Preliminaries}
\label{sec:pre}

We start with general notations,
continue to our modeling of objects, implementations, and programs (\cref{sec:objects}),
and finally recap the formal notions of preservation of hyperproperties via
strong observational refinement (\cref{sec:hyp}).

\begin{description}[leftmargin=0pt,itemsep=2pt]
\item[Sequences.]
For a finite alphabet $\Sigma$, we denote by $\Sigma^*$ the set of all (finite) sequences over $\Sigma$.
The length of a sequence $s$ is denoted by $\size{s}$.
We write $s[k]$ for the symbol at position $1 \leq k \leq \size{s}$ in $s$.
We write $\sigma\in s$ if $s[k]=\sigma$ for some $1 \leq k \leq \size{s}$.
We use ``$\cdot$'' for the concatenation of sequences.
We often identify symbols with sequences of length $1$ or their singletons
(\eg in expressions like $s \cdot \sigma$).
The \emph{restriction} of a sequence $s$ \wrt a set $\Gamma$, denoted by $s\rst{\Gamma}$,
is the longest subsequence of $s$ that consists only of symbols in $\Gamma$.
This notation is extended to sets by $S\rst{\Gamma} \defeq \set{s\rst{\Gamma} \st s\in S}$.
We write $s_1 \subseq s_2$ when $s_1$ is a subsequence of $s_2$,
and $s_1 \pref s_2$ when $s_1$ is a prefix of $s_2$.

\item[Labeled Transition Systems.]
A \emph{labeled transition system} (LTS, for short) is a tuple $\lts = \tup{Q,\Sigma,q_0,T}$,
where $Q$ is a set of \emph{states},
$\Sigma$ is a (possibly infinite) alphabet (whose elements are called \emph{transition labels}),
$q_0\in Q$ is an \emph{initial state},
and $T\suq Q\times \Sigma \times Q$ is a set of \emph{transitions}.
We denote by $\lts.\lQ$, $\lts.\lSigma$, $\lts.\linit$, and $\lts.\lT$ the components of an LTS $\lts$.
We write $\asteplab{\sigma}{\lts}$ for the relation
$\set{\tup{q,q'} \st \tup{q,\sigma,q'}\in \lts.\lT}$.
An \emph{execution} $\ex$ of $\lts$ is a (possibly empty) finite sequence of transitions in $\lts.\lT$
such that the first transition starts in $q_0$ and each other transition continues
from the target of the previous transition.
An execution $\ex$ induces a \emph{trace} $\tr\in \lts.\lSigma^*$,
where $\tr[i]$ is given by the label of $\ex[i]$ for every $1\leq i \leq \size{\ex}$.
We denote by $\executions{\lts}$ and  $\traces{\lts}$ the set of all executions of $\lts$
and  the set of all traces induced by executions of $\lts$ (respectively).
Note that we only consider finite executions and traces.

\item[Forward Simulations.]
	Given LTSs $\lts$ and $\abs\lts$ and
	a set $\Gamma\subseteq\lts.\lSigma$,

	a relation $\fsim\subseteq \lts.\lQ \times \abs\lts.\lQ$ is a
	\emph{$\Gamma$-forward simulation} from $\lts$ to $\abs\lts$ if
	\begin{enumerate*}
		\item[(i)]
		$\tup{\lts.\linit,\abs\lts.\linit}\in \fsim$; and
		\item[(ii)]
		if $q \asteplab{\sigma}{\lts} q'$ and $\tup{q,\abs q}\in \fsim$,
		then there exist ${\abs q}' \in \abs\lts.\lQ$ and $\tr\in\abs\lts.\lSigma^*$
		such that ${\abs q \asteplab{\tr[1]}{{\abs \lts}}  \ldots  \asteplab{\tr[\size{\tr}]}{{\abs \lts}} {\abs q}'}$,
		$\tr\rst{\Gamma}=\sigma\rst{\Gamma}$,
		and $\tup{q',{\abs q}'}\in \fsim$.
	\end{enumerate*}
We write $\lts\leqsim[\Gamma]\abs\lts$ when such relation exists.

\end{description}

\subsection{Objects, Implementations, and  Programs}
\label{sec:objects}

We review standard notions that are needed for our formal results.
We assume a set $\Tid$ of thread identifiers and an infinite set $\Id$ of action identifiers.

\begin{description}[leftmargin=0pt,itemsep=2pt]
\item[Objects.]
An \emph{object} is a pair $\obj=\tup{\Method,\Val}$, where
$\Method$ is a set of method names and
$\Val$ is a set of values.
An object $\obj$ is associated with actions divided into \emph{invocations} $\inv=\invact\tup{\method,\val,\tid,\id} \in \Inv(\obj)$
and \emph{responses} $\res=\resact\tup{\method,\val,\tid,\id} \in \Res(\obj)$,
where $\method\in\Method$, $\val\in\Val \cup \set{\bot}$, $\tid\in\Tid$, and $\id\in\Id$.
We let $\Inv\Res(\obj) \defeq \Inv(\obj) \cup \Res(\obj)$.

\item[Histories.]
A \emph{history} $\h$ of an object $\obj$ is a finite sequence over $\Inv\Res(\obj)$. %
A history $\h$ is \emph{sequential} if it alternates between invocations and responses 
(starting with an invocation), %
such that every consecutive $\inv$, $\res$ in $\h$ have the same method and thread identifiers,
and a unique action identifier across $\h$.
A history $\h$ is \emph{well-formed} if its restriction to actions of each $\tid\in\Tid$, denoted by $\h\rst{\tid}$, is sequential.
An invocation $\inv\in\h$ is \emph{pending} if there is no response in $\h$ with the same thread and action identifiers.
Otherwise, $\inv$ is \emph{complete}.
These notions are also applied on \emph{operations} $\op$,
which are either single invocations $\op=\inv$ or pairs of matching invocation and response $\op=\tup{\inv,\res}$.
We let $\completed(\h)$ denote the subsequence of $\h$ consisting of actions that are a part of completed operations.

\item[Real-time Order.]
The \emph{real time order} induced by
a well-formed history $\h$, denoted by $<_\h$,
is the partial order on operations defined by $\op_1 <_\h \op_2$ iff
$\op_1$'s response appears in $\h$ before $\op_2$'s invocation.

\item[Specifications.] 
A \emph{specification} of %
$\obj$ is a prefix-closed set of sequential histories of $\obj$.

\item[Registers.]
A register object is given by $\Reg=\tup{\set{\regread,\regwrite},\N}$.
Its specification, denoted by $\regspec$, is defined as usual, assuming that $0$ is the initial register value.

\item[Object Implementations.]
We assume a set $\Objint$ of labels for implementation internal actions and define an
\emph{implementation} $\impl$ of an object $\obj$ to be an LTS
 over the alphabet $\Inv\Res(\obj) \cup \Objint$.
We assume that the history induced by every execution $\ex$ of $\impl$, denoted by $\hs{\ex}$, is a well-formed history.
The pseudo-code presented in specific implementations in the paper is easily translatable
to formal LTSs, whose executions represent executions generated
by the methods' code when they are repeatedly and concurrently invoked with arbitrary arguments.

\item[Client Programs.]
We assume a set $\Progint$ of labels for client internal actions (disjoint from $\Objint$) and define a
client program $\prog$ for an object $\obj$ as an LTS over the alphabet $\Inv\Res(\obj) \cup \Progint$.
A program $\prog$ and implementation $\impl$ are \emph{linked} by taking ``interface parallel composition'', denoted by $\prog[\impl]$.
The resulting LTS interleaves the steps of $\prog$ and $\impl$ while forcing the two LTSs to synchronize
on labels from $\Inv\Res(\obj)$. The defining property of $\prog[\impl]$ is given by:

\begin{proposition}
$\tr\in\traces{\prog[\impl]}$ iff $\tr\rst{\impl.\lSigma}\in\traces{\impl}$ and $\tr\rst{\prog.\lSigma}\in\traces{\prog}$.
\end{proposition}

\end{description}

\subsection{Hyperproperties Preservation via Strong Observational Refinement}
\label{sec:hyp}

A \emph{hyperproperty} $\phi$ of a program $\prog$ is a set of sets of the program's traces
(\ie $\phi \subseteq \powerset{\traces{\prog}}$).
Such sets can capture probabilistic requirements,
such as the one informally described in \cref{sec:overview}, via suitable encodings of traces~\cite{ClarksonS10}.

The hyperproperties that are satisfied by an object implementation,
and accordingly, strong observational refinement between implementations, are defined using deterministic schedulers, which formalize the notion of a strong adversary~\cite{AE19}.

\begin{description}[leftmargin=0pt,itemsep=2pt]
\item[\em Schedulers.]
Given a program $\prog$ and an implementation $\impl$,
a \emph{scheduler} is a function $\scheduler : \executions{\prog[\impl]} \to \powerset{\prog[\impl].\lT}$.
An execution $\ex\in\executions{\prog[\impl]}$ is \emph{consistent with $\scheduler$} if
$\ex[j] \in \scheduler(\ex[1]\cdots\ex[j-1])$ for every $1 \leq j \leq \size{\ex}$.
We denote by $\executions{\prog[\impl],\scheduler}$ the set of executions of $\prog[\impl]$
that are consistent with $\scheduler$,
and by $\traces{\prog[\impl],\scheduler}$ the traces of executions in $\executions{\prog[\impl],\scheduler}$.
A scheduler is \emph{deterministic} if for every $\ex\in \executions{\prog[\impl]}$,
either $\size{\scheduler(\ex)}\leq 1$ or all transitions in $\scheduler(\ex)$ are labeled by actions in $\Progint$.

\begin{remark}
\label{rem:step_det}
Attiya and Enea~\cite{AE19} restricted their attention to \emph{step-deterministic} implementations in which a trace uniquely determines an execution
(which includes the intermediate states along the trace).
We avoid this technical restriction, and thus use executions instead of traces
in the definitions of schedulers, as well as of linearizability criteria below.
In particular, we define schedulers as functions from executions to sets of transitions
instead of functions from traces to sets of labels.
For step-deterministic implementations our definitions coincide with those of~\cite{AE19}.
\end{remark}

\item[\em Hyperproperty Satisfaction.]
An implementation $\impl$ \emph{satisfies a hyperproperty} $\phi$ of $\prog$,
denoted by $\impl \models_\prog \phi$,
if $\traces{\prog[\impl],\scheduler}\rst{\prog.\lSigma}\in\phi$ for every deterministic scheduler $\scheduler$.

\begin{example}
    For the client program $\progww$ (represented as an LTS) and the set of traces $\exsetww$ from \cref{fig:clients},
    we have that $\dlreg\models_{\progww}\powerset{\traces{\progww}}\setminus\set{\exsetww}$.
    This is because, as discussed in \cref{sec:overview}, there exists no scheduler $\scheduler$
    such that $\traces{\progww[\dlreg],\scheduler}\rst{\prog.\lSigma}=\exsetww$.
\lipicsEnd
\end{example}

\item[\em Strong Observational Refinement.]
An implementation $\impl$ \emph{strongly observationally refines} an implementation $\abs\impl$,
denoted by $\impl \srefine \abs\impl$,
if $\abs\impl\models_\prog \phi\implies\impl\models_\prog \phi$
for every program $\prog$ and hyperproperty $\phi$ of $\prog$.
The following alternative characterization follows from the definition.

\begin{lemma}
\label{lem:strong_refinement}
$\impl \srefine \abs\impl$ iff for every program $\prog$ and deterministic scheduler $\scheduler$,
there exists a deterministic scheduler $\abs\scheduler$ such that
$\traces{\prog[\impl],\scheduler}\rst{\prog.\lSigma} = \traces{\prog[\abs\impl],\abs\scheduler}\rst{\prog.\lSigma}$.
\end{lemma}

Attiya and Enea~\cite[Theorem 8]{AE19} show that $\Inv\Res(\obj)$-forward simulation between implementations is equivalent to strong observational refinement. (Their result applies to finite traces as we consider here; see \cite{DongolSW22} for a discussion on infinite traces.)
We adapt this result to our setting.
In the sequel, for implementations $\impl$ and $\abs\impl$ of an object $\obj$, we write $\impl\leqsim\abs\impl$ for $\impl\leqsim[\Inv\Res(\obj)]\abs\impl$.

\begin{theorem}
\label{thm:sref_iff_sim}
$\impl \srefine \abs\impl$ iff $\impl\leqsim\abs\impl$.
\end{theorem}

\begin{example}
It is easy to show that $\strongreg\leqsim\dlreg$,
and we obtain that $\strongreg \srefine \dlreg$.
Thus,  $\strongreg\models_{\progww}\powerset{\traces{\progww}}\setminus\set{\exsetww}$
follows from $\dlreg\models_{\progww}\powerset{\traces{\progww}}\setminus\set{\exsetww}$.
In addition, since $\dlreg \not\srefine \strongreg$ (see \cref{sec:overview}),
we have $\dlreg\not\leqsim\strongreg$.
Indeed, if a concurrent write is about to change the value of $\reg$ after a read of $\dlreg$ performs its first load,
$\strongreg$ has no matching action: if it performs its (single) load it will not be able
to return the right value in case $\dlreg$ returns the value read in the second load;
and similarly, if it waits, it will fail to return the same value if
$\dlreg$ returns the value of the first load.
\lipicsEnd
\end{example}

\end{description}

\section{Complete Implementations for Linearizability Classes}
\label{sec:comp}

Knowing that a given implementation is a member of a certain linearizability class
is only useful if it enables reasoning about programs that use that implementation without understanding the implementation itself.
For hyperproperties, such reasoning is made possible
if the implementation is known to strongly observationally refine a simpler implementation,
in which case the latter can be used instead of the actual implementation in the analysis.
To standardize the relation between linearizability classes and strong observational refinement,
we propose a definition of \emph{hard} and \emph{complete} implementations
in analogy to hardness and completeness \wrt complexity classes,
where instead of reductions, we use simulations, which ensure strong observational refinement:

\begin{definition2}
\label{def:comp}
Let $\class$ be a class of implementations of an object $\obj$
that is downward closed \wrt forward simulation
(\ie $\impl\in\class$ whenever $\impl'\in\class$ and $\impl\leqsim\impl'$).
An implementation $\abs\impl$ of $\obj$ is \emph{$\class$-hard} if $\impl\leqsim\abs\impl$ for every $\impl\in\class$.
It is \emph{$\class$-complete} (or \emph{complete for $\class$}) if we also have $\abs\impl\in\class$.
\end{definition2}

In addition to allowing reasoning about hyperproperties of implementations in $\class$,
an $\class$-complete implementation $\abs\impl$ also provides a sound and complete method
to establish the membership of an implementation $\impl$ in $\class$ by showing that $\impl\leqsim\abs\impl$.

In the following we %
take $\class$ to be the set of implementations of some object %
that satisfy certain linearizability criteria.

\smallskip
\noindent\textbf{Linearizability.}
Consider first standard linearizability~\cite{HW:1990,Sela21}:

\begin{definition2}
A history $\seqh$ of an object $\obj$ is a \emph{linearization} of a history $\h$ of $\obj$, denoted by $\h\linleq\seqh$,
if there exists a sequence of responses $\ressq$ %
for some of the pending invocations in $\h$ such that
the following hold for $\h'=\completed(\h\cdot\ressq)$:
\begin{enumerate*}[label=(\roman*)]
\item $\h'\rst{\tid}=\seqh\rst{\tid}$ for every $\tid\in\Tid$; and
\item $<_{\h'} \;\subseteq\; <_\seqh$.
\end{enumerate*}
A history $\h$ of $\obj$ is \emph{linearizable} \wrt a specification $\seqspec$ of $\obj$ if it has a linearization $\seqh\in\seqspec$.
An implementation $\impl$ of $\obj$ is \emph{linearizable} \wrt $\seqspec$ if $\hs{\ex}$ is linearizable \wrt $\seqspec$ for every $\ex\in\exec{\impl}$.
\end{definition2}

\begin{proposition}
\label{prop:comp-lin}
The class of linearizable implementations of an abject $\obj$ \wrt a specification $\seqspec$ %
is downward closed \wrt forward simulation,
and there exists a complete implementation for it.
\end{proposition}
\begin{proof}[Proof (sketch)]
Downward closedness follows from the fact that $\impl\leqsim\impl'$ implies that $\set{\hs{\ex} \st \ex\in\exec{\impl}} \suq \set{\hs{\ex} \st \ex\in\exec{\impl'}}$.
A complete implementation is the implementation that tracks in its internal state the history $\h$ generated so far.
When executing an invocation or response, the action is added in the end of the current history.
But, while invocations are always enabled, a response $\res$ is only enabled when $\h \cdot \res$ is linearizable \wrt $\seqspec$.
\end{proof}

The (theoretical)
construction in the above proof provides us with a complete implementation,
which may help in streamlining and mechanizing linearizability arguments as forward simulations
(\eg \cite{Jayanti24} utilized such implementation).
However, since it directly encodes the definition of the class,
it is unhelpful for reasoning about hyperproperties of implementations.
Thus, for the stronger classes considered below we are interested in identifying simple complete implementations that are not based on history tracking.

\smallskip
\noindent\textbf{Strong linearizability.}
Golab \etal~\cite{Golab11} proposed a strengthening of linearizability, called \emph{strong linearizability},
and showed that it is necessary and sufficient for reasoning on
probability distributions of outcomes that a strong adversary can generate.
Roughly speaking, while linearizability allows one to choose the linearization order ``after the fact''
in view of the whole execution,
strong linearizability requires the linearization of implementation histories into specification histories
to be done online in a prefix-preserving manner, that is, by continuously adding operations at the end of the linearized history.

\begin{definition2}
\label{def:slin}
A \emph{linearization mapping} for an implementation $\impl$ of an object $\obj$ \wrt a specification $\seqspec$ of $\obj$
is a function $\linmap: \exec{\impl}\to\seqspec$ such that $\hs{\ex}\linleq\linmap(\ex)$ for every $\ex\in\exec{\impl}$.
An implementation $\impl$ of $\obj$ is \emph{strongly linearizable} \wrt a specification $\seqspec$ of $\obj$
if there is a linearization mapping $\linmap$ %
for $\impl$ \wrt $\seqspec$
such that $\linmap(\ex_1)\pref\linmap(\ex_2)$ whenever $\ex_1\pref\ex_2$.
\end{definition2}

Since we aim to also capture non-deterministic implementations (and do not assume step-determinism),
our linearizations apply on \emph{executions} rather than traces (see also \cref{rem:step_det}).

\begin{example}
\label{ex:cl1_lin}
From the register implementations presented in \cref{sec:overview}, only $\strongreg$ is strongly linearizable.
We use the histories associated with the set $\exsetwr$ from \cref{fig:clients} to show that $\dlreg$ and $\ABD$ are not strongly linearizable.
Consider the following history $\h$,
its two possible extensions $\h_1$ and $\h_2$,
and its possible linearizations $\seqh_1,\seqh_2,\seqh_3$:
\begin{equation*}
	\h = \inarr{\wint{1}{15pt} \wint{2}{15pt} \\ \; \rintp{1}{50pt}}
	\quad\vrule\quad
	\inarr{
		\h_1 = \inarr{\wint{1}{15pt} \wint{2}{15pt} \\ \; \rint{1}{50pt}} \quad
		\h_2 = \inarr{\wint{1}{15pt} \wint{2}{15pt} \\ \; \rint{2}{50pt}}}
	\quad\vrule\quad
	\inarr{\seqh_1 = \text{\wint{1}{15pt} \rint{1}{15pt} \wint{2}{15pt}} \\
		\seqh_2 = \text{\wint{1}{15pt} \wint{2}{15pt} \rint{2}{15pt}} \\
		\seqh_3 = \text{\wint{1}{15pt} \wint{2}{15pt}}}\quad
\end{equation*}
Unlike $\strongreg$ (and $\tnsreg$), both $\dlreg$ and $\ABD$ have a \emph{single} execution $\ex$ that induces $\h$
and can be extended into two alternative executions that induce $\h_1$ and $\h_2$.
Then, $\linmap(\ex)$ can be $\seqh_1$, $\seqh_2$, or $\seqh_3$,
but any choice at this stage is doomed to fail:
\begin{enumerate*}[label=(\roman*)]
\item $\seqh_1$ fails if the execution continues to generate $\h_2$;
\item $\seqh_2$ fails if the execution continues to generate $\h_1$;
and \item $\seqh_3$ fails if the execution continues to generate $\h_1$
since we are only allowed to extend the current linearization by adding operations at its end.
The history of the common prefix of the traces in $\exsetww$
from \cref{fig:clients} %
can be similarly used to show that $\tnsreg$ is not strongly linearizable.
\end{enumerate*}
\lipicsEnd
\end{example}

Attiya and Enea~\cite{AE19} show that %
the class of strongly linearizable implementations %
is downward closed \wrt forward simulation,
and that every strongly linearizable implementation strongly observationally refines the atomic implementation (\eg $\strongreg$ for registers).
Together with \cref{thm:sref_iff_sim}, this result is restated as follows:

\begin{theorem}
\label{thm:slin_comp}
The atomic implementation for specification $\seqspec$ of an object $\obj$ is complete for
the class of strongly linearizable implementations of $\obj$ \wrt $\seqspec$.
\end{theorem}

\smallskip
\noindent\textbf{Additional linearizability classes.}
We observe that downward-closedness \wrt simulation, as well as the existence of a complete implementation,
generalize to a range of linearizability classes beyond linearizability and strong linearizability mentioned above.
These linearizability classes are parameterized by a preorder that must hold between the linearizations of an execution and its extensions.
Formally, given a preorder $\rel$ (\ie reflexive and transitive relation) on sequences, %
the class $\classrel{\rel}{\obj,\seqspec}$ consists of all implementations $\impl$ of $\obj$
for which there exists a linearization mapping
$\linmap:\exec{\impl}\to\seqspec$ such that $\tup{\linmap(\ex_1),\linmap(\ex_2)}\in\rel$ whenever $\ex_1\pref\ex_2$.
The class of all linearizable implementations of $\obj$ \wrt $\seqspec$ is obtained by taking $\rel = \seqspec \times \seqspec$,
whereas for all strongly linearizable implementations we take $\rel = \pref$.
Other classes defined in the rest of this paper are also instances of this definition.

\begin{lemma}
\label{lem:closure_rel}
For every preorder $\rel$ on sequences, %
the class $\classrel{\rel}{\obj,\seqspec}$ is downward closed \wrt forward simulation, and there exists a complete implementation for it.
\end{lemma}

The complete implementation for $\classrel{\rel}{\obj,\seqspec}$ is constructed similarly to the one in the proof of \cref{prop:comp-lin} (which is a special case),
except that here the state also tracks a linearization of the history so far, and ensures in each transition that the linearizations in the pre-state and post-state are related by $\rel$.

Similarly to the construction in \cref{prop:comp-lin}, the generic construction in \cref{lem:closure_rel} is not helpful for reasoning about hyperproperties.
In contrast, \cref{thm:slin_comp} proposes a simple and useful complete implementation for strong linearizability.
In the remainder of the paper we seek useful complete implementations for other linearizability classes of interest.

\section{Complete Implementation for Write Strong Linearizability}
\label{sec:ws}

Focusing on registers and identifying that useful register implementations are not strongly linearizable,
Hadzilacos \etal~\cite{HadzilacosHT21} have recently proposed a weakening of strong linearizability, called \emph{write strong-linearizability},
and showed that every linearizable \emph{single writer} register implementation,
including single-writer $\ABD$, is write strongly-linearizable.
However, they do not provide
a specification for write strong-linearizability that plays the role that the atomic register implementation plays for strong linearizability.

Write strong-linearizability weakens the prefix-preservation requirement of strong linearizability by applying it only
to writes, thus allowing reads to be linearized offline, and freely ``move around'' when more operations are added.
For the formal definition, we let $\seqh\rst{\regwrite}$ denote the restriction of $\seqh \in \regspec$ to write operations.

\begin{definition2}
\label{def:wslin}
	Let $\impl$ be a register implementation. 
	A linearization mapping $\linmap:\exec{\impl}\to \regspec$ is \emph{write strong}
	if $\linmap(\ex_1)\rst{\regwrite}\pref\linmap(\ex_2)\rst{\regwrite}$ whenever $\ex_1\pref\ex_2$.
We say that $\impl$ is \emph{write strongly-linearizable} %
if there exists a write strong linearization mapping $\linmap:\exec{\impl}\to \regspec$.
\end{definition2}

\begin{example}
\label{ex:cl2_lin}
From the implementations in \cref{sec:overview},
$\strongreg$ and $\dlreg$ are write strongly-linearizable.
(For $\dlreg$, for $\h$ from \cref{ex:cl1_lin}, we can pick $\seqh_3$, and later on, when the read returns, pick either $\seqh_1$, by adding a read in the middle, or $\seqh_2$ according to the returned value.)
We use the histories associated with the set $\exsetww$ from \cref{fig:clients} to show that $\tnsreg$ and $\ABD$ are not write strongly-linearizable.
(For $\ABD$ this also follows from the general result in~\cite{YuHHT21}.)
Consider the following history $\h$,
its two possible extensions $\h_1$ and $\h_2$,
and its possible linearizations $\seqh_1,\seqh_2,\seqh_3$:
\begin{equation*}
	\h = \inarr{\wint{1}{15pt}\\ \; \wintp{2}{20pt}}
	\qquad\vrule\qquad
	\inarr{
		\h_1 = \inarr{\;\wint{1}{15pt} \\ \wint{2}{30pt} \!\! \rint{2}{15pt}}\qquad
		\h_2 = \inarr{\;\wint{1}{15pt}\\ \wint{2}{30pt} \!\! \rint{1}{15pt}}}
	\qquad\vrule\qquad
	\inarr{
		\seqh_1 = \text{\wint{1}{15pt} \wint{2}{15pt}} \\
		\seqh_2 = \text{\wint{2}{15pt} \wint{1}{15pt}} \\
		\seqh_3 = \text{\wint{1}{15pt} }}
\end{equation*}
Unlike $\strongreg$ and $\dlreg$, both $\tnsreg$ and $\ABD$ have a \emph{single} execution $\ex$ that induces $\h$
and can be extended into two alternative executions that induce $\h_1$ or $\h_2$.
Then, $\linmap(\ex)$ can be $\seqh_1$, $\seqh_2$, or $\seqh_3$,
but any choice at this stage is doomed to fail:
\begin{enumerate*}[label=(\roman*)]
\item $\seqh_1$ fails if the execution continues to generate $\h_2$
since no extension of $\seqh_1$ linearizes $\h_2$;
\item $\seqh_2$ fails if the execution continues to generate $\h_1$
since no extension of $\seqh_2$ linearizes $\h_1$;
and \item $\seqh_3$ fails if the execution continues to generate $\h_2$
since no extension of $\seqh_3$, where write operations are only added after
the write operation in $\seqh_3$, linearizes $\h_1$.
\end{enumerate*}
\lipicsEnd
\end{example}

We denote by $\wslin$ the class of write strongly-linearizable register implementations.
By \cref{lem:closure_rel}
(with $\rel$ ordering histories using the prefix relation on the restriction to writes),
$\wslin$ is downward-closed \wrt simulation, and the notion of a complete implementation is well-defined.
\Cref{wsreg} presents our proposed complete implementation for this class.
\iflong
Appendix~\ref{subsec:wslts} formalizes the pseudo-code as an LTS.
\fi
Its construction is inspired by a specification given by~Attiya and Enea~\cite[§6]{AE19} for capturing the hyperproperties
of a specific snapshot implementation~\cite{AfekADGMS93}.
It is a generalization of $\dlreg$ from \cref{sec:overview}, where instead of loading twice, the reader
repeatedly loads from $\reg$ as long as new values are observed,
and non-deterministically decides which value to return.

\begin{remark}
One can define a sequence $\{\impl_k\}_{k=1}^\infty$ of implementations, all with atomic write,
and read that non-deterministically picks between $k$-loads (so $\strongreg=\impl_1$ and $\dlreg=\impl_2$).
It can be shown that all of these implementations are write strongly-linearizable,
but for every $k$, $\impl_{k+1}$ does not strongly observationally refine $\impl_k$.
The $\wsreg$ implementation is what one gets ``at the limit'' of this sequence,
 and every $\impl_k$ trivially
strongly observationally refines $\wsreg$.
\end{remark}

\setcounter{algocf}{0}
\begin{algorithm}[t]
\footnotesize
\caption{$\wsreg$: A complete implementation for write strongly-linearizable registers}
	\label{wsreg}
	\textbf{Shared Variables:}
	the current value $\reg$.
	
	Multi-assignments are executed atomically.
	
	\begin{multicols}{2}
		\METHOD{\Read{}}{
			$\valset\gets\set{\reg}$\;
			\Do{$\valset\neq\valsetprev$}
			{
				$\tup{\valsetprev,\valset}\gets\tup{\valset,\valset \cup \set{\reg}}$\;
			}
			$\retval\gets$ \Pick{$\val\in\valset$}\;
			\Return{$\retval$}\;
		}
		
		\METHOD{\Write{$\val$}}{
			$\reg\gets\val$\;
			\Return{}\;
		}
	\end{multicols}
	\vspace*{3pt}
\end{algorithm}

\begin{theorem}\label{thm:ws-complete}
$\wsreg$ is
complete for the class of
write strongly-linearizable register implementations.
\end{theorem}

As a consequence of \cref{thm:ws-complete}, we obtain that
single-writer \ABD strongly observationally refines $\wsreg$,
and so we can use $\wsreg$ to argue about the hyperproperties of client programs that use single-writer \ABD.

\section{Complete Implementation for Decisive Linearizability}
\label{sec:decisive}

In this section we identify a novel linearizability criterion, which we call \emph{decisive linearizability}.
Then, we present a complete implementation for the corresponding class of register implementations, which can serve as a hyperproperty-preserving specification for any implementation in the class.
Using this implementation, we show that multi-writer \ABD is decisively linearizable,
and that decisive linearizability (for registers) is weaker than write strong-linearizability.

\begin{definition2}
	\label{def:dlin}
	Let $\impl$ be an implementation of an object $\obj$
	and $\seqspec$ be a specification of $\obj$.
	A linearization mapping $\linmap:\exec{\impl}\to \seqspec$ is \emph{decisive}
	if $\linmap(\ex_1)\subseq\linmap(\ex_2)$ whenever $\ex_1\pref\ex_2$.
We say that $\impl$ is \emph{decisively linearizable} \wrt $\seqspec$
if there there exists a decisive linearization mapping $\linmap:\exec{\impl}\to \seqspec$.
\end{definition2}

Decisive linearizability, like strong and write strong-linearizability, requires the linearization process to be ``online''.
Nevertheless, unlike strong and write strong-linearizability, it does not require that the sequences of linearizations
produced in this process are increasing ``at the end'',
thus allowing operations to be added to the linearized history possibly
before operations that are already included in the linearized history.
The only requirement of decisive linearizability is that this process maintains
the relative order of already linearized operations: once the order between $\op_1$ and $\op_2$ has been decided,
it cannot be reverted.

\newcommand{\wlabel}[1]{\mathtt{w\text{#1}}}

\begin{example}
\label{ex:cl3_lin}
All implementations in \cref{sec:overview} are decisively linearizable:
$\strongreg$ and $\dlreg$ are already write strongly-linearizable
(which is a stronger condition, as we show below)
and for
$\tnsreg$ and \ABD, which are not write strongly-linearizable, this will be proven
later in the section.
To illustrate how a suitable linearization mapping can be obtained for these implementations,
we revisit the histories $\h$ and its extensions $\h_1$ and $\h_2$ from \cref{ex:cl2_lin}.
To linearize $\h$, we can pick $\seqh_3$; later on, if the execution continues according to $\h_1$,
we append $\wlabel{2}$ to the linearization, and if the execution continues to $\h_2$,
we add $\wlabel{2}$ to the linearization before $\wlabel{1}$---note that decisive linearizability allows this;
finally, when the read returns we add it immediately after the corresponding write.

For a ``non-example'',
we use the histories associated with the set $\exsetwwb$ from \cref{fig:clients} to show that the complete implementation for the class of
linearizable registers (see \cref{prop:comp-lin}) is not decisively linearizable.
Consider the following history $\h$,
its two possible extensions $\h_1$ and $\h_2$,
and its possible linearizations $\seqh_1$ and $\seqh_2$:
\begin{equation*}
	\h = \inarr{\; \wint{1}{15pt} \\ \wint{2}{15pt}}
	\qquad\qquad\vrule\qquad
	\inarr{
		\h_1 = \inarr{\; \wint{1}{15pt} \rint{1}{15pt} \\ \wint{2}{15pt}} \qquad
		\h_2 = \inarr{\; \wint{1}{15pt} \rint{2}{15pt} \\ \wint{2}{15pt}}}
	\qquad\qquad\vrule\qquad
	\inarr{
		\seqh_1 = \text{\wint{1}{15pt} \wint{2}{15pt} } \\
		\seqh_2 = \text{\wint{2}{15pt} \wint{1}{15pt} } }
\end{equation*}
Recall that in the complete implementation for standard linearizability,
an execution $\ex$ that induces $\h$ can be extended both to an execution $\ex_1$ that
induces $\h_1$ and to an execution $\ex_2$ that induces $\h_2$.
(In particular, this means that an adversary for $\progwwb$ from  \cref{fig:clients} can decide between these options after the coin toss, refuting the hyperproperty discussed in \cref{sec:overview}, which is satisfied when each of $\strongreg$, $\dlreg$, $\tnsreg$, \ABD and in fact any decisively linearizable implementation is used.)
If $\linmap(\ex) = \seqh_1$ then the
linearization of $\ex_1$ must reorder the writes in $\seqh_1$, violating decisiveness.
Similarly, if $\linmap(\ex) = \seqh_2$, then the linearization of $\ex_2$ must reorder the writes in $\seqh_2$, violating decisiveness.
Thus, no decisive linearization mapping exists.
\lipicsEnd
\end{example}

By \cref{lem:closure_rel} (with $\rel$ being the subsequence relation),
the class of
decisively linearizable implementations is downward-closed \wrt simulation and
a complete implementation exists for it, for any object.
Next, we present a complete implementation for the class
of decisively linearizable register implementations.
We note that while \cref{def:dlin} is not specific to registers (unlike \cref{def:wslin})
and \cref{lem:closure_rel} applies to any object,
the complete implementation we present is only for register implementations.
We denote by $\dlin$ the class of all decisively linearizable register implementations.
The complete implementation, $\dreg$, is presented in \cref{alg:dlin}.
\iflong
Appendix~\ref{subsec:dlts} formalizes the pseudo-code as an LTS.
\fi

\begin{algorithm}[t]
\footnotesize
	\caption{$\dreg$: A complete implementation for decisively linearizable registers
}
	\label{alg:dlin}

	\textbf{Shared Variables:}
	the current value $\reg$, the current version number $\ver$, and a lock flag $\lock$.
	
	\KwSty{await}~B~\KwSty{do}~C blocks until the condition $B$ is met,
	at which point the evaluation of $B$ and the body $C$ are atomically executed.
	Multi-assignments and \KwSty{atomic} blocks are executed atomically.
	\setlength{\columnsep}{-1cm}
	\begin{multicols}{2}
		\METHOD{\Read{}}{
			\lAwait{$\lock=0$}{$\tup{\start,\valset}\gets\tup{\ver,\set{\reg}}$}
			\Do{$\valset\neq\valsetprev$}
			{
				\Atomic{}{
					$\valsetprev\gets\valset$\;
					\lIf{$\ver\geq\start$}{$\valset\gets\valset\cup \set{\reg}$}
				}
			}
			$\retval \gets$ \Pick{$\val\in\valset$}\;
			\Return{$\retval$}\;
		}
		
		\METHOD{\Write{$\val$}}{
			\lAwait{$\lock=0$}{$\start\gets\ver$}
			
			\If{\normalfont{\texttt{*}}}{\lAwait{$\lock=0$}{$\tup{\reg,\ver}\gets\tup{\val,\ver+1}$}}
			\Else{
					\Await{$\lock=0\land\ver>\start$}{$\tup{\lock,\tmp,\reg,\ver}\gets\tup{1,\reg,\val,\ver-1}$\;}
					$\tup{\lock,\reg,\ver}\gets\tup{0,\tmp,\ver+1}$\;
			}
			\Return{}\;
		}
	\end{multicols}
	\vspace*{3pt}

\end{algorithm}

$\dreg$ stores the current value in $\reg$ and a corresponding version number in $\ver$.
Reads use repeated loads similarly to $\wsreg$, but add loaded values to $\valset$
only when their version number is not older than the version number when the read started
(stored in $\start$).
The return value is picked non-deterministically from $\valset$.

Writes are based on the idea used in $\tnsreg$, allowing stores to non-deterministically choose
to be overwritten by a concurrent write,
with two important differences. First, new stores by concurrent writes are identified based on version number
($\ver > \start$) rather than values (to avoid data dependencies).
Second, even if a write
chooses to be overwritten, the store to $\reg$ is not skipped but momentarily executed with a lower version number,
to allow concurrent reads to observe it.
This is done by a step that temporarily decreases $\ver$ and stores the input value to $\reg$,
followed by a step that restores $\ver$ and $\reg$
to their newer values.
The two steps are not executed atomically, letting concurrent reads to load the
intermediate value.
Importantly, a lock $\lock$ is used to prevent concurrent methods from setting their start version number ($\start$)
to a temporary version number,
and from updating $\ver$ based on a temporary version number.

To simplify the presentation, the pseudo-code is written such that a write makes the non-deterministic choice
whether to be overwritten or not before it determines that it can indeed be overwritten.
As a result, the execution may get stuck.
This does not affect linearizability, and
this behavior is impossible in our formulation of $\dreg$ as an LTS.

\begin{figure}
	\begin{minipage}{0.5\textwidth}
		\begin{equation*}
			\inarrIII{\Write{$1$};\\ \lreg \gets\Coin{};\\ \Barrier{};}
			{\Write{$2$};\\ \Barrier{}; \\ {\lrega \gets \Read{};}}
			{{\lregb \gets \Read{};}\\ \Barrier{};\\}
		\end{equation*}
	\end{minipage}
	\hfill
	\begin{minipage}{0.55\textwidth}
		$$\exsetwwr = \left\{
		\inarr{\; \;\wint{1}{10pt} \!\circled{2}\! \\ \; \wint{2}{30pt} \! \rint{2}{10pt}
			\\  \rint{2}{35pt} }
		\;,\;
		\inarr{\; \;\wint{1}{10pt} \!\circled{1}\!  \\ \; \wint{2}{30pt} \! \rint{1}{10pt}
			\\  \rint{2}{35pt} }
		\right\}$$
	\end{minipage}
	\caption{A program $\progwwr$ and a set $\exsetwwr$ of traces of the program\label{fig:client4}}
\end{figure}

\begin{example}
\label{ex:cl3}
Allowing concurrent reads to observe ``overwritten'' writes
is crucial for capturing all behaviors of decisively linearizable implementations such as multi-writer \ABD.
Consider the program $\progwwr$ and set of traces $\exsetwwr$ in \cref{fig:client4}.
The program $\progwwr$ extends $\progww$ from  \cref{fig:clients} with another thread,
and $\exsetwwr$ is similar to $\exsetww$
except that the additional thread observes the value 2 written by the middle thread,
even when this value ends up being overwritten.
Recall that $\exsetww$ can be generated by an adversary for both $\tnsreg$ and $\ABD$.
For $\tnsreg$, this leverages the ability of the adversary to postpone the decision
whether to store 2 or not until after the coin toss.
In contrast, $\exsetwwr$ is not possible for $\tnsreg$, since in the trace where the middle thread reads 1,
it must be the case that $\tnsreg$ chose to overwrite 2 and as a result has never stored 2 to $\reg$,
preventing concurrent threads from loading the value before it is overwritten.
($\dreg$ does perform a store in such a case, allowing $\exsetwwr$.)
Unlike $\tnsreg$, \ABD allows this behavior:
The adversary acts on the left and middle processes similarly to the adversary for $\progww$ that generates $\exsetww$ described in \cref{sec:overview},
with the added right process sending an additional query when $\Write{$2$}$ does so,
immediately receiving replies with the initial value from a set of process that excludes the middle process and is one-short from a quorum.
Then, when $\Write{$2$}$ sends its update, it also replies to the right process with the timestamp it chose.
Regardless of the chosen timestamp, it is larger than the initial timestamp,
causing the right process to return the value 2.
\lipicsEnd

\end{example}

\begin{theorem}
\label{thm:dlin_complete}
$\dreg$ is
complete for the class of
decisively linearizable register implementations.
\end{theorem}

While not immediate from the definitions, a corollary of \cref{thm:dlin_complete,thm:ws-complete},
together with the observation that $\wsreg\leqsim\dreg$,
is that every write strongly-linearizable implementation is also decisively linearizable.
That is, decisive linearizability is indeed weaker than write strong-linearizability.

Having constructed a complete implementation, we now leverage
it to show that other implementations are decisively linearizable:
all we need to do is prove that they are simulated by $\dreg$.
For example, $\tnsreg$ is trivially simulated by $\dreg$, and is therefore decisively linearizable.
We show that the same holds for multi-writer $\ABD$.

\begin{theorem}
\label{thm:abd-sim-dreg}
	$\ABD\leqsim\dreg$.
\end{theorem}

\begin{corollary}
	$\ABD$ is decisively linearizable.
\end{corollary}

Thus, $\dreg$ provides a shared-memory specification for multi-writer \ABD that enables reasoning about its hyperproperties.

\section{Related and Future Work}
\label{sec:related}

Since the observation that linearizability does not suffice for reasoning about randomized client programs
and the introduction of strong linearizability~\cite{Golab11}, many works have studied (im)possibility of implementing 
strongly linearizable objects under different progress conditions.
Helmi \etal~\cite{HelmiHW12} showed that lock-free strongly linearizable multi-writer registers, max registers, snapshots, and counters cannot be constructed from a single-writer registers.
Attiya \etal~\cite{AttiyaEW21} and Chan \etal~\cite{YuHHT21} adapted and extended these results for a fault-tolerant message passing setting.

Attiya \etal~\cite{AttiyaEW22} developed a methodology of making existing implementations probabilistically close to strongly linearizable ones
by repeating an effect-free preamble of every method and picking uniformly at random which outcome to continue with.
They introduced a correctness condition called \emph{tail strong linearizability} that ensures the effectiveness of this construction.
This criterion depends on the choice of the preamble and is thus not comparable to decisive linearizability.
Interestingly, the construction in \cite{AttiyaEW22} is not effective for our complete implementations ($\wsreg$ and $\dreg$).

The work of Hadzilacos \etal~\cite{HadzilacosHT21} is closer to our work in its aim to give up strong linearizability,
and study what existing implementations do provide.
In addition to what we have already discussed, \cite[Algorithm 4]{HadzilacosHT21} demonstrated a 
multi-writer register implementation that is not write strongly-linearizable.
This implementation is essentially a simplified version of \ABD,
and using forward simulation to \ABD, one can conclude that it is decisively linearizable.

As discussed in length, our work is heavily inspired by~\cite{AE19}
that uncovered the correspondence between strong observational refinement and simulation,
and suggested the use of non-atomic specifications for reasoning about non-strongly-linearizable implementations.
Derrick \etal~\cite{DerrickDDSW21} and Dongol \etal~\cite{DongolSW22} identified a gap in the way~\cite{AE19} handle infinite traces,
and show that in that case, while simulation is still necessary for strong observational refinement,
only a stronger relation, called \emph{(weak) progressive forward simulation} is sufficient.
We focus solely on finite traces, leaving infinite traces to future work.

Bouajjani \etal~\cite{BouajjaniEEM17} used forward simulations to non-atomic reference implementations as means to establish linearizability.
In particular, they developed abstract stack and queue specifications such that forward simulations to these specifications
is necessary and sufficient for establishing linearizability. In our terms, this gets close to complete implementations
for the class of linearizable stacks and queues, but, their results are, however, limited to implementations that have explicit
marking of linearization points (or so-called ``commit points'') in some of the methods. 
Their implementations are
highly beneficial in simplifying (and possibly automating) complex linearizability arguments, as the ones needed for
Herlihy\&Wing Queue~\cite{HW:1990} and the Time-Stamped Stack~\cite{DoddsHK15}.

Finally, we note that although we focused on registers, decisive linearizability is a general correctness criterion.
Investigating its applicability beyond registers is left for future work.
We believe that various implementations that are not strongly linearizable are still decisively linearizable
(but, there are known implementations that are not even decisively linearizable, such as the Time-Stamped Stack~\cite{DoddsHK15}, which allows concurrent complete push operations to remain unordered until a later
pop determines their order).
Identifying complete implementations for the class of decisively linearizable implementations of other objects
is an important (and challenging!) avenue for future work.
It would also be interesting to study (im)possibility for decisively linearizable implementations
with different progress guarantees.

\bibliography{biblio}

\appendix

\section{Proof Sketches}
\label{sec:sketches}

\begin{proof}[Proof (sketch) of \Cref{lem:closure_rel}]
	For downward closure, we observe that $\impl\leqsim\impl'$ implies that
	there exists a mapping $\pi: \exec{\impl} \to \exec{\impl'}$
	such that $\hs{\ex}= \hs{\pi(\ex)}$ for every $\ex\in\exec{\impl}$,
	and $\pi(\ex_1)\pref\pi(\ex_2)$ whenever $\ex_1\pref\ex_2$.
	Then, given a suitable linearization mapping $\linmap$ for $\impl'$,
	the composition $\linmap\circ\pi$ is a suitable linearization mapping for $\impl$.
	
	A complete implementation for $\classrel{\rel}{\obj,\seqspec}$ 
	is similar to the complete implementation for the class of all linearizable implementations presented in the proof of \cref{prop:comp-lin}, except that in addition to tracking in its internal state the history $\h$ it has observed so far, it also tracks a linearization $\seqh$ of $\h$.
	When executing an invocation or response $\alpha$,
	the linearization is non-deterministically updated to a linearization
	$\seqh'$ of $\h \cdot \alpha$ such that $\tup{\seqh,\seqh'} \in \rel$.
	If such $\seqh'$ does not exist, the $\alpha$ step is not enabled.
	This construction generalizes the implementation in \cite[Appendix C]{Attiya19arXiv}
	which uses the same set of states but only allows to append actions to the linearization.
\end{proof}

\begin{proof}[Proof (sketch) of \Cref{thm:ws-complete}]
	To show that $\wsreg\in\wslin$,
	we construct a linearization mapping $\linmap: \exec{\wsreg} \to \regspec$
	by assigning ``linearization points'':
	Write operations that have already stored their values in $\reg$ are linearized to the transition where they stored this value,
	and reads that picked a value to return are linearized to the transition where they first loaded that value.
	Other pending operations are not included in the linearization.
	
	Hardness is much more challenging.
	Given a write strongly-linearizable implementation $\impl$,
	we begin by instrumenting the state of $\impl$ with a ghost variable that tracks the full execution performed so far.
	Then, given an execution $\ex$ and a transition $\trans$,
	we compare $\linmap(\ex)$ and $\linmap(\ex\cdot\trans)$, where $\linmap$ is the given write strong-linearization mapping for $\impl$.
	A naive attempt to show $\impl \leqsim \wsreg$
	would execute a store in $\wsreg$ at the time the corresponding write operation $w$
	is added to the linearization.
	This fails since $w$ might be added to linearization immediately after all previous writes have been linearized,
	which can be before $w$ takes effect, and
	performing the store of $\wsreg$ at that step will not allow later reads (that are concurrent with $w$) to load earlier values.
	
	To overcome this, we prove the existence of a so-called \emph{lazy} linearization mapping.
	Informally, this mapping adds operations to the linearization only when it must,
	\eg when an operation completes, or when a write is needed to justify a completed read.
	More concretely, assuming arbitrary write strong-linearization mapping $\linmap$,
	we prove the existence of a write strong-linearization mapping $\linmap^*: \exec{\impl} \to \regspec$
	with the following additional properties:
	\begin{enumerate}
		\item $\linmap^*(\ex)=\linmap^*(\ex\cdot \trans)$ for every $\ex\cdot\trans\in\exec{\impl}$ such that the transition
		$\trans$ is not labeled with a response.
		
		\item For every $\ex\in\exec{\impl}$ and operation $\op$ in $\linmap^*(\ex)$,
		if $\op$ is not completed in $\hs{\ex}$, then it is a write operation and it is not last in $\linmap^*(\ex)$.
		
		\item $\linmap^*$ is decisive.

	\end{enumerate}
	Using $\linmap^*$, the simulation works.
	Invocation and response transitions are simulated by an identical invocation or response,
	where invocations of read operations also load the stored value once.
	The stores in write operations are executed when the write operations are added to the lazy linearization,
	after which all pending reads that did not already load the stored value load it.
	
	The crux of the proof is to justify that when a completed read is added to the linearization, the matching read in $\wsreg$ has already loaded the value it needs to pick to match the return value of that read.
	For this, it suffices to show that the value the read returns was written by a write that either was rightmost in the linearizaton of the prefix of the execution up to the transition that invoked the read, or was added to the linearization later (as these are exactly the writes whose stores we load as described above).
	The properties of the lazy linearization are used to establish this fact.
\end{proof}

\begin{proof}[Proof (sketch) of \Cref{thm:dlin_complete}]
	For inclusion, $\dreg\in\dlin$, %
	given an execution of $\dreg$ we construct a linearization that only includes writes that already stored their value in $\reg$ and reads that picked a value to return.  
	This is similar to the linearization in the proof that $\wsreg \in \wslin$. 
	However, the order in which the operations are linearized is more involved.
	We begin by assigning to each operation we intend to include in the linearization a version number:
	for a write operation it is the version number it wrote in the transition where it stored its value into $\reg$,
	and for a read operation it is the version number in the pre-state of the transition where it first loaded the value it later picked to return.
	Operations are ordered based on version number, with ties broken based on the ordering induced by the aforementioned transitions.
	The ordering between existing operations according to these rules does not change when new operations are added, and so the mapping we get is decisive.
	We can use the conditions guarding loads and \quotes{roll backs} to earlier version numbers to show the ordering according to the above rules respects real time order.
	
	For hardness, the proof closely follows the proof that $\wsreg$ is $\wslin$-hard.
	The main difference is that when a write appears for the first time in a linearization, it might not appear to the right of writes which already appeared earlier.
	We use the roll back mechanism to simulate these writes, thus maintaining an invariant that $\reg$ contains the value of the rightmost write in the linearization.
\end{proof}

\begin{proof}[Proof (sketch) of \Cref{thm:abd-sim-dreg}]
	The simulation keeps track of when a pair of value $\val$ and timestamp $\ts$ reaches a majority of other processes for the first time.
	This can happen due to either a write distributing its newly written value
	or a read distributing its decided read value.
	When this happens, we check whether $\ts$ is larger than the current maximal timestamp that reached a majority of processes.
	If so, we perform in $\dreg$ a store of $\val$, attached to a new, larger version number,
	and then load this value with all threads that are active in a read method.
	Otherwise, we use the ``overwritten value'' path of $\dreg$:
	temporarily store $\val$ with a lower version number, collect this value by concurrent readers that can see it,
	and finally restore $\reg$ to its latest value.
\end{proof}

\clearpage

\section{Appendix Overview, Definitions and Notations}
\label{app:pre}

The rest of the appendix is structured as follows.
In this section, we present definitions used in proofs in the appendix,
(extending \cref{sec:pre} but also repeating it to assist the reader).
In \cref{app:hyp} we prove the results presented in \cref{sec:hyp} about the correspondence between simulation and strong observational refinement.
In \cref{app:comp} we formally construct the complete implementation for the class $\classrel{\rel}{\obj,\seqspec}$ and prove \cref{lem:closure_rel}.
In \cref{app:ws_d_comp} we present the formal definitions of the complete implementations for write strong-linearizability ($\wsreg$) and decisive linearizability ($\dreg$), and prove their inclusion in the
respective classes.
In \cref{sec:hard} we discuss our general technique for proving the hardness of $\wsreg$ and $\dreg$ for their respective classes, using what we call ``lazy linearizations'',
and then provide formal proofs of hardness employing these linearizations.
In \cref{sec:lazy} we formally prove the existence of said ``lazy linearizations''.
In \cref{app:abd} we formally define the \ABD implementation,
and prove the existence of a forward simulation from \ABD to $\dreg$.

\subsection{Sequences and Relations}

\smallskip
\noindent\textbf{Sequences.}
For a finite alphabet $\Sigma$, we denote by $\Sigma^*$ the set of all (finite) sequences over $\Sigma$,
and by $\Sigma^+$ the set of all non-empty (finite) sequences over $\Sigma$.
We use $\emptyword$ to denote the empty sequence.
The length of a (finite) sequence $s$ is denoted by $\size{s}$ (in particular $\size{\emptyword} = 0$).
We write $s[k]$ for the symbol at position $1 \leq k \leq \size{s}$ in $s$,
$s\range{k}{j}$ for the subsequence $s[k]\til s[j]$ ($\emptyword$ if $k>j$),
and $s[\last]$ for the last element of $s$.
We write $\sigma\in s$ if $s[k]=\sigma$ for some $1 \leq k \leq \size{s}$.
We denote $a<_s b$ ($a\leq_s b$) if there exists $j<k$ ($j\leq k$) such that $a=s[j]$ and $b=s[k]$.
We use ``$\cdot$'' for the concatenation of sequences,
and lift this notation to sets of sequences pointwise
($S_1 \cdot S_2 \defeq \set{s_1\cdot s_2 \st s_1\in S_1, s_2\in S_2}$).
We often identify symbols with sequences of length $1$ or their singletons (\eg in expressions like $\sigma \cdot S$).
The \emph{restriction} of a sequence $s$ \wrt a set $\Gamma$, denoted by $s\rst{\Gamma}$,
is the longest subsequence of $s$ that consists only of symbols in $\Gamma$.
This notation is extended pointwise to sets of sequences
($S\rst{\Gamma} \defeq \set{s\rst{\Gamma} \st s\in S}$).
We write $s_1 \subseq s_2$ when $s_1$ is a subsequence of $s_2$,
and $s_1 \pref s_2$ when $s_1$ is a prefix of $s_2$.
When $s_1\pref s_2$, we denote by $\suf{s_1}{s_2}$ the suffix $s$ of $s_2$ such that $s_2=s_1\cdot s$.
Given a sequences $s_1$ and a subsequence $s_2=s_1\range{i}{j}$ we denote by $\prefof{s_1}{s_2}$ the maximal prefix of $s_1$ such that $\prefof{s_1}{s_2}\cdot s_2\pref s_1$.
Note that $\prefof{s_1}{\emptyword}=s_1$.

\smallskip
\noindent\textbf{Relations.}
Given two relations $S,R$, we write $S\seq R$ to denote their composition.
Given a set $A$, we write $\idrel{A}$ to denote the identity relation on $A$.
We also write $\idrel{a_1,\ldots,a_n}$ as a short for $\idrel{\set{a_1,\ldots,a_n}}$.

\subsection{Labeled Transition Systems and Forward Simulations}
\label{app:lts_defs}

\smallskip
\noindent\textbf{Labeled transition systems.}
A \emph{labeled transition system} (LTS, for short) is a tuple $\lts = \tup{Q,\Sigma,q_0,T}$,
where $Q$ is a set of \emph{states},
$\Sigma$ is a (possibly infinite) alphabet (whose elements are called \emph{transition labels}),
$q_0\in Q$ is an \emph{initial state},
and $T\suq Q\times \Sigma \times Q$ is a set of \emph{transitions}.
We denote by $\lts.\lQ$, $\lts.\lSigma$, $\lts.\linit$, and $\lts.\lT$ the components of an LTS $\lts$.
Given a transition $\trans=\tup{q,\sigma,q'}\in \lts.\lT$, we let
$\prestate(\trans)\defeq q$, $\lab(\trans)\defeq\sigma$, and $\poststate(\trans)\defeq q'$.
We write $\asteplab{\sigma}{\lts}$ for the relation
$\set{\tup{q,q'} \st \tup{q,\sigma,q'}\in \lts.\lT}$,
and $\astep{\lts}$ for $\bigcup_{\sigma\in\Sigma} \asteplab{\sigma}{\lts}$.
For a sequence $\tr \in \lts.\lSigma^*$,
we write $\bsteplab{\tr}{\lts}$ for the composition $\asteplab{\tr[1]}{\lts} \seq \ldots \seq \asteplab{\tr[\last]}{\lts}$.
A sequence $\tr \in \lts.\lSigma^*$ such that $\lts.\linit \bsteplab{\tr}{\lts} q$
for some $q\in\lts.\lQ$ is a \emph{trace} of $\lts$ (or an \emph{$\lts$-trace}).
We denote by $\traces{\lts}$ the set of all traces of $\lts$.
A state $q\in \lts.\lQ$ is \emph{reachable} in $\lts$ if
$\lts.\linit \bsteplab{\tr}{\lts} q$ for some $\tr \in \traces{\lts}$.
A transition label $\sigma\in\lts.\lSigma$ is \emph{enabled} in a state $q\in \lts.\lQ$
if $q \asteplab{\sigma}{\lts} q'$ for some $q'\in\lts.\lQ$.
A sequence $\ex\in \lts.\lT^*$ such that $\ex=\emptyword$, or $\prestate(\ex[1])=q_0$ and
$\poststate(\ex[i])=\prestate(\ex[i+1])$ for all $1\leq i<\size{\ex}$
is an \emph{execution} of $\lts$.
We denote by $\executions{\lts}$ the set of all executions of $\lts$.
An \emph{execution fragment}  $\transsq\in \lts.\lT^*$ is a sequence of transitions such that $\transsq=\emptyword$ or $\transsq=\ex\range{i}{j}$ for some $i\leq j$ and $\ex\in\exec{\lts}$.
We also define the pre-state and post-state of an execution fragment $\transsq\neq\emptyword$ by $\prestate(\transsq)=\prestate(\transsq[1])$ and $\poststate(\transsq)=\poststate(\transsq[\last])$.
The trace induced by an execution fragment $\transsq$ of $\lts$
is the trace $\traceof{\transsq}$ given by $(\traceof{\transsq})[i]=\lab(\transsq[i])$ for every $1\leq i \leq \size{\transsq}$.
Note that we only consider finite executions and traces.

\smallskip
\noindent\textbf{Forward simulations.}
Given two LTSs $\lts$ and $\abs\lts$ and
a set $\Gamma\subseteq\lts.\lSigma\cap\abs\lts.\lSigma$ of ``observable'' transition labels,
a relation $\fsim\subseteq \lts.\lQ \times \abs\lts.\lQ$ is a
\emph{$\Gamma$-forward simulation} from $\lts$ to $\abs\lts$ if:
\begin{enumerate*}[label=(\roman*)]
	\item $\tup{\lts.\linit,\abs\lts.\linit}\in \fsim$; and
	\item If $q \asteplab{\sigma}{\lts} q'$ and $\tup{q,\abs q}\in \fsim$,
	then there exist ${\abs q}' \in \abs\lts.\lQ$ and $\tr\in\abs\lts.\lSigma^*$
	such that $\abs q\bsteplab{\tr}{\abs \lts}{\abs q}'$,
	$\tr\rst{\Gamma}=\sigma\rst{\Gamma}$,
	and $\tup{q',{\abs q}'}\in \fsim$.
\end{enumerate*}
We write $\lts\leqsim[\Gamma]\abs\lts$ when such $\fsim$ exists.

\begin{definition2}[Execution Recorders]\label{exrecorder}
The \emph{execution recorder} of an LTS $\lts$, denoted by $\ltsg$, is the LTS given by:
\begin{itemize}
	\item $\ltsg.\lQ= \executions{\lts}$
	\item $\ltsg.\lSigma= \lts.\lSigma$
	\item $\ltsg.\linit=\emptyword$
	\item $\ltsg.\lT = \set{\tup{\ex,\lab(\trans),\ex\cdot\trans}\st\trans\in\lts.\lT\land \ex\cdot\trans\in\executions{\lts}}$
\end{itemize}
\end{definition2}

The LTS $\ltsg$ maintains the execution generated so far in each of its states.
The execution instrumentation does not affect the traces of $\ltsg$ and behaves similarly to \quotes{ghost variables}.
Thus, a simulation between $\lts$ and $\ltsg$ (and vice-versa) can be easily established:

\begin{proposition}
	\label{prop:ghost}
	For every LTS $\lts$,
	we have $\lts\leqsim[\lts.\lSigma]\ltsg$ and $\ltsg \leqsim[\lts.\lSigma] \lts$.
\end{proposition}
\begin{proof}
	The relation $\fsim \defeq \set{\tup{\lts.\linit,\emptyword}}\cup\set{\tup{q,\ex}\st \ex\in \executions{\lts}\land \poststate(\ex[\last])=q}$
	is an $\lts.\lSigma$-forward simulation from $\lts$ to $\ltsg$,
	and its inverse is an $\lts.\lSigma$-forward simulation from $\ltsg$ to $\lts$.
\end{proof}

Forward simulations are equivalent to mappings between executions with certain useful properties.
Given two LTSs $\lts,\abs\lts$, we say that a mapping $\pi: \exec{\lts} \to \exec{\abs\lts}$:
\begin{itemize}
	\item is \emph{prefix-preserving} if $\pi(\emptyword)=\emptyword$ and $\pi(\ex_1)\pref\pi(\ex_2)$ for every $\ex_1,\ex_2\in\exec{\lts}$ such that $\ex_1\pref\ex_2$.
	
	\item \emph{respects} a set $\Gamma \subseteq\lts.\lSigma\cap\abs\lts.\lSigma$ if $\traceof{\ex}\rst{\Gamma} = \traceof{\pi(\ex)}\rst{\Gamma}$ for every $\ex\in\exec{\lts}$.
\end{itemize}

\begin{lemma} \label{lem:sim_pp}
	$\lts\leqsim[\Gamma]\abs\lts$ iff there exists a prefix preserving mapping $\pi: \exec{\lts} \to \exec{\abs\lts}$
	that respects $\Gamma$.
\end{lemma}
\begin{proof}
	Given a $\Gamma$-forward simulation $\fsim$ from $\lts$ to $\abs\lts$,
	we obtain that there exists a choice function
	$f_\fsim : (\lts.\lT \times \abs\lts.\lQ) \to (\abs\lts.\lT)^*$
	such that $\tup{\prestate(\trans),\abs q}\in \fsim$
	implies that $\traceof{f_\fsim(\trans,\abs q)}\rst{\Gamma}=\lab(\trans)\rst{\Gamma}$ and
	if $f_\fsim(\trans, \abs q)=\emptyword$, then $\tup{\poststate(\trans),{\abs q}}\in \fsim$,
	and otherwise
	$\abs q = \prestate(f_\fsim(\trans, \abs q))$
	and $\tup{\poststate(\trans),\poststate(f_\fsim(\trans,\abs q))}\in \fsim$.
	Then, for $\ex=\trans_1 \cdots \trans_n\in\exec{\lts}$, we define $\pi(\ex) \defeq f_\fsim(\trans_1, \abs q_0)\cdot f_\fsim(\trans_2,\abs q_1) \cdots f_\fsim(\trans_n,\abs q_{n-1})$
	where $\abs q_0 \til \abs q_{n-1}$ are defined by
	$\abs q_0 \defeq \abs\lts.\linit$ and
	$\abs q_{i+1} \defeq\abs q_i$ if $f_\fsim(\trans_{i+1},\abs q_i)=\emptyword$ and otherwise $\abs q_{i+1} \defeq\poststate(f_\fsim(\trans_{i+1},\abs q_i)) $.
	It is straightforward to verify that $\pi$ satisfies the required properties.

	For the converse, by \cref{prop:ghost} it suffices to show $\ltsg\leqsim[\Gamma]\ghost{\abs\lts}$.
	Indeed, $\pi$ itself (taken as a relation, \ie $\fsim \defeq \set{\tup{\ex,\pi(\ex)}\st \ex\in\exec{\lts}}$)
	is a $\Gamma$-forward simulation from $\ltsg$ to $\ghost{\abs\lts}$.
\end{proof}

\smallskip
\noindent\textbf{Synchronizing LTSs.}
The \emph{interface parallel composition} of two LTSs $\lts_1$ and $\lts_2$,
denoted by $\intcomp{\lts_1}{\lts_2}$, is the LTS given by:
\begin{itemize}
	\item $\intcomp{\lts_1}{\lts_2}.\lQ = \lts_1.\lQ \times \lts_2.\lQ$.
	\item $\intcomp{\lts_1}{\lts_2}.\lSigma = \lts_1.\lSigma\cup\lts_2.\lSigma$.
	\item $\intcomp{\lts_1}{\lts_2}.\linit = \tup{\lts_1.\linit,\lts_2.\linit}$.
	\item $\intcomp{\lts_1}{\lts_2}.\lT$ is defined by:
	\begin{mathpar}
		\inferrule[\internalleft]{\sigma \in \lts_1.\lSigma\setminus\lts_2.\lSigma \\\\ q_1 \asteplab{\sigma}{} q'_1}
		{\tup{q_1,q_2} \asteplab{\sigma}{} \tup{q'_1,q_2}}
		\and
		\inferrule[\internalright]{\sigma \in \lts_2.\lSigma\setminus\lts_1.\lSigma \\\\ q_2 \asteplab{\sigma}{} q'_2}
		{\tup{q_1,q_2} \asteplab{\sigma}{} \tup{q_1,q'_2}}
		\and
		\inferrule[\interface]{\sigma \in \lts_1.\lSigma\cap\lts_2.\lSigma \\\\ q_1 \asteplab{\sigma}{} q'_1 \\ q_2 \asteplab{\sigma}{} q'_2}
		{\tup{q_1,q_2} \asteplab{\sigma}{} \tup{q'_1,q'_2}}
	\end{mathpar}
\end{itemize}
The labels in $\lts_1.\lSigma\cap\lts_2.\lSigma$ are called \emph{synchronization labels}.
Given an execution fragment $\transsq = \tup{\tup{q_1^1,q_2^1},\sigma_1,\tup{q_1^{2},q_2^{2}}} \til \tup{\tup{q_1^n,q_2^n},\sigma_n,\tup{q_1^{n+1},q_2^{n+1}}} \in(\intcomp{\lts_1}{\lts_2}).\lT^*$, we write $\transsq\rst{\lts_j}$ for the induced execution fragment of $\lts_j$,
that is
$\tup{q_j^{i_1},\sigma_{i_1},q_j^{{i_1}+1}} \til \tup{q_j^{i_{n_j}},\sigma_{i_{n_j}},q_j^{{i_{n_j}}+1}}$
where $i_1 \til i_{n_j}$ is an enumeration of all indices $1\leq i< n$ such that $\sigma_i\in\lts_j.\lSigma$.
This notation is lifted to sets in the obvious way.

\begin{proposition}
	\label{prop:linking}
	$\ex\in\exec{\intcomp{\lts_1}{\lts_2}}$ iff $\ex\rst{\lts_1}\in\exec{\lts_1}$ and $\ex\rst{\lts_2}\in\exec{\lts_2}$.
\end{proposition}

Equality of induced executions of a composed LTS implies equality of traces restricted to the alphabet of said LTS:

\begin{lemma}
	\label{lem:executions_to_traces}
	Let $\execset\subseteq\exec{\intcomp{\lts}{\lts'}}$ and $\abs\execset\subseteq\exec{\intcomp{\abs\lts}{\lts'}}$ be sets of executions such that $\execset\rst{\lts'}=\abs\execset\rst{\lts'}$.
	Then $\traces{\execset}\rst{\lts'.\lSigma}=\traces{\abs\execset}\rst{\lts'.\lSigma}$.
\end{lemma}
\begin{proof}
	Given $\tr \in \traces{\execset}\rst{\lts'.\lSigma}$,
	there exists an execution
	$\ex\in\execset$ such that
	$\tr=\traceof{\ex}\rst{\lts'.\lSigma}$.
	Since $\traceof{\ex\rst{\lts'}}=\traceof{\ex}\rst{\lts'.\lSigma}$, we obtain that $\tr=\traceof{\ex\rst{\lts'}}$.
	As $\ex\in\execset$ we obtain that
	$\ex\rst{\lts'}\in\execset\rst{\lts'}= \abs\execset\rst{\lts'}$.
	That is, there exists an execution $\abs\ex\in\abs\execset$ such that
	${\abs\ex}\rst{\lts'}=\ex\rst{\lts'}$ and thus
	$\tr=\traceof{\ex\rst{\lts'}}=\traceof{\abs\ex\rst{\lts'}}$.
	Since $\traceof{\abs\ex\rst{\lts'}}=\traceof{\abs\ex}\rst{\lts'.\lSigma}$, we obtain that
	$\tr=\traceof{\abs\ex}\rst{\lts'.\lSigma}\in\traces{\abs\execset}\rst{\lts'.\lSigma}$.
	The converse inclusion is proven symmetrically,
	and we obtain that $\traces{\execset}\rst{\lts'.\lSigma}=\traces{\abs\execset}\rst{\lts'.\lSigma}$.
\end{proof}

In the following we show that a simulation $\lts\leqsim[\Gamma]\abs\lts$ induces a simulation $\intcomp{\lts}{\lts'}\leqsim[\lts'.\lSigma]\intcomp{\abs\lts}{\lts'}$
for every LTS $\lts'$ such that $\lts'.\lSigma \cap\lts.\lSigma=\lts'.\lSigma \cap\abs\lts.\lSigma=\Gamma$.
Similarly to \cref{lem:sim_pp}, this in turn induces a mapping $\pi:\exec{\intcomp{\lts}{\lts'}}\to\exec{\intcomp{\abs\lts}{\lts'}}$,
with several useful properties in addition to those guaranteed by \cref{lem:sim_pp}:

\begin{theorem}
	\label{thm:int_par}
	Suppose that $\lts\leqsim[\Gamma]\abs\lts$
	and let $\lts'$ be an LTS such that $\lts'.\lSigma \cap\lts.\lSigma=\lts'.\lSigma \cap\abs\lts.\lSigma=\Gamma$.
	Then, there exists a prefix preserving mapping $\pi:\exec{\intcomp{\lts}{\lts'}}\to\exec{\intcomp{\abs\lts}{\lts'}}$
	with the following additional properties:
	\begin{itemize}
		\item $\ex\rst{\lts'} = \pi(\ex)\rst{\lts'}$ for every $\ex\in\exec{\lts}$.
		
		\item
		$\lab(\pi(\ex)[\last])\in\lts'.\lSigma$ or $\pi(\ex)=\emptyword$ for every $\ex\in\exec{\intcomp{\lts}{\lts'}}$.
		
		\item $\pi(\ex\cdot\trans)=\pi(\ex)$ for every $\ex\in\exec{\intcomp{\lts}{\lts'}}$ and $\trans\in(\intcomp{\lts}{\lts'}).\lT$ such that $\ex\cdot\trans\in\exec{\intcomp{\lts}{\lts'}}$ and $\lab(\trans)\notin\lts'.\lSigma$.
		
		\item
		For every $\ex\in\exec{\intcomp{\lts}{\lts'}}$ and $\trans\in(\intcomp{\lts}{\lts'}).\lT$ such that $\ex\cdot\trans\in\exec{\intcomp{\lts}{\lts'}}$ and $\lab(\trans)\in\lts'.\lSigma\setminus\Gamma$,
		there exists $\abs\trans\in(\intcomp{\abs\lts}{\lts'}).\lT$ such that $\lab(\abs\trans)=\lab(\trans)$ and  $\pi(\ex\cdot\trans)=\pi(\ex)\cdot\abs\trans$.
	\end{itemize}
\end{theorem}

\begin{proof}
	Let $\fsim$ be a $\Gamma$-forward simulation from $\lts$ to $\abs\lts$.
	We define a relation $\fsim'$ by:
	$$\fsim'=\set{\tup{\tup{q,q'},\tup{\abs q,q'}}\st  q'\in\lts'.\lQ\land \exists {\abs q}'\in\abs\lts.\lQ, \abs\tr\in(\abs\lts.\lSigma\setminus\Gamma)^*\ldotp \tup{q,{\abs q}'}\in\fsim \wedge \abs q \bsteplab{\abs\tr}{\abs\lts} {\abs q}'}$$
	It is straightforward to verify that $\fsim'$ is an $\lts'.\lSigma$-forward simulation from $\intcomp{\lts}{\lts'}$ to $\intcomp{\abs\lts}{\lts'}$.
	Moreover, similarly to the proof of \cref{lem:sim_pp} we can use $\fsim'$ to obtain a choice function 
	$f_{\fsim'} : ((\intcomp{\lts}{\lts'}).\lT \times (\intcomp{\abs\lts}{\lts'}).\lQ) \to (\intcomp{\abs\lts}{\lts'}).\lT^*$
	with the following properties (stronger than those obtained when constructing a choice function from an arbitrary forward simulation):
	\begin{itemize}
		\item $\tup{\prestate(\trans),\tup{\abs q,q'}}\in \fsim'$
		implies that $f_{\fsim'}(\trans,\tup{\abs q,q'})\rst{\lts'}=\trans\rst{\lts'}$ and
		if $f_{\fsim'}(\trans, \tup{\abs q,q'})=\emptyword$, then $\tup{\poststate(\trans),\tup{\abs q,q'}}\in \fsim'$,
		and otherwise
		$\tup{\abs q,q'} = \prestate(f_{\fsim'}(\trans, \tup{\abs q,q'}))$
		and\\ $\tup{\poststate(\trans),\poststate(f_{\fsim'}(\trans,\tup{\abs q,q'}))}\in \fsim'$.
		
		\item $(f_{\fsim'}(\trans,\tup{\abs q,q'})[\last])\rst{\lts'}=\trans\rst{\lts'}$ if $\lab(\trans)\in\lts'.\lSigma$ and $f_{\fsim'}(\trans,\tup{\abs q,q'})=\emptyword$ otherwise.
		
		\item $f_{\fsim'}(\trans,\tup{\abs q,q'})\rst{\lts'}=\trans\rst{\lts'}$ if $\lab(\trans)\in\lts'.\lSigma\setminus\Gamma$.
	\end{itemize}
	
	Then, we proceed to construct a mapping $\pi$ from $f_{\fsim'}$ as done in the proof of \cref{lem:sim_pp}.
	It is straightforward to verify $\pi$ has all required properties.
\end{proof}

\subsection{Objects and Programs}
\label{app:objects}

We assume a set $\Tid$ of thread identifiers and an infinite set $\Id$ of action identifiers.

\begin{description}[leftmargin=0pt,itemsep=2pt]
	\item[\em Objects.]
	An \emph{object} is a pair $\obj=\tup{\Method,\Val}$, where
	$\Method$ is a set of method names and
	$\Val$ is a set of values.
	An object $\obj$ is associated with actions divided to \emph{invocations} $\inv=\invact\tup{\method,\val,\tid,\id} \in \Inv(\obj)$
	and \emph{responses} $\res=\resact\tup{\method,\val,\tid,\id} \in \Res(\obj)$,
	where $\method\in\Method$, $\val\in\Val \cup \set{\bot}$, $\tid\in\Tid$, and $\id\in\Id$.
	We use the functions $\methodf$, $\valf$, $\tidf$, and $\idf$,
	to retrieve the method name ($\method$),
	value ($\val$),
	thread identifier ($\tid$),
	and identifier ($\id$)
	of a given action $\act\in\Act$.
	We use $\bot$ as argument/return value for methods without input/output.
	We let $\Inv\Res(\obj) \defeq \Inv(\obj) \cup \Res(\obj)$.
	We sometimes lift the notation defined for actions to LTS transitions whose labels are these actions, e.g. $\idf(\trans)=\idf(\lab(\trans))$.
	
	\item[\em Histories.]
	A \emph{history} $\h$ of an object $\obj$ is a finite sequence of invocations and responses of $\obj$.
	A history $\h$ is \emph{sequential} if it is alternating between invocations and responses (starting with an invocation and ending either with an invocation or a response),
	such that every consecutive $\inv$, $\res$ in $\h$ has the same method and thread identifiers,
	and a unique action identifier across $\h$.
	A history $\h$ is \emph{well-formed} if for every $\tid\in\Tid$, the restriction of $\h$ to actions of $\tid$, denoted by $\h\rst{\tid}$, is sequential.
	An invocation $\inv\in\h$ is \emph{pending} if there is no response in $\h$ with the same thread and action identifiers.
	Otherwise, $\inv$ is \emph{complete}.
	We use $\h\rst{\method}$ for the restriction of $\h$ to actions with method $\method$.
	
	\item[\em Operations.]
	We apply the above notions also on \emph{operations} $\op$,
	which are either single invocations $\op=\inv$ (called a \emph{pending} operation) or pairs of matching invocation and response $\op=\tup{\inv,\res}$ (called a \emph{completed} operation).
	We let $\completed(\h)$ denote the subsequence of $\h$ consisting of actions that are a part of completed operations.
	The functions $\methodf$, $\tidf$, $\idf$ are lifted to operations by defining:
	$\methodf(\op)\defeq\methodf(\op[1])$, $\tidf(\op)\defeq\tidf(\op[1])$, and $\idf(\op)\defeq\idf(\op[1])$.
	An \emph{operation identifier} is a pair in $\Opid \defeq \Tid\times\Id$.
	The operation identifier of an operation $\op$ is given by $\opidf(\op) \defeq \tup{\tidf(\op),\idf(\op)}$.
	An \emph{operation's input} is defined by $\inputf(\op)\defeq \valf(\op[1])$,
	and an \emph{operation's output} is defined by $\outputf(\op)\defeq\valf(\op[\last])$ (when $\op$ is completed).
	Given an operation $\op$ and a history $\h$, we sometimes write $\op\in\h$ instead of $\op\subseq\h$.
	We let $\setof{\h}$ denote the set of all operations that occur in $\h$.
	Given a method name $\method\in\Method$, we write $\idrel{\method}$ as a short for $\idrel{\set{\op\st\methodf(\op)=\method}}$- the identity relation over operations of method $\method$.
	Given two sequential histories $\seqh_1,\seqh_2$ we define their disjoint concatenation $\seqh_1\cdotnew \seqh_2$ as the sequence with $\seqh_1$ followed by the subsequence of $\seqh_2$ composed of operations not present in $\seqh_1$. Formally, $\seqh_1\cdotnew\seqh_2=\seqh_1\cdot(\seqh_2\rst{\setof{\seqh_2}\setminus\setof{\seqh_1}})$.
	
	\item[\em Real-time Order.]
	The \emph{real time order} induced by
	a well-formed history $\h$, denoted by $<_\h$,
	is the partial order on operations defined by $\op_1 <_\h \op_2$ iff
	the response of $\op_1$ appears in $\h$ before the invocation of $\op_2$.
	The notation $\op_1 <_\h \op_2$ is also used when $\op_2$ is a completed operation but only its invocation appears in $\h$.
	We also define $\op_1 \leq_\h \op_2\defiff\op_1 <_\h \op_2\lor \op_1=\op_2$

	\item[\em Specifications.]
	A \emph{specification} of an object $\obj$ is a prefix closed set of sequential histories of $\obj$.

	\item[\em Object Implementations.]
	We assume a set $\Objint$ of labels for implementation internal actions and define an
	\emph{implementation} $\impl$ of an object $\obj$ to be an LTS over the alphabet $\Inv\Res(\obj) \cup \Objint$.
	For all implementations presented in the paper we will use
	$\Objint=\set{\objact\tup{\tid}\st \tid\in\Tid}$
	We assume that the history induced by every execution $\ex$ of $\impl$, denoted by $\hs{\ex}$, is a well-formed history.
	Intuitively, executions of $\impl$ represent executions generated
	by the methods' implementations when they repeatedly and concurrently invoked with arbitrary arguments.
	
	\item[\em Client Programs.]
	We assume a set $\Progint$ of labels for client internal actions (disjoint from $\Objint$) and define a
	client program $\prog$ for an object $\obj$ as an LTS over the alphabet $\prog.\lSigma$.	

	\item[\em Linking Client Programs and Object Implementations.]
	The \emph{linking} of a program $\prog$ and implementation $\impl$, denoted by $\prog[\impl]$,
	is the interface parallel composition of $\prog$ and $\impl$,
	\ie $\prog[\impl]=\intcomp{\impl}{\prog}$.

\end{description}

\subsection{Linearizability}\label{subsec:lin}

\begin{definition2}
	\label{def:lin}
	A sequential history $\seqh\in\seqspec$ is a \emph{linearization} of a history $\h$, denoted $\h\linleq\seqh$,
	if there exist a sequence of responses $\ressq\in\Res^*$ for pending invocations in $\h$ such that the following hold for $\h'=\completed(\h\cdot\ressq)$:
	\begin{itemize}
		\item $\h'\rst{\tid}=\seqh\rst{\tid}$ for every $\tid\in\Tid$.
		\item $<_{\h'} \;\subseteq\; <_\seqh$.
	\end{itemize}
	A history $\h$ is called \emph{linearizable} \wrt $\seqspec$ if it has a linearization $s\in\seqspec$.
\end{definition2}

The first condition on $\h'$ and $\seqh$ can be simplified, as follows:

\begin{lemma}
	\label{obs:lin}
	$\h\linleq\seqh$ iff there exists a sequence $\ressq\in\Res^*$ of responses for pending invocations in $\h$ such that the following hold for $\h'=\completed(\h\cdot\ressq)$:
	\begin{itemize}
		\item $\setof{\h'}=\setof{\seqh}$.
		\item $<_{\h'} \;\subseteq\; <_\seqh$.
	\end{itemize}
\end{lemma}

\begin{proof}
	In one direction, since every operation belongs to some thread, we have that $\forall\tid\in\Tid\ldotp\h'\rst{\tid}=\seqh\rst{\tid}\implies \setof{\h'}=\setof{\seqh}$.
	For the converse, we have $\setof{\h'}=\setof{\seqh}\implies\setof{\h'\rst{\tid}}=\setof{\seqh\rst{\tid}}$ for every $\tid\in\Tid$.
	As $\h'\rst{\tid}$ is sequential, we obtain that both $<_{\h'\rst{\tid}}, <_{\seqh\rst{\tid}}$ are total orders over the same set of operations.
	Using $<_{\h'} \;\subseteq\; <_\seqh$ we get that $<_{\h'\rst{\tid}} \;\subseteq\; <_{\seqh\rst{\tid}}$ and so overall $\h'\rst{\tid}=\seqh\rst{\tid}$.
\end{proof}

\begin{definition2}
	\label{def:linmap}
	Let $\impl$ be an implementation and $\seqspec$ a specification.
	A function $\linmap: \exec{\impl}\to\seqspec$ is a \emph{linearization mapping}
	if $\hs{\ex}\linleq\linmap(\ex)$ for every $\ex\in\exec{\impl}$.
	We say $\impl$ is \emph{linearizable} \wrt $\seqspec$
	if there exists a linearization mapping $\linmap:\exec{\impl}\to \seqspec$.
\end{definition2}

Next, we prove several lemmas that simplify the construction of linearization mappings from existing ones.
The first is used to show that a sequential history linearizes some history,
given that some other sequential history linearizes that history.

\begin{lemma}\label{lem:lin}
	If $\seqh,\seqh'\in\seqspec$, $\setof{\seqh}=\setof{\seqh'}$, and for some history $\h$ we have $\h\linleq\seqh$ and $\op<_{\h}\op'$ implies $\op<_{\seqh'}\op'$ for all $\op,\op'\in\setof{\seqh'}$,
	then $\h\linleq\seqh'$.
\end{lemma}

\begin{proof}
	Let $\ressq$ be the sequence of responses such that the conditions in \cref{obs:lin} hold for $\h'=\completed(\h\cdot\ressq)$ and $\seqh$.
	That is, $\setof{\h'}=\setof{\seqh}$ and  $<_{\h'}\subseteq<_{\seqh}$.
	We claim the same conditions hold for $\h'$ and $\seqh'$.
	Indeed,
	by definition of linearizability, $\setof{\seqh}=\setof{\h'}$, and as $\setof{\seqh}=\setof{\seqh'}$, we get that $\setof{\h'}=\setof{\seqh'}$.
	It remains to verify that $<_{\h'}\subseteq<_{\seqh'}$.
	Indeed, assume $\op<_{\h'}\op'$. Using the lemma's assumptions it is sufficient to verify that $\op,\op'\in\setof{\seqh'}$ and $\op<_{\h}\op'$.
	The former follows from $\setof{\h'}=\setof{\seqh'}$, and the latter holds since real time order is not effected by the addition of a suffix of responses, which means $<_{\h}=<_{\h\cdot\ressq}\supseteq<_{\h'}$.
\end{proof}

The next lemma also enables generating a linearization from an existing linearization
by removing operations that are not completed in the execution inducing the linearized history.

\begin{lemma}\label{lem:remove_pending}
	If $\seqh,\hat{\seqh}\in\seqspec$, $\hat{\seqh}\subseq\seqh$, $\h\linleq \seqh$ for some history $\h$ and $\setof{\completed(\h)}\subseteq\setof{\hat{\seqh}}$ then $\h\linleq \hat{\seqh}$.
\end{lemma}
\begin{proof}
	From $\h\linleq\seqh$ there exist a sequence of responses $\ressq\in\Res^*$ for pending invocations in $\h$ such that the conditions in \cref{obs:lin} hold for $\h'=\completed(\h\cdot\ressq)$ and $\seqh$.
	That is, $\setof{\h'}=\setof{\seqh}$ and  $<_{\h'}\subseteq<_{\seqh}$.
	For every $\op\in\setof{\seqh}\setminus\setof{\hat{\seqh}}$ it must be the case that $\op\notin\setof{\completed(\h)}$ and so for $\inv\in\Inv$, $\res\in\Res$ the inv-res pair such that $\op=\inv\cdot\res$
	we have $\inv\in\h$, $\res\in\ressq$.
	If we take $\ressq'\subseq\ressq$ to be the subsequence with each such $\res$ removed, then for $\h''\defeq\completed(\h\cdot\ressq')$
	we have $\setof{\h'}\setminus\setof{\h''}=\setof{\seqh}\setminus\setof{\hat{\seqh}}$.
	As $\setof{\h'}=\setof{\seqh}$, we can deduce that $\setof{\h''}=\setof{\hat{\seqh}}$.
	Moreover, $<_{\h''}=<_{}$
	
	the same operations are removed from $\h''\defeq\completed(\h\cdot\ressq')$ compared to $\h'$ and from $\hat{\seqh}$ compared to $\seqh$, which means the conditions in \cref{obs:lin} hold for $\h''$ and $\hat{\seqh}$, and we get that $\h\linleq \hat{\seqh}$.
\end{proof}

Finally, if a sequential history linearizes the history of some execution, and we append a non-response transition,
then the same sequential history linearizes the extended execution:

\begin{lemma}\label{lem:non_response}
	If $\impl$ is an implementation,
	$\ex\in\exec{\impl}$ is an execution,
	$\trans\in\impl.\lT$ is a transition such that $\ex\cdot\trans\in\exec{\impl}$ and $\trans\notin\Res$,
	$\seqh\in\seqspec$,
	and $\hs{\ex}\linleq\seqh$,
	then $\hs{\ex\cdot\trans}\linleq\seqh$.
\end{lemma}
\begin{proof}
	From $\hs{\ex}\linleq\seqh$,
	there exists a sequence $\ressq\in\Res^*$ such that the conditions in \cref{def:lin} hold for $\h'=\completed(\hs{\ex}\cdot \ressq)$ and $\seqh$.
	It is sufficient to show that $\h'=\completed(\hs{\ex\cdot\trans} \cdot\ressq)$, and thus $\hs{\ex\cdot\trans}\linleq\seqh$ follows using the same sequence $\ressq$.
	
	Indeed, denoting $\act=\lab(\trans)$ we have $\act\in\Objint$ or $\act\in\Inv$.
	If $\act\in\Objint$ then $\hs{\ex\cdot\trans}=\hs{\ex}$ and we are done.
	Otherwise, $\act\in\Inv$ and we have $\hs{\ex\cdot\trans}=\hs{\ex}\cdot\act$.
	We have that $\act$ does not have a matching response in $\ressq$ since $\ressq$ only contains responses for pending invocations in $\hs{\ex}$.
	Hence,
	$\completed(\hs{\ex\cdot\trans}\cdot \ressq)=\completed(\hs{\ex}\cdot\act\cdot \ressq)=\completed(\hs{\ex}\cdot \ressq)$.
\end{proof}

\subsection{Registers}\label{subsec:regspec}

A register object is given by $\Reg=\tup{\set{\regread,\regwrite},\N}$.
Its specification, denoted by $\regspec$, assumes that $0$ is the initial register value,
and is formally defined as follows.

\begin{definition2}
\emph{The value observed after a sequential history $\seqh$ of $\Reg$},
denoted by $\obs(\seqh)\in\Val$,
is defined as follows:
	\begin{equation*}
		\obs(\seqh)\defeq\begin{cases}
			0 & \seqh\rst{\regwrite}=\emptyword \\
			\inputf(\lastop(\seqh\rst{\regwrite})) & \seqh\rst{\regwrite}\neq\emptyword\\
		\end{cases}
	\end{equation*}
\end{definition2}

\begin{definition2}
	\label{def:regspec}
	The set $\regspec$ consists of all sequential histories $\seqh$ such that for all operations $\op\in\seqh$:
	\begin{itemize}
		\item If $\methodf(\op)=\regread$, then $\inputf(\op)=\bot$ and $\outputf(\op)=\obs(\prefof{\seqh}{\op})$.
		\item If $\methodf(\op)=\regwrite$, then $\inputf(\op)\in\Val$ and $\outputf(\op)=\bot$.
	\end{itemize}
\end{definition2}

Note that $\seqh\cdot\invact\tup{\regread,\bot,\tid,\id}\cdot \resact\tup{\regread,\obs(\seqh),\tid,\id}\in\regspec$ for any operation identifier $\tup{\tid,\id}\in\Opid$ that is not the identifier of some operation in $\seqh$.

The following lemmas are useful for constructing register linearizations from existing linearizations.
First, register sequential histories remain legal when removing read operations from them.

\begin{lemma}\label{lem:remove_read}
	If $\seqh\in\regspec$, then $\hat{\seqh}\in\regspec$ for every sequential history $\hat{\seqh}$ such that $\seqh\rst{\regwrite}\subseq\hat{\seqh}\subseq\seqh$.
\end{lemma}

\begin{proof}
	Follows from the fact that $\obs$ only depends on write operations and the sequence of write operations before every read is the same in $\hat{\seqh}$ and $\seqh$.
\end{proof}

The next lemma is used to show that a sequence is a valid register sequential history, using other valid histories.
This is done by showing each read comes after a sequence of writes which it also comes after in some other valid sequential history.

\begin{lemma}\label{lem:regspec}
	If $\seqh$ is a sequential history of $\Reg$ such that $\setof{\seqh}\subseteq\setof{\seqh'}$ for some $\seqh'\in\regspec$ and for each read operation $\rop\in\seqh\rst{\regread}$ we have $\prefof{\seqh}{\rop}\rst{\regwrite}=\prefof{\seqh_\rop}{\rop}\rst{\regwrite}$ for some $\seqh_\rop\in\regspec$ such that $\rop\in\setof{\seqh_\rop}$,
	then $\seqh\in\regspec$.
\end{lemma}

\begin{proof}
	First, the condition $\setof{\seqh}\subseteq\setof{\seqh'}$ and $\seqh'\in\regspec$ ensure that the inputs of reads and outputs of writes are $\bot$, and inputs of writes are in $\Val$.
	
	Moreover, for each read operation $\rop\in\seqh\rst{\regread}$ we have $\prefof{\seqh}{\rop}\rst{\regwrite}=\prefof{\seqh_\rop}{\rop}\rst{\regwrite}$ which implies
	$\obs(\prefof{\seqh}{\rop})=\obs(\prefof{\seqh_\rop}{\rop})$.
	As $\seqh_\rop\in\regspec$ we have that $\outputf(\rop)=\obs(\prefof{\seqh_\rop}{\rop})$ and thus overall $\outputf(\rop)=\obs(\prefof{\seqh}{\rop})$ as required.
\end{proof}

\subsection{Hyperproperties, schedulers and strong observational refinement}
\begin{description}[leftmargin=0pt,itemsep=2pt]
	\item[\em Schedulers.]
	Given a program $\prog$ and an implementation $\impl$,
	a \emph{scheduler} is a function $\scheduler : \executions{\prog[\impl]} \to \powerset{\prog[\impl].\lT}$.
	An execution $\ex\in\executions{\prog[\impl]}$ is \emph{consistent with $\scheduler$} if
	$\ex[j] \in \scheduler(\ex[1]\cdots\ex[j-1])$ for every $1 \leq j \leq \size{\ex}$.
	We denote by $\executions{\prog[\impl],\scheduler}$ the set of executions of $\prog[\impl]$
	that are consistent with $\scheduler$,
	and by $\traces{\prog[\impl],\scheduler}$ the traces of executions in $\executions{\prog[\impl],\scheduler}$.
	A scheduler is \emph{deterministic} if for every $\ex\in \executions{\prog[\impl]}$,
	either $\size{\scheduler(\ex)}\leq 1$ or all transitions in $\scheduler(\ex)$ are labeled by actions in $\Progint$.
\end{description}

Deterministic schedulers have the following useful property: 
if two executions generated by a deterministic scheduler induce the same program execution,
then one of the executions must be a prefix of the other:

\begin{lemma}
	\label{lem:det_rst}
	Let $\impl$ be an implementation of an object $\obj$ and $\prog$ a client program for $\obj$,
	and let $\scheduler:\exec{\prog[\impl]}\to\powerset{\prog[\impl].\lT}$ be a deterministic scheduler.
	Then $\ex_1\pref\ex_2$ or $\ex_2\pref\ex_1$
	for every $\ex_1,\ex_2\in\executions{\prog[\impl],\scheduler}$
	such that $\ex_1\rst{\prog}=\ex_2\rst{\prog}$.
\end{lemma}
\begin{proof}
	Assume $\ex_1,\ex_2$ are not prefixes of each other.
	Let $j$ be the minimal index such that $\ex_1[j]\neq\ex_2[j]$.
	Then $\ex_1[j],\ex_2[j]\in\scheduler(\ex_1\range{1}{j-1})$
	and thus $\size{\scheduler(\ex_1\range{1}{j-1})}>1$.
	From $\scheduler$ being deterministic,
	it must be the case that $\lab(\ex_{1}[j]),\lab(\ex_{2}[j])\in\Progint$,
	contradicting $\ex_1\rst{\prog}=\ex_2\rst{\prog}$.
\end{proof}

Next, we show that
the %
mapping of executions generated by a deterministic scheduler to their induced program executions is a bijection
when considering only %
executions ending in a program transition (\ie an invocation, a response, or a program internal transition).
Moreover, both the mapping and its inverse are prefix preserving:

\begin{lemma}
	\label{thm::det_rst}
	Let $\impl$ be an implementation of an object $\obj$ and $\prog$ a client program for $\obj$,
	and let $\scheduler:\exec{\prog[\impl]}\to\powerset{\prog[\impl].\lT}$ be a deterministic scheduler.
	Then the function 
	$\progrst:(\set{\emptyword}\cup\set{\ex\in\exec{\prog[\impl],\scheduler}\st\lab(\ex[\last])\in\prog.\lSigma})\to \exec{\prog[\impl],\scheduler}\rst{\prog}$
	defined by $\progrst(\ex)=\ex\rst{\prog}$ is a bijection
	and both $\progrst$ and $\progrst^{-1}$ are prefix preserving.
\end{lemma}
\begin{proof}
	For injectivity, let $\ex_1,\ex_2 \in\set{\emptyword}\cup\set{\ex\in\exec{\prog[\impl],\scheduler}\st\lab(\ex[\last])\in\prog.\lSigma}$ 
	be two executions such that $\ex_1\rst{\prog}=\ex_2\rst{\prog}$.
	By \cref{lem:det_rst} we obtain \wlg that $\ex_1\pref\ex_2$.
	Assume for contradiction that $\ex_1\prefneq\ex_2$.
	Then $\ex_2\neq\emptyword$ and in particular $\ex_2[\last]$ is a program transition (\ie $\lab(\ex_2[\last])\in\prog.\lSigma$).
	We deduce that $\ex_1\rst{\prog}\prefneq\ex_2\rst{\prog}$,
	a contradiction.
	Therefore, $\ex_1=\ex_2$.
	
	For surjectivity, let $\progex\in\exec{\prog[\impl],\scheduler}\rst{\prog}$.
	Then by definition there exists some execution $\ex\in\exec{\prog[\impl],\scheduler}$ such that $\progex=\ex\rst{\prog}$.
	Let $\ex'$ be the maximal prefix of $\ex$ such that $\lab(\ex'[\last])\in\prog.\lSigma$, or $\ex'=\emptyword$ if no such prefix exists.
	Then we get that $\progrst(\ex')=\ex'\rst{\prog}=\ex\rst{\prog}=\progex$.
	
	The fact that $\progrst$ is prefix preserving is immediate.
	To show that $\progrst^{-1}$ is also prefix preserving, due to bijectivity of $\progrst$ it suffices to prove:
	\begin{itemize}
		\item $\progrst(\emptyword)=\emptyword$:
		
		Indeed $\progrst(\emptyword)=\emptyword\rst{\prog}=\emptyword$.
		
		\item  $\progrst(\ex_1)\pref \progrst(\ex_2)\implies \ex_1\pref\ex_2$ for all $\ex_1,\ex_2\in\set{\emptyword}\cup\set{\ex\in\exec{\prog[\impl],\scheduler}\st\lab(\ex[\last])\in\prog.\lSigma}$:
		
		Indeed, $\progrst(\ex_1)\pref \progrst(\ex_2)$ means $\ex_1\rst\prog\pref\ex_2\rst{\prog}$.
		Taking $\ex'_2\pref\ex_2$ the minimal prefix of $\ex_2$ such that $\ex_1\rst\prog\pref\ex'_2\rst{\prog}$, we obtain that $\ex_1\rst\prog=\ex'_2\rst{\prog}$ %
		and either $\ex'_2=\emptyword$ or $\lab(\ex'_2[\last])\in\prog.\lSigma$.
		$\ex'_2$ is consistent with $\scheduler$ since $\ex_2$ is consistent with $\scheduler$, and so $\ex'_2$ is in the domain of $\progrst$ and $\progrst(\ex'_2) = \ex'_2\rst\prog$. Thus, from $\ex_1\rst\prog=\ex'_2\rst{\prog}$, we deduce that $\progrst(\ex_1)=\progrst(\ex'_2)$. Combined with the injectivity of $\progrst$ proven above we get that $\ex_1=\ex'_2\pref\ex_2$ as required.\qedhere
	\end{itemize}
\end{proof}

\begin{description}[leftmargin=0pt,itemsep=2pt]
	\item[\em Hyperproperties.] A \emph{hyperproperty} $\phi$ of a program $\prog$ is a set of sets of the program's traces
	(\ie $\phi \subseteq \powerset{\traces{\prog}}$).
	An implementation $\impl$ \emph{satisfies a hyperproperty} $\phi$ of $\prog$,
	denoted by $\impl \models_\prog \phi$,
	if $\traces{\prog[\impl],\scheduler}\rst{\prog.\lSigma}\in\phi$ for every deterministic scheduler $\scheduler$.
	
	\item[\em Strong Observational Refinement.]
	An implementation $\impl$ \emph{strongly observationally refines} an implementation $\abs\impl$,
	denoted by $\impl \srefine \abs\impl$,
	if $\abs\impl\models_\prog \phi\implies\impl\models_\prog \phi$
	for every program $\prog$ and hyperproperty $\phi$ of $\prog$.
	The following alternative characterization follows from the definition.
	
	\begin{lemma}
		$\impl \srefine \abs\impl$ iff for every program $\prog$ and deterministic scheduler $\scheduler$,
		there exists a deterministic scheduler $\abs\scheduler$ such that
		$\traces{\prog[\impl],\scheduler}\rst{\prog.\lSigma} = \traces{\prog[\abs\impl],\abs\scheduler}\rst{\prog.\lSigma}$.
	\end{lemma}
	
	\begin{proof}
		Assume $\impl \srefine \abs\impl$.
		Let $\prog$ be a program and consider the hyperproperty $$\phi=\set{\traces{\prog[\abs\impl],\abs\scheduler}\rst{\prog.\lSigma}\st \abs\scheduler \text{ is a deterministic scheduler}}$$
		It is straightforward to verify that $\phi \subseteq \powerset{\traces{\prog}}$ and $\abs\impl \models_\prog \phi$.
		Therefore, $\impl \models_\prog \phi$.
		This implies that for every deterministic scheduler $\scheduler$,
		$\traces{\prog[\impl],\scheduler}\rst{\prog.\lSigma}\in\phi$ and thus there exists a deterministic scheduler $\abs\scheduler$ such that $\traces{\prog[\impl],\scheduler}\rst{\prog.\lSigma} = \traces{\prog[\abs\impl],\abs\scheduler}\rst{\prog.\lSigma}$.
		
		The converse is immediate from the definitions.
	\end{proof}
	
\end{description}

\section{Proof of \cref{thm:sref_iff_sim} ($\impl \srefine \abs\impl\iff\impl\leqsim\abs\impl$)}
\label{app:hyp}

In this appendix we prove \cref{thm:sref_iff_sim}.
A similar theorem is proved in \cite[Theorem 8]{AE19}
for slightly different definitions of linearization mappings and schedulers, operating on traces instead of executions,
and for implementations that are step-deterministic LTSs.
For the sake of completeness, we reprove the result for our setting,
adapting in some cases (and somewhat simplifying in others) the proofs from \cite{AE19}.

First, we prove that forward simulation is sufficient for strong observational refinement.

\begin{theorem}
	Let $\impl,\abs\impl$ be implementations of an object $\obj$.
	If $\impl\leqsim\abs\impl$, then $\impl \srefine \abs\impl$.
\end{theorem}

\begin{proof}
	Let $\prog$ be a client program and $\scheduler:\exec{\prog[\impl]}\to\powerset{\prog[\impl].\lT}$ be a deterministic scheduler.
	From \cref{thm:int_par}, there exists a prefix preserving mapping  $\pi: \exec{\prog[\impl]} \to  \exec{\prog[\abs\impl]}$ such that:
	\begin{itemize}
		\item (respects $\prog$) $\ex\rst{\prog} = \pi(\ex)\rst{\prog}$ for every $\ex\in\exec{\prog}$.
		
		\item ($\prog$-ending) $\lab(\pi(\ex)[\last])\in\prog.\lSigma$ or $\pi(\ex)=\emptyword$ for every $\ex\in\exec{\prog[\impl]}$.
		
		\item ($\prog$-activated) $\pi(\ex\cdot\trans)=\pi(\ex)$ for every $\ex\in\exec{\prog[\impl]}$ and $\trans\in\prog[\impl].\lT$ such that $\ex\cdot\trans\in\exec{\prog[\impl]}$ and $\lab(\trans)\notin\prog.\lSigma$.
		
		\item ($\Progint$-copying)
		For every $\ex\in\exec{\prog[\impl]}$ and $\trans\in\prog[\impl].\lT$ such that $\ex\cdot\trans\in\exec{\prog[\impl]}$ and $\lab(\trans)\in\Progint$,
		there exists $\abs\trans\in\prog[\abs\impl].\lT$ such that $\abs\trans\rst{\prog}=\trans\rst{\prog}$ and  $\pi(\ex\cdot\trans)=\pi(\ex)\cdot\abs\trans$.
	\end{itemize}
	
	We define a scheduler $\abs\scheduler:\exec{\prog[\abs\impl]}\to\powerset{\prog[\abs\impl].\lT}$ as follows:
	$$\abs\scheduler(\abs\ex) \defeq\set{\abs\trans\in\abs\lts.\lT\st \exists \ex\in\exec{\prog[\impl],\scheduler}\ldotp \abs\ex\cdot\abs\trans\pref \pi(\ex)}$$
	The scheduler $\abs\scheduler$ allows transitions when they extend the current execution towards some execution obtained by applying $\pi$
	on some execution allowed by $\scheduler$.

	The following claim states the main useful property of $\abs\scheduler$:
	
	\begin{claim2}
		\label{claim:inverse_ex}
		For all $\abs\trans\in\abs\scheduler(\abs\ex)$, there exists an execution $\ex=\ex'\cdot\trans\in\exec{\prog[\impl],\scheduler}\setminus\set{\emptyword}$ such that the following hold:
		\begin{itemize}
			\item $\abs\ex\cdot\abs\trans\pref\pi(\ex)$
			
			\item $\lab(\trans)\in\prog.\lSigma$.
			
			\item $\abs\ex\rst{\prog}=\ex'\rst{\prog}$.
			
			\item If $\lab(\abs\trans)\in\prog.\lSigma$, then $\abs\ex\cdot\abs\trans=\pi(\ex)$.
		\end{itemize}
	\end{claim2}
	\begin{claimproof}
		By definition of $\abs\scheduler$ there exists an execution $\ex$ satisfying $\abs\ex\cdot\abs\trans\pref\pi(\ex)$.
		Let $\exmin$ be the minimal execution (\wrt $\pref$) satisfying $\abs\ex\cdot\abs\trans\pref\pi(\exmin)$.
		Since $\pi$ is $\prog$-activated,
		it must be the case that $\lab(\exmin[\last])\in\prog.\lSigma$.
		Denote $\exmin=\ex'\cdot\trans$.
		Let $\ex''\pref\ex'$ be the largest prefix of $\ex'$ (\wrt $\pref$) such that $\lab(\ex''[\last])\in\prog.\lSigma$.
		Then from $\pi$ being $\prog$-activated we have $\pi(\ex')=\pi(\ex'')$ and from minimality of $\exmin$ we have that
		$\pi(\ex'')\prefneq\abs\ex\cdot\abs\trans\pref\pi(\exmin)$.
		Removing the transition $\abs\trans$ we obtain that $\pi(\ex')=\pi(\ex'')\pref\abs\ex\prefneq\pi(\exmin)$.
		As $\lab(\pi(\exmin)[\last])\in\prog.\lSigma$,
		restricting to program transitions we deduce
		$\pi(\ex')\rst{\prog}\pref\abs\ex\rst{\prog}\prefneq\pi(\exmin)\rst{\prog}$.
		Using the fact that $\pi$ respects $\prog$, we obtain:
		$\ex'\rst{\prog}\pref\abs\ex\rst{\prog}\prefneq\exmin\rst{\prog}$.
		As $\exmin\rst{\prog}=\ex'\rst{\prog}\cdot\trans$ and $\lab(\trans)\in\prog.\lSigma$
		we get overall that $\abs\ex\rst{\prog}=\ex'\rst{\prog}$.
		
		Assume now additionally that $\abs\trans\in\prog.\lSigma$ and
		assume for contradiction $\abs\ex\cdot\abs\trans\neq\pi(\ex)$.
		Returning to the inequality
		$\pi(\ex')\prefneq\abs\ex\cdot\abs\trans\pref\pi(\exmin)$
		we deduce that
		$\pi(\ex')\prefneq\abs\ex\cdot\abs\trans\prefneq\pi(\exmin)$.
		All executions here end in a program transitions and thus $\prefneq$ is maintained when restricting to these transitions and similarly to above we can deduce
		$\ex'\rst{\prog}\prefneq\abs\ex\rst{\prog}\cdot\abs\trans\prefneq\exmin\rst{\prog}$,
		contradicting the fact that $\exmin\rst{\prog}=\ex'\rst{\prog}\cdot\trans$.
	\end{claimproof}
	
	When restricting $\pi$'s domain and codomain to executions ending in program transitions, $\pi$ is a bijection:
	
	\begin{claim2}
		\label{claim:pi_biject}
		$\pi$ is a bijection from $\set{\emptyword}\cup\set{\ex\in\exec{\prog[\impl],\scheduler}\st\lab(\ex[\last])\in\prog.\lSigma}$ to $\set{\emptyword}\cup\set{\abs\ex\in\exec{\prog[\abs\impl],\abs\scheduler}\st\lab(\abs\ex[\last])\in\prog.\lSigma}$.
	\end{claim2}
	\begin{claimproof}
		As $\pi(\emptyword)=\emptyword$,
		it suffices to show $\pi$ is a bijection from $\set{\ex\in\exec{\prog[\impl],\scheduler}\st\lab(\ex[\last])\in\prog.\lSigma}$ to $\set{\abs\ex\in\exec{\prog[\abs\impl],\abs\scheduler}\st\lab(\abs\ex[\last])\in\prog.\lSigma}$.
		Indeed, we have:
		\begin{itemize}
			\item The codomain is valid:
			
			Let $\ex\in\set{\ex\in\exec{\prog[\impl],\scheduler}\st\lab(\ex[\last])\in\prog.\lSigma}$.
			Since $\pi$ is $(\prog.\lSigma)$-ending,
			$\lab(\pi(\ex)[\last])\in\prog.\lSigma$,
			and it remains to show $\pi(\ex)$ is consistent with $\abs\scheduler$.
			Indeed, let $1\leq j\leq\size{\pi(\ex)}$.
			We have $\pi(\ex)\range{1}{j-1} \cdot \pi(\ex)[j] \pref\pi(\ex)$,
			and so $\pi(\ex)[j] \in \abs\scheduler(\pi(\ex)\range{1}{j-1})$.
			
			\item $\pi$ is injective:
			
			Let $\ex_1,\ex_2\in \set{\ex\in\exec{\prog[\impl],\scheduler}\st\lab(\ex[\last])\in\prog.\lSigma}$
			be two traces such that $\pi(\ex_1)=\pi(\ex_2)$.
			Then $\pi(\ex_1)\rst{\prog}=\pi(\ex_2)\rst{\prog}$.
			Since $\pi$ respects $\prog$
			we have that $\ex_j\rst{\prog}=\pi(\ex_j)\rst{\prog}$.
			Overall we obtain that	$\progrst(\ex_1)=\ex_1\rst{\prog}=\pi(\ex_1)\rst{\prog}=\pi(\ex_2)\rst{\prog}=\ex_2\rst{\prog}=\progrst(\ex_2)$,
			and thus $\ex_1=\ex_2$ follows from the injectivity of $\progrst$, proven in \cref{thm::det_rst}.
			
			\item $\pi$ is surjective:
			
			Let $\abs\ex\in\set{\abs\ex\in\exec{\prog[\abs\impl],\abs\scheduler}\st\lab(\abs\ex[\last])\in\prog.\lSigma}$.
			Then $\abs\ex[\last]\in\abs\scheduler(\abs\ex\range{1}{\size{\abs\ex}-1})$,
			and applying \cref{claim:inverse_ex} (using $\lab(\abs\ex[\last])\in\prog.\lSigma$) we obtain an execution $\ex\in\set{\ex\in\exec{\prog[\impl],\scheduler}\st\lab(\ex[\last])\in\prog.\lSigma}$ such that $\pi(\ex)=\abs\ex$.\qedhere
		\end{itemize}
	\end{claimproof}

	Next, we show $\exec{\prog[\impl],\scheduler}\rst{\prog}= \exec{\prog[\abs\impl],\abs\scheduler}\rst{\prog}$.
	This is done by establishing a bijection $f:\exec{\prog[\impl],\scheduler}\rst{\prog}\to \exec{\prog[\abs\impl],\abs\scheduler}\rst{\prog}$,
	and then showing it is the identity function.
	
	For this purpose, using \cref{thm::det_rst} we have that the functions
	\begin{align*}
		\progrst:&\set{\emptyword}\cup\set{\ex\in\exec{\prog[\impl],\scheduler}\st\lab(\ex[\last])\in\prog.\lSigma}\to \exec{\prog[\impl],\scheduler}\rst{\prog} \\
		\abs\progrst:&\set{\emptyword}\cup\set{\abs\ex\in\exec{\prog[\abs\impl],\abs\scheduler}\st\lab(\abs\ex[\last])\in\prog.\lSigma}\to \exec{\prog[\abs\impl],\abs\scheduler}\rst{\prog}
	\end{align*}
	defined by $p(\ex)=\ex\rst{\prog}$ and $\abs p(\abs\ex)=\abs\ex\rst{\prog}$
	are bijections.
	Applying \ref{claim:pi_biject} we obtain that the function
	$f:\exec{\prog[\impl],\scheduler}\rst{\prog}\to\exec{\prog[\abs\impl],\abs\scheduler}\rst{\prog}$
	defined by $f(\progex)=\abs{\progrst}(\pi(\progrst^{-1}(\progex)))$
	is a bijection.
	Moreover, using the fact that $\pi$ respects $\prog$ we obtain for every $\progex\in\exec{\prog[\impl],\scheduler}\rst{\prog}$:
	\begin{equation*}
		{f(\progex)}
		={\abs{\progrst}(\pi(\progrst^{-1}(\progex)))}
		={\pi(\progrst^{-1}(\progex))\rst{\prog}}
		={\progrst^{-1}(\progex)\rst{\prog}}
		= {\progrst(\progrst^{-1}(\progex))}={\progex}
	\end{equation*}
	
	We deduce $\exec{\prog[\impl],\scheduler}\rst{\prog}= \exec{\prog[\abs\impl],\abs\scheduler}\rst{\prog}$.
	Applying \cref{lem:executions_to_traces} this implies $\traces{\prog[\impl],\scheduler}\rst{\prog.\lSigma} = \traces{\prog[\abs\impl],\abs\scheduler}\rst{\prog.\lSigma}$.
	
	Now, we show that $\abs\scheduler$ is deterministic.
	It suffices to prove that
	for every $\abs\trans_1,\abs\trans_2\in\abs\scheduler(\abs\ex)$,
	either $\abs\trans_1=\abs\trans_2$ or $\lab(\abs\trans_1),\lab(\abs\trans_2)\in\Progint$.
	
	Applying \cref{claim:inverse_ex} to both $\abs\trans_1$ and $\abs\trans_2$ we obtain executions $\ex_1=\ex_1'\cdot\trans_1$,$\ex_2=\ex_2'\cdot\trans_2$ such that
	$\ex_1,\ex_2\in\exec{\prog[\impl],\scheduler}\setminus\set{\emptyword}$,
	$\abs\ex\cdot\abs\trans_1\pref\pi(\ex_1)$, $\abs\ex\cdot\abs\trans_2\pref\pi(\ex_2)$, $\lab(\trans_1),\lab(\trans_2)\in\prog.\lSigma$,
	and $\abs\ex\rst{\prog}=\ex'_1\rst{\prog}=\ex'_2\rst{\prog}$.
	We show that $\ex'_1=\ex'_2$.
	From $\ex'_1\rst{\prog}=\ex'_2\rst{\prog}$ we obtain
	using \cref{lem:det_rst} that \wlg $\ex'_1\pref\ex'_2$.
	Assume for contradiction $\ex'_1\prefneq\ex'_2$, and let $\trans'_2=\ex'_2[\size{\ex'_1}+1]$.
	From $\ex'_1\rst{\prog}=\ex'_2\rst{\prog}$ we get that $\lab(\trans'_2)\notin\prog.\lSigma$,
	and thus $\trans_1\neq\trans'_2$.
	Therefore from $\trans_1,\trans'_2\in\scheduler(\ex'_1)$ and from determinism of $\scheduler$
	we get that $\lab(\trans'_2)\in\Progint$, a contradiction.
	Hence it must be the case that $\ex'_1=\ex'_2$.	
	Therefore, $\trans_1,\trans_2\in\scheduler(\ex'_1)=\scheduler(\ex'_2)$.
	From determinism of $\scheduler$,
	this means either $\trans_1=\trans_2$ or $\trans_1,\trans_2\in\Progint$.
	In the first case, we obtain that $\ex_1=\ex_2$ and from $\abs\ex\cdot\abs\trans_1\pref\pi(\ex_1)$ and $\abs\ex\cdot\abs\trans_2\pref\pi(\ex_2)$
	we get that $\abs\trans_1=\abs\trans_2$.
	In the second case, from $\pi$ being $\Progint$-copying we obtain that
	$\pi(\ex_1)=\pi(\ex'_1)\cdot{\abs\trans_1}'$ and $\pi(\ex_2)=\pi(\ex'_2)\cdot{\abs\trans_2}'$ for transitions ${\abs\trans_1}',{\abs\trans_2}'$ such that $\lab({\abs\trans_1}')=\lab(\trans_1)$ and $\lab({\abs\trans_2}')=\lab(\trans_2)$.
	As $\ex'_1=\ex'_2$, the executions $\pi(\ex_1)$ and $\pi(\ex_2)$ only differ in the last transition.
	This is only possible if either $\abs\trans_1=\abs\trans_2$,
	or $\abs\ex\cdot\abs\trans_1=\pi(\ex_1)$ and $\abs\ex\cdot\abs\trans_2=\pi(\ex_2)$.
	The latter option implies $\abs\trans_1={\abs\trans_1}'$ and $\abs\trans_2={\abs\trans_2}'$,
	which implies that $\lab(\abs\trans_1)=\lab({\abs\trans_1}')=\lab(\trans_1)\in\Progint$ and $\lab(\abs\trans_2)=\lab({\abs\trans_2}')=\lab(\trans_2)\in\Progint$,
	as required.
\end{proof}

We proceed now with the converse:
assuming strong observational refinement, we construct a forward simulation.
This is done by first applying strong observational refinement to a ``most-general client'' program \mgc.
In addition to enabling arbitrary invocations of object methods, as one would expect from a most general client, the client program \mgc\ also tosses coins (called ``guesses''), that are marked by arbitrary transitions of the implementation.
Intuitively, this lets us capture any non-deterministic behavior of the implementation as a hyperproperty of \mgc.

Formally, given an implementation $\impl$ of an object $\obj$ we define
$\mgc(\impl,\obj)$ by:
\begin{itemize}
	\item $\mgc(\impl,\obj).\lQ = \set{\mgcinit}$  for an arbitrary state $\mgcinit$.
	\item $\mgc(\impl,\obj).\linit=\mgcinit$.
	\item $\mgc(\impl,\obj).\lSigma = \Inv\Res(\obj)\cup \set{\guess(\act)\st\act\in\impl.\lT}$. (\ie $\Progint=\set{\guess(\act)\st\act\in\impl.\lT}$)
	\item $\mgc(\impl,\obj).\lT=\set{\tup{\mgcinit,\sigma,\mgcinit}\st\sigma\in\mgc(\impl,\obj).\lSigma}$.
\end{itemize}

\begin{theorem}
	\label{thm:abs_comp}
	Let $\impl,\abs\impl$ be implementations of an object $\obj$.
	If $\impl \srefine \abs\impl$, then $\impl \leqsim \abs\impl$.
\end{theorem}
\begin{proof}[Proof Sketch.]
	Assume $\impl \srefine \abs\impl$ and let $\prog=\mgc(\impl,\obj)$.
	Consider the following scheduler $\scheduler$ for $\prog[\impl]$:
	\begin{equation*}
		\scheduler(\ex)=\begin{cases}
			\set{\tup{\tup{\prestate(\trans),\mgcinit},\guess(\trans),\tup{\prestate(\trans),\mgcinit}}\st \trans\in\impl.\lT} & \ex=\emptyword \lor \lab(\ex[\last])\in\impl.\lSigma\\
			\set{\tup{\tup{\prestate(\trans),\mgcinit},\lab(\trans),\tup{\poststate(\trans),\mgcinit}}} & \lab(\ex[\last]) = \guess(\trans)
		\end{cases}
	\end{equation*}
	This scheduler is clearly deterministic.
	Using this scheduler,
	we get traces of $\prog[\impl]$ where each implementation transition (including internal implementation transitions)
	is preceded by an internal program transition ``guessing'' it.
	Formally, for an execution $\ex\in\exec{\impl}$ let $\guess(\ex)$ be an execution where each transition $\ex[i]$ is replaced with a pair of transitions $\tup{\tup{\prestate(\ex[i]),\mgcinit},\guess(\ex[i]),\tup{\prestate(\ex[i]),\mgcinit}},\tup{\tup{\prestate(\ex[i]),\mgcinit},\lab(\ex[i]),\tup{\poststate(\ex[i]),\mgcinit}}$.
	It is straightforward to verify that $\guess(\ex)\in \exec{\prog[\impl],\scheduler}$.
	Now, by our assumption, there exists a deterministic scheduler $\abs\scheduler$
	such that $\traces{\prog[\impl],\scheduler}\rst{\prog.\lSigma} = \traces{\prog[\abs\impl],\abs\scheduler}\rst{\prog.\lSigma}$.
	As every trace of $\prog$ induces a unique execution of $\prog$ by adding pre-state and post-state $\mgcinit$ to each transition label, we deduce that $\exec{\prog[\impl],\scheduler}\rst{\prog}=\exec{\prog[\abs\impl],\abs\scheduler}\rst{\prog}$.

	Using \cref{thm::det_rst}, there exists a bijection
	\[\abs\progrst:\set{\emptyword}\cup\set{\abs\ex\in\exec{\prog[\abs\impl],\abs\scheduler}\st\lab(\abs\ex[\last])\in\prog.\lSigma}\to \exec{\prog[\abs\impl],\abs\scheduler}\rst{\prog}\]
	defined by $\abs\progrst(\abs\ex)=\abs\ex\rst{\prog}$, such that ${\abs\progrst}^{-1}$ is prefix preserving.
	
	We define a mapping $\pi:\exec{\impl}\to\exec{\abs\impl}$ by
	$\pi(\ex) = {\abs\progrst}^{-1}(\guess(\ex)\rst{\prog})\rst{\abs\impl}$.
	First, observe that indeed $\pi(\ex) \in \exec{\abs\impl}$ for every $\ex\in\exec{\impl}$:
	\begin{align*}
		\ex\in\exec{\impl} &\implies
		\guess(\ex)\in \exec{\prog[\impl],\scheduler} \\
		&\implies\guess(\ex)\rst{\prog}\in \exec{\prog[\impl],\scheduler}\rst{\prog}=\exec{\prog[\abs\impl],\abs\scheduler}\rst{\prog}\\
		&\implies {\abs\progrst}^{-1}(\guess(\ex)\rst{\prog})\in\exec{\prog[\abs\impl],\abs\scheduler} \subseteq \exec{\prog[\abs\impl]} \\
		&\implies {\abs \progrst}^{-1}(\guess(\ex)\rst{\prog})\rst{\abs\impl}\in\exec{\abs\impl}
	\end{align*}
	The last implication is due to \cref{prop:linking}.
	
	Moreover, $\pi$ is prefix preserving since $\guess$ is prefix preserving,
	${\abs\progrst}^{-1}$ is prefix preserving,
	and restriction to an induced execution is prefix preserving.
	Finally, we note that $\traceof{\ex}\rst{\Inv\Res(\obj)}=\traceof{\pi(\ex)}\rst{\Inv\Res(\obj)}$ (\ie $\hs{\ex}=\hs{\pi(\ex)}$) for every $\ex\in\exec{\impl}$ as $\guess$ and restrictions to LTSs with an alphabet that includes $\Inv\Res(\obj)$ are all operations that respect the restriction to $\Inv\Res(\obj)$.
	Overall we obtain a prefix preserving mapping $\pi$ that respects $\Inv\Res(\obj)$ and due to \cref{lem:sim_pp} we have $\impl \leqsim \abs\impl$.
\end{proof}

\section{Complete Implementation for the Class $\classrel{\rel}{\obj,\seqspec}$}
\label{app:comp}

In this section we prove \cref{lem:closure_rel}.
Let $\obj=(\Method,\Val)$ be an object, $\seqspec$ a specification, and $\rel$ a preorder over sequences.
First we prove downward closure:

\begin{lemma}
	$\classrel{\rel}{\obj,\seqspec}$ is downward closed \wrt forward simulation.
\end{lemma}
\begin{proof}
	Let $\impl\in\classrel{\rel}{\obj,\seqspec}$ be an implementation and let $\impl'$ be another implementation such that $\impl'\leqsim\impl$.
	From $\impl\in\classrel{\rel}{\obj,\seqspec}$ there exists a linearization mapping $\linmap:\exec{\impl}\to\seqspec$ such that $\tup{\linmap(\ex_1),\linmap(\ex_2)}\in\rel$ whenever $\ex_1\pref\ex_2$.
	From $\impl'\leqsim\impl$ and \cref{lem:sim_pp} there exists a mapping $\pi:\exec{\impl'}\to\exec{\impl}$ such that $\hs{\ex}= \hs{\pi(\ex)}$ for every $\ex\in\exec{\impl'}$
	and $\pi(\ex_1)\pref\pi(\ex_2)$ for all $\ex_1,\ex_2\in\exec{\impl'}$ such that $\ex_1\pref\ex_2$.
	
	Define a mapping $\linmap':\exec{\impl'}\to\seqspec$ by $\linmap'=\linmap\circ\pi$.
	We obtain that $\linmap'$ is a linearization mapping since $\hs{\ex}=\hs{\pi(\ex)}\linleq\linmap(\pi(\ex))=\linmap'(\ex)$ for all $\ex\in\exec{\impl'}$,
	and $\ex_1\pref\ex_2\implies\pi(\ex_1)\pref\pi(\ex_2)\implies\tup{\linmap(\pi(\ex_1)),\linmap(\pi(\ex_1)}\in\rel\implies \tup{\linmap'(\ex_1),\linmap'(\ex_2)}\in\rel$ for all $\ex_1,\ex_2\in\exec{\impl'}$ such that $\ex_1\pref\ex_2$.
\end{proof}

Next, we define an $\classrel{\rel}{\obj,\seqspec}$-complete object $\relcomp{\rel}{\obj,\seqspec}$,
as follows:

	\begin{itemize}
		\item $\relcomp{\rel}{\obj,\seqspec}.\lQ = \Inv\Res(\obj)^*\times\seqspec \times (\Tid \to \Id^*) \\ \text{\qquad \qquad \qquad \qquad } \times (\Tid\to \set{\main,\mbegin{\method}{\val},\mpend{\method},\mend{\method}{\val}\st\method\in\Method,\val\in\Val})$
		
		(An history, a valid sequential history, used operation IDs for each thread, and a local state of each thread: $\main$ denotes no pending method, $\mbegin{\method}{\val}$ denotes that a method of type $\method$ with argument $\val$ was invoked with the invocation not yet ``logged'' in the history, $\mpend{\method}$ denotes that a method of type $\method$ is pending, and  $\mend{\method}{\val}$ denotes that a method of type $\method$ with return value $\val$ had its response ``logged'' in the history and has not responded yet.)
		
		\item $\relcomp{\rel}{\obj,\seqspec}.\lSigma = \Inv\Res(\obj)\cup\set{\objact\tup{\tid}\st\tid\in\Tid}$

		\item $\relcomp{\rel}{\obj,\seqspec}.\linit = \tup{\emptyword,\emptyword,\lambda \tid \ldotp \emptyword,\lambda \tid \ldotp \main}$
		\item The transitions are as follows:
		\begin{mathparpagebreakable}
			\inferrule[\invrule]
			{\inv = \invact\tup{\method,\val,\tid,\id} \\\\
				\Tmap(\tid)=\main \\ k\notin \Tmapid(p) \\\\
				\Tmapid'=\Tmapid[\tid \mapsto \Tmapid(\tid)\cdot\id]\\\\
				\Tmap'=\Tmap[\tid \mapsto \mbegin{\method}{\val}]}
			{\tup{\h,\seqh,\Tmapid,\Tmap} \asteplab{\inv}{}
				\tup{\h,\seqh,\Tmapid', \Tmap'}}
			\and
			\inferrule[\loginv]{\inv = \invact\tup{\method,\val,\tid,\id} \\\\
				\Tmap(\tid)=\mbegin{\method}{\val} \\ \id = \Tmapid(\tid)[\last] \\\\
				\h\cdot\inv \linleq \seqh' \\ \tup{\seqh,\seqh'}\in\rel \\\\
				\h' = \h\cdot\inv \\ \Tmap'=\Tmap[\tid \mapsto \mpend{\method}]}
			{\tup{\h,\seqh,\Tmapid,\Tmap} \asteplab{\objact\tup{\tid}}{}
				\tup{\h',\seqh',\Tmapid,\Tmap'}}
			\and
			\inferrule[\logres]{
				\res = \resact\tup{\method,\val,\tid,\id} \\\\	
				\Tmap(\tid)=\mpend{\method} \\ \id = \Tmapid(\tid)[\last]  \\\\
				\h\cdot\res \linleq \seqh' \\ \tup{\seqh,\seqh'}\in\rel \\\\
				\h'=\h\cdot\res \\ \Tmap'=\Tmap[\tid \mapsto \mend{\method}{\val}]}
			{\tup{\h,\seqh,\Tmapid,\Tmap} \asteplab{\objact\tup{\tid}}{}
				\tup{\h',\seqh',\Tmapid,\Tmap'}}
			\and
			\inferrule[\resrule]{\res = \resact\tup{\method,\val,\tid,\id} \\\\
				\Tmap(\tid)=\mend{\method}{\val} \\ \id = \Tmapid(\tid)[\last] \\\\
				\Tmap'=\Tmap[\tid \mapsto \main]}
			{\tup{\h,\seqh,\Tmapid,\Tmap} \asteplab{\res}{}
				\tup{\h,\seqh,\Tmapid,\Tmap'}}
		\end{mathparpagebreakable}
	\end{itemize}

\begin{theorem}
	$\relcomp{\rel}{\obj,\seqspec}$ is $\classrel{\rel}{\obj,\seqspec}$-complete
\end{theorem}
\begin{proof}
	
	For inclusion $\relcomp{\rel}{\obj,\seqspec}\in\classrel{\rel}{\obj,\seqspec}$,
	let $\tup{\h_\ex,\seqh_\ex,\Tmapid_\ex,\Tmap_\ex}$ denote the state reached by an execution $\ex\in\exec{\relcomp{\rel}{\obj,\seqspec}}$ (\ie the initial state if $\ex=\emptyword$ and otherwise $\poststate(\ex)$).
	By transitivity of $\rel$ we have that $\tup{\seqh_{\ex_1},\seqh_{\ex_2}}\in\rel$ for all $\ex_1,\ex_2\in\exec{\relcomp{\rel}{\obj,\seqspec}}$ such that $\ex_1\pref\ex_2$.
	
	Let $\execset$ be the set of executions such that every thread is either outside a method or in a ``pending state'' in the state they reach.
	Formally: $$\execset\defeq\set{\ex\in\exec{\relcomp{\rel}{\obj,\seqspec}}\st\forall\tid\in\Tid\ldotp\Tmap_\ex(\tid)\in\set{\main}\cup\set{\mpend{\method}\st\method\in\Method.}}$$
	
	It is straightforward to verify that $\h_\ex=\hs{\ex}$ and $\h_\ex\linleq\seqh_\ex$ for all $\ex\in\execset$.
	
	Next, we claim there exists a prefix preserving mapping $g:\exec{\relcomp{\rel}{\obj,\seqspec}}\to \execset$
	such that $\hs{\ex}=\hs{g(\ex)}$ for all $\ex\in \exec{\relcomp{\rel}{\obj,\seqspec}}$.
	
	Indeed, given an execution $\ex\in\exec{\relcomp{\rel}{\obj,\seqspec}}$
	and a transition $\trans$ such that $\ex\cdot\trans\in\exec{\relcomp{\rel}{\obj,\seqspec}}$ and $\lab(\trans)\in\Inv\Res(\obj)$,
	let $\logtrans$ be a transition such that:
	\begin{itemize}
		\item If $\lab(\trans)\in\Inv(\obj)$, then $\logtrans$ is the transition using the rule \loginv\ that becomes available following the transition $\trans$ in the execution $\ex\cdot\trans$.
	
		\item If $\lab(\trans)\in\Res(\obj)$, then $\logtrans$ is the transition using the rule \logres\ such that $\trans$ becomes available following the transition $\logtrans$ in the execution $\ex\cdot\logtrans$.
	\end{itemize}
	We define $g$ inductively by $g(\emptyword)=\emptyword$ and:
	\begin{equation*}
		g(\ex\cdot \trans)=\begin{cases}
			g(\ex) & \lab(\trans)=\objact\tup{\tid} \\
			g(\ex)\cdot \trans\cdot\logtrans & \lab(\trans)\in\Inv(\obj) \\
			g(\ex)\cdot \logtrans\cdot\trans & \lab(\trans)\in\Res(\obj)
		\end{cases}
	\end{equation*}
	
	It is straightforward to verify that $g$ returns executions in $\execset$ and that it satisfies all required properties.
	
	Now,
	we define a linearization mapping $\linmap:\exec{\relcomp{\rel}{\obj,\seqspec}}\to \seqspec$ by $\linmap(\ex)=\seqh_{g(\ex)}$.
	We obtain that $\hs{\ex}=\hs{g(\ex)}=\h_{g(\ex)}\linleq\seqh_{g(\ex)}=\linmap(\ex)$ for all $\ex\in\exec{\relcomp{\rel}{\obj,\seqspec}}$,
	and that $g(\ex_1)\pref g(\ex_2)$ and thus $\tup{\linmap(\ex_1),\linmap(\ex_2)}=\tup{\seqh_{g(\ex_1)},\seqh_{g(\ex_2)}}\in\rel$ for all $\ex_1,\ex_2\in\exec{\relcomp{\rel}{\obj,\seqspec}}$ such that $\ex_1\pref\ex_2$.
	
	For hardness, let $\impl\in\classrel{\rel}{\obj,\seqspec}$ be an implementation, and $\linmap:\exec{\impl}\to\seqspec$ a linearization mapping such that $\tup{\linmap(\ex_1),\linmap(\ex_2)}\in\rel$ whenever $\ex_1\pref\ex_2$.
	
	Let $d:\exec{\impl}\to\exec{\impl}$ be a mapping taking each $\ex\in\exec{\impl}$ to its maximal prefix ending in an invocation or response transition,
	and denote $\linmap'=\linmap\circ d$.
	It is straightforward to verify that $\linmap':\exec{\impl}\to\seqspec$ is also a linearization mapping such that $\tup{\linmap'(\ex_1),\linmap'(\ex_2)}\in\rel$ whenever $\ex_1\pref\ex_2$,
	with the additional property that $\linmap(\ex)=\linmap(\ex\cdot\trans)$ whenever $\lab(\trans)\notin\Inv\Res(\obj)$.
	
	By \cref{prop:ghost} we have that $\impl\leqsim\ghost\impl$, and it suffices to prove that  $\ghost\impl\leqsim\relcomp{\rel}{\obj,\seqspec}$.
	For this we use the following simulation relation:
	\begin{equation*}
		\left\{\tup{\ex,\tup{\hs{\ex},\linmap'(\ex),\Tmapid^\ex,\Tmap^\ex}}\st\ex\in\exec{\impl}\right\}
	\end{equation*}
	where for each $\tid\in\Tid$, $\Tmapid^\ex(\tid)$ is the sequence of identifiers used for operations in $\hs{\ex}\rst{\tid}$,
	and $\Tmap^\ex(\tid)$ is $\mpend{\method}$ if $\tid$ is within a pending method $\method$ at the state reached by $\ex$, and $\main$ otherwise.
	
	A simple case analysis over $\lab(\trans)$ shows that given a pair of states $\tup{\ex,\tup{\hs{\ex},\linmap'(\ex),\Tmapid^\ex,\Tmap^\ex}}$ and a transition $\ex\asteplab{\lab(\trans)}{\ghost\impl}\ex\cdot\trans$,
	there exists a sequence of transitions\\ $\tup{\hs{\ex},\linmap'(\ex),\Tmapid^\ex,\Tmap^\ex}\bsteplab{\lab(\trans)}{\relcomp{\rel}{\obj,\seqspec}} \tup{\hs{\ex\cdot\trans},\linmap'(\ex\cdot\trans),\Tmapid^{\ex\cdot\trans},\Tmap^{\ex\cdot\trans}}$.
\end{proof}

\section{Definitions of Complete Implementations}
\label{app:ws_d_comp}

We present the formal definitions of the complete implementation's LTSs.
In both implementations, each thread's state corresponds to a location in the pseudocode,
and includes the values of shared and local variables.
Transitions represent parts of the code that are executed atomically.

We also prove that each implementation is included in the appropriate linearizability class.
Their hardness is proved in \cref{sec:hard}.

\subsection{Complete Implementation for Write Strong Registers: $\wsreg$}
\label{subsec:wslts}

Definition of $\wsreg$:

\begin{itemize}
	\item $\Pc_\ws=\set{\main,\wbegin,\wend,\rbegin,\rlisten,\rend}$ (locations in the pseudocode from \cref{wsreg})
	\item $\wsreg.\lQ = \Val \times (\Tid \to \Id^*) \times  (\Tid \to \Pc_\ws) \times (\Tid \to \powerset{\Val})  \times (\Tid \to \Val)$ (register value, used operation IDs, program counter, set of values each read collected, input/output argument of the method)
	\item $\wsreg.\lSigma = \Inv\Res(\Reg)\cup
	\set{\objact\tup{\tid}\st \tid\in\Tid}$
	\item $\wsreg.\linit = \tup{0,\lambda \tid \ldotp \emptyword,\lambda \tid \ldotp \main,\lambda \tid \ldotp \emptyset,\lambda \tid \ldotp 0}$
	\item The transitions are as follows:
	\begin{mathparpagebreakable}
		\inferrule[\readinv]
		{\act = \invact\tup{\regread,\bot,\tid,\id} \\\\
		\Tmappc(\tid)=\main \\ k\notin \Tmapid(p) \\\\
		\Tmapid'=\Tmapid[\tid \mapsto \Tmapid(\tid)\cdot\id]\\\\ \Tmappc'=\Tmappc[\tid \mapsto \rbegin]}
		{\tup{\val,\Tmapid,\Tmappc,\Tmapvalset,\Tmaparg} \asteplab{\act}{} {}\\\\ \qquad \tup{\val,\Tmapid', \Tmappc',\Tmapvalset,\Tmaparg}}
		\and
		\inferrule[\readinit]{ \Tmappc(\tid)=\rbegin \\\\
		\Tmappc'=\Tmappc[\tid \mapsto \rlisten] \\\\ \Tmapvalset'=\Tmapvalset[\tid\mapsto\set{\val}]}
		{\tup{\val,\Tmapid,\Tmappc,\Tmapvalset,\Tmaparg} \asteplab{\objact\tup{\tid}}{} {}\\\\ \qquad \tup{\val,\Tmapid,\Tmappc',\Tmapvalset',\Tmaparg}}
		\and
		\inferrule[\readloop]{\Tmappc(\tid)=\rlisten \\ \val\notin\Tmapvalset(\tid) \\\\
		\Tmapvalset'=\Tmapvalset[\tid \mapsto \Tmapvalset(\tid)\uplus\set{\val}]}
		{\tup{\val,\Tmapid,\Tmappc,\Tmapvalset,\Tmaparg} \asteplab{\objact\tup{\tid}}{} {}\\\\ \qquad \tup{\val,\Tmapid,\Tmappc,\Tmapvalset',\Tmaparg}}
		\and
		\inferrule[\readpick]{\Tmappc(\tid)=\rlisten \\  \valout\in\Tmapvalset(\tid) \\\\
		\Tmappc'=\Tmappc[\tid \mapsto \rend] \\\\
		\Tmaparg'=\Tmaparg[\tid\mapsto \valout]}
		{\tup{\val,\Tmapid,\Tmappc,\Tmapvalset,\Tmaparg} \asteplab{\objact\tup{\tid}}{} {}\\\\ \qquad \tup{\val,\Tmapid,\Tmappc',\Tmapvalset,\Tmaparg'}}
		\and
		\inferrule[\readres]{\act = \resact\tup{\regread,\Tmaparg(\tid),\tid,\id} \\\\
		\Tmappc(\tid)=\rend \\ \id = \Tmapid(\tid)[\last] \\\\
		\Tmappc'=\Tmappc[\tid \mapsto \main] \\\\
		\Tmapvalset'=\Tmapvalset[\tid\mapsto\emptyset]\\\\
		\Tmaparg'=\Tmaparg[\tid\mapsto0]}
		{\tup{\val,\Tmapid,\Tmappc,\Tmapvalset,\Tmaparg} \asteplab{\act}{} {}\\\\ \qquad \tup{\val,\Tmapid,\Tmappc',\Tmapvalset',\Tmaparg'}}
		\and
		\inferrule[\writeinv]{\act = \invact\tup{\regwrite,\valin,\tid,\id} \\\\ \Tmappc(\tid)=\main \\ k\notin \Tmapid(p)\\\\
		\Tmapid'=\Tmapid[\tid \mapsto \Tmapid(\tid)\cdot\id] \\\\
		\Tmappc' = \Tmappc[\tid \mapsto\wbegin] \\\\
		\Tmaparg'=\Tmaparg[\tid\mapsto\valin]}
		{\tup{\val,\Tmapid,\Tmappc,\Tmapvalset,\Tmaparg} \asteplab{\act}{} {}\\\\ \qquad \tup{\val,\Tmapid', \Tmappc',\Tmapvalset,\Tmaparg'}}
		\and
		\inferrule[\writeint]{\Tmappc(\tid)=\wbegin \\\\
		\Tmappc'=\Tmappc[\tid \mapsto \wend]
		\\\\ \val'=\Tmaparg(\tid)}
		{\tup{\_,\Tmapid,\Tmappc,\Tmapvalset,\Tmaparg} \asteplab{\objact\tup{\tid}}{} {}\\\\ \qquad \tup{\val',\Tmapid,\Tmappc',\Tmapvalset,\Tmaparg}}
		\and
		\inferrule[\writeres]{\act = \resact\tup{\regwrite,\bot,\tid,\id}  \\\\ \Tmappc(\tid)=\wend \\ \id = \Tmapid(\tid)[\last] \\\\
		\Tmappc'=\Tmappc[\tid \mapsto \main] \\\\
		\Tmapvalset' = \Tmapvalset[\tid\mapsto\emptyset] \\\\
		\Tmaparg' = \Tmaparg[\tid\mapsto0]}
		{\tup{\val,\Tmapid,\Tmappc,\Tmapvalset,\Tmaparg} \asteplab{\act}{} {}\\\\ \qquad \tup{\val,\Tmapid,\Tmappc',\Tmapvalset',\Tmaparg'}}
	\end{mathparpagebreakable}
\end{itemize}

We show that $\wsreg$ is indeed write strongly-linearizable.
The idea is to linearize write operations that already stored their value (in a $\writeint$ transition) and read operations that already picked a value to return (in a $\readpick$ transition). 
Intuitively, these transitions are when operations ``decide'' their relative ordering w.r.t already linearized operations.
The linearized operations are ordered by marking linearization points. 

\begin{theorem}\label{thm:wsincluded}
	$\wsreg\in\wslin$.
\end{theorem}
\begin{proof}
	\label{thm:wsreg_inclusion}
	The write strong-linearization $\linmap$ for $\wsreg$ is constructed as follows.
	Given an execution $\ex$,
	let $\writeopset(\ex)$ be the set of all write operations that performed a $\writeint$ transition in $\ex$, where pending write operations are completed with a matching response.
	Similarly, let $\readopset(\ex)$ be the set of all read operations that performed a $\readpick$ transition, where pending read operations are completed with a matching response containing the value that was picked in said transition.
	
	Let $\opset(\ex)=\writeopset(\ex)\cup\readopset(\ex)$ and define a mapping $\linpoint: \opset(\ex)\to\set{\trans\st\trans\in\ex}$ by:
	\begin{enumerate}
		\item If $\op\in\opset(\ex)$ is a write operation, then $\linpoint(\op)$ is the operation's $\writeint$ transition.
		
		\item If $\op\in\opset(\ex)$ is a read operation, then $\linpoint(\op)$ is the $\readloop$ transition where it loads the value it later picks in its $\readpick$ transition.
	\end{enumerate}
	
	We define $\linmap(\ex)$ to be the sequence over $\opset(\ex)$ induced by $\linpoint$, \ie $\op_1<_{\linmap(\ex)}\op_2$ iff $\linpoint(\op_1)<_{\ex}\linpoint(\op_2)$.
	We prove that all required properties are satisfied:
	
	\begin{itemize}
		\item $\linmap(\ex)\in \regspec$:
		
		Follows from the fact the last write (if any) before each read in $\linmap(\ex)$ is the last write that stored in a $\writeint$ transition before the read loaded the value in a $\readloop$ transition and later picked it in a $\readpick$ transition, hence that read's output in $\linmap(\ex)$ is equal to the input of said write (or the initial value, if no such write exists).
		
		\item $\hs{\ex}\linleq\linmap(\ex)$:

		We show the conditions in \cref{obs:lin} are satisfied. 
		Let $\ressq$ be a subsequence of $\linmap(\ex)$ composed of response for operations in $\linmap(\ex)$ that are pending in $\hs{\ex}$.
		Denoting $\h'=\completed(\hs{\ex}\cdot\ressq)$, we have as required:
		\begin{itemize}
			\item $\setof{\h'}=\setof{\linmap(\ex)}$:
			
			Let $\op\in\setof{\h'}$.
			If $\op\in \completed(\hs{\ex})$, then it already performed either a  $\writeint$ transition or a $\readpick$ transition, and thus $\op\in \linmap(\ex)$.
			Otherwise, $\ressq$ contains the response of $\op$ which implies by definition that $\op\in \linmap(\ex)$.
			
			For the converse, let $\op\in \linmap(\ex)$. 
			If $\op\in \completed(\hs{\ex})$, then $\completed(\hs{\ex})\subseq\completed(\hs{\ex}\cdot\ressq)$ and thus $\op\in\setof{\h'}$.
			Otherwise, the response of $\op$ is in $\ressq$ by definition,
			and so again $\op\in\setof{\h'}$.
			
			\item $<_{\h'} \;\subseteq\; <_{\linmap(\ex)}$:
			
			Let $\op_1,\op_2\in\setof{\h'}$ be operations such that $\op_1<_{\h'}\op_2$.
			Let $\res_1$ be  $\op_1$'s response and let $\inv_2$ be $\op_2$'s invocation.
			Then we have that $\linpoint(\op_1)<_{\ex}\res_1<_{\ex}\inv_2<_{\ex}\linpoint(\op_2)$ 
			and so $\op_1<_{\linmap(\ex)}\op_2$ as required.
			
		\end{itemize}
		
		\item $\linmap$ is write strong:
		
		Let $\ex_1,\ex_2\in\exec{\impl}$ be executions such that $\ex_1\pref\ex_2$.
		Then $\linmap(\ex_1)\rst{\regwrite}\pref\linmap(\ex_2)\rst{\regwrite}$
		follows from the fact that writes in $\linmap(\ex_2)\rst{\regwrite}$ that are not in $\linmap(\ex_1)\rst{\regwrite}$
		all perform their $\writeint$ transition after all transitions in $\ex_1$ are executed. \qedhere
	\end{itemize}
\end{proof}

\subsection{Complete Implementation for Decisive Registers: $\dreg$}
\label{subsec:dlts}

Definition of $\awsreg$:

\begin{itemize}
	\item $\Pc_\aws = \set{\main,\wbegin,\wwait,\wrollback,\wend,\rbegin,\rlisten,\rend}$
	\item $\awsreg.\lQ = \Val \times\N\times \set{0,1}\times (\Tid \to \Id^*) \times (\Tid \to \Pc_\aws) \times (\Tid \to \powerset{\Val}) \times (\Tid \to \N)\times (\Tid \to \Val)$ (Register value, version, lock, used operation IDs, program counter, set of values, input/output argument of the method (also used to hold a temporary value))
	\item $\awsreg.\lSigma = \Inv\Res(\Reg)\cup
	\set{\objact\tup{\tid}\st \tid\in\Tid}$
	\item $\awsreg.\linit = \tup{0,0,0,\lambda \tid \ldotp \emptyword,\lambda \tid \ldotp \main,\lambda \tid \ldotp \emptyset,\lambda \tid \ldotp 0,\lambda \tid \ldotp 0}$
	\item The transitions are as follows:
	\begin{mathparpagebreakable}
		\inferrule[\readinv]
		{\act = \invact\tup{\regread,\bot,\tid,\id} \\\\  \Tmappc(\tid)=\main \\ \id\notin \Tmapid(\tid) \\\\
		\Tmappc'=\Tmappc[\tid \mapsto \rbegin] \\\\
		\Tmapid'=\Tmapid[\tid \mapsto \Tmapid(\tid)\cdot\id]}
		{\tup{\val,\ver,\lockLTS,\Tmapid,\Tmappc,\Tmapvalset,\Tmapstart,\Tmaparg} \asteplab{\act}{} {}\\\\ \qquad \tup{\val,\ver,\lockLTS,\Tmapid', \Tmappc',\Tmapvalset,\Tmapstart,\Tmaparg}}
		\and
		\inferrule[\readinit]{\Tmappc(\tid)=\rbegin \\ \lockLTS = 0 \\\\
		\Tmappc'=\Tmappc[\tid \mapsto \rlisten] \\\\
		\Tmapvalset'=\Tmapvalset[\tid\mapsto\set{\val}] \\\\
		\Tmapstart'=\Tmapstart[\tid\mapsto\ver]}
		{\tup{\val,\ver,\lockLTS,\Tmapid,\Tmappc,\Tmapvalset,\Tmapstart,\Tmaparg} \asteplab{\objact\tup{\tid}}{} {}\\\\ \qquad \tup{\val,\ver,\lockLTS,\Tmapid,\Tmappc',\Tmapvalset',\Tmapstart',\Tmaparg}}
		\and
		\inferrule[\readloop]{ \Tmappc(\tid)=\rlisten \\ \val\notin\Tmapvalset(\tid) \\\\
		\ver\geq\Tmapstart(\tid) \\\\
		\Tmapvalset'=\Tmapvalset[\tid \mapsto \Tmapvalset(\tid)\uplus\set{\val}]}
		{\tup{\val,\ver,\lockLTS,\Tmapid,\Tmappc,\Tmapvalset,\Tmapstart,\Tmaparg} \asteplab{\objact\tup{\tid}}{} {}\\\\ \qquad \tup{\val,\ver,\lockLTS,\Tmapid,\Tmappc,\Tmapvalset',\Tmapstart,\Tmaparg}}
		\and
		\inferrule[\readpick]{ \Tmappc(\tid)=\rlisten \\ \valout\in\Tmapvalset(\tid) \\\\
		\Tmappc'=\Tmappc[\tid \mapsto \rend] \\\\
		\Tmaparg'=\Tmaparg[\tid\mapsto \valout]}
		{\tup{\val,\ver,\lockLTS,\Tmapid,\Tmappc,\Tmapvalset,\Tmapstart,\Tmaparg} \asteplab{\objact\tup{\tid}}{} {}\\\\ \qquad \tup{\val,\ver,\lockLTS,\Tmapid,\Tmappc',\Tmapvalset,\Tmapstart,\Tmaparg'}}
		\and
		\inferrule[\readres]{\act = \resact\tup{\regread,\Tmaparg(\tid),\tid,\id} \\\\ \Tmappc(\tid)=\rend \\ \id = \Tmapid(\tid)[\last] \\\\
		\Tmappc'=\Tmappc[\tid \mapsto \main] \\\\
		\Tmapvalset'=\Tmapvalset[\tid\mapsto\emptyset] \\\\
		\Tmapstart' = \Tmapstart[\tid\mapsto 0] \\\\
		\Tmaparg'=\Tmaparg[\tid\mapsto 0]}
		{\tup{\val,\ver,\lockLTS,\Tmapid,\Tmappc,\Tmapvalset,\Tmapstart,\Tmaparg} \asteplab{\act}{} {}\\\\ \qquad \tup{\val,\ver,\lockLTS,\Tmapid,\Tmappc',\Tmapvalset',\Tmapstart',\Tmaparg'}}
		\and
		\inferrule[\writeinv]{\act = \invact\tup{\regwrite,\valin,\tid,\id} \\\\ \Tmappc(\tid)=\main \\ k\notin \Tmapid(\tid) \\\\
		\Tmapid' = \Tmapid[\tid \mapsto \Tmapid(\tid)\cdot\id] \\\\
		\Tmappc'=\Tmappc[\tid \mapsto\wbegin] \\\\
		\Tmaparg'=\Tmaparg[\tid\mapsto\valin]}
		{\tup{\val,\ver,\lockLTS,\Tmapid,\Tmappc,\Tmapvalset,\Tmapstart,\Tmaparg} \asteplab{\act}{} {}\\\\ \qquad \tup{\val,\ver,\lockLTS,\Tmapid', \Tmappc', \Tmapvalset,\Tmapstart,\Tmaparg'}}
		\and
		\inferrule[\writeinit]{ \Tmappc(\tid)=\wbegin \\ \lockLTS=0 \\\\
		\Tmappc'=\Tmappc[\tid \mapsto \wwait] \\\\
		\Tmapstart' = \Tmapstart[\tid\mapsto\ver]}
		{\tup{\val,\ver,\lockLTS,\Tmapid,\Tmappc,\Tmapvalset,\Tmapstart,\Tmaparg} \asteplab{\objact\tup{\tid}}{} \\\\ \qquad \tup{\val,\ver,\lockLTS,\Tmapid,\Tmappc',\Tmapvalset,\Tmapstart',\Tmaparg}}
		\and
		\inferrule[\writetop]{ \Tmappc(\tid)=\wwait \\ \lockLTS = 0 \\\\
		\val'=\Tmaparg(\tid)\\	\ver'=\ver+1\\\\
		\Tmappc'=\Tmappc[\tid \mapsto \wend]}
		{\tup{\_,\ver,\lockLTS,\Tmapid,\Tmappc,\Tmapvalset,\Tmapstart,\Tmaparg} \asteplab{\objact\tup{\tid}}{} \\\\ \qquad \tup{\val',\ver',\lockLTS,\Tmapid,\Tmappc',\Tmapvalset,\Tmapstart,\Tmaparg}}
		\and
		\inferrule[\writerollback]{ \Tmappc(\tid)=\wwait \\ \ver>\Tmapstart(\tid) \\ \lockLTS = 0 \\\\
		\val'=\Tmaparg(\tid)\\	\ver'=\ver-1\\\\
		\lockLTS'=1 \\ \Tmappc'=\Tmappc[\tid \mapsto \wrollback] \\\\
		\Tmaparg'=\Tmaparg[\tid\mapsto\val]}
		{\tup{\val,\ver,\lockLTS,\Tmapid,\Tmappc,\Tmapvalset,\Tmapstart,\Tmaparg} \asteplab{\objact\tup{\tid}}{} \\\\ \qquad \tup{\val',\ver',\lockLTS',\Tmapid,\Tmappc',\Tmapvalset,\Tmapstart,\Tmaparg'}}
		\and
		\inferrule[\writerollforward]{ \Tmappc(\tid)=\wrollback \\	\lockLTS = 1 \\\\
		\val'=\Tmaparg(\tid)\\	\ver'=\ver+1\\\\
		\lockLTS'=0 \\	\Tmappc'=\Tmappc[\tid \mapsto \wend]}
		{\tup{\_,\ver,\lockLTS,\Tmapid,\Tmappc,\Tmapvalset,\Tmapstart,\Tmaparg} \asteplab{\objact\tup{\tid}}{} \\\\ \qquad \tup{\val',\ver',\lockLTS',\Tmapid,\Tmappc',\Tmapvalset,\Tmapstart,\Tmaparg}}
		\and
		\inferrule[\writeres]{\act = \resact\tup{\regwrite,\bot,\tid,\id}  \\\\ \Tmappc(\tid)=\wend \\ \id = \Tmapid(\tid)[\last] \\\\
		\Tmappc'=	\Tmappc[\tid \mapsto \main] \\\\
		\Tmapvalset'=\Tmapvalset[\tid\mapsto\emptyset] \\\\
		\Tmapstart'=\Tmapstart[\tid\mapsto 0] \\\\
		\Tmaparg'=\Tmaparg[\tid\mapsto 0]}
		{\tup{\val,\ver,\lockLTS,\Tmapid,\Tmappc,\Tmapvalset,\Tmapstart,\Tmaparg} \asteplab{\act}{} \\\\ \qquad \tup{\val,\ver,\lockLTS,\Tmapid,\Tmappc',\Tmapvalset',\Tmapstart',\Tmaparg'}}
	\end{mathparpagebreakable}
\end{itemize}

We show that $\dreg$ is indeed decisively linearizable.
The set of operations to linearize is chosen similarly to the proof of \cref{thm:wsincluded}: writes that already stored their value and reads that already picked a value to return.
The ordering of these operations is more involved in this case, 
as it is less clear how to assign linearization points. 
Specifically, for writes that choose to be ``overwritten'' by a concurrent write, their linearization point should be somewhere between their invocation and the linearization point of the write that overwrote them,
but it is not immediately clear how to choose a specific transition within that interval.
Instead of resolving this question directly, we order the operations based on a version number associated with each, and use the relative order of specific transitions only to resolve version number ties.

\begin{theorem}
	$\dreg\in\dlin$.
\end{theorem}
\begin{proof}
	The decisive linearization $\linmap$ for $\wsreg$ is constructed as follows.
	Given an execution $\ex$,
	let $\writeopset(\ex)$ be the set of all write operations that performed
	either a $\writetop$ or a $\writerollback$ transition in $\ex$, where pending write operations are completed with a matching response.
	Let $\readopset(\ex)$ be the set of all read operations that performed a $\readpick$ transition, where pending read operations are completed with a matching response containing the value they picked in said transition.
	
	Denote $\opset(\ex)=\writeopset(\ex)\cup\readopset(\ex)$, and define mappings $\rellinpoint:\opset(\ex)\to\ex$, $\verf:\opset(\ex)\to\N$, by:
	\begin{enumerate}
		\item If $\op\in\opset(\ex)$ is a write operation, then $\rellinpoint(\op)$ is either the $\writetop$ or the $\writerollback$ transition performed by $\op$ in $\ex$,
		and $\verf(\op)$ is the version number stored to $\ver$ in said transition.
		
		\item If $\op\in\opset(\ex)$ is a read operation, then $\rellinpoint(\op)$ is the $\readloop$ transition where it loads the value it later picks in its $\readpick$ transition,
		and $\verf(\op)$ is the value of $\ver$ in the state when said $\readloop$ transition is performed.
	\end{enumerate}
	
	We define $\linmap(\ex)$ to be the sequence over $\opset(\ex)$ induced by $\verf$, using $\rellinpoint$ for tie-breaking.
	\ie $\op_1<_{\linmap(\ex)}\op_2$ if and only if either $\verf(\op_1)<\verf(\op_2)$, or $\verf(\op_1)=\verf(\op_2)$ and 
	$\rellinpoint(\op_1)<_{\ex}\rellinpoint(\op_2)$.
	
	Note that 
	the leftmost operation in $\linmap(\ex)$ associated with any version number greater than 0 is the write operation that stored this number in a $\writetop$ transition, and all other writes associated with the same version number are writes that stored this number later in $\ex$ in a $\writerollback$ transition.
	In particular, all write operations are mapped to a version number strictly greater than $0$.

	We prove that all required properties are satisfied:
	
	\begin{itemize}
		\item $\linmap(\ex)\in \regspec$:
		
		Let $\rop\in\linmap(\ex)$ be a read operation.
		We need to show
		$\outputf(\rop)=\obs(\prefof{\linmap(\ex)}{\op})$.
		
		Equivalently, defining $\wop_1=\lastop(\prefof{\linmap(\ex)}{\rop}\rst{\regwrite})$  (\ie the last write operation before $\rop$ in $\linmap(\ex)$, if such exists, and otherwise $\emptyword$),
		we need to show that $\outputf(\rop)=\inputf(\wop_1)$ if $\wop_1\neq\emptyword$ and $\outputf(\rop)=0$ if $\wop_1=\emptyword$.
		
		Let $\wop_2$ be a write operation such that $\rellinpoint(\wop_2)<_{\ex}\rellinpoint(\rop)$,
		$\rellinpoint(\wop_2)$ is either a $\writetop$ transition or a $\writerollback$ transition with no matching $\writerollforward$ transition before the transition $\rellinpoint(\rop)$,
		and there is no write operation $\wop'$ satisfying this condition such that $\rellinpoint(\wop_2)<_{\ex}\rellinpoint(\wop')$,
		or $\wop_2=\emptyword$ if no such write exists.
		
		Consider two cases:
		\begin{itemize}
			\item If $\wop_2\neq\emptyword$:
			
			By definition of $\wop_2$ we clearly have $\rellinpoint(\wop_2)<_{\ex}\rellinpoint(\rop)$, $\verf(\wop_2)=\verf(\rop)$, and
			$\outputf(\rop)=\inputf(\wop_2)$.
			Therefore, $\wop_2<_{\linmap(\ex)}\rop$.
			By definition of $\wop_1$ this implies that $\wop_1\neq\emptyword$ and in particular $\wop_2\leq_{\linmap(\ex)}\wop_1$.
			As $\wop_1<_{\linmap(\ex)}\rop$ and $\verf(\wop_2)=\verf(\rop)$ we obtain that $\verf(\wop_1)=\verf(\rop)$
			and so $\rellinpoint(\wop_1)<_{\ex}\rellinpoint(\rop)$.
			Consider two cases:
			\begin{itemize}
				\item If $\rellinpoint(\wop_1)$ is a $\writerollback$ transition with a matching $\writerollforward$ transition before $\rellinpoint(\rop)$ in $\ex$,
				that $\writerollforward$ transition's post-state has a version number larger than $\verf(\wop_1)=\verf(\rop)$, 
				which implies $\rop$ loads in the transition $\rellinpoint(\rop)$ a value written later by an additional $\writerollback$ transition of a write operation $\wop'$ with no matching response before $\rellinpoint(\rop)$. 
				We deduce from definition of $\wop_2$ that $\wop'=\wop_2$ and as $\rellinpoint(\wop_1)<_{\ex}\rellinpoint(\wop')$ we obtain that $\rellinpoint(\wop_1)<_{\ex}\rellinpoint(\wop_2)$.
				
				\item Otherwise, $\rellinpoint(\wop_1)$ is either a $\writetop$ transition, or a $\writerollback$ transition with no matching $\writerollforward$ transition before $\rellinpoint(\rop)$ in $\ex$.
				In both cases we deduce from the definition of $\wop_2$ that necessarily $\rellinpoint(\wop_1)\leq_{\ex}\rellinpoint(\wop_2)$.
			\end{itemize}
			
			We obtain that $\rellinpoint(\wop_1)\leq_{\ex}\rellinpoint(\wop_2)$ and $\verf(\wop_1)=\verf(\wop_2)$ and thus $\wop_1\leq_{\linmap(\ex)}\wop_2$ and overall $\wop_1=\wop_2$.
			This gives as required that $\outputf(\rop)=\inputf(\wop_2)=\inputf(\wop_1)$.
		
			\item If $\wop_2=\emptyword$, then clearly $\outputf(\rop)=0$, $\verf(\rop)=0$,
			and it suffices to show $\wop_1=\emptyword$.
			Indeed, if we assume otherwise then from $\wop_1<_{\linmap(\ex)}\rop$ we obtain that $\verf(\wop_1)=0$,
			contradicting the fact that all write operations are assigned a positive version number.

		\end{itemize}

		\item $\hs{\ex}\linleq\linmap(\ex)$:
		
		We show the conditions in \cref{obs:lin} are satisfied. 
		Let $\ressq$ be a subsequence of $\linmap(\ex)$ composed of response for operations in $\linmap(\ex)$ that are pending in $\hs{\ex}$.
		Denoting $\h'=\completed(\hs{\ex}\cdot\ressq)$, 
		the fact that $\setof{\h'}=\setof{\linmap(\ex)}$ is proved exactly like in the proof of \cref{thm:wsincluded},
		and it remains to show that $<_{\h'} \;\subseteq\; <_{\linmap(\ex)}$.
		Let $\op_1,\op_2\in\setof{\h'}$ be operations such that $\op_1<_{\h'}\op_2$.
		First, we show that $\verf(\op_1)\leq\verf(\op_2)$.
		Let $\trans_2$ be the transition where $\op_2$ performs either a \readinit\ or \writeinit\ transition.
		This is well defined as $\op_2\in\opset(\ex)$ implies that $\op_2$ already performed in $\ex$ the transition $\rellinpoint(\op_2)$,
		and we have $\trans_2<_{\ex}\rellinpoint(\op_2)$.
		Since $\op_1$ responds before $\op_2$ is invoked in $\ex$,
		we must have that $\rellinpoint(\op_1)<_{\ex}\trans_2$.
		As $\trans_2$ is performed under a precondition that $\lockLTS=0$,
		there exists a transition $\trans$ such that $\rellinpoint(\op_1)\leq_{\ex}\trans<_{\ex}\trans_2$ and
		$\lockLTS=0$ in the post-state of $\trans$.
		Let $\trans_1$ be the first (\wrt $<_{\ex}$) such transition.
		Let $\ver_i$ denote the version number in the post-state of transition $\trans_i$.
		We have:
		\begin{itemize}
			\item $\verf(\op_1)\leq\ver_1$:
			
			If $\rellinpoint(\op_1)=\trans_1$, we are done.
			Otherwise, the transition $\rellinpoint(\op_1)$ is performed after some write performed its $\writerollback$ transition (possibly $\rellinpoint(\op_1)$ is exactly this transition), but before that write performed its $\writerollforward$ transition.
			It must be the case that $\trans_1$ is said $\writerollforward$ transition, as no other operation can update the lock from $1$ to $0$.
			As no transition prior to $\trans_1$ can update the version number,
			and it increase by 1 after $\trans_1$, we obtain that $\verf(\op_1)<\ver_1$.
				
			\item $\ver_1\leq\ver_2$:
			
			Follows from the fact that $\trans_1<_{\ex}\trans_2$, the version number is monotonically non-decreasing when considering points in the execution where $\lockLTS=0$, and indeed $\lockLTS=0$ in the post-states of $\trans_1$ and $\trans_2$.

			\item $\ver_2\leq\verf(\op_2) $:
			
			Consider the following cases:
			\begin{itemize}
				\item If $\rellinpoint(\op_2)$ is a $\writetop$ transition, this follows from $\trans_2<_{\ex}\rellinpoint(\op_2)$, using again the fact that the version number is monotonically non-decreasing when considering points in the execution where $\lockLTS=0$, and indeed $\lockLTS=0$ in the post-states of $\trans_2$ and $\rellinpoint(\op_2)$.
				
				\item If $\rellinpoint(\op_2)$ is a $\writerollback$ transition,
				then $\verf(\op_2)$ is smaller by 1 then the version number $\ver$ in pre-state of $\rellinpoint(\op_2)$,
				and this transition is done under a condition that $\ver$ is strictly larger than the version number $\op_2$ read in the \writeinit\ transition, which is $\ver_2$.

				\item  If $\rellinpoint(\op_2)$ is a $\readloop$ transition, then $\verf(\op_2)$ is loaded under a condition that indeed it is larger than or equal to the version number $\op_2$ read in the \readinit\ transition, which is $\ver_2$.
			\end{itemize}
		\end{itemize}
		
		We obtain that $\verf(\op_1)\leq\verf(\op_2)$.
		Next we prove that $\op_1<_{\linmap(\ex)}\op_2$. Consider the following cases:
		\begin{itemize}
			\item If $\verf(\op_1)<\verf(\op_2)$, then by definition $\op_1<_{\linmap(\ex)}\op_2$ as required.
			
			\item If $\verf(\op_1)=\verf(\op_2)$, 
			let $\res_1$ be  $\op_1$'s response and let $\inv_2$ be $\op_2$'s invocation.
			Then we have that $\rellinpoint(\op_1)<_{\ex}\res_1<_{\ex}\inv_2<_{\ex}\rellinpoint(\op_2)$ 
			and so overall $\op_1<_{\linmap(\ex)}\op_2$ as required.
		\end{itemize}
		
		\item $\linmap$ is decisive:
		
		Let $\ex_1,\ex_2\in\exec{\impl}$ be executions such that $\ex_1\pref\ex_2$.
		We clearly have that $\opset(\ex_1)\subseteq\opset(\ex_2)$ as additional transitions can only add operations to $\opset$.
		Moreover, the ordering between operations in $\opset(\ex_1)$ is not effected by the additional operations added,
		and thus, $\linmap(\ex_1)\subseq\linmap(\ex_2)$.\qedhere
	\end{itemize}
\end{proof}

\section{Hardness of Complete Implementations}
\label{sec:hard}

In this section we prove the hardness of $\wsreg$ and $\dreg$ for their perspective linearizability classes.
We being in \cref{subsec:hardness_motivation} with a motivating discussion.
In \cref{subsec:lazy_def} we formally define ``lazy linearizations'': linearizations with additional properties that are useful for the hardness proof, presented in \cref{subsec:hard_proof}.

\subsection{Motivation and Proof Strategy}\label{subsec:hardness_motivation}

To motivate our approach, we begin by sketching a proof attempt for the $\wslin$-hardness of $\wsreg$.

Given a write strong implementation $\impl$, we wish to prove that $\impl\leqsim\wsreg$.
Let $\linmap:\exec{\impl}\to\regspec$ be a write strong-linearization for $\impl$.
We begin by attaching to the state of $\impl$ a ghost variable keeping the execution performed so far (see \cref{exrecorder} for more details).
The resulting implementation $\ghost{\impl}$ is simulation-equivalent to $\impl$ and so it is sufficient to show $\ghost{\impl}\leqsim\wsreg$.
While we have no direct information about the structure of $\impl$'s state,
given a state up to which an execution $\ex$ was performed, and a transition $\trans$, we can look at the linearizations $\linmap(\ex)$ and $\linmap(\ex\cdot\trans)$, and try to learn the \quotes{effect} of the transition $\trans$ from the difference between them.

A first attempt to simulate $\ghost{\impl}$ is to match the invocations and responses, and perform in $\wsreg$ the internal step of writes in response to a transition $\trans$ after which they first appear in the linearization $\linmap(\ex\cdot\trans)$.
For reads, we can load the currently stored value every time it is changed to a value not loaded in previous transitions, in the hope that when the reads appear in the linearization, the return value described in that linearization is included in the set of values loaded by the read method, which would allow us to simulate the response of that read.

This attempt fails. To show this, consider the following register implementation:

\begin{minipage}{\textwidth}
	\begin{algorithm}[H]
		\textbf{Shared Variables:}
		the current register value $\reg$, a counter $\cnt_1$ for tickets taken by writers, and a counter $\cnt_2$ for executed writes.
		\setlength{\columnsep}{-1cm}
		\begin{multicols}{2}
		\METHOD{\Read{}}{
			$\retval\gets\reg$\;
			\Return{$\retval$}\;
		}
		
		\METHOD{\Write{$\val$}}{
			$\tup{\queueturn,\cnt_1}\gets\tup{\cnt_1,\cnt_1+1}$\;
			\Wait{$\queueturn=\cnt_2$}\;
			$\tup{\reg,\cnt_2}\gets\tup{\val,\cnt_2+1}$\;
			\Return{}\;
		}
		\caption{\!\! Ticket Protocol}
		\end{multicols}
		\vspace*{5pt}
	\end{algorithm}
\end{minipage}

In this implementation, each write takes a ticket from an increasing shared counter $\cnt_1$, and later stores its input value when a second shared counter $\cnt_2$ reaches its ticket number (initially $\cnt_1=\cnt_2=0$).

This implementation is write strong. A possible write strong linearization mapping for it includes in the linearization all writes that \quotes{took a ticket}, ordered according to their ticket number, and linearizes reads after the write which stored the value they loaded.

For the simulation attempt described above, when a write is invoked and takes a ticket, it appears in the linearization and thus we do an internal write step in $\wsreg$.
As the write only took a ticket and did not store yet, it is then possible for a newly invoked read to load an older value, and we will not be able to simulate this with a matching newly invoked read in $\wsreg$.
The problem here is that writes were added to the linearization mapping before their effect was ``final'' in the sense of preventing reads from loading older values.

Our solution to this issue is to base the simulation on what we call a \emph{lazy} linearization mapping. Informally, this is a linearization mapping that only puts operations in the linearization of some execution when it \emph{must}, either because the operation returned, or since it is necessary for validity of the sequential history (i.e. a write before a read returning its value), or for satisfying the prefix condition of write strong linearizability w.r.t linearizations of shorter executions.

When using a lazy linearization, the sequence of linearized write operations describes more accurately what values invoked read operations might return in the future,
and the simulation sketch described above can be proven correct.

\subsection{Lazy Linearizations}\label{subsec:lazy_def}

We define the following properties:

\begin{definition2}
	Let $\impl$ be an implementation, $\seqspec$ be a specification,
	and $\linmap:\exec{\impl}\to\seqspec$ be a linearization mapping.
	$\linmap$ is called:\footnote{The last two properties
		only apply when $\impl$ is a register implementation and $\seqspec=\regspec$.}
	\begin{itemize}
		\item \emph{response-activated} if $\linmap(\ex)=\linmap(\ex\cdot \trans)$ for every $\ex\in\exec{\impl}$ and $\trans\in\impl.\lT$ such that $\ex\cdot\trans\in\exec{\impl}$ and $\trans\notin\Res$.
		(The linearization remains the same when appending a non-response transition to an execution.)
		
		\item \emph{last-completed} if $\lastop(\linmap(\ex))\in\setof{\completed(\hs{\ex})}$  for every $\ex\in\exec{\impl} \setminus \set{\emptyword}$. (The last operation in any non-empty linearization is completed in the linearized execution.)
		
		\item \emph{pending-read-free} if $\setof{\linmap(\ex)\rst{\regread}}\subseteq\setof{\completed(\hs{\ex})}$ for every $\ex\in\exec{\impl}$. (Every linearized read operation is completed in the linearized execution.)
		
			\item \emph{lazy} if $\linmap$ is decisive, response-activated, last-completed, and pending-read-free.
				
	\end{itemize}
\end{definition2}

The idea behind the last-completed and pending-read-free properties is that a linearization that is ``lazy'' will not end with a pending operation,
or contain a pending read, as a sequence where such an operation is removed is also a valid linearization of the same execution.

The idea behind the response-activated and decisive properties is that a lazy linearization does not ``switch around'' operations once they are linearized, and in particular keeps the exact same linearization when that is possible.

In \cref{sec:lazy}, we prove the following theorems:

\begin{theorem}\label{lem:wsnice}
	If $\impl\in\wslin$, then there exists a lazy write strong-linearization mapping $\linmap:\exec{\impl}\to\regspec$.
\end{theorem}

\begin{theorem}\label{lem:awsnice}
	If $\impl\in\dlin$, then there exists a lazy linearization mapping $\linmap:\exec{\impl}\to\regspec$.
\end{theorem}

Note that a corollary of \cref{lem:wsnice} is that every write strong-linearizable implementation is also decisively linearizable.

The following theorem presents the main useful property of lazy linearizations:
for every read operation in a linearization,
its return value is the input of some write that was added to the linearization after the read was invoked, or was the latest write in the linearization at the time of invocation.
(If no such write exists, the initial value is possibly returned.)

To formalize this claim, we first define the set of values ``observable between'' two sequential histories.

\begin{definition2}
	The set of values \emph{observable between} two sequential histories $\seqh_1,\seqh_2\in\regspec$ such that $\seqh_1\subseq\seqh_2$,
	denoted by $\obsbetween(\seqh_1,\seqh_2)$, is defined as follows:
	\begin{equation*}
		\obsbetween(\seqh_1,\seqh_2)=\set{\obs(\seqh_1)}\cup
		\set{\inputf(\wop)\st \wop\in\seqh_2\rst{\regwrite}\land\forall \hat{\wop}\in\seqh_1\rst{\regwrite}\ldotp \hat{\wop}<_{\seqh_2}\wop}
	\end{equation*}
\end{definition2}

The idea is that if $\seqh_1=\linmap(\ex_1)$ is the lazy linearization of an execution $\ex_1$ up to the invocation of some read,
and $\seqh_2=\linmap(\ex_2)$ where $\ex_2$ a continuation of $\ex_1$ up to a response of that read,
then that read always returns a value in $\obsbetween(\seqh_1,\seqh_2)$.
Formally, we prove the following:

\begin{theorem}\label{thm:lazyisgood}
	If $\impl$ is a register implementation, $\linmap:\exec{\impl}\to\regspec$ is a decisive and last-completed linearization mapping,
	$\ex\in\exec{\impl}$ is an execution,
	$\rop\in\linmap(\ex)\rst{\regread}$ is a read operation,
	and $\ex'$ is the prefix of $\ex$ up to the invocation of $\rop$ (exclusive),
	then $\outputf(\rop)\in\obsbetween(\linmap(\ex'),\linmap(\ex))$.
\end{theorem}

In other words,
either the read operation $\rop$ returns the value that can be read after $\linmap(\ex')$, the linearization of the execution up to its invocation,
or it returns the input of some write in $\linmap(\ex)$ which appears after all writes which were already present in $\linmap(\ex')$.

\begin{proof}
	By \cref{def:regspec} we have $\outputf(\rop)=\obs(\prefof{\linmap(\ex)}{\rop})$.
	Thus it is sufficient to assume
	\begin{equation}\label{eq:lazyisgood_implication}
		\obs(\prefof{\linmap(\ex)}{\rop})\notin\set{\inputf(\wop)\st \wop\in\linmap(\ex)\rst{\regwrite}\land\forall \hat{\wop}\in\linmap(\ex')\rst{\regwrite}\ldotp \hat{\wop}<_{\linmap(\ex)}\wop}
	\end{equation}
	and prove that $\obs(\prefof{\linmap(\ex)}{\rop})=\obs(\linmap(\ex'))$.
	
	Denote $\wop=\lastop(\prefof{\linmap(\ex)}{\rop}\rst{\regwrite})$, $\wop'=\lastop(\linmap(\ex')\rst{\regwrite})$,
	and $\op'=\lastop(\linmap(\ex'))$.
	
	By definition of $\obs$ it is sufficient to show that $\wop=\wop'$. (Possibly both are $\emptyword$, if the relevant sequences contain no writes.)
	Consider the following cases:
	\begin{itemize}
		\item If $\wop\neq\emptyword$, then
		$\wop\in\prefof{\linmap(\ex)}{\rop}$ and we obtain that $\wop\in\linmap(\ex)$. By definition, $\obs(\prefof{\linmap(\ex)}{\rop})=\inputf(\wop)$.
		Thus, from assumption (\ref{eq:lazyisgood_implication}) we obtain that there must exist some $\hat{\wop}\in\linmap(\ex')$ such that $\wop\leq_{\linmap(\ex)}\hat{\wop}$.
		As $\hat{\wop}\in\linmap(\ex')$,
		using the definition of $\wop'$ we obtain $\hat{\wop}\leq_{\linmap(\ex')}\wop'$.
		Since $\linmap$ is decisive, $\linmap(\ex')\subseq\linmap(\ex)$
		and we obtain that $\hat{\wop}\leq_{\linmap(\ex)}\wop'$ and overall $\wop\leq_{\linmap(\ex)}\wop'$.
		We also obtain in particular that $\wop',\op'\neq\emptyword$, and so by definition we have
		$\wop'\leq_{\linmap(\ex')}\op'$, which similarly to the above implies $\wop'\leq_{\linmap(\ex)}\op'$.
		As $\op'=\lastop(\linmap(\ex'))$ and $\linmap$ is last-completed, we must have $\op'\in\setof{\completed(\hs{\ex'})}$.
		That is, $\op'$ is completed in $\ex'$ and as $\rop$ is not yet invoked in $\ex'$ we obtain that $\op'<_{\hs{\ex}}\rop$ and thus $\op'<_{\linmap(\ex)}\rop$.
		We deduce that $\wop'\leq_{\linmap(\ex)}\rop$ and so by definition of $\wop$ we obtain that $\wop'\leq_{\linmap(\ex)}\wop$.
		As we have already shown $\wop\leq_{\linmap(\ex)}\wop'$ we deduce that $\wop=\wop'$.

		\item If $\wop=\emptyword$, then to show $\wop'=\emptyword$ it is sufficient to show that $\linmap(\ex')\rst\regwrite=\emptyword$.
		Indeed, if we assume otherwise, there exists some $\hat{\wop}\in\linmap(\ex')$.
		Like we did in the previous case, we can deduce that $\hat{\wop}\leq_{\linmap(\ex)}\wop'$ and thus $\wop',\op'\neq\emptyword$.
		From that we can deduce exactly like above that $\wop'\leq_{\linmap(\ex)}\rop$, contradicting the fact that $\wop=\emptyword$ and so no write operation appears before $\rop$ in $\linmap(\ex)$.
		We obtain that $\wop'=\emptyword=\wop$, as required. \qedhere
	\end{itemize}
\end{proof}

\subsection{Hardness Proof}\label{subsec:hard_proof}

We begin by proving the $\awslin$-hardness of $\dreg$.
Following this proof, we explain how the proof is modified to demonstrate $\wslin$-hardness of $\wsreg$.

\begin{theorem}\label{thm:awshard}
	$\awsreg$ is $\awslin$-hard.
\end{theorem}

\begin{proof}
	Let $\impl\in\awslin$ be a decisively linearizable register.
	By \cref{lem:awsnice}, there exists a lazy linearization mapping
	$\linmap:\exec{\impl}\to\regspec$.
	Using \cref{prop:ghost}, it suffices to prove $\ghost\impl\leqsim\awsreg$.
	
	Given an execution $\ex\in\exec{\impl}$ and a thread $\tid\in\Tid$, we write $\inv^{\ex}_\tid$ to denote the single pending invocation of thread $\tid$ in $\hs{\ex}$, or $\emptyword$ if such an invocation does not exist.
	We abuse notation and write $\prefinv{\inv^\ex_\tid}$ to denote $\prefof{\ex}{\trans}$ where $\trans=\emptyword$ if $\inv^{\ex}_\tid=\emptyword$, and otherwise $\trans\in\ex\rst{\tid}$ is the transition such that $\lab(\trans)=\inv^{\ex}_\tid$.
	
	A simulation relation $\fsim$ is defined by
	$\tup{\ex,\tup{\val,\ver,\lockLTS,\Tmapid,\Tmappc,\Tmapvalset,\Tmapstart,\Tmaparg}}\in\fsim$ iff
	the following conditions hold:

	\begin{enumerate}[label=(\condletter\arabic*),ref=\arabic*,leftmargin=1cm]
		\item\label{awshard:1} $\val=\obs(\linmap(\ex))$.
		
		\item\label{awshard:2} $\lockLTS=0$.

		\item\label{awshard:3} For every $\tid\in\Tid$ and $\id\in\Id$, $\id\in\Tmapid(\tid)\iff\exists\act\in\hs{\ex}\rst{\tid}\ldotp\idf(\act)=\id$.

		\item\label{awshard:4} $0\leq\Tmapstart(\tid)\leq\ver$ for every $\tid\in\Tid$.
		
		\item\label{awshard:5} For every $\tid\in\Tid$, if $\Tmapstart(\tid)=\ver$, then $\lastop(\linmap(\ex)\rst{\regwrite})=\lastop(\linmap(\prefinv{\inv^\ex_\tid})\rst{\regwrite})$.

		\item\label{awshard:6} For every $\tid\in\Tid$, denoting $\id=\Tmapid(\tid)[\last]$, one of the following holds:
		\begin{enumerate}[label=(\condletter\theenumi\alph*),ref=\theenumi\alph*,leftmargin=1cm]
			
			\item\label{awshard:6a} $\Tmappc(\tid)=\main$ and $\inv^{\ex}_\tid=\emptyword$.
			
			\item\label{awshard:6b} $\Tmappc(\tid)=\rlisten$, $\inv^{\ex}_\tid=\invact\tup{\regread,\bot,\tid,\id}$, $\inv^{\ex}_\tid\in\hs{\ex}$, $\inv^{\ex}_\tid\notin\linmap(\ex)$, and \\
			$\obsbetween(\linmap(\prefinv{\inv^\ex_\tid}),\linmap(\ex))\subseteq\Tmapvalset(\tid)$.
			
			\item\label{awshard:6c} $\Tmappc(\tid)=\wwait$, $\inv^{\ex}_\tid=\invact\tup{\regwrite,\Tmaparg(\tid),\tid,\id}$, $\inv^{\ex}_\tid\in\hs{\ex}$, and $\inv^{\ex}_\tid\notin\linmap(\ex)$.
			
			\item\label{awshard:6d} $\Tmappc(\tid)=\wend$ and $\inv^{\ex}_\tid=\invact\tup{\regwrite,\_,\tid,\id}\in\linmap(\ex)$.
			
		\end{enumerate}
	\end{enumerate}
	
	Clearly, for the initial states we have
	$\tup{\emptyword,\awsreg.\linit}\in\fsim$.
	
	Given states $\tup{q,\abs q}=\tup{\ex,\tup{\val,\ver,\lockLTS,\Tmapid,\Tmappc,\Tmapvalset,\Tmapstart,\Tmaparg}}\in\fsim$
	and a transition $\ex \asteplab{\act}{\ghost\impl} \ex\cdot\trans=q'$ where $\act=\lab(\trans)$, we present a sequence of transitions $\abs q\bsteplab{\actsq}{\awsreg}{\abs q}'$ such that $\act\rst{\Inv\Res(\Reg)}=\actsq\rst{\Inv\Res(\Reg)}$,
	and prove that $\tup{q',{\abs q}'}\in\fsim$.
	We omit the justification of conditions for which all relevant values are unchanged between the pre-state and post-state.
	In particular:
	\begin{itemize}
		\item If $\act\notin\Res$ then from $\linmap$ being response-activated we have $\linmap(\ex\cdot\trans)=\linmap(\ex)$,
		and so conditions for which no other values are changed are omitted.
		
		\item For conditions on each $\tid\in\Tid$, we justify them only for thread ids for which the relevant values were changed.
		Specifically we note that $\inv^{\ex}_\tid=\inv^{\ex\cdot\trans}_\tid$ whenever $\trans$ is not an invocation or response of thread $\tid$, and in such cases as $\hs{\ex}\pref\hs{\ex\cdot\trans}$ we also have $\inv^{\ex}_\tid\in\hs{\ex}\implies\inv^{\ex\cdot\trans}_\tid\in\hs{\ex\cdot\trans}$
	\end{itemize}
	We divide into cases based on $\act$:
	\begin{itemize}
		\item If $\act\in\Inv$, then it is of the form $\act=\invact\tup{\method,\valin,\tid,\id}$.
		From $\ex\cdot\trans$ being a well formed execution we get that $\id$ does not appear in any action in $\hs{\ex}\rst{\tid}$ and thus from \cref{awshard:3} we have $\id\notin\Tmapid(\tid)$.
		We also obtain that necessarily $\inv^{\ex}_\tid=\emptyword$ and thus from \cref{awshard:6a} we have $\Tmappc(\tid)=\main$.
		\begin{itemize}
			\item If $\method=\regread$, then $\valin=\bot$ and a transition using \readinv\ is available:\\
			$\tup{\val,\ver,\lockLTS,\Tmapid,\Tmappc,\Tmapvalset,\Tmapstart,\Tmaparg} \asteplab{\act}{\awsreg} \tup{\val,\ver,\lockLTS,\Tmapid[\tid \mapsto \Tmapid(\tid)\cdot\id],\Tmappc[\tid \mapsto \rbegin],\Tmapvalset,\Tmapstart,\Tmaparg}$.
			
			Then, as $\lockLTS=0$, a transition using \readinit\ is available:\\
			$\tup{\val,\ver,\lockLTS,\Tmapid[\tid \mapsto \Tmapid(\tid)\cdot\id],\Tmappc[\tid \mapsto \rbegin],\Tmapvalset,\Tmapstart,\Tmaparg} \asteplab{\objact\tup{\tid}}{\awsreg}\\
			\tup{\val,\ver,\lockLTS,\Tmapid[\tid \mapsto \Tmapid(\tid)\cdot\id],\Tmappc[\tid \mapsto \rlisten],\Tmapvalset[\tid\mapsto\set{\val}],\Tmapstart[\tid\mapsto\ver],\Tmaparg}$.
			
			For \cref{awshard:3}, note that $\idf(\act)=\id\in\Tmapid[\tid \mapsto \Tmapid(\tid)\cdot\id](\tid)$ and indeed $\act\in\hs{\ex\cdot\trans}\rst{\tid}$.
			For \cref{awshard:4} we use the fact that $\Tmapstart[\tid\mapsto\ver](\tid)=\ver$.
			For \cref{awshard:5} we note that $\inv^{\ex\cdot\trans}_\tid=\act$,
			and thus
			$\prefinv{\inv^{\ex\cdot\trans}_\tid}=\prefof{\ex\cdot\trans}{\trans}=\ex$, and we obtain that
			$\lastop(\linmap(\prefinv{\inv^{\ex\cdot\trans}_\tid})\rst{\regwrite})=\lastop(\linmap(\ex)\rst{\regwrite})$, as required.
			For \cref{awshard:6}, we show \cref{awshard:6b} holds. We have $\Tmappc[\tid \mapsto \rlisten](\tid)=\rlisten$. As $\inv^{\ex\cdot\trans}_\tid=\act$, we have that $\inv^{\ex\cdot\trans}_\tid=\invact\tup{\regread,\bot,\tid,\id}$, where as required $\id=\Tmapid[\tid \mapsto \Tmapid(\tid)\cdot\id](\tid)[\last]$.
			Moreover,  $\inv^{\ex\cdot\trans}_\tid=\act\in\hs{\ex\cdot\trans}$, and from $\linmap(\ex)$ being a valid linearization $\inv^{\ex\cdot\trans}_\tid\notin\linmap(\ex)=\linmap(\ex\cdot\trans)$.
			Finally, we need to show for \cref{awshard:6b} that\\
			$\obsbetween(\linmap(\prefinv{\inv^{\ex\cdot\trans}_\tid}),\linmap({\ex\cdot\trans}))\subseteq\Tmapvalset[\tid\mapsto\set{\val}](\tid) $.
			
			As $\prefinv{\inv^{\ex\cdot\trans}_\tid}=\ex$ and $\linmap(\ex)=\linmap(\ex\cdot\trans)$
			we have that $\obsbetween(\linmap(\prefinv{\inv^{\ex\cdot\trans}_\tid}),\linmap({\ex\cdot\trans}))=\obsbetween(\linmap(\ex),\linmap(\ex))$.
			As there is no $\wop\in\linmap(\ex)\rst{\regwrite}$ such that $\wop'<_{\linmap(\ex)}\wop$ for all $\wop'\in\linmap(\ex)\rst{\regwrite}$,
			we obtain that $\obsbetween(\linmap(\ex),\linmap(\ex))=\set{\obs(\linmap(\ex))}$.
			Overall the requirement simplifies to
			$\set{\obs(\linmap(\ex))}\subseteq\set{\val}$, which follows from \cref{awshard:1}.
			
			\item If $\method=\regwrite$, then  a transition using \writeinv\ is available:\\
			$\tup{\val,\ver,\lockLTS,\Tmapid,\Tmappc,\Tmapvalset,\Tmapstart,\Tmaparg} \asteplab{\act}{\awsreg}\\ \tup{\val,\ver,\lockLTS,\Tmapid[\tid \mapsto \Tmapid(\tid)\cdot\id],\Tmappc[\tid \mapsto \wbegin],\Tmapvalset,\Tmapstart,\Tmaparg[\tid\mapsto \valin]}$.
			
			Then, as $\lockLTS=0$, a transition using \writeinit\ is available:\\
			$\tup{\val,\ver,\lockLTS,\Tmapid[\tid \mapsto \Tmapid(\tid)\cdot\id],\Tmappc[\tid \mapsto \wbegin],\Tmapvalset,\Tmapstart,\Tmaparg[\tid\mapsto \valin]} \asteplab{\objact\tup{\tid}}{\awsreg}\\
			\tup{\val,\ver,\lockLTS,\Tmapid[\tid \mapsto \Tmapid(\tid)\cdot\id],\Tmappc[\tid \mapsto \wwait],\Tmapvalset,\Tmapstart[\tid\mapsto\ver],\Tmaparg[\tid\mapsto \valin]}$.
			
			Conditions \cref{awshard:3},\cref{awshard:4}, and \cref{awshard:5} are justified exactly like in the case of $\method=\regread$.
			For \cref{awshard:6}, we show \cref{awshard:6c} holds. We have $\Tmappc[\tid \mapsto \wwait](\tid)=\wwait$. As explained in the case of $\method=\regread$, $\inv^{\ex\cdot\trans}_\tid=\act$. Here this implies that $\inv^{\ex\cdot\trans}_\tid=\invact\tup{\regwrite,\valin,\tid,\id}$, where as required $\id=\Tmapid[\tid \mapsto \Tmapid(\tid)\cdot\id](\tid)[\last]$ and $\valin=\Tmaparg[\tid\mapsto \valin](\tid)$.
			The fact that $\inv^{\ex\cdot\trans}_\tid\in\hs{\ex\cdot\trans}$ and $\inv^{\ex\cdot\trans}_\tid\notin\linmap(\ex\cdot\trans)$ is justified exactly like in the case of $\method=\regread$.
			
		\end{itemize}
		\item If $\act\in\Objint$, we perform an empty sequence of transitions.
		Note that $\hs{\ex}=\hs{\ex\cdot\trans}$,
		and more generally no values relevant for the conditions are changed, so all conditions are maintained.
		
		\item If $\act\in\Res$, it is of the form $\act=\resact\tup{\method,\valout,\tid,\id}$.
		As $\hs{\ex\cdot\trans}$ is well formed, this response corresponds to an invocation. That is, $\inv^{\ex}_\tid\neq\emptyword$ and $\op\defeq\inv^{\ex}_\tid\cdot\act$ is an inv-res pair (\ie an operation).
		Specifically $\tid=\tidf(\inv^{\ex}_\tid)$, $\id=\idf(\inv^{\ex}_\tid)$, and from \cref{awshard:6} we get $\id=\Tmapid(\tid)[\last]$.

		The sequence of transitions we perform will be composed of two parts:
		\begin{enumerate}
			\item A sequence of write and read transitions, with a structure that depends on whether a new write was linearized to the right of all existing writes (\ie $\lastop(\linmap(\ex\cdot\trans)\rst{\regwrite})\neq\lastop(\linmap(\ex)\rst{\regwrite})$).
			
			\item A sequence with a response transition matching $\act$.
		\end{enumerate}
		
		For the first part, we begin by proving two technical claims about transitions available for particular sets of thread identifiers.
		
		Denote $\writeopset=\setof{\linmap(\ex\cdot\trans)\rst{\regwrite}}\setminus\setof{\linmap(\ex)\rst{\regwrite}}$
		the set of writes which appear in the linearization of $\ex\cdot\trans$ but not in the linearization of $\ex$.
		Also denote $\writetidset=\set{\tidf(\wop)\st \wop\in\writeopset}$.
		\begin{claim2}
			\label{claim:new_write_set}
			Let $\tid'\in\writetidset$. The following holds:
			\begin{enumerate*}
				\item $\Tmappc(\tid')=\wwait$, $\inv^{\ex}_{\tid'}=\invact\tup{\regwrite,\Tmaparg(\tid'),\tid',\id'}$ where $\id'=\Tmapid(\tid')[\last]$,
				and $\Tmaparg(\tid')=\inputf(\wop')$ for a unique write operation $\wop'\in\writeopset$ such that $\tidf(\wop')=\tid'$.
				\item $\set{\Tmaparg(\tid')\st\tid'\in\writetidset}=\set{\inputf(\wop')\st \wop'\in\writeopset}$.
				\item If additionally $\lastop(\linmap(\ex\cdot\trans)\rst{\regwrite})=\lastop(\linmap(\ex)\rst{\regwrite})$, then $\Tmapstart(\tid')<\ver$.
			\end{enumerate*}
		\end{claim2}
		\newline
		
		The idea here is that writes that were just added to the linearization following the last transition are all pending writes which did not store their value yet before said transition.
		
		\begin{claimproof}
			Let $\wop'\in\writeopset$ be some write operation such that $\tidf(\wop')=\tid'$. We prove all required properties:
			\begin{enumerate}
				\item Let $\inv_{\wop'}$ be the invocation of $\wop'$.
				As $\wop'\in\setof{\linmap(\ex\cdot\trans)}$, by definition of linearizability it holds that $\inv_{\wop'}\in\hs{\ex\cdot\trans}$.
				As $\act=\lab(\trans)\in\Res$ we have that $\inv_{\wop'}\neq\act$ and so $\inv_{\wop'}\in\hs{\ex\cdot\trans}=\hs{\ex}\cdot\act$ implies that $\inv_{\wop'}\in\hs{\ex}$.
				As $\wop'\notin\setof{\linmap(\ex)\rst{\regwrite}}$
				we obtain that $\inv_{\wop'}=\inv^{\ex}_{\tid'}$ as the invocation must be pending in $\ex$ to not appear in the linearization $\linmap(\ex)$.
				This also means $\wop'$ is a unique operation in $\writeopset$ such that $\tidf(\wop')=\tid'$, as no thread has two pending operations at once.
				Overall we have $\inv^{\ex}_{\tid'}\in\hs{\ex}$, $\inv^{\ex}_{\tid'}\notin\linmap(\ex)$ and $\methodf(\inv^{\ex}_{\tid'})=\regwrite$ and thus \cref{awshard:6c} must hold and we obtain that $\Tmappc(\tid')=\wwait$ and $\inv^{\ex}_{\tid'}=\invact\tup{\regwrite,\Tmaparg(\tid'),\tid',\id'}$.
				In particular as $\inv_{\wop'}=\inv^{\ex}_{\tid'}=\invact\tup{\regwrite,\Tmaparg(\tid'),\tid',\id'}$ we obtain that $\Tmaparg(\tid')=\inputf(\wop')$.
				
				\item As every $\wop'\in\writeopset$ has a corresponding $\tidf(\wop')\in\writetidset$, and as we have shown $\Tmaparg(\tidf(\wop'))=\inputf(\wop')$,
				we overall obtain that
				$\set{\Tmaparg(\tid')\st\tid'\in\writetidset}=\set{\inputf(\wop')\st \wop'\in\writeopset}$.
				
				\item Assume additionally $\lastop(\linmap(\ex\cdot\trans)\rst{\regwrite})=\lastop(\linmap(\ex)\rst{\regwrite})$.
				Using \cref{awshard:2} we have $\Tmapstart(\tid')\leq\ver$ and thus it suffices to show that  $\Tmapstart(\tid')\neq\ver$.
				
				Assume otherwise.
				As $\wop'\in\setof{\linmap(\ex\cdot\trans)\rst{\regwrite}}$ we have in particular that $\linmap(\ex\cdot\trans)\rst{\regwrite}\neq\emptyword$ and so $\wop^\last\eqdef\lastop(\linmap(\ex\cdot\trans)\rst{\regwrite})\neq\emptyword$.
				Using \cref{awshard:5} and $\Tmapstart(\tid')=\ver$
				we obtain that
				$\lastop(\linmap(\ex)\rst{\regwrite})=\lastop(\linmap(\prefinv{\inv^\ex_{\tid'}})\rst{\regwrite})$.
				As we assume $\lastop(\linmap(\ex\cdot\trans)\rst{\regwrite})=\lastop(\linmap(\ex)\rst{\regwrite})$,
				we overall obtain that $\wop^\last=\lastop(\linmap(\prefinv{\inv^\ex_{\tid'}})\rst{\regwrite})$.
				Now, denote $\op^\last\defeq\lastop(\linmap(\prefinv{\inv^\ex_{\tid'}}))$
				we obtain that
				$\wop^\last\leq_{\linmap(\prefinv{\inv^\ex_{\tid'}})}\op^\last$.
				Using decisiveness of $\linmap$ we have that $\linmap(\prefinv{\inv^\ex_{\tid'}})\subseq\linmap(\ex\cdot\trans)$.
				We deduce that $\wop^\last\leq_{\linmap(\ex\cdot\trans)}\op^\last$.
				From $\linmap$ being last-completed we also have that $\op^\last\in\setof{\completed(\hs{\prefinv{\inv^\ex_{\tid'}}})}$.
				That is, the response of $\op^\last$ appears in $\ex$ before the invocation of $\wop'$.
				This implies $\op^\last<_{\hs{\ex\cdot\trans}}\wop'$ and thus $\op^\last<_{\linmap(\ex\cdot\trans)}\wop'$.
				Overall $\wop^\last<_{\linmap(\ex\cdot\trans)}\wop'$, contradicting the definition of $\wop^\last$.\claimqedhere
			\end{enumerate}
		\end{claimproof}
		
		Denote $\readtidset=\set{\hat{\tid}\in\Tid\st \Tmappc(\hat{\tid})=\rlisten}$ the set of thread identifiers of pending read operations currently executing their loop.
		
		\begin{claim2}
			\label{claim:listening_reads}
			$\obsbetween(\linmap(\prefinv{\inv^\ex_{\hat{\tid}}}),\linmap(\ex\cdot\trans))\subseteq\Tmapvalset(\hat{\tid})\cup \set{\Tmaparg(\tid')\st\tid'\in\writetidset} $
			for every $\hat{\tid}\in\readtidset$.
			If additionally $\Tmapstart(\hat{\tid})=\ver$ and $\lastop(\linmap(\ex\cdot\trans)\rst{\regwrite})=\lastop(\linmap(\ex)\rst{\regwrite})$, then
			$\obsbetween(\linmap(\prefinv{\inv^\ex_{\hat{\tid}}}),\linmap(\ex\cdot\trans))\subseteq\Tmapvalset(\hat{\tid})$.
		\end{claim2}
		\newline
		
		The idea here is that any potential return value for a pending read is either already present in the set of values loaded by that read, or is the input value of some write which was just added to the linearization.
		The second option is not needed if the read method started on the latest version number,
		and the added writes were inserted between existing writes, as those are essentially performed with an earlier version numbers.
		
		\begin{claimproof}
			Let $\val'\in\obsbetween(\linmap(\prefinv{\inv^\ex_{\hat{\tid}}}),\linmap(\ex\cdot\trans))$.
			Using the definition of $\obsbetween$, there are two cases to consider:
			\begin{itemize}
				\item If $\val'=\obs(\linmap(\prefinv{\inv^\ex_{\hat{\tid}}}))$:
				
				By definition $\val'\in\obsbetween(\linmap(\prefinv{\inv^\ex_{\hat{\tid}}}),\linmap(\ex))$.
				Since $\Tmappc(\hat{\tid})=\rlisten$ we obtain from \cref{awshard:6b} that	$\obsbetween(\linmap(\prefinv{\inv^\ex_{\hat{\tid}}}),\linmap(\ex))\subseteq\Tmapvalset(\hat{\tid})$,
				and thus $\val'\in\Tmapvalset(\hat{\tid})$.
				
				\item If $\val'=\inputf(\wop')$ for some $\wop'\in\linmap(\ex\cdot\trans)\rst{\regwrite}$
				such that $\hat{\wop}<_{\linmap(\ex\cdot\trans)}\wop'$ for all $\hat{\wop}\in\linmap(\prefinv{\inv^\ex_{\hat{\tid}}})\rst{\regwrite}$:
				
				We consider two sub-cases:
				\begin{itemize}
					\item If $\wop'\notin\writeopset=\setof{\linmap(\ex\cdot\trans)\rst{\regwrite}}\setminus\setof{\linmap(\ex)\rst{\regwrite}}$:
					
					As $\wop'\in\linmap(\ex\cdot\trans)\rst{\regwrite}$, it must be the case that $\wop'\in\linmap(\ex)\rst{\regwrite}$.
					Since $\linmap$ is decisive we have $\linmap(\prefinv{\inv^\ex_{\hat{\tid}}})\subseq\linmap(\ex)\subseq\linmap(\ex\cdot\trans)$.
					Therefore, for every $\hat{\wop}\in\linmap(\prefinv{\inv^\ex_{\hat{\tid}}})$,
					$\hat{\wop}\in\linmap(\ex)$ and  $\hat{\wop}<_{\linmap(\ex\cdot\trans)}\wop'$ implies that $\hat{\wop}<_{\linmap(\ex)}\wop'$.
					Overall by definition of $\obsbetween$ we obtain that $\val'=\inputf(\wop')\in \obsbetween(\linmap(\prefinv{\inv^\ex_{\hat{\tid}}}),\linmap(\ex))$,
					which as before implies that $\val'\in\Tmapvalset(\hat{\tid})$.
					
					\item If $\wop'\in\writeopset$:
					
					By definition $\val'\in\set{\inputf(\wop)\st \wop\in\writeopset}$.
					Using \cref{claim:new_write_set} we have
					$\set{\Tmaparg(\tid')\st\tid'\in\writetidset}=\set{\inputf(\wop)\st \wop\in\writeopset}$,
					and so $\val'\in\set{\Tmaparg(\tid')\st\tid'\in\writetidset}$.
				\end{itemize}
				
			\end{itemize}
			
			Assume now additionally that $\Tmapstart(\hat{\tid})=\ver$ and $\lastop(\linmap(\ex\cdot\trans)\rst{\regwrite})=\lastop(\linmap(\ex)\rst{\regwrite})$.
			It suffices to show that the case $\wop'\in\writeopset$ is now impossible.
			(Indeed, in all other cases we proved that $\val'\in\Tmapvalset(\hat{\tid})$.)
			Assume otherwise.
			From $\Tmapstart(\hat{\tid})=\ver$, using \cref{awshard:5} we deduce that $\lastop(\linmap(\ex)\rst{\regwrite})=\lastop(\linmap(\prefinv{\inv^\ex_{\hat{\tid}}})\rst{\regwrite})$.
			Thus using the assumption that $\lastop(\linmap(\ex\cdot\trans)\rst{\regwrite})=\lastop(\linmap(\ex)\rst{\regwrite})$
			we overall obtain that $\lastop(\linmap(\ex\cdot \trans)\rst{\regwrite})=\lastop(\linmap(\prefinv{\inv^\ex_{\hat{\tid}}})\rst{\regwrite})$.
			Since $\wop'\in\linmap(\ex\cdot\trans)$ we have that
			$\wop^\last\eqdef\lastop(\linmap(\ex\cdot \trans)\rst{\regwrite})\neq\emptyword$,
			and as $\lastop(\linmap(\ex\cdot \trans)\rst{\regwrite})=\lastop(\linmap(\prefinv{\inv^\ex_\tid})\rst{\regwrite})$
			we deduce in particular that $\wop^\last\in\linmap(\prefinv{\inv^\ex_\tid})\rst{\regwrite}$
			By choice of $\wop'$ this implies $\wop^\last<_{\linmap(\ex\cdot\trans)}\wop'$,
			contradicting the definition of $\wop^\last$.
		\end{claimproof}

		We use the following claim to capture all relevant properties of the first part of the sequence we construct.
		This then enables us to argue about the second part and the satisfaction of the required conditions in the post-states in a way which captures both cases for the first part at once.
		
		\begin{claim2} \label{claim:write_read_seq}
			There exists a sequence of transitions $\tup{\val,\ver,\lockLTS,\Tmapid,\Tmappc,\Tmapvalset,\Tmapstart,\Tmaparg} \Longrightarrow^*\\
			\tup{\val',\ver',\lockLTS,\Tmapid,\Tmappc[\forall\tid'\in\writetidset\ldotp\tid'\mapsto\wend],\Tmapvalset',\Tmapstart,\Tmaparg'}$
			such that the following holds:
			\begin{enumerate}
				\item $\val'=\obs(\linmap(\ex\cdot\trans))$.
				\item $\ver'\geq\ver$.
				\item $\ver'>\ver \iff \lastop(\linmap(\ex\cdot\trans)\rst{\regwrite})\neq\lastop(\linmap(\ex)\rst{\regwrite})$.
				\item $\obsbetween(\linmap(\prefinv{\inv^\ex_{\hat{\tid}}}),\linmap(\ex\cdot\trans))\subseteq\Tmapvalset'(\hat{\tid})$ for all $\hat{\tid}\in\readtidset$.
			\end{enumerate}
		\end{claim2}
		\begin{claimproof}
			Denote $\wop^\last\eqdef\lastop(\linmap(\ex\cdot \trans)\rst{\regwrite})$.
			Consider the following cases:
			\begin{itemize}
				\item If $\lastop(\linmap(\ex\cdot\trans)\rst{\regwrite})\neq\lastop(\linmap(\ex)\rst{\regwrite})$:
				
				First, we claim that $\wop^\last\neq\emptyword$ and $\tid^\last\defeq\tidf(\wop^\last)\in\writetidset$.
				
				Indeed, $\linmap(\ex)\subseq\linmap(\ex\cdot\trans)$ due to $\linmap$ being decisive, and from $\lastop(\linmap(\ex\cdot\trans)\rst{\regwrite})\neq\lastop(\linmap(\ex)\rst{\regwrite})$ necessarily $\linmap(\ex)\rst{\regwrite}\subsneq\linmap(\ex\cdot\trans)\rst{\regwrite}$ and thus $\wop^\last=\lastop(\linmap(\ex\cdot \trans)\rst{\regwrite})\neq\emptyword$.
				
				Moreover, it is not possible that $\wop^\last\in \linmap(\ex)$ because then we would get from $\linmap(\ex)\rst{\regwrite}\subseq\linmap(\ex\cdot\trans)\rst{\regwrite}$ that $\wop^\last=\lastop(\linmap(\ex)\rst{\regwrite})$, contradicting $\lastop(\linmap(\ex\cdot\trans)\rst{\regwrite})\neq\lastop(\linmap(\ex)\rst{\regwrite})$.
				Therefore, we obtain that $\wop^\last\in\writeopset$ and $\tid^\last\in\writetidset$.
				
				Now, from \cref{awshard:2} we have $\lockLTS=0$, and as we have shown in \cref{claim:new_write_set} that $\Tmappc(\tid')=\wwait$ for all $\tid'\in\writetidset$, a transition using the rule \writetop\ is available for them.
				
				Using \cref{awshard:4}, a transition using \readloop\ is available for each $\hat{\tid}\in\readtidset$ where the current value in the register is not in its set of loaded values already.
				As $\ver$ only increases and $\Tmapstart(\hat{\tid})$ is not changed by \writetop\ or \readloop\ transitions, \readloop\ remains available after performing these transitions (again dependent on the current register value).
				
				Overall, we can perform a \writetop\ transition for every $\tid'\in\writetidset$, doing the transition for $\tid^\last$ last, where each such transition is followed by a \readloop\ transition for every $\hat{\tid}\in\readtidset$ that does not already contain the value stored in the \writetop\ transition.
				After performing this sequence, for every $\hat{\tid}\in\readtidset$, every value in $\set{\Tmaparg(\tid')\st\tid'\in\writetidset}$ is either added to $\Tmapvalset(\hat{\tid})$ or is already present in that set.
				This overall gives:
				$\tup{\val,\ver,\lockLTS,\Tmapid,\Tmappc,\Tmapvalset,\Tmapstart,\Tmaparg} \Longrightarrow^*\\
				\tup{\Tmaparg(\tid^\last),\ver+\size{\writetidset},\lockLTS,\Tmapid,\Tmappc[\forall\tid'\in\writetidset\ldotp\tid'\mapsto\wend],\Tmapvalset',\Tmapstart,\Tmaparg} $\\
				where $\Tmapvalset'=\Tmapvalset[\forall\hat{\tid}\in\readtidset\ldotp\hat{\tid} \mapsto \Tmapvalset(\hat{\tid})\cup\set{\Tmaparg(\tid')\st\tid'\in\writetidset}]$.
				We show all conditions required for the claim hold:
				\begin{enumerate}
					\item $\Tmaparg(\tid^\last)=\obs(\linmap(\ex\cdot\trans))$:
					
					As we defined $\tid^\last=\tidf(\wop^\last)$ where $\wop^\last=\lastop(\linmap(\ex\cdot \trans)\rst{\regwrite})$, we have $\obs(\linmap(\ex\cdot\trans))=\inputf(\wop^\last)$.
					As we have shown $\tid^\last\in\writetidset$,
					we obtain from \cref{claim:new_write_set} that $\Tmaparg(\tid^\last)=\inputf(\wop^\last)$, and the equality $\Tmaparg(\tid^\last)=\obs(\linmap(\ex\cdot\trans))$ follows.
					
					\item $\ver+\size{\writetidset}\geq\ver$:
					
					Immediate from $\size{\writetidset}\geq 0$.
					
					\item $\ver+\size{\writetidset}>\ver \iff \lastop(\linmap(\ex\cdot\trans)\rst{\regwrite})\neq\lastop(\linmap(\ex)\rst{\regwrite})$:
					
					We have shown $\tid^\last\in\writetidset$.
					This implies $\size{\writetidset}>0$ and so the LHS is true.
					The RHS is true by assumption.

					\item $\obsbetween(\linmap(\prefinv{\inv^\ex_{\hat{\tid}}}),\linmap(\ex\cdot\trans))\subseteq\Tmapvalset'(\hat{\tid}) $ for all $\hat{\tid}\in\readtidset$:
					
					Let $\hat{\tid}\in\readtidset$. Using  \cref{claim:listening_reads} we have that
					$\obsbetween(\linmap(\prefinv{\inv^\ex_{\hat{\tid}}}),\linmap(\ex\cdot\trans))\subseteq\Tmapvalset(\hat{\tid})\cup \set{\Tmaparg(\tid')\st\tid'\in\writetidset}=\Tmapvalset'(\hat{\tid})$.
				\end{enumerate}
				
				\item If $\lastop(\linmap(\ex\cdot\trans)\rst{\regwrite})=\lastop(\linmap(\ex)\rst{\regwrite})$:
				
				Using \cref{claim:new_write_set} we have all neccessary conditions to perform a \writerollback\ transition for every $\tid'\in\writetidset$.
				(In particular, $\Tmapstart(\tid')<\ver$ holds since we assume $\lastop(\linmap(\ex\cdot\trans)\rst{\regwrite})=\lastop(\linmap(\ex)\rst{\regwrite})$.)
				
				Doing this transition for a given $\tid'\in\writetidset$ gives
				$\tup{\val,\ver,\lockLTS,\Tmapid,\Tmappc,\Tmapvalset,\Tmapstart,\Tmaparg} \asteplab{\objact\tup{\tid'}}{}\\
				\tup{\Tmaparg(\tid'),\ver-1,1,\Tmapid,\Tmappc[\tid' \mapsto \wrollback],\Tmapvalset,\Tmapstart,\Tmaparg[\tid'\mapsto\val]}$
				
				Define $\readtidset'\eqdef\set{\hat{\tid}\st \Tmappc(\hat{\tid})=\rlisten\land \Tmapstart(\hat{\tid})\leq\ver-1}$.
				After the above transition, a \readloop\ transition can be done for each $\hat{\tid}\in\readtidset'$ for whom the set of loaded values does not contain the current register value,
				and this remains true when performing other \readloop\ transitions.
				Preforming such a transition for each relevant $\hat{\tid}\in\readtidset'$ gives:
				$\tup{\Tmaparg(\tid'),\ver-1,1,\Tmapid,\Tmappc[\tid' \mapsto \wrollback],\Tmapvalset,\Tmapstart,\Tmaparg[\tid'\mapsto\val]} \Longrightarrow^* \\ \tup{\Tmaparg(\tid'),\ver-1,1,\Tmapid,\Tmappc[\tid' \mapsto \wrollback],\tilde{\Tmapvalset},\Tmapstart,\Tmaparg[\tid'\mapsto\val]}$
				
				Where
				$\tilde{\Tmapvalset}=\Tmapvalset[\forall\hat{\tid}\in\readtidset'\ldotp\hat{\tid} \mapsto \Tmapvalset(\hat{\tid})\cup\set{\Tmaparg(\tid')}]$.
				
				After that, a \writerollforward\ transition is available for the $\tid'$ we did \writerollback\ for, giving:
				$\tup{\Tmaparg(\tid'),\ver-1,1,\Tmapid,\Tmappc[\tid' \mapsto \wrollback],\tilde{\Tmapvalset},\Tmapstart,\Tmaparg[\tid'\mapsto\val]}\asteplab{\objact\tup{\tid'}}{}\\
				\tup{\Tmaparg[\tid'\mapsto\val](\tid'),\ver-1+1,0,\Tmapid,\Tmappc[\tid' \mapsto \wend],\tilde{\Tmapvalset},\Tmapstart,\Tmaparg[\tid'\mapsto\val]}$
				
				Simplifying this (using $\lockLTS=0$) we obtain a state\\
				$\tup{\val,\ver,\lockLTS,\Tmapid,\Tmappc[\tid' \mapsto \wend],\tilde{\Tmapvalset},\Tmapstart,\Tmaparg[\tid'\mapsto\val]}$.
				
				After such a sequence, a similar sequence for other $\tid''\in\writetidset\setminus\set{\tid'}$ starting with \writerollback\ , followed by multiple \readloop\ transitions and ending with \writerollforward\ is similarly available, as the relevant variables enabling this sequence are unchanged.
				
				Overall, we can perform this sequence of each $\tid'\in\writetidset$, and get:\\
				$\tup{\val,\ver,\lockLTS,\Tmapid,\Tmappc,\Tmapvalset,\Tmapstart,\Tmaparg} \Longrightarrow^* \\ \tup{\val,\ver,\lockLTS,\Tmapid,\Tmappc[\forall\tid'\in\writetidset\ldotp\tid'\mapsto\wend],\Tmapvalset',\Tmapstart,\Tmaparg[\forall\tid'\in\writetidset\ldotp\tid'\mapsto\val]}$
				
				Where
				$\Tmapvalset'=\Tmapvalset[\forall\hat{\tid}\in\readtidset'\ldotp\hat{\tid} \mapsto \Tmapvalset(\hat{\tid})\cup\set{\Tmaparg(\tid')\st\tid'\in\writetidset}]$.
				Note that it is possible that $\writetidset=\emptyset$, and in this case the above sequence is empty and the post-state is equal to the pre-state.
				
				We show all conditions required for the claim hold:
				\begin{enumerate}
					\item $\val=\obs(\linmap(\ex\cdot\trans))$:
					
					This follows from \cref{awshard:1} holding in the pre-state and the fact that $\lastop(\linmap(\ex\cdot\trans)\rst{\regwrite})=\lastop(\linmap(\ex)\rst{\regwrite})$
					and thus $\obs(\linmap(\ex))=\obs(\linmap(\ex\cdot\trans))$
					
					\item $\ver\geq\ver$:
					
					Immediate.
					
					\item $\ver>\ver \iff \lastop(\linmap(\ex\cdot\trans)\rst{\regwrite})\neq\lastop(\linmap(\ex)\rst{\regwrite})$:
					
					The RHS is false by assumption and thus this statement is true.
					
					\item $\obsbetween(\linmap(\prefinv{\inv^\ex_{\tid}}),\linmap(\ex\cdot\trans))\subseteq\Tmapvalset'(\tid)$ for all $\hat{\tid}\in\readtidset$:
					
					\begin{itemize}
						\item If $\hat{\tid}\in\readtidset'$:
						
						As $\readtidset'\subseteq\readtidset$ we can use \cref{claim:listening_reads} and obtain that
						$\obsbetween(\linmap(\prefinv{\inv^\ex_{\hat{\tid}}}),\linmap(\ex\cdot\trans))\subseteq\Tmapvalset(\hat{\tid})\cup \set{\Tmaparg(\tid')\st\tid'\in\writetidset}=\Tmapvalset'(\hat{\tid})$.
						
						\item If $\hat{\tid}\in\readtidset\setminus\readtidset'$:
						
						We deduce that $\Tmapvalset'(\hat{\tid}) =\Tmapvalset(\hat{\tid})$ and $\Tmapstart(\hat{\tid})>\ver-1$.
						From \cref{awshard:4} we have $\Tmapstart(\hat{\tid})\leq\ver$ and thus overall $\Tmapstart(\hat{\tid})=\ver$.
						Therefore we can apply \cref{claim:listening_reads} and deduce that
						$\obsbetween(\linmap(\prefinv{\inv^\ex_{\hat{\tid}}}),\linmap(\ex\cdot\trans))\subseteq\Tmapvalset(\hat{\tid})=\Tmapvalset'(\hat{\tid})$.
						\claimqedhere
					\end{itemize}
				\end{enumerate}
			\end{itemize}
		\end{claimproof}
		
		We move on to argue about the second part of the sequence where we perform a sequence ending with a response labeled $\act$.
		The structure of this sequence depends on $\methodf(\act)$.
		
		\begin{claim2}\label{claim:matching_res}
			Starting from the post-state described in \cref{claim:write_read_seq},
			there exists a sequence of transitions:
			$\tup{\val',\ver',\lockLTS,\Tmapid,\Tmappc[\forall\tid'\in\writetidset\ldotp\tid'\mapsto\wend],\Tmapvalset',\Tmapstart,\Tmaparg'} \Longrightarrow^*\\
			\tup{\val',\ver',\lockLTS,\Tmapid,\Tmappc[\forall\tid'\in\writetidset\ldotp\tid'\mapsto\wend][\tid\mapsto\main],\Tmapvalset'[\tid\mapsto\emptyset],\Tmapstart[\tid\mapsto0],\Tmaparg'[\tid\mapsto0]}$
		\end{claim2}
		\begin{claimproof}
			As $\act=\resact\tup{\method,\valout,\tid,\id}\in\Res$, we consider two cases:
			\begin{itemize}
				\item If $\method=\regread$:
				
				We obtain that \cref{awshard:6b} must hold and thus $\Tmappc(\tid)=\rlisten$, $\inv^{\ex}_\tid=\invact\tup{\regread,\bot,\tid,\id}$,
				$\inv^{\ex}_\tid\in\hs{\ex}$, $\inv^{\ex}_\tid\notin\linmap(\ex)$, and from \cref{claim:write_read_seq} we have
				$\obsbetween(\linmap(\prefinv{\inv^\ex_{\tid}}),\linmap(\ex\cdot\trans))\subseteq\Tmapvalset'(\tid)$.
				
				Moreover, as $\op\in\setof{\completed(\hs{\ex\cdot\trans})}$ it must be the case that $\op\in\linmap(\ex\cdot\trans)$ and as $\methodf(\op)=\regread$
				we obtain from \cref{thm:lazyisgood} that $\valout\in\obsbetween(\linmap(\prefinv{\inv^\ex_{\tid}}),\linmap(\ex\cdot\trans))$ and thus $\valout\in\Tmapvalset'(\tid)$.
				
				As $\Tmappc(\tid)=\rlisten$ and $\tid\notin\writetidset$,
				$\Tmappc[\forall\hat{\tid}\in\writetidset\ldotp\hat{\tid}\mapsto\wend](\tid)=\rlisten$
				and a transition using \readpick\ is available:
				
				$\tup{\val',\ver',\lockLTS,\Tmapid,\Tmappc[\forall\tid'\in\writetidset\ldotp\tid'\mapsto\wend],\Tmapvalset',\Tmapstart,\Tmaparg'} \asteplab{\objact\tup{\tid}}{\awsreg}\\
				\tup{\val',\ver',\lockLTS,\Tmapid,\Tmappc[\forall\tid'\in\writetidset\ldotp\tid'\mapsto\wend][\tid\mapsto\rend],\Tmapvalset',\Tmapstart,\Tmaparg'[\tid\mapsto\valout]}$	
				
				Then, as $\Tmappc[\forall\tid'\in\writetidset\ldotp\tid'\mapsto\wend][\tid\mapsto\rend](\tid)=\rend$ and $\Tmaparg'[\tid\mapsto\valout](\tid)=\valout$ a \readres\ transition with label $\act$ is available:
				
				$\tup{\val',\ver',\lockLTS,\Tmapid,\Tmappc[\forall\tid'\in\writetidset\ldotp\tid'\mapsto\wend][\tid\mapsto\rend],\Tmapvalset',\Tmapstart,\Tmaparg'[\tid\mapsto\valout]} \asteplab{\act}{\awsreg}\\
				\tup{\val',\ver',\lockLTS,\Tmapid,\Tmappc[\forall\tid'\in\writetidset\ldotp\tid'\mapsto\wend][\tid\mapsto\main],\Tmapvalset'[\tid\mapsto\emptyset],\Tmapstart[\tid\mapsto0],\Tmaparg'[\tid\mapsto0]}$

				\item If $\method=\regwrite$:
				
				We have two options:
				The first option is that $\op\in\linmap(\ex)$, in which case \cref{awshard:6d} must hold, and we get that $\Tmappc(\tid)=\wend$.
				As $\op\in\linmap(\ex)$, we also get that $\tid\notin\writetidset$ and so $\Tmappc[\forall\tid'\in\writetidset\ldotp\tid'\mapsto\wend](\tid)=\Tmappc(\tid)=\wend$.
				The second option is that $\op\notin\linmap(\ex)$. As $\op\in\setof{\completed(\hs{\ex\cdot\trans})}$, we must have $\op\in\linmap(\ex\cdot\trans)$ and thus $\tid\in\writetidset$ and we deduce
				$\Tmappc[\forall\tid'\in\writetidset\ldotp\tid'\mapsto\wend](\tid)=\wend$.
				In both cases we obtain that a transition using the rule \writeres\ with label $\act$ is available:
				
				$\tup{\val',\ver',\lockLTS,\Tmapid,\Tmappc[\forall\tid'\in\writetidset\ldotp\tid'\mapsto\wend],\Tmapvalset',\Tmapstart,\Tmaparg'} \asteplab{\act}{\awsreg}\\
				\tup{\val',\ver',\lockLTS,\Tmapid,\Tmappc[\forall\tid'\in\writetidset\ldotp\tid'\mapsto\wend][\tid\mapsto\main],\Tmapvalset'[\tid\mapsto\emptyset],\Tmapstart[\tid\mapsto0],\Tmaparg'[\tid\mapsto0]} $. \claimqedhere
			
			\end{itemize}
		\end{claimproof}
		
		Using \cref{claim:write_read_seq,claim:matching_res}, we reach a state of the form
		
		$\tup{\val',\ver',\lockLTS,\Tmapid,\Tmappc',\Tmapvalset'[\tid\mapsto\emptyset],\Tmapstart[\tid\mapsto0],\Tmaparg'[\tid\mapsto0]}$
		
		where $\Tmappc'=\Tmappc[\forall\tid'\in\writetidset\ldotp\tid'\mapsto\wend][\tid\mapsto\main]$,
		such that:
		\begin{enumerate}
			\item $\val'=\obs(\linmap(\ex\cdot\trans))$.
			\item $\ver'\geq\ver$.
			\item $\ver'>\ver \iff \lastop(\linmap(\ex\cdot\trans)\rst{\regwrite})\neq\lastop(\linmap(\ex)\rst{\regwrite})$
			\item $\obsbetween(\linmap(\prefinv{\inv^\ex_{\tid}}),\linmap(\ex\cdot\trans))\subseteq\Tmapvalset'(\tid)$ for all $\hat{\tid}\in\readtidset$.
		\end{enumerate}
		
		We show that all necessary conditions are satisfied.
		For \cref{awshard:1}, we have that indeed $\val'=\obs(\linmap(\ex\cdot\trans))$.
		For \cref{awshard:4} we use the fact that \cref{awshard:4} holds in the pre-states:
		We get that $\Tmapstart[\tid\mapsto0](\tid)=0\leq\ver\leq\ver'$,
		and for all $\tid'\neq\tid$ we have $\Tmapstart[\tid\mapsto0](\tid')=\Tmapstart(\tid')\leq\ver\leq\ver'$.
		For \cref{awshard:5}, consider two cases:
		\begin{itemize}
			\item If $\ver'>\ver$:
			
			Using the fact that $\Tmapstart[\tid\mapsto0](\tid')\leq\ver<\ver'$ for all $\tid'\in\Tid$, the condition $\Tmapstart[\tid\mapsto0](\tid')=\ver'$ is always false and  \cref{awshard:5} holds vacuously.
			
			\item  If $\ver'=\ver$:
			
			We deduce that $\lastop(\linmap(\ex\cdot\trans)\rst{\regwrite})=\lastop(\linmap(\ex)\rst{\regwrite})$.
			
			For $\tid'\neq\tid$, $\inv^{\ex}_{\tid'}=\inv^{\ex\cdot\trans}_{\tid'}$ and so overall no values relevant for \cref{awshard:5} have changed.
			For $\tid'=\tid$ we have $\inv^{\ex\cdot\trans}_{\tid}=\emptyword$ and thus $\prefinv{\inv^{\ex\cdot\trans}_\tid}=\ex$ and the condition follows.
		\end{itemize}
		For \cref{awshard:6}, let $\tid'\in\Tid$.
		Consider the following cases:
		\begin{itemize}
			\item If $\tid'\neq\tid$ and $\Tmappc'(\tid)=\main$:
			
			We deduce that $\tid'\notin\writetidset$ and so $\Tmappc(\tid)=\main$ and \cref{awshard:6a} holds in the pre-states. As $\inv^{\ex\cdot\trans}_{\tid'}=\inv^{\ex}_{\tid'}=\emptyword$, \cref{awshard:6a} continues to hold in the post-states.
			
			\item If $\tid'=\tid$:
			
			We have $\Tmappc'(\tid)=\main$ and as $\op$ is completed and $\ex\cdot\trans$ is well formed we deduce that $\inv^{\ex\cdot\trans}_\tid=\emptyword$ and \cref{awshard:6a} holds.
			
			\item If $\tid'\in\readtidset\setminus\set{\tid}$ and $\Tmappc'(\tid')=\rlisten$:
			
			As $\tid'\neq\tid$ and $\Tmappc(\tid')=\rlisten$, using \cref{awshard:6b} we have $\inv^{\ex\cdot\trans}_{\tid'}=\inv^{\ex}_{\tid'}\in\hs{\ex}\pref\hs{\ex\cdot\trans}$.
			To show $\inv^{\ex\cdot\trans}_{\tid'}\notin\linmap(\ex\cdot\trans)$, using the fact that $\linmap$ is pending-read-free it is sufficient to prove that $\inv^{\ex\cdot\trans}_{\tid'}\notin\completed(\hs{\ex\cdot \trans})$.
			Indeed, if we assume otherwise, as $\trans$ is the response for $\tid\neq\tid'$ we get $\inv^{\ex}_{\tid'}\in\completed(\hs{\ex})$.
			Since \cref{awshard:6b} holds in the pre-state we have $\inv^{\ex}_{\tid'}\notin\linmap(\ex)$, and overall this contradicts $\linmap$ being a linearization mapping.
			We get that $\inv^{\ex\cdot\trans}_{\tid'}\notin\linmap(\ex\cdot\trans)$.
			Additionally, since $\tid'\in\readtidset$ we have shown in \cref{claim:listening_reads} that $\obsbetween(\linmap(\prefinv{\inv^\ex_{\tid'}}),\linmap(\ex\cdot\trans))\subseteq\Tmapvalset'(\tid')$.
			As $\inv^{\ex\cdot\trans}_{\tid'}=\inv^{\ex}_{\tid'}$ and $\Tmapvalset'[\tid\mapsto\emptyset](\tid')=\Tmapvalset'(\tid')$, we deduce that the condition on $\Tmapvalset$ holds and overall \cref{awshard:6b} holds.
			
			\item If $\tid'\in\writetidset$:
			
			We have $\Tmappc'(\tid')=\wend$ and by definition of $\writetidset$ we obtain that $\wop\in\linmap(\ex\cdot\trans)$ for the write operation $\wop\in\writeopset$ such that $\tidf(\wop)=\tid'$. (It is unique due to \cref{claim:new_write_set}.)
			As $\inv^{\ex\cdot\trans}_{\tid'}$ is the invocation of $\wop$ we get $\inv^{\ex\cdot\trans}_{\tid'}\in\linmap(\ex\cdot\trans)$ and $\methodf(\inv^{\ex\cdot\trans}_{\tid'})=\regwrite$, and so \cref{awshard:6d} holds.
			
			\item Otherwise:
			
			It must be the case that $\inv^{\ex\cdot\trans}_{\tid'}$ is a write invocation and $\tid'\notin\writetidset$.
			By definition of $\writetidset$ this means $\inv^{\ex\cdot\trans}_{\tid'}\notin\linmap(\ex\cdot\trans)$ and overall we get \cref{awshard:6c} from this condition holding in the pre-states.\qedhere
		\end{itemize}
	\end{itemize}

\end{proof}

\begin{theorem}
	$\wsreg$ is $\wslin$-hard.
\end{theorem}
\begin{proof}[Proof Sketch]
	The proof follows along very similar lines to \cref{thm:awshard}.
	The simulation relation is the same, except for the omission of \cref{awshard:2}, \cref{awshard:4}, and \cref{awshard:5} as they are irrelevant to the state of $\wsreg$.
	For a write strong implementation, we get via \cref{lem:wsnice} a linearization mapping $\linmap$ that is write strong in addition to being lazy.
	
	Using the fact that $\linmap$ is write strong, in the second case considered in the proof of \cref{claim:write_read_seq} ($\lastop(\linmap(\ex\cdot\trans)\rst{\regwrite})=\lastop(\linmap(\ex)\rst{\regwrite})$),
	we can deduce that $\linmap(\ex)\rst{\regwrite}=\linmap(\ex\cdot\trans)\rst{\regwrite}$.
	That is, no new writes are added to the linearization following the transition $\trans$, and so $\writetidset=\emptyset$.
	As the \quotes{rollback} mechanism of $\awsreg$ is only used to construct a sequence in the proof of \cref{claim:write_read_seq} under an assumption that  $\lastop(\linmap(\ex\cdot\trans)\rst{\regwrite})=\lastop(\linmap(\ex)\rst{\regwrite})$ and $\writetidset\neq\emptyset$,
	and without using it the mechanism involving version numbers is redundant, as they can only increase,
	we overall have that the mechanisms in $\awsreg$ that are not available in $\wsreg$ are not needed for the simulation, and we are able to use $\wsreg$ to simulate the implementation.
\end{proof}

\section{Lazy Linearizations}
\label{sec:lazy}

In this section we prove the existence of linearizations with the properties described in \cref{lem:wsnice,lem:awsnice}.

In \cref{subsec:minimal} we discuss the existence and usefulness of linearization mappings that are minimal \wrt preorders satisfying certain properties.
In \cref{subsec:wslazy}/\cref{subsec:declazy} we prove several properties of write strong/decisive linearizations that are minimal \wrt certain preorders.
In \cref{subsec:lazy_conclusion} we use the above to prove \cref{lem:wsnice,lem:awsnice}.

\subsection{Minimal Linearization Mappings}\label{subsec:minimal}

In this section we describe the general framework we use to construct linearization mappings with desirable properties. This framework is not specific to register objects.
The construction is generally done via the following stages:
\begin{enumerate}
	\item Define a notion of minimality among linearization mapping that satisfy a certain condition. (\eg decisive linearization mappings.) Intuitively, the notions of minimality we use represent ``delaying as much as possible decision on ordering between operations'', hence minimal mappings are also called ``lazy''.
	
	\item Prove that a minimal mapping exists.
	
	\item Prove that minimal mappings have desirable properties. (In the write strong case we also show they can be modified to obtain additional properties not guaranteed by all minimal mappings.)
\end{enumerate}

We begin by formally defining notions related to preorders and equivalences:

\begin{description}[leftmargin=0pt,itemsep=2pt]
\item[\em Preorders.]
	A preorder $\preleq$ on a set $A$ is a reflexive and transitive relation over $A$.
	An element $x\in A$ is \emph{$\preleq$-minimal}
	if $y\preleq x$ implies $x\preleq y$ for every $y\in A$.
	$\preleq$ is \emph{well-founded} if every non-empty $B\subseteq A$ has a $\preleq$-minimal element.
	A \emph{$\preleq$-chain} is a subset $C\subseteq A$ such that $\preleq$ is a total order (not just preorder)
	on $C$.
	An element $x\in A$ is a \emph{$\preleq$-lower bound} of a set $B\subseteq A$ if $x\preleq y$ for all $y\in B$.
	If we also have $x\in B$, we say that $x$ is a (not necessarily unique) \emph{$\preleq$-minimum} of $B$.
	
	To find minima we use the following classical lemma by Zorn (stated here for preorders and minima):
	
	\begin{lemma}[Zorn]
		If $\preleq$ is a preorder on a set $A$
		and every $\preleq$-chain $C\subseteq A$ has a $\preleq$-lower bound,
		then $A$ contains a $\preleq$-minimal element.
	\end{lemma}

\item[\em Equivalences.] An equivalence $\preeq$ on a set $A$ is a reflexive, transitive and symmetric relation over $A$.
$\preeq$ is a \emph{congruence} \wrt a relation $\rel\subseteq A\times A$ if $\tup{a,c}\in\rel\iff\tup{b,d}\in\rel$ whenever $a\preeq b$ and $c\preeq d$.
\end{description}

In the sequel, we use $\preleq$ to denote an arbitrary preorder, and ${\preeq} \defeq {\preleq}\,\cap\,{\preleq^{-1}}$ to denote the equivalence relation induced by $\preleq$.

Due to transitivity, a preorder and the induced equivalence satisfy the following:

\begin{lemma}\label{lem:self_congruent}
	Let $\preleq$ be a preorder and let ${\preeq} \defeq {\preleq}\,\cap\,{\preleq^{-1}}$ be the induced equivalence. 
	Then,
	$\preeq$ is a \emph{congruence} \wrt $\preleq$.
\end{lemma}

Preorders on specifications are lifted point-wise to preorders on linearization mappings:

\begin{notation}
	Given an implementation $\impl$, a specification $\seqspec$,
	linearization mappings $\linmap_1,\linmap_2:\exec{\impl}\to\seqspec$,
	and a preorder $\preleq$ on $\seqspec$,
	we write $\linmap_1\preleq\linmap_2$
	if $\linmap_1(\ex)\preleq\linmap_2(\ex)$ for every $\ex\in\exec{\impl}$.
\end{notation}

Note that if $\preleq$ is a preorder on specifications,
then its point-wise lift is a preorder on linearization mappings.

We are interested in families of linearization mappings for a specific implementation that satisfy a condition on the linearizations of prefixes.
(Strong, write strong and decisive linearizability are all instances of such a definition.)
This is formalized as follows:
\begin{definition2}\label{def:respectrel}
	Let $\impl$ be an implementation of an object $\obj$ and let $\seqspec$ be a specification of $\obj$.
	A linearization mapping $\linmap:\exec{\impl}\to\seqspec$ \emph{respects} a relation $\rel\suq \seqspec\times\seqspec$ if $\tup{\linmap(\ex_1),\linmap(\ex_2)}\in\rel$ for every
	$\ex_1,\ex_2\in\exec{\impl}$ such that $\ex_1\pref\ex_2$.

\end{definition2}

Let $\linmaps_{\seqspec}^\rel(\impl)$ denote the set of all linearization mappings $\linmap:\exec{\impl}\to\seqspec$ that respect $\rel$.
The following lemma is used to prove the existence of linearization mappings that are $\preleq$-minimal in $\linmaps_{\seqspec}^\rel(\impl)$:

\begin{lemma}\label{lem:proc}
	Let $\impl$ be an implementation.
	Let $\preleq$ be a well-founded preorder on a specification $\seqspec$ such that $\preeq \defeq {\preleq}\,\cap\,{\preleq^{-1}}$ is a congruence \wrt a relation $\rel\subseteq\seqspec\times\seqspec$.
	If $\linmaps_{\seqspec}^\rel(\impl)\neq\emptyset$, then $\linmaps_{\seqspec}^\rel(\impl)$
	contains some $\preleq$-minimal linearization mapping.
\end{lemma}

Note that while $\preleq$ is assumed to be well-founded, 
this property is not generally preserved when lifting point-wise to linearization mappings,
as they may be decreased infinitely often \wrt $\preleq$, by decreasing the linearizations of different executions.
Thus an additional argument is required to obtain a $\preleq$-minimal mapping.

\begin{proof}
	In the following proof, 
	to avoid confusion between $\preleq$ and its point-wise lifting to linearization mappings, we will use $\preleq^*$ to denote the point-wise lifting.
	
	Using Zorn's Lemma, it is sufficient to prove that any $\preleq^*$-chain of linearization mappings in $\linmaps_{\seqspec}^\rel(\impl)$ has a $\preleq^*$-lower bound.
	
	Let $\chain\subseteq\linmaps_{\seqspec}^\rel(\impl)$ be a $\preleq^*$-chain.

	If $\chain=\emptyset$, then any $\linmap\in\linmaps_{\seqspec}^\rel(\impl) $ is trivially a $\preleq^*$-lower bound of $\chain$.
	
	Otherwise, we define a $\preleq^*$-lower bound $\boundlin:\exec{\impl}\to\seqspec$ as follows:

	Let $\ex\in\exec{\impl}$.
	Since $\preleq$ is well-founded
	and the set $\set{\linmap(\ex)\st\linmap\in\chain}$
	is non-empty,
	it contains a (not necessarily unique) $\preleq$-minimal element.
	Since $\preleq^*$ is a total order on $\chain$, $\preleq$ is a total preorder on $\set{\linmap(\ex)\st\linmap\in\chain}$
	and so any such $\preleq$-minimal element is a $\preleq$-minimum.
	We define $\boundlin(\ex)$ to be an arbitrarily chosen  $\preleq$-minimum of $\set{\linmap(\ex)\st\linmap\in\chain}$.
	
	Clearly, $\boundlin$ is a $\preleq^*$-lower bound on $\chain$. %
	
	It remains to show that $\boundlin\in\linmaps_{\seqspec}^\rel(\impl)$.
	It is immediate that it is a linearization mapping since for every $\ex\in\exec{\impl}$,
	$\boundlin(\ex)=\linmap'(\ex)$ for some $\linmap'\in\linmaps_{\seqspec}^\rel(\impl)$.
	
	Let $\ex_1,\ex_2\in\exec{\impl}$ such that $\ex_1\pref\ex_2$.
	We show that $\tup{\boundlin(\ex_1),\boundlin(\ex_2)}\in\rel$.
	
	Let $\linmap_1,\linmap_2\in\chain$ be mappings such that $\boundlin(\ex_1)=\linmap_1(\ex_1)$ and $\boundlin(\ex_2)=\linmap_2(\ex_2)$.
	
	We show $\tup{\linmap_1(\ex_1),\linmap_2(\ex_2)}\in\rel$.
	
	Since $\preleq^*$ is a total order on $\chain$, one of the following holds:
	\begin{itemize}
		\item $\linmap_1\preleq^*\linmap_2$:
		
		In this case, we have $\linmap_1(\ex_2)\preleq\linmap_2(\ex_2)$.
		As $\linmap_1\in\chain$ and $\boundlin$ is a $\preleq^*$-lower bound of $\chain$
		we obtain that $\linmap_2(\ex_2)=\boundlin(\ex_2)\preleq\linmap_1(\ex_2)$.
		Thus, $\linmap_1(\ex_2)\preeq\linmap_2(\ex_2)$ and from congruence of $\preeq$ \wrt $\rel$ we obtain that
		$\tup{\linmap_1(\ex_1),\linmap_2(\ex_2)}\in\rel\iff
		\tup{\linmap_1(\ex_1),\linmap_1(\ex_2)}\in\rel$.
		Then, $\tup{\linmap_1(\ex_1),\linmap_1(\ex_2)}\in\rel$ follows from the fact that $\linmap_1\in\linmaps_{\seqspec}^\rel(\impl)$.
		
		\item $\linmap_2\preleq^*\linmap_1$:
		We similarly obtain that $\linmap_1(\ex_1)\preeq\linmap_2(\ex_1)$,
		and the desired property follows from $\linmap_2\in\linmaps_{\seqspec}^\rel(\impl)$.
\qedhere	
	\end{itemize}
\end{proof}

For (partial) orders, we obtain the following corollary:
\begin{corollary}\label{cor:proc}
	Let $\impl$ be an implementation.
	Let $\partleq$ be a well-founded (partial) order on a specification $\seqspec$,
	and let $\rel$ be a relation on $\seqspec$.
	If $\linmaps_{\seqspec}^\rel(\impl)\neq\emptyset$, then $\linmaps_{\seqspec}^\rel(\impl)$ contains some $\partleq$-minimal linearization mapping.
\end{corollary}
\begin{proof}
	Follows from \cref{lem:proc}, using the fact that $\partleq\cap\partleq^{-1}$ is equality, which is a congruence \wrt any relation $\rel$.
\end{proof}

An immediate consequence of the definition of minimality is that one cannot strictly decrease sequential histories returned by minimal mappings that respect $\rel$ without violating the respect of $\rel$.
In other words, any decrease which respects $\rel$ is not a strict decrease.
This fact will be useful to establish properties of these sequential histories.

\begin{lemma}\label{lem:local_change}
	Let $\impl$ be an implementation of an object $\obj$,
	let $\seqspec$ be a specification of $\obj$,
	let $\preleq$ be a preorder on $\seqspec$,
	and let $\rel$ be a reflexive relation on $\seqspec$.
	Suppose that $\linmap\in\linmaps_{\seqspec}^\rel(\impl)$ is $\preleq$-minimal,
	$\ex\in\exec{\impl}$,
	and $\seqh$ is a sequential history of $\obj$,
	such that:
	\begin{enumerate}
		\item $\seqh\in\seqspec$,
		\item $\hs{\ex}\linleq\seqh$,
		\item $\seqh\preleq \linmap(\ex)$,
		\item $\tup{\linmap(\ex'),\seqh}\in\rel$ for all $\ex'\in\exec{\impl}$ such that $\ex'\prefneq\ex$, and
		\item $\tup{\seqh,\linmap(\ex')}\in\rel$ for all $\ex'\in\exec{\impl}$ such that $\ex\prefneq\ex'$.
	\end{enumerate}
	Then, $\seqh\preeq\linmap(\ex)$.
\end{lemma}

\begin{proof}
	Consider the linearization mapping $\linmap':\exec{\impl}\to\seqspec$ defined by
	\begin{equation*}
		\linmap'(\ex')=\begin{cases}
			\linmap(\ex') & \ex'\neq\ex \\
			\seqh & \ex'=\ex
		\end{cases}
	\end{equation*}
	It is easy to verify that $\linmap'\in\linmaps_{\seqspec}^\rel(\impl)$ and $\linmap'\preleq\linmap$.
	Since $\linmap$ is $\preleq$-minimal, we deduce that $\linmap'\preeq\linmap$
	and in particular $\seqh=\linmap'(\ex)\preeq\linmap(\ex)$.
\end{proof}

\subsection{Minimal Linearizations for Write-Strong Registers}\label{subsec:wslazy}

In this section we apply the framework described in \cref{subsec:minimal} to write strong-linearizability.
First, we notice write strong-linearizability is an instance of \cref{def:respectrel}.
Indeed, define: $$\wsleq\eqdef\set{\tup{\seqh_1,\seqh_2}\in\regspec\times\regspec\st\seqh_1\rst{\regwrite}\pref\seqh_2\rst{\regwrite}}$$
Then $\linmaps_{\regspec}^{\wsleq}(\impl)$ is the set of write strong-linearizations of the implementation $\impl$,
and the class of write strongly-linearizable register implementations satisfies: $$\impl\in\wslin\iff\linmaps_{\regspec}^{\wsleq}(\impl)\neq\emptyset$$
This is a restatement of the definition from the body of the paper, using the extra notation we now have.

Now, we define a preorder on $\regspec$ (which provides a preorder on write strong linearizations via point-wise lifting):

\begin{definition2}
	For $\seqh_1,\seqh_2\in\regspec$, we define $\seqh_1\permleq\seqh_2$ if $\setof{\seqh_1}\subseteq\setof{\seqh_2}$ and $(\writeid\seq<_{\seqh_1})\subseteq<_{\seqh_2}$.
\end{definition2}
Note that $\permleq$ is a preorder.
We define $\permeq\defeq(\permleq\cap{(\permleq)}^{-1})$, the induced equivalence.
The condition $(\writeid\seq<_{\seqh_1})\subseteq<_{\seqh_2}$ means that if a write operation appears before some other operation is $\seqh_1$, than this is also the case in $\seqh_2$.
The idea behind the condition $\setof{\seqh_1}\subseteq\setof{\seqh_2}$ is that linearizations that are $\permleq$-minimal do not linearize operations unless they must.
Indeed, if we can get a valid linearization $\seqh_1$ (respecting conditions on prefixes) by removing some operation from $\seqh_2$,
then a linearization mapping using $\seqh_1$ instead of $\seqh_2$ will be strictly smaller since $\setof{\seqh_1}\subseteq\setof{\seqh_2}$ (and $(\writeid\seq<_{\seqh_1})\subseteq<_{\seqh_2}$) but $\setof{\seqh_2}\nsubseteq\setof{\seqh_1}$.
The condition $(\writeid\seq<_{\seqh_1})\subseteq<_{\seqh_2}$ will be later used to show that minimal linearization avoid in some cases the re-ordering of already linearized operations in the linearization of an extension of an execution.
This will be a useful step towards showing that minimal linearizations can be modified to become decisive, 
in which case the ordering between operations is never changed once they are linearized.

The following follow immediately from definition:

\begin{lemma}\label{lem:subseq_permleq}
	Let $\seqh_1,\seqh_2\in\regspec$ be register sequential histories.
	If $\seqh_1\subseq\seqh_2$, then $\seqh_1\permleq\seqh_2$.
\end{lemma}
\begin{lemma}\label{lem:permleq_writesubseq}
	Let $\seqh_1,\seqh_2\in\regspec$ be register sequential histories.
	If $\seqh_1\permleq\seqh_2$, then $\seqh_1\rst{\regwrite}\subseq\seqh_2\rst{\regwrite}$.
\end{lemma}

Using \cref{lem:proc} we obtain the existence of write strong-linearizations that are $\permleq$-minimal:

\begin{lemma}\label{lem:wsexistproc}
	If $\impl\in\wslin$,
	then $\linmaps_{\regspec}^{\wsleq}(\impl)$ contains a $\permleq$-minimal linearization.
\end{lemma}
\begin{proof}
	By \cref{lem:proc} it suffices to show that $\permeq$ is a congruence \wrt $\wsleq$.
	Indeed,
	if $\seqh_i\permeq\seqh'_i$ for $i\in\set{1,2}$,
	then from \cref{lem:permleq_writesubseq} we have $\seqh_i\rst{\regwrite}=\seqh'_i\rst{\regwrite}$ and thus $\seqh_1\wsleq\seqh_2\iff\seqh_1\rst{\regwrite}\pref\seqh_2\rst{\regwrite}\iff\seqh'_1\rst{\regwrite}\pref\seqh'_2\rst{\regwrite}\iff\seqh'_1\wsleq\seqh'_2$.
\end{proof}

We wish to use \cref{lem:local_change} to argue about properties of these minimal linearizations.
Due to \cref{lem:self_congruent}, conditions 4 and 5 in \cref{lem:local_change} can be replaced with a requirement that $\seqh\rst{\regwrite}=\linmap(\ex)\rst{\regwrite}$.
(In other words, replacing a linearization with another one with the same sequence of writes maintains the respect of $\wsleq$ between prefixes.)
This gives the following corollary of \cref{lem:local_change}:

\begin{lemma}\label{lem:local_change_ws}
	If $\impl\in\wslin$, $\linmap\in\linmaps_{\regspec}^{\wsleq}(\impl)$ is $\permleq$-minimal,	$\ex\in\exec{\impl}$,
	and $\seqh$ is a sequential history of $\Reg$,
	such that:
	\begin{enumerate}
		\item $\seqh\in\regspec$,
		\item $\hs{\ex}\linleq\seqh$,
		\item $\seqh\permleq \linmap(\ex)$, and
		\item $\seqh\rst{\regwrite}=\linmap(\ex)\rst{\regwrite}$.
	\end{enumerate}
	Then, $\seqh\permeq\linmap(\ex)$.
\end{lemma}

We now begin the process of proving several properties of $\permleq$-minimal write strong linearizations.
First, we show that they contain no pending reads:

\begin{lemma}\label{lem:readcompleted_ws}
	If $\impl\in\wslin$ and $\linmap\in\linmaps_{\regspec}^{\wsleq}(\impl)$ is $\permleq$-minimal,
	then $\linmap$ is pending-read-free.
\end{lemma}
\begin{proof}
	We need to show $\setof{\linmap(\ex)\rst{\regread}}\subseteq\setof{\completed(\hs{\ex})}$ for every $\ex\in\exec{\impl}$.
	
	Denote 	$\seqh=\linmap(\ex)\rst{\setof{\linmap(\ex)}\setminus\left(\setof{\linmap(\ex)\rst{\regread}}\setminus\setof{\completed(\hs{\ex})}\right)}$ the subsequence of $\linmap(\ex)$ where every pending read operation is removed.
	It is sufficient to show that $\seqh=\linmap(\ex)$.
	We apply \cref{lem:local_change_ws} to $\ex$ and $\seqh$. 
	We have:
	\begin{enumerate}
		\item $\seqh\in\regspec$: 
		Follows from \cref{lem:remove_read}.
		
		\item $\hs{\ex}\linleq\seqh$: 
		Follows from \cref{lem:remove_pending}, using the fact that $\hs{\ex}\linleq\linmap(\ex)$.
		
		\item $\seqh\permleq\linmap(\ex)$:
		Using \cref{lem:subseq_permleq} it is sufficient to verify that $\seqh\subseq\linmap(\ex)$, which follows immediately from the construction of $\seqh$.
		
		\item $\seqh\rst{\regwrite}=\linmap(\ex)\rst{\regwrite}$:
		This is again immediate from the construction, as we only remove reads from the sequence.
	\end{enumerate}
	Overall, we can apply \cref{lem:local_change_ws} and deduce that $\seqh\permeq\linmap(\ex)$.
	This implies that $\setof{\seqh}=\setof{\linmap(\ex)}$ and as $\seqh\subseq\linmap(\ex)$ we obtain $\seqh=\linmap(\ex)$.
\end{proof}

Next, we prove several properties related to maintaining operations and the ordering between them in the linearizations when extending an execution.
These properties are proven as steps in the way to obtaining a minimal write strong-linearization mapping that is also decisive, a property which is stronger than those discussed in the lemmas leading up to that point.

We start this process by showing that once operations appear in the minimal linearization of some execution, they remain in the linearization for all extensions:

\begin{lemma}\label{lem:wssetinclusion}
	If $\impl\in\wslin$ and $\linmap\in\linmaps_{\regspec}^{\wsleq}(\impl)$ is $\permleq$-minimal,
	then $\setof{\linmap(\ex_1)}\suq\setof{\linmap(\ex_2)}$ for all $\ex_1,\ex_2\in\exec{\impl}$ such that $\ex_1\pref\ex_2$.
\end{lemma}
\begin{proof}
	Let $\op\in\linmap(\ex_1)$. We consider two cases:
	If $\op$ is a write operation, then from $\linmap$ being write strong we have $\op\in\linmap(\ex_1)\rst{\regwrite}\pref\linmap(\ex_2)\rst{\regwrite}$
	and thus $\op\in\linmap(\ex_2)$.
	If $\op$ is a read operation, from \cref{lem:readcompleted_ws} we have  $\op\in\setof{\completed(\hs{\ex_1})}\subseteq\setof{\completed(\hs{\ex_2})}$, and by definition of linearizability we must get $\op\in\linmap(\ex_2)$, as required.
\end{proof}

Now, we show that not only is the set of operations respected by extensions of an execution, 
but also the ordering between a write operation and a read operation is maintained.
(This is also true for two write operations, as an immediate consequence of the mapping being write strong.)

Intuitively, there are two types of violations we must consider,
given two executions $\ex_1,\ex_2$ such that $\ex_1\pref\ex_2$,
a write operation $\wop$, and a read operation $\rop$:
The first is that $\wop\leq_{\linmap(\ex_1)}\rop$ and $\rop\leq_{\linmap(\ex_2)}\wop$,
and the second is that $\rop\leq_{\linmap(\ex_1)}\wop$ and $\wop\leq_{\linmap(\ex_2)}\rop$.

We wish to exploit the symmetry between arguments required to show both of these violations cannot occur,
and also handle the case where $\wop$ does not appear in $\linmap(\ex_1)$,
but as we show is still guaranteed to appear after $\rop$ in $\linmap(\ex_2)$.

This is done by comparing the sequences $\seqh_1\eqdef\linmap(\ex_1)\cdotnew(\linmap(\ex_2)\rst{\regwrite})$ (concatenating write operations not preset in $\linmap(\ex_1)$ to the end of $\linmap(\ex_1)$)
and $\seqh_2\eqdef\linmap(\ex_2)$.
These sequences can be handled symmetrically, using the observation that they both have the same subsequence of write operations:

\begin{observation}\label{lem:ws_preserve_writes}
	Let $\impl\in\wslin$, let $\linmap\in\linmaps_{\regspec}^{\wsleq}(\impl)$,
	let $\ex_1,\ex_2\in\exec{\impl}$ be executions such that $\ex_1\pref\ex_2$,
	and define $\seqh_1=\linmap(\ex_1)\cdotnew(\linmap(\ex_2)\rst{\regwrite})$, $\seqh_2=\linmap(\ex_2)$.
	Then,
	$\seqh_1\rst{\regwrite}=\seqh_2\rst{\regwrite}$.
\end{observation}
\begin{proof}
	Since $\linmap$ is write strong we have $\linmap(\ex_1)\rst{\regwrite}\pref\linmap(\ex_2)\rst{\regwrite}$.
		
	That is, $\seqh_2\rst{\regwrite}=\linmap(\ex_2)\rst{\regwrite}=\linmap(\ex_1)\rst{\regwrite}\cdot\linmap(\ex_2)\rst{\setof{\linmap(\ex_2)\rst{\regwrite}}\setminus\setof{\linmap(\ex_1)\rst{\regwrite}}}=\linmap(\ex_1)\rst{\regwrite}\cdotnew\linmap(\ex_2)\rst{\regwrite}=(\linmap(\ex_1)\cdotnew\linmap(\ex_2)\rst{\regwrite})\rst{\regwrite}=\seqh_1\rst{\regwrite}$
\end{proof}

\begin{lemma}\label{lem:ws_write_read_order}
	Let $\impl\in\wslin$ and $\linmap\in\linmaps_{\regspec}^{\wsleq}(\impl)$ such that $\linmap$ is $\permleq$-minimal.
	For all $\ex_1,\ex_2\in\exec{\impl}$ such that $\ex_1\pref\ex_2$,
	let $\seqh_1=\linmap(\ex_1)\cdotnew(\linmap(\ex_2)\rst{\regwrite})$ and $\seqh_2=\linmap(\ex_2)$.
	Then,
	$(\writeid\seq(<_{\seqh_{i}}\cap<_{\seqh_{\bar{i}}}^{-1}))=\emptyset$
	for $(i,\bar{i})\in\set{(1,2),(2,1)}$.
	
\end{lemma}

The idea is that a member of $(\writeid\seq(<_{\seqh_{i}}\cap<_{\seqh_{\bar{i}}}^{-1}))$ for some $i$ is a pair of operations $\tup{\wop,\rop}$ such that the sequences $\seqh_1$, $\seqh_2$ disagree on the ordering between them,
$\wop$ is a write operation, and $\rop$ must be a read operation (due to $\linmap$ being write strong).
This would be exactly a violation of the desired property as described above.

\begin{proof}
	
	Denote $\linrel\eqdef(\writeid\seq(<_{\seqh_{i}}\cap<_{\seqh_{\bar{i}}}^{-1}))$.
	Assume for contradiction $\linrel\neq\emptyset$. 
	We begin by proving properties of this relation:
	\begin{claim2}\label{claim:linrel_props}
		For every $\tup{\wop,\rop}\in\linrel$:
		\begin{enumerate*}
			\item $\methodf(\wop)=\regwrite$,
			\item $\methodf(\rop)=\regread$,
			\item $\wop,\rop\in\setof{\linmap(\ex_{i})}$,
			\item $\rop\in\setof{\linmap(\ex_{\bar{i}})}$,
			\item $\wop<_{\linmap(\ex_{i})}\rop$, and
			\item $\rop<_{\seqh_{\bar{i}}}\wop$.
		\end{enumerate*}
	\end{claim2}
	\begin{claimproof}
		First, from definition of $\linrel$ we get that $\methodf(\wop)=\regwrite$, $\wop<_{\seqh_{i}}\rop$, and $\rop<_{\seqh_{\bar{i}}}\wop$.
		
		Due to \cref{lem:ws_preserve_writes} we have $\seqh_{i}\rst{\regwrite}=\seqh_{\bar{i}}\rst{\regwrite}$.
		This means $\wop<_{\seqh_{i}}\rop$ and $\rop<_{\seqh_{\bar{i}}}\wop$ together imply $\methodf(\rop)=\regread$.
		
		Now, since $\wop,\rop$ are ordered by $\seqh_1$ we have in particular $\rop\in\setof{\linmap(\ex_1)\cdotnew(\linmap(\ex_2)\rst{\regwrite})}$.
		As $\methodf(\rop)=\regread$, it must be the case that $\rop\in\setof{\linmap(\ex_1)}$.
		Since $\wop,\rop$ are ordered by $\seqh_2$ we also have $\wop,\rop\in\setof{\linmap(\ex_2)}$.
		If $i=1$ then $\wop<_{\seqh_{i}}\rop\in\setof{\linmap(\ex_1)}$ implies that $\wop\in\setof{\linmap(\ex_1)}$.
		This means in any case we have $\wop,\rop\in\setof{\linmap(\ex_{i})}$ and $\rop\in\setof{\linmap(\ex_{\bar{i}})}$.
		From $\wop,\rop\in\setof{\linmap(\ex_{i})}$, $\wop<_{\seqh_{i}}\rop$, and $\linmap(\ex_{i})\pref\seqh_{i}$ we deduce that $\wop<_{\linmap(\ex_{i})}\rop$.
	\end{claimproof}
	
	Define a partial order $\partleq$ on $\linrel$ by $\tup{\wop',\rop'}\partleq\tup{\wop,\rop}\defiff\rop'\leq_{\seqh_{i}}\rop\land\wop'\leq_{\seqh_{\bar{i}}}\wop$.
	Since $\linrel$ is finite and non-empty,
	it contains a minimal element $\tup{\wopmin,\ropmin}$.
	This minimal element has the following useful properties:
	
	\begin{claim2}\label{claim:write_min}
		There is no write operation $\wop$ such that $\ropmin<_{\seqh_{\bar{i}}}\wop<_{\seqh_{\bar{i}}}\wopmin$.
	\end{claim2}
	\begin{claimproof}
		Otherwise, by \cref{lem:ws_preserve_writes} we have $\seqh_{i}\rst{\regwrite}=\seqh_{\bar{i}}\rst{\regwrite}$ and thus $\wop<_{\seqh_{i}}\wopmin$. As $\wopmin<_{\seqh_{i}}\ropmin$ we get $\wop<_{\seqh_{i}}\ropmin$ and together with  $\ropmin<_{\seqh_{\bar{i}}}\wop$ this implies $\tup{\wop,\ropmin}\in\linrel$.
		Since $\wop<_{\seqh_{\bar{i}}}\wopmin$ we have $\tup{\wop,\ropmin}\partlneq\tup{\wopmin,\ropmin}$,
		contradicting minimality.
	\end{claimproof}

	\begin{claim2}\label{claim:read_min}
		There is no operation $\op$ such that $\wopmin\leq_{\linmap(\ex_{i})}\op<_{\hs{\ex_{i}}}\ropmin$.
	\end{claim2}
	\begin{claimproof}
		Assume otherwise.
		From $\op<_{\hs{\ex_{i}}}\ropmin$ we have by definition that $\op$ responds before $\ropmin$ us invoked in the execution $\ex_i$.
		We also have that $\ropmin\in\setof{\linmap(\ex_1)}$, 
		and so it must be the case that $\ropmin$ was already invoked in $\ex_1$.
		This means that $\op$ already responded before that invocation in $\ex_1$.
		Therefore, by definition $\op<_{\hs{\ex_{1}}}\ropmin$ and as $<_{\hs{\ex_{1}}}\subseteq<_{\hs{\ex_{2}}}$ we obtain $\op<_{\hs{\ex_{2}}}\ropmin$.
		This also means $\op\in\completed(\hs{\ex_1})\subseteq\completed(\hs{\ex_2})$ and thus
		$\op\in\setof{\linmap(\ex_{1})}\cap\setof{\linmap(\ex_{2})}$.
		Overall we obtain for any $j\in\set{1,2}$ that 
		$\op<_{\hs{\ex_{j}}}\ropmin$ and as $\ropmin,\op\in\setof{\linmap(\ex_{j})}$ we have $\op<_{\linmap(\ex_{j})}\ropmin$.
		In particular $\op<_{\linmap(\ex_{\bar{i}})}\ropmin$ and
		as $\linmap(\ex_{\bar{i}})\pref\seqh_{\bar{i}}$ we deduce $\op<_{\seqh_{\bar{i}}}\ropmin$.
		As $\ropmin<_{\seqh_{\bar{i}}}\wopmin$ we obtain $\op<_{\seqh_{\bar{i}}}\wopmin$.
		In particular $\op\neq\wopmin$ and so from $\wopmin\leq_{\linmap(\ex_{i})}\op$ we obtain $\wopmin<_{\linmap(\ex_{i})}\op$.
		As $\linmap(\ex_{i})\pref\seqh_{i}$ we obtain $\wopmin<_{\seqh_{i}}\op$. 
		Now, consider two cases:
		\begin{itemize}
			\item If $\methodf(\op)=\regwrite$, $\op<_{\seqh_{\bar{i}}}\wopmin$ and $\wopmin<_{\seqh_{i}}\op$ together contradict $\seqh_{i}\rst{\regwrite}=\seqh_{\bar{i}}\rst{\regwrite}$, which holds due to \cref{lem:ws_preserve_writes}.
			
			\item 	If $\methodf(\op)=\regread$, we get that $\tup{\wopmin,\op}\in\linrel$.
			As shown above we have $\op<_{\hs{\ex_{i}}}\ropmin$ and thus $\op<_{\linmap(\ex_{i})}\ropmin$, which implies $\op<_{\seqh_{i}}\ropmin$.
			Overall $\tup{\wopmin,\op}\partlneq\tup{\wopmin,\ropmin}$,
			contradicting minimality. \claimqedhere
		\end{itemize}
	\end{claimproof}

	Now, let $\opset=\setof{\prefof{\linmap(\ex_{i})}{\wopmin}}$ and  $\bar{\opset}=\setof{\linmap(\ex_{i})}\setminus(\opset\cup\set{\ropmin})$.
	Let $\seqh$ be the sequence over $\setof{\linmap(\ex_{i})}$ induced by the total relation:
	$$(\opset\times\set{\ropmin})\cup(\set{\ropmin}\times\bar{\opset})\cup(<_{\linmap(\ex_{i})}\cap(\setof{\linmap(\ex_{i})}\setminus\set{\ropmin})^2)$$
	
	We proceed to show $\seqh$ satisfies the conditions of \cref{lem:local_change_ws}.
	
	\begin{enumerate}
		\item $\seqh\in\regspec$:

		We use \cref{lem:regspec}.
		As $\setof{\seqh}=\setof{\linmap(\ex_{i})}$
		and $\linmap(\ex_{j})\in\regspec$ for $j\in\set{1,2}$
		it is sufficient to show for all $\op\in\seqh\rst{\regread}$ that
		$\prefof{\seqh}{\op}\rst{\regwrite}=\prefof{\linmap(\ex_{j})}{\op}\rst{\regwrite}$ for some $j\in\set{1,2}$ such that $\op\in\setof{\linmap(\ex_{j})}$.
		
		Indeed:
		\begin{itemize}
			\item If $\op\neq\ropmin$ 
			then taking $j=i$ we have that $\op\in\setof{\seqh}=\setof{\linmap(\ex_{i})}$, and we can show $\prefof{\seqh}{\op}\rst{\regwrite}=\prefof{\linmap(\ex_{i})}{\op}\rst{\regwrite}$.
			This is true since for all $\wop\in \setof{\seqh}=\setof{\linmap(\ex_{i})}$ such that $\methodf(\wop)=\regwrite$ we have
			$\wop,\op\in\setof{\linmap(\ex_{i})}\setminus\set{\ropmin}$ and thus
			$\wop<_{\seqh}\op\iff \wop'<_{\linmap(\ex_{i})}\op$.
			
			\item If $\op=\ropmin$ we take $j=\bar{i}$.
			From \cref{claim:linrel_props} we have that $\op=\ropmin\in\setof{\linmap(\ex_{\bar{i}})}$.
			We show that
			$\prefof{\seqh}{\ropmin}\rst{\regwrite}=\prefof{\linmap(\ex_{\bar{i}})}{\ropmin}\rst{\regwrite}$.
			Indeed, we have:
			$$\prefof{\seqh}{\ropmin}\rst{\regwrite}=\prefof{\linmap(\ex_{i})}{\wopmin}\rst{\regwrite}=\prefof{\seqh_i}{\wopmin}\rst{\regwrite}=\prefof{\seqh_{\bar{i}}}{\wopmin}\rst{\regwrite}$$
			where the equalities are justified by:
			The construction of $\seqh$,
			the fact due to \cref{claim:linrel_props} and the definition of $\seqh_i$ that $\wopmin\in\linmap(\ex_{i})\pref\seqh_{i}$,
			and the fact due to \cref{lem:ws_preserve_writes} that
			$\seqh_{i}\rst{\regwrite}=\seqh_{\bar{i}}\rst{\regwrite}$.

			Similarly, since $\ropmin\in\linmap(\ex_{\bar{i}})\pref\seqh_{\bar{i}}$ we have
			$\prefof{\linmap(\ex_{\bar{i}})}{\ropmin}\rst{\regwrite}=\prefof{\seqh_{\bar{i}}}{\ropmin}\rst{\regwrite}$.
			
			Overall it remains to show that
			$\prefof{\seqh_{\bar{i}}}{\wopmin}\rst{\regwrite}=\prefof{\seqh_{\bar{i}}}{\ropmin}\rst{\regwrite}$.
			Equivalently, we need to show that for a write operation $\wop$, we have $\wop<_{\seqh_{\bar{i}}}\wopmin\iff \wop<_{\seqh_{\bar{i}}}\ropmin$.
			Indeed,
			from \cref{claim:linrel_props} we have $\ropmin<_{\seqh_{\bar{i}}}\wopmin$ which gives the implication $\wop<_{\seqh_{\bar{i}}}\ropmin\implies \wop<_{\seqh_{\bar{i}}}\wopmin$.
			If we assume for contradiction the other implication does not hold, then there must exist a write operation $\wop$ such that $\ropmin<_{\seqh_{\bar{i}}}\wop<_{\seqh_{\bar{i}}}\wopmin$, but this is impossible due to \cref{claim:write_min}.
		\end{itemize}
		
		\item $\hs{\ex_i}\linleq \seqh$:
		
		As $\hs{\ex_i}\linleq\linmap(\ex_{i})$ and by construction $\setof{\seqh}=\setof{\linmap(\ex_{i})}$,
		using \cref{lem:lin} it is sufficient to show that $\op'<_{\hs{\ex_i}}\op$ implies $\op'<_{\seqh}\op$ for all $\op,\op'\in\setof{\seqh}$.
		
		Indeed, from $\op'<_{\hs{\ex_i}}\op$ and the fact that $\op,\op'\in\setof{\seqh}=\setof{\linmap(\ex_{i})}$
		we get by definition of linearizability that $\op'<_{\linmap(\ex_{i})}\op$.
		Consider the following cases:
		\begin{itemize}
			\item If $\op,\op'\in\setof{\linmap(\ex_{i})}\setminus\set{\ropmin}$ then by construction $\op'<_{\seqh}\op$.
			
			\item If $\op'=\ropmin$ then using \cref{claim:linrel_props} we have $\wopmin<_{\linmap(\ex_{i})}\ropmin=\op'<_{\linmap(\ex_{i})}\op$ and so  $\wopmin<_{\linmap(\ex_{i})}\op$,
			This means $\op\notin\prefof{\linmap(\ex_{i})}{\wopmin}=\opset$.
			$\op=\ropmin$ is impossible as $\op'<_{\linmap(\ex_{i})}\op$ implies that $\op\neq\op'=\ropmin$,
			and thus overall $\op\in\bar{\opset}$,
			which implies $\op'=\ropmin<_{\seqh}\op$.
			
			\item If $\op=\ropmin$, due to \cref{claim:read_min} it is not possible that $\wopmin\leq_{\linmap(\ex_{i})}\op'<_{\hs{\ex_{i}}}\ropmin=\op$,
			and so we must have that $\op'<_{\linmap(\ex_{i})}\wopmin$ and thus
			$\op'\in\prefof{\linmap(\ex_{i})}{\wopmin}=\opset$
			This implies $\op'<_{\seqh}\ropmin=\op$.
		\end{itemize}

		\item $\seqh\permleq\linmap(\ex_{i})$:

		By construction $\setof{\seqh}=\setof{\linmap(\ex_{i})}$.
		It remains to show that $(\writeid\seq<_{\seqh})\subseteq<_{\linmap(\ex_{i})}$.
		
		Let $\op,\op'$ be operations such that $\methodf(\op)=\regwrite$ and $\op<_{\seqh}\op'$.
		In particular $\op\neq\ropmin$.
		Consider the following options:
		\begin{itemize}
			\item If $\op'\neq\ropmin$, then $\op,\op'\in\setof{\linmap(\ex_{i})}\setminus\set{\ropmin}$ and we must have by construction that  $\op<_{\linmap(\ex_{i})}\op'$.
			
			\item If $\op'=\ropmin$, then $\op<_{\seqh}\op'=\ropmin$ implies that $\op\in\opset$ and thus $\op<_{\linmap(\ex_{i})}\wopmin$.
			Due to \cref{claim:linrel_props} we have $\wopmin<_{\linmap(\ex_{i})}\ropmin$ and thus  $\op<_{\linmap(\ex_{i})}\ropmin=\op'$, as required.
		\end{itemize}

		\item $\seqh\rst{\regwrite}=\linmap(\ex_{i})\rst{\regwrite}$:

		Given any pair of write operations $\wop,\wop'\in\setof{\seqh}$ we have that $\wop,\wop'\in\setof{\linmap(\ex_{i})}\setminus\set{\ropmin}$ and so by construction $\wop<_{\seqh}\wop'\iff \wop<_{\linmap(\ex_{i})}\wop'$.
	
	\end{enumerate}
	
	From \cref{lem:local_change_ws} we deduce that $\seqh\permeq\linmap(\ex_{i})$
	and in particular $(\writeid\seq<_{\linmap(\ex_{i})})\subseteq<_{\seqh}$.
	Due to \cref{claim:linrel_props} we have $\wopmin<_{\linmap(\ex_{i})}\ropmin$,
	and so using the inclusion we get $\wopmin<_{\seqh}\ropmin$.
	By construction of $\seqh$ this means $\wopmin\in\opset=\setof{\prefof{\linmap(\ex_{i})}{\wopmin}}$, a contradiction.
	
	We conclude that $\linrel=\emptyset$.\qedhere

\end{proof}

We now extend the above result,
and show the ordering between \emph{any two operations} is maintained, as long as there is some write operation between them. 
(The previous lemma is then the specific case where one of the operations is itself a write.)
Formally we prove the following:

\begin{lemma}\label{lem:write_between_ops}
	If $\impl\in\wslin$,
	$\linmap\in\linmaps_{\regspec}^{\wsleq}(\impl)$ is $\permleq$-minimal,
	and $\wop$ is a write operation,
	then\\ $(\leq_{\seqh_{1}}\seq\idrel{\wop}\seq\leq_{\seqh_{1}})=(\leq_{\seqh_{2}}\seq\idrel{\wop}\seq\leq_{\seqh_{2}})\cap(\setof{\seqh_1}\times\setof{\seqh_1})$
	for all $\ex_1,\ex_2\in\exec{\impl}$ such that $\ex_1\pref\ex_2$,
	where $\seqh_1=\linmap(\ex_1)\cdotnew(\linmap(\ex_2)\rst{\regwrite})$ and $\seqh_2=\linmap(\ex_2)$.
\end{lemma}

That is, every two operations present in $\seqh_1$ (and thus in $\seqh_2$) are ordered the same way by $\seqh_1$,$\seqh_2$,
as long as there is some write operation $\wop$ between them.
(Which leaves only the ordering between reads with no write between them, which might not be maintained.)

\begin{proof}
	As $(\leq_{\seqh_{1}}\seq\idrel{\wop}\seq\leq_{\seqh_{1}})=(\leq_{\seqh_{1}}\seq\idrel{\wop}\seq\leq_{\seqh_{1}})\cap(\setof{\seqh_1}\times\setof{\seqh_1})$ it is sufficient to show that
	$(\leq_{\seqh_{i}}\seq\idrel{\wop}\seq\leq_{\seqh_{i}})\cap(\setof{\seqh_1}\times\setof{\seqh_1})\subseteq(\leq_{\seqh_{\bar{i}}}\seq\idrel{\wop}\seq\leq_{\seqh_{\bar{i}}})$
	for $(i,\bar{i})\in\set{(1,2),(2,1)}$.
	
	Let $\tup{\op,\op'}\in(\leq_{\seqh_{i}}\seq\idrel{\wop}\seq\leq_{\seqh_{i}})\cap(\setof{\seqh_1}\times\setof{\seqh_1})$.
	Then $\op\leq_{\seqh_{i}}\wop\leq_{\seqh_{i}}\op'$ and in particular  $\op,\wop,\op'\in\setof{\seqh_{i}}$.
	We can show $\op,\wop,\op'\in\setof{\seqh_{\bar{i}}}$.
	Indeed,	Consider both options for $i$:	
	\begin{itemize}
		\item If $i=1$ then $\op,\wop,\op'\in\setof{\seqh_{2}}$ follows from  $\op,\wop,\op'\in\setof{\seqh_{1}}$ and $\setof{\seqh_1}\subseteq\setof{\seqh_2}$, which follows from  \cref{lem:wssetinclusion}.
		
		\item If $i=2$ then $\op,\op'\in\setof{\seqh_{1}}$ follows from $\tup{\op,\op'}\in(\setof{\seqh_1}\times\setof{\seqh_1})$ and $\wop\in\setof{\seqh_{1}}$ follows from $\wop\in\setof{\seqh_{2}}$ and $\seqh_1\rst{\regwrite}=\seqh_2\rst{\regwrite}$, which follows from \cref{lem:ws_preserve_writes}.
	\end{itemize}

	We deduce that $\op,\wop,\op'$ are ordered by $\leq_{\seqh_{\bar{i}}}$.
	Moreover, we have:
	
	\begin{itemize}
		
		\item $\op\leq_{\seqh_{\bar{i}}}\wop$:
		
		If $\op=\wop$ this is trivial.
		Otherwise, $\op<_{\seqh_{i}}\wop$.
		If we assume for contradiction $\wop<_{\seqh_{\bar{i}}}\op$ we get 	$\tup{\wop,\op}\in(\writeid\seq(<_{\seqh_{\bar{i}}}\cap<_{\seqh_{i}}^{-1}))$,
		contradicting \cref{lem:ws_write_read_order}.
		We deduce $\op<_{\seqh_{\bar{i}}}\wop$.
		
		\item $\wop\leq_{\seqh_{\bar{i}}}\op'$:
		
		If $\op'=\wop$ this is trivial.
		Otherwise, $\wop<_{\seqh_{i}}\op'$.
		If we assume for contradiction $\op'<_{\seqh_{\bar{i}}}\wop$ we get 	$\tup{\wop,\op'}\in(\writeid\seq(<_{\seqh_{i}}\cap<_{\seqh_{\bar{i}}}^{-1}))$,
		contradicting \cref{lem:ws_write_read_order}.
		We deduce $\wop<_{\seqh_{\bar{i}}}\op'$.
		
	\end{itemize}
	
	Overall, $\op\leq_{\seqh_{\bar{i}}}\wop\leq_{\seqh_{\bar{i}}}\op'$ and thus $\tup{\op,\op'}\in(\leq_{\seqh_{\bar{i}}}\seq\idrel{\wop}\seq\leq_{\seqh_{\bar{i}}})$
	as required.
\end{proof}

Finally, we use the above to show that while a $\permleq$-minimal write strong linearization does not have to be decisive,
one can transform it into a decisive linearization which is also minimal, thus not losing any of the other properties established so far.
Specifically, we modify an arbitrary minimal mapping to ensure that the ordering between reads that have no write between them does not change once they are linearized.
As \cref{lem:write_between_ops} ensures the ordering between operations is maintained in all other cases,
we overall obtain decisiveness.

\begin{lemma}\label{lem:wssubsq}
	If $\impl\in\wslin$, then there exists a write strong and decisive linearization %
	that is $\permleq$-minimal in $\linmaps_{\regspec}^{\wsleq}(\impl)$.
\end{lemma}
\begin{proof}
	
	Let $\linmap\in\linmaps_{\regspec}^{\wsleq}(\impl)$ be a $\permleq$-minimal write strong linearization as guaranteed to exist by \cref{lem:wsexistproc}.
	Let $\ex\in\exec{\impl}$.

	First we define the following relation:
	$$\linrel(\ex)=(\leq_{\linmap(\ex)}\seq\writeid\seq\leq_{\linmap(\ex)})\setminus [\setof{\linmap(\ex)}]$$
	
	As we later prove, this is a strict partial order.
	$\linrel(\ex)$ orders operations the same way $<_{\linmap(\ex)}$ does,
	as long as the operations have some write operation between them.
	Otherwise, the ordering between the operations is discarded.
	The idea is to only keep the ordering that will be preserved when extending an execution due to \cref{lem:write_between_ops}, 
	and to reorder the reads between any two writes in a way which overall guarantees decisiveness.
	
	Specifically, we will use the following ordering for reads with no write between them:
	$$\readrespondsbefore(\ex)=\set{\tup{\rop,\rop'}\in\setof{\linmap(\ex)\rst{\regread}}^2\st \rop \text{ responds before }\rop' \text{ in }\ex}\setminus(\linrel(\ex)\cup\linrel(\ex)^{-1})$$
	
	As we show later, this is also a strict partial order.
	This relation orders reads that have no write between them in $\linmap(\ex)$, based on which one responded first.
	
	We use the above relations to construct a linearization mapping $\linmap'\in\linmaps_{\regspec}^{\wsleq}(\impl)$,
	where $\linmap'(\ex)$ is defined for every $\ex\in\exec{\impl}$ as the sequence over $\setof{\linmap(\ex)}$ induced by the relation $\linrel(\ex)\cup\readrespondsbefore(\ex)$,
	that, as we prove later, is a strict total order.
	The idea is that operations with a write between then in $\linmap(\ex)$ are ordered like in $\linmap(\ex)$,
	and operations with no write between them in $\linmap(\ex)$ are two reads, ordered based on which one responded first.

	Finally, we prove that $\linmap'$ has all desired properties.
	We proceed now with filling out the details of the proof.
	First we show that $\linrel(\ex)$ and $\readrespondsbefore(\ex)$ are indeed strict partial orders:
	
	\begin{claim2}\label{claim:linrel_order}
		$\linrel(\ex)$ is a strict partial order.
	\end{claim2}
	\begin{claimproof}
		Clearly, $\linrel(\ex)\subseteq <_{\linmap(\ex)}$.
		Irreflexivity 
		then immediately follows from the inclusion in a strict partial order.
		For transitivity, assume $\tup{\op_1,\op_2},\tup{\op_2,\op_3}\in\linrel(\ex)$.
		From $\tup{\op_1,\op_2}\in\linrel(\ex)$ we have that there exists a a write operation $\wop$ such that $\op_1\leq_{\linmap(\ex)}\wop\leq_{\linmap(\ex)}\op_2$.
		From $\tup{\op_2,\op_3}\in\linrel(\ex)$ we get that $\op_2\leq_{\linmap(\ex)}\op_3$,
		and so overall $\op_1\leq_{\linmap(\ex)}\wop\leq_{\linmap(\ex)}\op_3$
		which implies $\tup{\op_1,\op_3}\in\linrel(\ex)$.
	\end{claimproof}
	
	\begin{claim2}\label{claim:rfr_order}
		$\readrespondsbefore(\ex)$ is a strict partial order.
	\end{claim2}
	\begin{claimproof}
		First, the ``responds before'' relation is irreflexive 
		as a relation over all operations, and so $\readrespondsbefore(\ex)$ is also irreflexive.
		For transitivity, assume
		$\tup{\op_1,\op_2},\tup{\op_2,\op_3}\in\readrespondsbefore(\ex)$.
		Clearly, $\methodf(\op_1)=\methodf(\op_2)=\methodf(\op_3)=\regread$ and $\op_1$ responds before $\op_3$ in $\ex$.
		If we assume for contradiction $\tup{\op_1,\op_3}\in(\linrel(\ex)\cup\linrel(\ex)^{-1})$,
		we get that there exists a write operation $\wop$ between $\op_1$ and $\op_3$ in $\linmap(\ex)$.
		However, due to $\tup{\op_1,\op_2},\tup{\op_2,\op_3}\in\readrespondsbefore(\ex)$, this write operation cannot be between $\op_1$ and $\op_2$ or between $\op_2$ and $\op_3$, which is a contradiction for any relative ordering of $\op_1$,$\op_2$ and $\op_3$ in $\linmap(\ex)$,
		as the union of the segments between $\op_1$ and $\op_2$ and between $\op_2$ and $\op_3$ always contains the segment between $\op_1$ and $\op_3$.
	\end{claimproof}

	Now, we show that $\linrel(\ex)\cup\readrespondsbefore(\ex)$ is a strict total order, thus inducing a sequence of operations.
	
	\begin{claim2}
		$\linrel(\ex)\cup\readrespondsbefore(\ex)$ is a strict total order over $\linmap(\ex)$.
	\end{claim2}
	\begin{claimproof}
		Denote $\linreltotal(\ex)\eqdef\linrel(\ex)\cup\readrespondsbefore(\ex)$
		First, we show $\linreltotal(\ex)$ is a strict (partial) order.
		Irreflexivity follows from the fact that $\linrel(\ex)$ and $\readrespondsbefore(\ex)$ are irreflexive.
		Transitivity follows from a simple case analysis.
		Let $\op_1,\op_2,\op_3$ be operations such that
		$\tup{\op_1,\op_2},\tup{\op_2,\op_3}\in\linreltotal(\ex)$.
		There are four cases to consider:
		\begin{itemize}
			\item If $\tup{\op_1,\op_2},\tup{\op_2,\op_3}\in\linrel(\ex)$, then by transitivity of $\linrel(\ex)$ we obtain that $\tup{\op_1,\op_3}\in\linrel(\ex)\subseteq\linreltotal(\ex)$.
			
			\item Similarly, if $\tup{\op_1,\op_2},\tup{\op_2,\op_3}\in\readrespondsbefore(\ex)$, then by transitivity of $\readrespondsbefore(\ex)$ we obtain that $\tup{\op_1,\op_3}\in\readrespondsbefore(\ex)\subseteq\linreltotal(\ex)$.
			
			\item If $\tup{\op_1,\op_2}\in\linrel(\ex)$, $\tup{\op_2,\op_3}\in\readrespondsbefore(\ex)$,
			then there exists a write operation $\wop$ such that
			$\op_1\leq_{\linmap(\ex)}\wop\leq_{\linmap(\ex)}\op_2$.
			It must be the case that $\wop<_{\linmap(\ex)}\op_3$,
			as otherwise $\op_3\leq_{\linmap(\ex)}\wop\leq_{\linmap(\ex)}\op_2$ which implies $\tup{\op_3,\op_2}\in\linrel(\ex)$, contradicting $\tup{\op_2,\op_3}\in\readrespondsbefore(\ex)$.
			Overall we obtain that $\op_1\leq_{\linmap(\ex)}\wop\leq_{\linmap(\ex)}\op_3$
			and thus $\tup{\op_1,\op_3}\in\linrel(\ex)\subseteq\linreltotal(\ex)$.
			
			\item If $\tup{\op_1,\op_2}\in\readrespondsbefore(\ex)$, $\tup{\op_2,\op_3}\in\linrel(\ex)$,
			then there exists a write operation $\wop$ such that
			$\op_2\leq_{\linmap(\ex)}\wop\leq_{\linmap(\ex)}\op_3$.
			It must be the case that $\op_1<_{\linmap(\ex)}\wop$,
			as otherwise $\op_2\leq_{\linmap(\ex)}\wop\leq_{\linmap(\ex)}\op_1$ which implies $\tup{\op_2,\op_1}\in\linrel(\ex)$, contradicting $\tup{\op_1,\op_2}\in\readrespondsbefore(\ex)$.
			Overall we obtain that $\op_1\leq_{\linmap(\ex)}\wop\leq_{\linmap(\ex)}\op_3$
			and thus $\tup{\op_1,\op_3}\in\linrel(\ex)\subseteq\linreltotal(\ex)$.
		\end{itemize}
		
		In all cases we got that $\tup{\op_1,\op_3}\in\linreltotal(\ex)$ and so $\linreltotal(\ex)$ is transitive.
		
		Next we show $\linreltotal(\ex)$ is total. Let $\op\neq\op'\in\setof{\linmap(\ex)}$.
		As $<_{\linmap(\ex)}$ is a strict total order we can assume \wlg that $\op<_{\linmap(\ex)}\op'$.
		If $\tup{\op,\op'}\in\linrel(\ex)$, then $\tup{\op,\op'}\in\linreltotal(\ex)$ and we are done.
		Otherwise, $\tup{\op,\op'}\notin\linrel(\ex)$.
		Therefore, as $\op<_{\linmap(\ex)}\op'$, it must be the case that $\methodf(\op)=\methodf(\op')=\regread$.
		From $\op<_{\linmap(\ex)}\op'$ we can additionally deduce that $\tup{\op,\op'}\notin<_{\linmap(\ex)}^{-1}$,
		and as $\linrel(\ex)\subseteq<_{\linmap(\ex)}$ we obtain that $\tup{\op,\op'}\notin\linrel(\ex)^{-1}$.
		Moreover, from \cref{lem:readcompleted_ws} we have that $\op,\op'\in \completed(\hs{\ex})$ and thus both operations have a response in $\ex$.
		Overall $\op$ and $\op'$ are ordered by $\readrespondsbefore(\ex)$ and thus ordered by $\linreltotal(\ex)$.
	\end{claimproof}
	
	We obtain a mapping $\linmap':\exec{\impl}\to\regspec$ where $\linmap'(\ex)$ is the sequence over $\setof{\linmap(\ex)}$ induced by the strict total order $\linrel(\ex)\cup\readrespondsbefore(\ex)$.
	For the purpose of showing $\linmap'$ is indeed a linearization mapping with all required properties, 
	we first prove the following claim, 
	which says $\linmap$ and $\linmap'$ agree on the ordering between write operations and operations which come after them:

	\begin{claim2}\label{claim:write-then-op}
		$(\writeid\seq<_{\linmap(\ex)})=(\writeid\seq<_{\linmap'(\ex)})$
		for all $\ex\in\exec{\impl}$.
	\end{claim2}
	\begin{claimproof}
		By construction we have $\setof{\linmap(\ex)}=\setof{\linmap'(\ex)}$.
		Thus, it suffices to show that if $\wop,\op\in\setof{\linmap(\ex)}$ and $\methodf(\wop)=\regwrite$ then $\wop<_{\linmap(\ex)}\op\iff\wop<_{\linmap'(\ex)}\op$.
		
		In one direction, by definition of $\linrel(\ex)$ we have $\wop<_{\linmap(\ex)}\op\implies\tup{\wop,\op}\in\linrel(\ex)\implies\wop<_{\linmap'(\ex)}\op$.
		
		For the converse, as $\wop$ is a write operation, $\tup{\wop,\op}\notin\readrespondsbefore(\ex)$
		and so $\wop<_{\linmap'(\ex)}\op\implies\tup{\wop,\op}\in\linrel(\ex)$.
		As $\linrel(\ex)\subseteq <_{\linmap(\ex)}$ we deduce $\wop<_{\linmap(\ex)}\op$.
	\end{claimproof}
	
	Now, we show $\linmap'$ has all required properties:
	
	\begin{itemize}
		\item $\linmap'$ is a linearization mapping:
		
		Let $\ex\in\exec{\impl}$. First, we use \cref{lem:regspec} to show $\linmap'(\ex)\in\regspec$.
		Indeed, $\linmap(\ex)\in\regspec$, $\setof{\linmap'(\ex)}=\setof{\linmap(\ex)}$ and
		given any $\rop\in \setof{\linmap'(\ex)\rst{\regread}}$ using \cref{claim:write-then-op}
		we get that $(\writeid\seq<_{\linmap(\ex)}\seq\idrel{\rop})=(\writeid\seq<_{\linmap'(\ex)}\seq\idrel{\rop})$,
		or equivalently
		$\prefof{\linmap'(\ex)}{\rop}\rst{\regwrite}=\prefof{\linmap(\ex)}{\rop}\rst{\regwrite}$ and thus from \cref{lem:regspec} we get $\linmap'(\ex)\in\regspec$.
		
		Now, we show $\hs{\ex}\linleq\linmap'(\ex)$.
		As $\hs{\ex}\linleq\linmap(\ex)$ and by construction $\setof{\linmap(\ex)}=\setof{\linmap'(\ex)}$,
		using \cref{lem:lin} it is sufficient to show that $\op<_{\hs{\ex}}\op'\implies \op<_{\linmap'(\ex)}\op'$ for all $\op,\op'\in\setof{\linmap'(\ex)}$.
		
		We show the contra-positive: let $\op,\op'\in\setof{\linmap'(\ex)}$ such that
		$\op<_{\linmap'(\ex)}\op'$ does not hold.
		As $<_{\linmap'(\ex)}$ is total, we must have $\op'<_{\linmap'(\ex)}\op$, and by construction there are two options:
		\begin{itemize}
			\item $\tup{\op',\op}\in\linrel(\ex)$:
			
			In this case as $\linrel(\ex)\subseteq<_{\linmap(\ex)}$ we get $\op'<_{\linmap(\ex)}\op$ which since $\linmap$ is a legal linearization mapping implies $\op<_{\hs{\ex}}\op'$ does not hold.
			
			\item $\tup{\op',\op}\in\readrespondsbefore(\ex)$:
			
			In this case $\op'$ responds before $\op$ in $\ex$ and thus $\op<_{\hs{\ex}}\op'$ does not hold.
		\end{itemize}
		
		\item $\linmap'$ is write-strong:
		
		From \cref{claim:write-then-op} we get that $(\writeid\seq<_{\linmap(\ex)}\seq\writeid)=(\writeid\seq<_{\linmap'(\ex)}\seq\writeid)$ or equivalently $\linmap'(\ex)\rst{\regwrite}=\linmap(\ex)\rst{\regwrite}$ for all $\ex\in\exec{\impl}$.
		
		Therefore, using the fact that $\linmap$ is write-strong we get $\ex_1\pref\ex_2\implies \linmap(\ex_1)\rst{\regwrite}\pref\linmap(\ex_2)\rst{\regwrite}\implies \linmap'(\ex_1)\rst{\regwrite}\pref\linmap'(\ex_2)\rst{\regwrite}$
		for all $\ex_1,\ex_2\in\exec{\impl}$.
		
		\item $\linmap'$ is $\permleq$-minimal in $\linmaps_{\regspec}^{\wsleq}(\impl)$:
		
		From \cref{claim:write-then-op} and the fact that by construction $\setof{\linmap(\ex)}=\setof{\linmap'(\ex)}$,
		we  get have $\linmap'(\ex)\permeq\linmap(\ex)$ for all $\ex\in\exec{\impl}$ and thus $\linmap'\permeq\linmap$.
		The claim then follows from $\permleq$-minimality of $\linmap$.
		
		\item $\linmap'$ is decisive:
		
		Let $\ex_1,\ex_2\in\exec{\impl}$ be executions such that $\ex_1\pref\ex_2$.
		We need to show $\linmap'(\ex_1)\subseq\linmap'(\ex_2)$, 
		or equivalently $<_{\linmap'(\ex_1)}\subseteq<_{\linmap'(\ex_2)}$.
		By definition of $\linmap'$, it is sufficient to prove that $\linrel(\ex_1)\subseteq\linrel(\ex_2)$ and $\readrespondsbefore(\ex_1)\subseteq\readrespondsbefore(\ex_2)$.
		Indeed:
		\begin{itemize}
			\item $\linrel(\ex_1)\subseteq\linrel(\ex_2)$:
			
			Let $\seqh_1=\linmap(\ex_1)\cdotnew(\linmap(\ex_2)\rst{\regwrite})$.
			As $\linmap(\ex_1)\pref\seqh_1$ we have $(\leq_{\linmap(\ex_1)}\seq\writeid\seq\leq_{\linmap(\ex_1)})\subseteq(\leq_{\seqh_1}\seq\writeid\seq\leq_{\seqh_1})$.
			Using \cref{lem:write_between_ops}, taking union over all write operations, we have
			$(\leq_{\seqh_1}\seq\writeid\seq\leq_{\seqh_1})\subseteq(\leq_{\linmap(\ex_2)}\seq\writeid\seq\leq_{\linmap(\ex_2)})$.
			It follows that
			$\linrel(\ex_1)\subseteq\linrel(\ex_2)$.
			
			\item $\readrespondsbefore(\ex_1)\subseteq\readrespondsbefore(\ex_2)$:
			
			Let $\tup{\op,\op'}\in\readrespondsbefore(\ex_1)$.
			Then $\methodf(\op)=\methodf(\op')=\regread$ and $\op$ responds before $\op'$ in $\ex_1$.
			As $\ex_1\pref\ex_2$ this is also the case in $\ex_2$.
			It remains to show that $\tup{\op,\op'}\notin(\linrel(\ex_2)\cup\linrel(\ex_2)^{-1})$.
			Assume otherwise.
			If $\tup{\op,\op'}\in\linrel(\ex_2)$, then there exists a write operation $\wop\in\setof{\linmap(\ex_2)}$ such that $\op\leq_{\linmap(\ex_2)}\wop\leq_{\linmap(\ex_2)}\op'$.
			From \cref{lem:write_between_ops} we have
			$(\leq_{\linmap(\ex_2)}\seq\idrel{\wop}\seq\leq_{\linmap(\ex_2)})\cap(\setof{\seqh_1}\times\setof{\seqh_1})\subseteq(\leq_{\seqh_{1}}\seq\idrel{\wop}\seq\leq_{\seqh_{1}})$
			where $\seqh_1=\linmap(\ex_1)\cdotnew(\linmap(\ex_2)\rst{\regwrite})$.
			As $\op,\op'$ have a response in $\ex_1$ we have $\op,\op'\in\setof{\linmap(\ex_1)}\subseteq\setof{\seqh_1}$, and we obtain that
			$\tup{\op,\op'}\in(\leq_{\seqh_{1}}\seq\idrel{\wop}\seq\leq_{\seqh_{1}})$.
			That is, $\op\leq_{\seqh_{1}}\wop\leq_{\seqh_{1}}\op'$.
			As $\linmap(\ex_1)\pref\seqh_{1}$
			and $\op'\in\setof{\linmap(\ex_1)}$ we must have  $\op\leq_{\linmap(\ex_1)}\wop\leq_{\linmap(\ex_1)}\op'$
			and thus $\tup{\op,\op'}\in\linrel(\ex_1)$, 
			contradicting $\tup{\op,\op'}\in\readrespondsbefore(\ex_1)$.
			
			The case where $\tup{\op,\op'}\in\linrel(\ex_2)^{-1}$ similarly leads to $\tup{\op,\op'}\in\linrel(\ex_1)^{-1}$, which is again a contradiction. \qedhere
		\end{itemize}
	\end{itemize}
\end{proof}

\subsection{Minimal Linearizations for Decisive Registers}\label{subsec:declazy}

Using the notation we developed, every $\linmaps_{\regspec}^{\subseq}(\impl)$ is the set of decisive linearization of the implementation $\impl$, and we have:

$$\impl\in\dlin\iff\linmaps_{\regspec}^{\subseq}(\impl)\neq\emptyset$$

As shown in \cref{lem:wssubsq}, every write strong implementation is decisive.
Fittingly, for decisive implementations we use a weaker notion of minimality.
Specifically, we can get $\subseq$-minimal linearizations:

\begin{lemma}\label{lem:awsexistproc}
	If $\impl\in\dlin$, then $\linmaps_{\regspec}^{\subseq}(\impl)$ contains a $\subseq$-minimal linearization.
\end{lemma}
\begin{proof}
	We apply \cref{cor:proc}, using the fact that
	$\subseq$ is a well-founded partial order on $\regspec$.
\end{proof}

As already shown in the write strong case, minimal linearizations contain no pending reads:

\begin{lemma}\label{lem:readcompleted_aws}
	If $\impl\in\dlin$ and $\linmap\in\linmaps_{\regspec}^{\subseq}(\impl)$ is $\subseq$-minimal,
	then $\linmap$ is pending-read-free.
\end{lemma}
\begin{proof}
	We need to show $\setof{\linmap(\ex)\rst{\regread}}\subseteq\setof{\completed(\hs{\ex})}$ for every $\ex\in\exec{\impl}$.
	
	We proceed by induction on $\size{\ex}$. For $\ex=\emptyword$ we have by definition of linearizability that $\linmap(\ex)=\emptyword$, and we are done.
	
	Let $\ex$ be an execution of length $\size{\ex}=n$.
	As in the proof of \cref{lem:readcompleted_ws}, it is sufficient to show that $\seqh=\linmap(\ex)$,
	where
	$\seqh=\linmap(\ex)\rst{\setof{\linmap(\ex)}\setminus\left(\setof{\linmap(\ex)\rst{\regread}}\setminus\setof{\completed(\hs{\ex})}\right)}$ is the subsequence of $\linmap(\ex)$ where every pending read operation is removed.	
	
	We wish to apply \cref{lem:local_change}. We have:
	\begin{itemize}
		\item From \cref{lem:remove_read} we have that $\seqh\in\regspec$ and from \cref{lem:remove_pending} we get $\hs{\ex}\linleq\seqh$.
		
		\item  $\linmap(\ex')\subseq\seqh$ for all $\ex'\in\exec{\impl}$ such that $\ex'\prefneq\ex$:
		
		As $\linmap$ is decisive we have $\linmap(\ex')\subseq\linmap(\ex)$.
		By definition we also have $\seqh\subseq\linmap(\ex)$.
		Therefore, it is sufficient to show that $\setof{\linmap(\ex')}\subseteq\setof{\seqh}$.
		
		Let $\op\in\setof{\linmap(\ex')}$
		\begin{itemize}
			\item if $\methodf(\op)=\regwrite$ then from $\linmap(\ex')\subseq\linmap(\ex)$ we have $\op\in\setof{\linmap(\ex)\rst{\regwrite}}$ and as  $\linmap(\ex)\rst{\regwrite}=\seqh\rst{\regwrite}$ we get $\op\in\setof{\seqh}$.
			
			\item If $\methodf(\op)=\regread$ then due to our inductive assumption (applied to the strictly smaller $\ex'$) we have $\op\in\setof{\completed(\hs{\ex'})}$.
			
			As $\setof{\completed(\hs{\ex'})}\subseteq\setof{\completed(\hs{\ex})}$ we deduce that $\op\in\setof{\completed(\hs{\ex})}$ and thus by construction $\op\in\setof{\seqh}$.
		\end{itemize}
		
		\item  $\seqh\subseq\linmap(\ex')$ for all $\ex'\in\exec{\impl}$ such that $\ex\prefneq\ex'$:
			
		Follows from $\seqh\subseq\linmap(\ex)\subseq\linmap(\ex')$, where the former is by definition, and the latter is from $\linmap$ being decisive.
		
	\end{itemize}
	Overall, we can apply \cref{lem:local_change} and get that $\seqh=\linmap(\ex)$, as the equivalence induced by $\subseq$ is equality.
\end{proof}

\Cref{lem:permeq,lem:lastcompleted} are stated in a way which applies both to the set of write strong linearizations and to the set of decisive linearizations.
To enable this simultaneous argument for both cases, we first show that $\permleq$-minimal linearizations are also $\subseq$-minimal:

\begin{lemma}\label{lem:min}
	If $\impl\in\wslin$ and
	$\linmap$ is $\permleq$-minimal in $\linmaps_{\regspec}^{\wsleq}(\impl)$,
	then it is also $\subseq$-minimal in $\linmaps_{\regspec}^{\wsleq}(\impl)$
\end{lemma}
\begin{proof}
	Assume otherwise and let $\linmap'\in\linmaps_{\regspec}^{\wsleq}(\impl)$ be a linearization mapping such that $\linmap'\subsneq \linmap$.
 	Then, there exist an execution $\ex\in\exec{\impl}$ and a (possibly pending) operation $\op\in\hs{\ex}$ such that $\linmap'(\ex)\subsneq\linmap(\ex)$ and $\tilde{\op}\in\setof{\linmap(\ex)}\setminus\setof{\linmap'(\ex)}$,
 	where $\tilde{\op}$ is $\op$ extended to a completed operation (if necessary).
 	Consider the following cases:
 	\begin{itemize}
 		\item If $\op$ is a read operation, from $\tilde{\op}\notin\setof{\linmap'(\ex)}$ we deduce that $\op$ is pending.
 		As $\tilde{\op}\in\setof{\linmap(\ex)}$, this contradicts \cref{lem:readcompleted_ws}
 		
 		\item  If $\op$ is a write operation,
 		note that due to \cref{lem:subseq_permleq} lifted point-wise we obtain from $\linmap'\subsneq\linmap$ that $\linmap'\permleq\linmap$.
 		Since $\linmap$ is $\permleq$-minimal in $\linmaps_{\regspec}^{\wsleq}(\impl)$ and $\linmap'\in\linmaps_{\regspec}^{\wsleq}(\impl)$ we deduce that
 		$\linmap'\permeq\linmap$ and in particular $\linmap'(\ex)\permeq\linmap(\ex)$.
 		Applying \cref{lem:permleq_writesubseq} this implies that $\linmap'(\ex)\rst{\regwrite}=\linmap(\ex)\rst{\regwrite}$,
 		which contradicts $\tilde{\op}\in\setof{\linmap(\ex)\rst{\regwrite}}\setminus\setof{\linmap'(\ex)\rst{\regwrite}}$.
 	\end{itemize}
	In both cases we reach a contradiction and the lemma follows.
\end{proof}

\Cref{lem:permeq,lem:lastcompleted} now refer to a decisive linearization that is $\subseq$-minimal in $\linmaps_{\regspec}^{\rel}(\impl)$.
They can be applied both to the write-strong case (by taking $\rel=\permleq$ and using \cref{lem:wssubsq} and \cref{lem:min}) and to the decisive case (by taking $\rel=\subseq$).
For both choices it can be easily verified that $\rel$ satisfies the conditions required by the lemmas.

The first property we prove in this way is that minimal linearizations are response-activated. 
(\ie we do not add additional operations to the linearization when extending an execution with an additional transition, 
unless that transition is a response.)

\begin{lemma}\label{lem:permeq}
	If $\impl$ is an implementation and
	$\linmap\in\linmaps_{\regspec}^{\rel}(\impl)$ is $\subseq$-minimal and decisive,
	where $\rel$ is reflexive,
	then $\linmap$ is response-activated.
\end{lemma}
\begin{proof}
	Let $\ex\in\exec{\impl}$ and $\trans\in\impl.\lT$ such that $\ex\cdot\trans\in\exec{\impl}$ and $\trans\notin\Res$.
	We show $\linmap(\ex)=\linmap(\ex\cdot\trans)$.
	We wish to apply \cref{lem:local_change} to $\ex\cdot\trans$ and $\linmap(\ex)$.
	Indeed, we have the following:
	\begin{enumerate}
		\item $\linmap(\ex)\in\regspec$:
		
		Immediate from $\linmap$ being a valid linearization mapping.
		
		\item $\hs{\ex\cdot\trans}\linleq\linmap(\ex)$:
		
		Follows from \cref{lem:non_response} and $\hs{\ex}\linleq\linmap(\ex)$.
		
		\item $\linmap(\ex)\subseq \linmap(\ex\cdot\trans)$:
		
		Follows from $\linmap$ being decisive. %
		
		\item  $\tup{\linmap(\ex'),\linmap(\ex)}\in\rel$ for all $\ex'\in\exec{\impl}$ such that $\ex'\prefneq\ex\cdot\trans$:
		
		$\ex'\prefneq\ex\cdot\trans\implies\ex'\pref\ex$ and so this follows from the fact that $\linmap$ respects $\rel$. 
		
		\item  $\tup{\linmap(\ex),\linmap(\ex')}\in\rel$ for all $\ex'\in\exec{\impl}$ such that $\ex\cdot\trans\prefneq\ex'$:
		
		$\ex\cdot\trans\prefneq\ex'\implies\ex\pref\ex'$ and so this follows from the fact that $\linmap$ respects $\rel$.  
		
	\end{enumerate}
	From \cref{lem:local_change} we deduce that $\linmap(\ex)=\linmap(\ex\cdot\trans)$,
	as the equivalence induced by the partial order $\subseq$ is equality.
\end{proof}

Next, we show that minimal linearizations are last-completed.
(\ie the last operation in a minimal linearization is some operation that is completed in the linearized execution.)

\begin{lemma}\label{lem:lastcompleted}
	If $\impl$ is an implementation, $\linmap\in\linmaps_{\regspec}^{\rel}(\impl)$ is $\subseq$-minimal and decisive,
	and $\rel$ is a relation such that:
	\begin{enumerate}
		\item $\tup{\seqh_1,\seqh_2\cdot\op}\in\rel\land \op\notin\seqh_1 \implies\tup{\seqh_1,\seqh_2}\in\rel$ for every $\seqh_1,\seqh_2\in\regspec$ and operation $\op$.
		
		\item $\pref\subseteq\rel$.
		
		\item $\rel$ is transitive.
	\end{enumerate}
	then $\linmap$ is last-completed.
\end{lemma}

\begin{proof}
	We need to show $\lastop(\linmap(\ex))\in\setof{\completed(\hs{\ex})}$
	for all $\ex\in\exec{\impl}$ such that $\linmap(\ex)\neq\emptyword$.
	Assume for contradiction this is false, and let $\ex$ be an execution of minimal size such that $\lastop(\linmap(\ex))\neq\emptyword$ and $\lastop(\linmap(\ex))\notin\setof{\completed(\hs{\ex})}$.
	Denote $\op\eqdef\lastop(\linmap(\ex))$.
	Let $\seqh=\prefof{\linmap(\ex)}{\op}$ be the prefix of $\linmap(\ex)$ with $\op$ removed.
	From $\pref\subseteq\rel$ we get that $\rel$ is reflexive.
	We wish to use \cref{lem:local_change}.
	We have:
	\begin{enumerate}
		\item $\seqh\in\regspec$:
		
		Follows from $\seqh\pref\linmap(\ex)\in\regspec$ and $\regspec$ being prefix closed.
		
		\item $\hs{\ex}\linleq\seqh$:
		
		Follows from \cref{lem:remove_pending} using $\hs{\ex}\linleq\linmap(\ex)$
		
		\item $\seqh\subseq\linmap(\ex)$:
		
		Holds since $\seqh\pref\linmap(\ex)$.
		
		\item $\tup{\linmap(\ex'),\seqh}\in\rel$ for all $\ex'\in\exec{\impl}$ such that $\ex'\prefneq\ex$:
		
		As $\tup{\linmap(\ex'),\linmap(\ex)}\in\rel$,
		using the first assumption on $\rel$ it is sufficient to show that $\op\notin\linmap(\ex')$.
		Indeed, if we assume for contradiction $\op\in\linmap(\ex')$ then since $\setof{\completed(\hs{\ex'})}\subseteq\setof{\completed(\hs{\ex})}$
		we deduce that $\op\notin\setof{\completed(\hs{\ex'})}$.
		From $\linmap(\ex')\subseq\linmap(\ex)$, $\op\in\linmap(\ex')$ ,and $\op=\lastop(\linmap(\ex))$ we can deduce that $\op=\lastop(\linmap(\ex'))$
		and thus $\lastop(\linmap(\ex'))\notin\setof{\completed(\hs{\ex'})}$,
		contradicting the minimality of $\ex$.
		From this contradiction we deduce that $\linmap(\ex')\subseq\seqh$
		
		\item $\tup{\seqh,\linmap(\ex')}\in\rel$ for all $\ex'\in\exec{\impl}$ such that $\ex\prefneq\ex'$:
		
		We have that $\seqh\pref\linmap(\ex)$
		and so due to our assumptions $\tup{\seqh,\linmap(\ex)}\in\rel$.
		As $\tup{\linmap(\ex),\linmap(\ex')}\in\rel$
		we obtain from transitivity that 
		$\tup{\seqh,\linmap(\ex')}\in\rel$.
		
	\end{enumerate}
	From \cref{lem:local_change} we have $\seqh=\linmap(\ex)$, as the equivalence induced by $\subseq$ is equality.
	This is a contradiction, and so the lemma holds.
\end{proof}

\subsection{Conclusion}\label{subsec:lazy_conclusion}

We obtain the existence of the lazy  linearizations required for \cref{sec:hard}:

\begin{proof}[Proof Of \cref{lem:wsnice}]
	From $\cref{lem:wssubsq}$ we obtain a $\permleq$-minimal mapping which is write strong and decisive.
	It is also pending-read-free due to \cref{lem:readcompleted_ws}.
	Using \cref{lem:min},
	 we have that the mapping is also $\subseq$-minimal.
	Therefore, it is response-activated due to \cref{lem:permeq} and last-completed due to \cref{lem:lastcompleted}.
\end{proof}

\begin{proof}[Proof Of \cref{lem:awsnice}]
	From \cref{lem:awsexistproc} we obtain a $\subseq$-minimal mapping which is decisive.
	It is pending-read free due to \cref{lem:readcompleted_aws},
	response-activated due to \cref{lem:permeq},
	and last-completed due to \cref{lem:lastcompleted}.
\end{proof}

\section{ABD}
\label{app:abd}

Now we formally define multi-writer $\ABD$.
First we need some notation:

\begin{enumerate}
	\item $\TS=\N\times(\Tid\cup\set{\bot})$ (The timestamps used in ABD, composed of an integer and a thread id. $\Tid$ is a finite and ordered set, and $\bot$ is defined as smaller than all thread ids).
	
	\item $\Msg=\inarr{\set{\query(\tid,\id) \st \tid\in\Tid, \id\in\Id} \cup {} \\
		\set{ \update(\tid,\id,\tup{\val,\ts}) \st \tid\in\Tid,\id\in\Id,\val\in\Val,\ts\in\TS}}$
	
	\item $\Acks=(\Tid\times(\Val\times\TS))\cup\Tid$
	
	\item $\Pc_\ABD = \set{\main,\wbegin,\wquery,\wupdate,\wend,\rbegin,\rquery,\rupdate,\rend}$
\end{enumerate}

We define order on timestamps and lift it to order on pairs in $\Val \times \TS$:
\begin{align*}
	\pair{\tm_1}{\tid_1} < \pair{\tm_2}{\tid_2} & \Longleftrightarrow \tm_1 < \tm_2 \lor (\tm_1=\tm_2 \land \tid_1 < \tid_2)
	\\
	\pair{\val_1}{\ts_1} < \pair{\val_2}{\ts_2} & \Longleftrightarrow \ts_1 < \ts_2
\end{align*}

Now, we define the LTS:

\begin{itemize}
	\item $\ABD.\lQ = (\Msg\pfn\powerset{\Acks}) \times (\Tid \to \Id) \times (\Tid \to (\Val\times\TS)) \times (\Tid \to \Pc_\ABD) \times (\Tid \to \Val)$ (mapping from messages to acks, used operation IDs,  most recent value and timestamp known, program counter, input/output argument)
	\item $\ABD.\lSigma = \Inv\Res(\Reg)\cup
	\set{\objact\tup{\tid}\st\tid\in\Tid}$
	\item $\ABD.\linit = \tup{\emptyset,\lambda \tid \ldotp 0,\lambda \tid \ldotp \main,\lambda \tid \ldotp \pair{0}{\pair{0}{0}},\lambda \tid \ldotp 0}$
	
	\item The transitions are as follows:
\end{itemize}

\begin{mathparpagebreakable}
	\inferrule[\readinv]
	{\act = \invact\tup{\regread,\bot,\tid,\id} \\\\  \Tmappc(\tid)=\main \\ \id\notin \Tmapid(\tid)}
	{\tup{\msgmap,\Tmapid,\Tmapval,\Tmappc,\Tmaparg} \asteplab{\act}{} \tup{\msgmap,\Tmapid[\tid \mapsto \Tmapid(\tid)\cdot\id], \Tmapval, \Tmappc[\tid \mapsto \rbegin],\Tmaparg}}
	\and
	\inferrule[\readquery]{ \Tmappc(\tid)=\rbegin \\ \id = \Tmapid(\tid)[\last]
		\\\\ \msg=\query(\tid,\id)}
	{\tup{\msgmap,\Tmapid,\Tmapval,\Tmappc,\Tmaparg} \asteplab{\objact\tup{\tid}}{} \tup{ \msgmap[\msg\mapsto\emptyset],\Tmapid,\Tmapval,\Tmappc[\tid \mapsto \rquery],\Tmaparg}}
	\and
	\inferrule[\readupdate]{\Tmappc(\tid)=\rquery \\ \id = \Tmapid(\tid)[\last]
		\\\\ \quorum\subseteq\set{\tid\st\tup{\tid,\_}\in\msgmap(\query(\tid,\id))} \\ \size{\quorum}>\frac{\size{\Tid}}{2}
		\\\\ \tup{\valmax,\tsmax}=\max\set{\tup{\val',\ts'}\st\exists\tid'\in\quorum\ldotp\tup{\tid',\tup{\val',\ts'}}\in \msgmap(\query(\tid,\id))}
		\\\\ \msg = \update(\tid,\id,\tup{\valmax,\tsmax})}
	{\tup{\msgmap,\Tmapid,\Tmapval,\Tmappc,\Tmaparg} \asteplab{\objact\tup{\tid}}{} \tup{\msgmap[\msg\mapsto\emptyset],\Tmapid,\Tmapval,\Tmappc[\tid \mapsto \rupdate],\Tmaparg[\tid\mapsto\valmax]}}
	\and
	\inferrule[\readend]{ \Tmappc(\tid)=\rupdate \\ \id = \Tmapid(\tid)[\last]
		\\\\ \size{\msgmap(\update(\tid,\id,\_)}>\frac{\size{\Tid}}{2}
	}
	{\tup{\msgmap,\Tmapid,\Tmapval,\Tmappc,\Tmaparg} \asteplab{\objact\tup{\tid}}{} \tup{\msgmap,\Tmapid,\Tmapval,\Tmappc[\tid \mapsto \rend],\Tmaparg}}
	\and
	\inferrule[\readres]{\act = \resact\tup{\regread,\Tmaparg(\tid),\tid,\id} \\\\ \Tmappc(\tid)=\rend \\ \id = \Tmapid(\tid)[\last]
	}
	{\tup{\msgmap,\Tmapid,\Tmapval,\Tmappc,\Tmaparg} \asteplab{\act}{} \tup{\msgmap,\Tmapid,\Tmapval,\Tmappc[\tid \mapsto \main],\Tmaparg[\tid\mapsto 0]}}
	\and
	\inferrule[\writeinv]{\act = \invact\tup{\regwrite,\valin,\tid,\id} \\\\
		\Tmappc(\tid)=\main \\ \id\notin \Tmapid(\tid)}
	{\tup{\msgmap,\Tmapid,\Tmapval,\Tmappc,\Tmaparg} \asteplab{\act}{} \tup{\msgmap,\Tmapid[\tid \mapsto \Tmapid(\tid)\cdot\id], \Tmapval, \Tmappc[\tid \mapsto\wbegin],\Tmaparg[\tid\mapsto\valin]}}
	\and
	\inferrule[\writequery]{\Tmappc(\tid)=\wbegin \\ \id = \Tmapid(\tid)[\last]
		\\\\ \msg = \query(\tid,\id)}
	{\tup{\msgmap,\Tmapid,\Tmapval,\Tmappc,\Tmaparg} \asteplab{\objact\tup{\tid}}{} \tup{ \msgmap[\msg\mapsto\emptyset],\Tmapid,\Tmapval,\Tmappc[\tid \mapsto \wquery],\Tmaparg}}
	\and
	\inferrule[\writeupdate]{ \Tmappc(\tid)=\wquery \\ \id = \Tmapid(\tid)[\last]
		\\\\ \quorum\subseteq\set{\tid\st\tup{\tid,\_}\in\msgmap(\query(\tid,\id))}  \\ \size{\quorum}>\frac{\size{\Tid}}{2}
		\\\\ \tup{\tm,\_}=\max\set{\ts'\st\exists\tid'\in\quorum,\val'\ldotp\tup{\tid',\tup{\val',\ts'}}\in \msgmap(\query(\tid,\id))}
		\\\\ \msg = \update(\tid,\id,\tup{\Tmaparg(\tid),\tup{\tm+1,\tid}}}
	{\tup{\msgmap,\Tmapid,\Tmapval,\Tmappc,\Tmaparg} \asteplab{\objact\tup{\tid}}{} \tup{\msgmap[\msg\mapsto\emptyset],\Tmapid,\Tmapval,\Tmappc[\tid \mapsto \wupdate],\Tmaparg}}
	\and
	\inferrule[\writeend]{\Tmappc(\tid)=\wupdate  \\ \id = \Tmapid(\tid)[\last]
		\\\\ \size{\msgmap(\update(\tid,\id,\_)}>\frac{\size{\Tid}}{2}
	}
	{\tup{\msgmap,\Tmapid,\Tmapval,\Tmappc,\Tmaparg} \asteplab{\objact\tup{\tid}}{} \tup{\msgmap,\Tmapid,\Tmapval,\Tmappc[\tid \mapsto \wend],\Tmaparg}}
	\and
	\inferrule[\writeres]{\act = \resact\tup{\regwrite,\bot,\tid,\id}  \\\\
		\Tmappc(\tid)=\wend \\ \id = \Tmapid(\tid)[\last]}
	{\tup{\msgmap,\Tmapid,\Tmapval,\Tmappc,\Tmaparg} \asteplab{\act}{} \tup{\msgmap,\Tmapid,\Tmapval,\Tmappc[\tid \mapsto \main],\Tmaparg[\tid\mapsto 0]}}
	\and	
	\inferrule[\serverqueryack]{
		\msg = \query(\_,\_)\in\dom{\msgmap} \\\\
		\forall \tup{\val',\ts'}. \tup{\tid,\tup{\val',\ts'}}\notin\msgmap(\msg) \\\\
		\msgmap'=\msgmap[\msg\mapsto\msgmap(\msg)\uplus\set{\tup{\tid,\Tmapval(\tid)}}]}
	{\tup{\msgmap,\Tmapid,\Tmapval,\Tmappc,\Tmaparg} \asteplab{\objact\tup{\tid}}{} \tup{\msgmap',\Tmapid,\Tmapval,\Tmappc,\Tmaparg}}
	\and
	\inferrule[\serverupdateack]{
		\msg = \update(\_,\_,\tup{\val',\ts'})\in\dom{\msgmap} \\\\
		\msgmap'=\msgmap[\msg\mapsto\msgmap(\msg)\uplus\set{\tid}] \\\\
		\Tmapval' = \Tmapval[\tid\mapsto\max\set{\Tmapval(\tid),\tup{\val',\ts'}}]
	}
	{\tup{\msgmap,\Tmapid,\Tmapval,\Tmappc,\Tmaparg} \asteplab{\objact\tup{\tid}}{}
		\tup{\msgmap',\Tmapid,\Tmapval',\Tmappc,\Tmaparg}}
\end{mathparpagebreakable}

We show $\ABD$ can be simulated by $\awsreg$, and it follows that it is decisively linearizable.

\begin{theorem}
	$\ABD\leqsim\awsreg$
\end{theorem}

\begin{proof}

	First, we define the following notations:
	\begin{itemize}
		\item Define $\fsim_\pc\subseteq\Pc_\ABD\times \Pc_\aws$ by
		\begin{equation*}
			\fsim_\pc=\set{\tup{\main,\main},\tup{\rbegin,\rbegin},\tup{\rquery,\rlisten},\tup{\rupdate,\rlisten},\tup{\rupdate,\rend},\tup{\rend,\rend},\tup{\wbegin,\wbegin},\tup{\wquery,\wwait},\tup{\wupdate,\wwait},\tup{\wupdate,\wend},\tup{\wend,\wend}}
		\end{equation*}
		
		\item Given a mapping $\msgmap:\Msg\pfn\powerset{\Acks}$ define
		\begin{equation*}
			\tshist(\msgmap)=\set{\tup{\val',\ts'}\st \exists \msg=\update(\_,\_,\tup{\val',\ts'})\in\dom\msgmap\ldotp\size{\msgmap(\msg)}>\frac{\size{\Tid}}{2}}\cup\set{\tup{0,\tup{0,0}}}
		\end{equation*}
	\end{itemize}

	The simulation relation is defined by\\
	$\tup{q,\abs q}=\tup{\tup{\msgmap,\Tmapid,\Tmapval,\Tmappc,\Tmaparg},\tup{\val,\ver,\lockLTS,\abs\Tmapid,\abs\Tmappc,\Tmapvalset,\Tmapstart,\abs\Tmaparg}}\in\fsim$ iff the following conditions hold:
	
	\noindent Conditions on $q$:
		\begin{enumerate}[label=(\condletter\arabic*),ref=\arabic*,leftmargin=1cm]
			
			\item\label{abdaws:1a} For every $\msg=\update(\_,\_,\tup{\val,\ts})\in\dom{\msgmap}$ and $\tid'\in\msgmap(\msg)$ we have $\Tmapval(\tid)\geq \tup{\val,\ts}$
			
			\item\label{abdaws:1b} For every $\quorum\subseteq\Tid$ such that $\size{\quorum}>\frac{\size{\Tid}}{2}$ we have $\max\set{\Tmapval(\tid)\st\tid\in\quorum}\geq\max \tshist(\msgmap)$
			
			\item For every $\tid\in\Tid$ we have:
			\begin{enumerate}[label=(R\theenumi\alph*),ref=\theenumi\alph*,leftmargin=1cm]
				\item\label{abdaws:1c1} If $\id\notin\Tmapid(\tid)$, or $\id=\Tmapid(\tid)[\last]$ with $\Tmappc(\tid)\in\set{\rbegin,\wbegin,\rquery,\wquery}$, then there exists no pair $\tup{\val,\ts}$ such that $\update(\tid,\id,\tup{\val,\ts})\in \dom{\msgmap}$
				
				\item\label{abdaws:1c2} If $\id=\Tmapid(\tid)[\last]$ and $\Tmappc(\tid)\in\set{\wupdate,\rupdate}$ then there exists a single pair, denoted $\tup{\val_\tid,\ts_\tid}$ in the following, such that $\msg=\update(\tid,\id,\tup{\val_\tid,\ts_\tid})\in \dom{\msgmap}$ and  $\val_\tid=\Tmaparg(\tid)$.
				If $\Tmappc(\tid)=\wupdate$ then $\ts_\tid=\tup{\_,\tid}$
				
			\end{enumerate}
			
			\item\label{abdaws:1d} For every pair $\tup{\val,\tup{\tm,\tid}}$ in the set
			$\set{\tup{\val',\ts'}\st \update(\_,\_,\tup{\val',\ts'})\in\dom\msgmap}\cup\\
			\set{ \tup{\val',\ts'}\st\msg=\query(\_,\_)\in\dom\msgmap\land \tup{\_,\tup{\val',\ts'}}\in\msgmap(\msg)}\cup\set{\Tmapval(\tid')\st\tid'\in\Tid}$\\
			such that $\tup{\val,\tup{\tm,\tid}}\notin\tshist(\msgmap)$ it holds that
			$\Tmappc(\tid)=\wupdate$ and $\tup{\val,\tup{\tm,\tid}}=\tup{\val_\tid,\ts_\tid}$

		\end{enumerate}
		
		\noindent Conditions on $\abs q$:
		\begin{enumerate}[resume*]
			\item\label{abdaws:2a}$\lockLTS=0$
			
			\item\label{abdaws:2b} $\ver\geq0$
			
			\item\label{abdaws:2c} For all $\tid\in\Tid$, $\Tmapstart(\tid)\leq\ver$
		\end{enumerate}

		\noindent Conditions on $\tup{q,\abs q}$: There exists a monotonically increasing function $\Fver:\set{0\til\ver}\to\Val\times\TS$ with $\Fver(0)=\tup{0,\tup{0,0}}$ such that the following holds:
		
		\begin{enumerate}[resume*]
			\item\label{abdaws:3a} $\Tmapid=\abs\Tmapid$
			
			\item\label{abdaws:3b} $\Fver(\ver)=\max \tshist(\msgmap)=\tup{\val,\_}$.  
			
			\item For every $\tid\in\Tid$:
			\begin{enumerate}[label=(\condletter\theenumi\alph*),ref=\theenumi\alph*,leftmargin=1cm]
				\item\label{abdaws:3c1} $\tup{\Tmappc(\tid),\abs\Tmappc(\tid)}\in\fsim_\pc$
				
				\item\label{abdaws:3c2} If $\abs\Tmappc(\tid)\in\set{\rend,\wbegin,\wwait}$ then $\Tmaparg(\tid)=\abs\Tmaparg(\tid)$.
				
				\item\label{abdaws:3c3} If $\abs\Tmappc(\tid)=\rlisten$ then\\ $\Tmapvalset(\tid) = \set{\val'\st \exists \ts'\ldotp \Fver(\Tmapstart(\tid))\leq\tup{\val',\ts'}\leq\Fver(\ver)\land \tup{\val',\ts'}\in \tshist(\msgmap)}$
				
				\item\label{abdaws:3c4} If $\Tmappc(\tid)\in\set{\rquery,\wquery}$ then denoting $\msg=\query(\tid,\id)$ where $\id=\Tmapid(\tid)[\last]$, for every $\quorum\subseteq\set{\tid'\st \tup{\tid',\_}\in\msgmap(\msg)}$ and $\quorum'\subseteq\Tid\setminus\quorum$ with $\size{\quorum\uplus \quorum'}>\frac{\size{\Tid}}{2}$ we have\\
				$\max\set{\tup{\val',\ts'}\st\exists\tid'\in\quorum\ldotp\tup{\tid',\tup{\val',\ts'}}\in \msgmap(\msg)}\cup\set{\Tmapval(\tid')\st\tid'\in\quorum'}\geq\Fver(\Tmapstart(\tid))$
				
				\item\label{abdaws:3c5} If $\Tmappc(\tid)\in\set{\rupdate,\wupdate}$ and $\tup{\val_\tid,\ts_\tid}\notin \tshist(\msgmap)$ then $\Fver(\Tmapstart(\tid))\leq\tup{\val_\tid,\ts_\tid}$ and $\abs\Tmappc(\tid)\in\set{\rlisten,\wwait}$.
				
				\item\label{abdaws:3c6} If $\Tmappc(\tid)\in\set{\rupdate,\wupdate}$ and $\tup{\val_\tid,\ts_\tid}\in \tshist(\msgmap)$ then  $\abs\Tmappc(\tid)\in\set{\rend,\wend}$.
			\end{enumerate}
		\end{enumerate}

	Clearly, for the initial states we have %
	$\tup{\ABD.\linit,\awsreg.\linit}\in\fsim$.
	
	Given states $\tup{q,\abs q}=\tup{\tup{\msgmap,\Tmapid,\Tmapval,\Tmappc,\Tmaparg},\tup{\val,\ver,\lockLTS,\abs\Tmapid,\abs\Tmappc,\Tmapvalset,\Tmapstart,\abs\Tmaparg}}\in\fsim$,
	a monotonic function $\Fver$ such that the above conditions hold,
	and a transition $\tup{\msgmap,\Tmapid,\Tmapval,\Tmappc,\Tmaparg} \asteplab{\act}{\ABD} q'=\tup{\msgmap',\_,\_,\_,\_}$, we present the following:
	\begin{enumerate}
		\item Proof that $q'$ upholds the necessary conditions.
		
		\item A sequence of transitions\\ $\tup{\val,\ver,\lockLTS,\abs\Tmapid,\abs\Tmappc,\Tmapvalset,\Tmapstart,\abs\Tmaparg}\bsteplab{\actsq}{\awsreg}{\abs q}'=\tup{\val',\ver',\_,\_,\_,\_,\_,\_}$ such that $\act\rst{\Inv\Res(\Reg)}=\actsq\rst{\Inv\Res(\Reg)}$, and a proof that ${\abs q}'$ upholds necessary conditions.
		
		\item A proof using a function $\Fver':\set{0\til\ver'}\to\Val\times\TS$ that the pair $\tup{q',{\abs q}'}$ upholds necessary conditions.
	\end{enumerate}
	
	Denote $\tshist=\tshist(\msgmap)$,$\tshist'=\tshist(\msgmap')$.
	
	When all values relevant for certain items in the list of conditions are unchanged between $\tup{q,{\abs q}}$ and $\tup{q',{\abs q}'}$, they trivially still hold, and we omit their justification.
	
	In particular, for all transition rules except possibly for \serverupdateack\ and \writeupdate\, we will have $\ver'=\ver$,$\val'=\val$ take $\Fver'=\Fver$, and it is easy to check that $\tshist'=\tshist$, so there is no need to justify \cref{abdaws:3b}.
	
	We also omit the justification of \cref{abdaws:3c1} as it is easily verifiable given $\tup{q',{\abs q}'}$,
	and justification for conditions that clearly hold vacuously.
	
	In the following we only mention conditions that require extra justification.
	
	We consider all cases for the transition rule used:
	
	\begin{itemize}
		\item \readinv: 
		
		Here $\act=\invact\tup{\regread,\bot,\tid,\id}$, $q'=\tup{\msgmap,\Tmapid[\tid \mapsto \Tmapid(\tid)\cdot\id], \Tmapval, \Tmappc[\tid \mapsto \rbegin],\Tmaparg}$,
		and we have preconditions $\Tmappc(\tid)=\main$ and $\id\notin \Tmapid(\tid)$.
		
		For \cref{abdaws:1c1}, $\id$ is required to satisfy it as $\Tmappc[\tid \mapsto \rbegin](\tid)=\rbegin$, and indeed we can deduce this from $\id\notin \Tmapid(\tid)$.
		
		We have that $\abs\Tmapid=\Tmapid$ which implies $\id\notin \abs\Tmapid(\tid)$, and $\abs\Tmappc(\tid)=\Tmappc(\tid)=\main$, and thus using \readinv\ we get \\
		$\tup{\val,\ver,\lockLTS,\abs\Tmapid,\abs\Tmappc,\Tmapvalset,\Tmapstart,\abs\Tmaparg} \asteplab{\act}{} \tup{\val,\ver,\lockLTS,\abs\Tmapid[\tid \mapsto \abs\Tmapid(\tid)\cdot\id], \abs\Tmappc[\tid \mapsto \rbegin],\Tmapvalset,\Tmapstart,\abs\Tmaparg}$.

		For \cref{abdaws:3a} we have that $\Tmapid=\abs\Tmapid$ and thus $\Tmapid[\tid \mapsto \Tmapid(\tid)\cdot\id]=\abs\Tmapid[\tid \mapsto \abs\Tmapid(\tid)\cdot\id]$.
		
		\item \readquery: 
		
		Here $\act=\objact\tup{\tid}$, $q'=\tup{ \msgmap[\msg\mapsto\emptyset],\Tmapid,\Tmapval, \Tmappc[\tid \mapsto \rquery], \Tmaparg}$,
		and we have preconditions $\Tmappc(\tid)=\rbegin$, $\id = \Tmapid(\tid)[\last]$, and $\msg = \query(\tid,\id)$.
		
		All conditions on $q$ are maintained under the change to $\msgmap$ and $\Tmappc$.
		
		We get that $\abs\Tmappc(\tid)=\rbegin$ and since $\lockLTS=0$ we can use the transition \readinit\ and get \\
		$\tup{\val,\ver,\lockLTS,\abs\Tmapid, \abs\Tmappc,\Tmapvalset,\Tmapstart,\abs\Tmaparg} \asteplab{\objact\tup{\tid}}{}\\
		\tup{\val,\ver,\lockLTS,\abs\Tmapid, \abs\Tmappc[\tid \mapsto \rlisten],\Tmapvalset[\tid\mapsto\set{\val}],\Tmapstart[\tid\mapsto\ver],\abs\Tmaparg}$.
		
		For \cref{abdaws:2c} we have that $\Tmapstart[\tid\mapsto\ver](\tid)=\ver$.
		
		Given that $\Tmappc[\tid \mapsto \rquery](\tid)=\rquery$
		and $\abs\Tmappc[\tid \mapsto \rlisten](\tid)=\rlisten$ it remains to show:
		\begin{itemize}
			\item \cref{abdaws:3c3}: we need to show
			
			$\Tmapvalset[\tid\mapsto\set{\val}](\tid) = \set{\val'\st \exists \ts'\ldotp \Fver(\Tmapstart[\tid\mapsto\ver](\tid))\leq\tup{\val',\ts'}\leq\Fver(\ver)\land \tup{\val',\ts'}\in \tshist'}$:
			
			As $\tshist'=\tshist$, simplifying both sides we get
			$\set{\val}=\set{\val'\st  \tup{\val',\ts'}=\Fver(\ver)\in \tshist}$ which follows from \cref{abdaws:3b}.
			
			\item \cref{abdaws:3c4}:
			
			Indeed, taking $\quorum\subseteq\set{\tid'\st \tup{\tid',\_}\in \msgmap[\msg\mapsto\emptyset](\msg)}=\set{\tid'\st \tup{\tid',\_}\in \emptyset}$ we must have $\quorum=\emptyset$, and it remains to show that for any $\quorum'\subseteq\Tid$ with $\size{\quorum'}>\frac{\size{\Tid}}{2}$.
			we have
			$\max\set{\Tmapval(\tid')\st\tid'\in\quorum'}\geq\Fver(\Tmapstart(\tid))$.
			Indeed, we get:
			\begin{equation*}
				\max\set{\Tmapval(\tid')\st\tid'\in\quorum'}\geq\max\tshist=\Fver(\ver)\geq\Fver(\Tmapstart(\tid))
			\end{equation*}
			Where we used \cref{abdaws:1b},\cref{abdaws:3b}, \cref{abdaws:2c} and the monotonicity of $\Fver$.
		\end{itemize}
		
		\item \readupdate: 
		
		Here $\act=\objact\tup{\tid}$, $q'=\tup{\msgmap[\msg\mapsto\emptyset],\Tmapid,\Tmapval,\Tmappc[\tid \mapsto \rupdate],\Tmaparg[\tid\mapsto\valmax]}$,
		and we have preconditions
		$\Tmappc(\tid)=\rquery$, $\id = \Tmapid(\tid)[\last]$,
		$\quorum\subseteq\set{\tid\st\tup{\tid,\_}\in\msgmap(\query(\tid,\id))}$,  $\size{\quorum}>\frac{\size{\Tid}}{2}$,
		$\tup{\valmax,\tsmax}=\max\set{\tup{\val',\ts'}\st\exists\tid'\in\quorum\ldotp\tup{\tid',\tup{\val',\ts'}}\in \msgmap(\query(\tid,\id))}$,
		
		and $\msg = \update(\tid,\id,\tup{\valmax,\tsmax})$.
		
		Note that by definition, and using the fact that $\frac{\size{\Tid}}{2}>0$ we get $\tshist'=\tshist(\msgmap[\msg\mapsto\emptyset])=\tshist(\msgmap)=\tshist$, and so \cref{abdaws:1b} still holds.
		
		Using \cref{abdaws:1c1} and the above preconditions, we get \cref{abdaws:1c2} and specifically
		$\tup{\valmax,\tsmax}=\tup{\val_\tid,\ts_\tid}$,
		noticing that indeed $\val_\tid=\valmax=\Tmaparg[\tid\mapsto\valmax](\tid)$.
		
		For \cref{abdaws:1d} note that from the above there exists some $\tid'$ such that $\tup{\tid',\tup{\valmax,\tsmax}}\in \msgmap(\query(\tid,\id))$, and so the set mentioned in \cref{abdaws:1d} is unchanged by the addition of $\update(\tid,\id,\tup{\valmax,\tsmax})$.
		For every element in that set, there exists a thread $\tid'$ such that $\Tmappc(\tid')=\wupdate\neq\rquery$, and in particular $\tid'\neq\tid$ and so the condition for each element continues to hold.
		
		Now, as $\Tmappc(\tid)=\rquery$, using \cref{abdaws:3c4} taking the above $\quorum$ and $\quorum'=\emptyset$ we get\\ $\max\set{\tup{\val',\ts'}\st\exists\tid'\in\quorum\ldotp\tup{\tid',\tup{\val',\ts'}}\in \msgmap(\query(\tid,\id))}\geq\Fver(\Tmapstart(\tid))$.\\
		That is, $\tup{\val_\tid,\ts_\tid}\geq\Fver(\Tmapstart(\tid))$.
		
		Consider two cases for the post-state $q'$:
		\begin{itemize}
			\item If $\tup{\val_\tid,\ts_\tid}\notin\tshist'$:
			
			We perform an empty sequence of transitions.
			
			We have that $\Tmappc[\tid \mapsto \rupdate](\tid)=\rupdate$, $\tup{\val_\tid,\ts_\tid}\notin\tshist'$, and from $\Tmappc(\tid)=\rquery$ we get $\abs\Tmappc(\tid)=\rlisten$.
			Togather with the fact that $\tup{\val_\tid,\ts_\tid}\geq\Fver(\Tmapstart(\tid))$ we get \cref{abdaws:3c5}.
			
			The relevant values for \cref{abdaws:3c3} are unchanged as $\tshist'=\tshist$.

			\item If $\tup{\val_\tid,\ts_\tid}\in\tshist'$:
			
			From the fact that $\Tmappc(\tid)=\rquery$ we get that $\abs\Tmappc(\tid)=\rlisten$ and thus from \cref{abdaws:3c3} we have $\Tmapvalset(\tid) = \set{\val'\st \exists \ts'\ldotp \Fver(\Tmapstart(\tid))\leq\tup{\val',\ts'}\leq\Fver(\ver)\land \tup{\val',\ts'}\in \tshist}$.
			
			We claim that $\valmax\in\Tmapvalset(\tid)$. As $\tup{\valmax,\tsmax}=\tup{\val_\tid,\ts_\tid}\in\tshist'=\tshist$, it is sufficient to show that
			$\Fver(\Tmapstart(\tid))\leq\tup{\val_\tid,\ts_\tid}\leq\Fver(\ver)$.
			Indeed, we have shown that  $\tup{\val_\tid,\ts_\tid}\geq\Fver(\Tmapstart(\tid))$,
			and the fact that $\tup{\val_\tid,\ts_\tid}\leq\Fver(\ver)$ follows from $\Fver(\ver)=\max \tshist$.
			Overall, we get that $\abs\Tmappc(\tid)=\rlisten$ and $\valmax\in\Tmapvalset(\tid)$,
			and using the rule \readpick\ , we can perform a transition\\
			$\tup{\val,\ver,\lockLTS,\abs\Tmapid,\abs\Tmappc,\Tmapvalset,\Tmapstart,\abs\Tmaparg} \asteplab{\objact\tup{\tid}}{}\\ \tup{\val,\ver,\lockLTS,\abs\Tmapid,\abs\Tmappc[\tid \mapsto \rend],\Tmapvalset,\Tmapstart,\abs\Tmaparg[\tid\mapsto \valmax]}$.
			
			For \cref{abdaws:3c6} we have $\Tmappc[\tid \mapsto \rupdate](\tid)=\rupdate$, $\tup{\val_\tid,\ts_\tid}\in \tshist'$,
			and $\abs\Tmappc[\tid \mapsto \rend](\tid)=\rend$.
			
			For \cref{abdaws:3c2} we have $\Tmaparg[\tid\mapsto\valmax](\tid)=\abs\Tmaparg[\tid\mapsto\valmax](\tid)=\valmax$.
		\end{itemize}
		
		\item \readend: 
		
		Here $\act=\objact\tup{\tid}$, $q'=\tup{\msgmap,\Tmapid,\Tmapval,\Tmappc[\tid \mapsto \rend],\Tmaparg}$,
		and we have preconditions
		$\Tmappc(\tid)=\rupdate$, $\id = \Tmapid(\tid)[\last]$,
		and $\size{\msgmap(\update(\tid,\id,\_)}>\frac{\size{\Tid}}{2}$.
		
		We perform an empty sequence of transitions.
		
		The precondition $\size{\msgmap(\update(\tid,\id,\_)}>\frac{\size{\Tid}}{2}$ guarantees by definition that $\tup{\val_\tid,\ts_\tid}\in\tshist$, and together with $\Tmappc(\tid)=\rupdate$ we get that $\abs\Tmappc(\tid)=\rend$, and all conditions hold.
		
		\item \readres: 
		
		Here $\act=\resact\tup{\regread,\Tmaparg(\tid),\tid,\id}$, $q'=\tup{\msgmap,\Tmapid,\Tmapval,\Tmappc[\tid \mapsto \main],\Tmaparg[\tid\mapsto 0]}$,
		and we have preconditions $\Tmappc(\tid)=\rend$ and $\id = \Tmapid(\tid)[\last]$.
		
		We have that $\abs\Tmapid=\Tmapid$ which implies $\id = \abs\Tmapid(\tid)[\last]$,
		$\Tmappc(\tid)=\rend$ implies $\abs\Tmappc(\tid)=\rend$,
		and thus using \readres\ we get \\
		$\tup{\val,\ver,\lockLTS,\abs\Tmapid,\abs\Tmappc,\Tmapvalset,\Tmapstart,\abs\Tmaparg} \asteplab{\act}{}\\ \tup{\val,\ver,\lockLTS,\abs\Tmapid,\abs\Tmappc[\tid \mapsto \main],\Tmapvalset[\tid\mapsto\emptyset],\Tmapstart[\tid\mapsto 0],\abs\Tmaparg[\tid\mapsto 0]}$.
		
		For \cref{abdaws:2c} we use \cref{abdaws:2b} and get $\Tmapstart[\tid\mapsto0](\tid)=0\leq\ver$.
		
		\item \writeinv: 
		
		Here $\act=\invact\tup{\regwrite,\valin,\tid,\id}$, $q'=\tup{\msgmap,\Tmapid[\tid \mapsto \Tmapid(\tid)\cdot\id], \Tmapval, \Tmappc[\tid \mapsto \wbegin],\Tmaparg[\tid\mapsto\valin]}$,
		and we have preconditions $\Tmappc(\tid)=\main$ and $\id\notin \Tmapid(\tid)$.
		
		\cref{abdaws:1c1} is justified exactly like for a \readinv\ transition.
		
		We have that $\abs\Tmapid=\Tmapid$ which implies $\id\notin \abs\Tmapid(\tid)$, and $\abs\Tmappc(\tid)=\Tmappc(\tid)=\main$, and thus using \writeinv\ we get \\
		$\tup{\val,\ver,\lockLTS,\abs\Tmapid,\abs\Tmappc,\Tmapvalset,\Tmapstart,\abs\Tmaparg} \asteplab{\act}{}\\ \tup{\val,\ver,\lockLTS,\Tmapid[\tid \mapsto \Tmapid(\tid)\cdot\id], \Tmappc[\tid \mapsto\wbegin], \Tmapvalset,\Tmapstart,\Tmaparg[\tid\mapsto\valin]}$
		
		\cref{abdaws:3a} is justified like for a \readinv\ transition.
		For \cref{abdaws:3c2} note that $\Tmaparg[\tid\mapsto\valin](\tid)=\abs\Tmaparg[\tid\mapsto\valin](\tid)=\valin$.

		\item \writequery:
		
		Here $\act=\objact\tup{\tid}$, $q'=\tup{ \msgmap[\msg\mapsto\emptyset],\Tmapid,\Tmapval, \Tmappc[\tid \mapsto \wquery], \Tmaparg}$,
		and we have preconditions $\Tmappc(\tid)=\wbegin$, $\id = \Tmapid(\tid)[\last]$, and $\msg = \query(\tid,\id)$.
		
		All conditions on $q$ are maintained under the change to $\msgmap$ and $\Tmappc$.
		
		We get that $\abs\Tmappc(\tid)=\wbegin$ and since $\lockLTS=0$ we can use the transition \writeinit\ and get \\
		$\tup{\val,\ver,\lockLTS,\abs\Tmapid,\abs\Tmappc,\Tmapvalset,\Tmapstart,\abs\Tmaparg} \asteplab{\objact\tup{\tid}}{} \tup{\val,\ver,\lockLTS,\abs\Tmapid,\abs\Tmappc[\tid \mapsto \wwait],\Tmapvalset,\Tmapstart[\tid\mapsto\ver],\abs\Tmaparg}$.
		
		For \cref{abdaws:2c} we have that $\Tmapstart[\tid\mapsto\ver](\tid)=\ver$.
		
		Given that $\Tmappc[\tid \mapsto \wquery](\tid)=\wquery$,
		$\abs\Tmappc[\tid \mapsto \wwait](\tid)=\wwait$, and $\Tmaparg(\tid)=\abs\Tmaparg(\tid)$, we get \cref{abdaws:3c2} and it remains to show that \cref{abdaws:3c4} holds. The proof of this follows exactly like the \readquery\ case.

		\item \writeupdate:
		
		Here $\act=\objact\tup{\tid}$, $q'=\tup{\msgmap[\msg\mapsto\emptyset],\Tmapid,\Tmapval,\Tmappc[\tid \mapsto \wupdate],\Tmaparg}$,
		and we have preconditions
		$\Tmappc(\tid)=\wquery$, $\id = \Tmapid(\tid)[\last]$,
		$\quorum\subseteq\set{\tid\st\tup{\tid,\_}\in\msgmap(\query(\tid,\id))}$,  $\size{\quorum}>\frac{\size{\Tid}}{2}$,\\
		$\tup{\tm,\_}=\max\set{\ts'\st\exists\tid'\in\quorum,\val'\ldotp\tup{\tid',\tup{\val',\ts'}}\in \msgmap(\query(\tid,\id))}$,\\
		and $\msg = \update(\tid,\id,\tup{\Tmaparg(\tid),\tup{\tm+1,\tid}})$.
		
		Note that by definition, and using the fact that $\frac{\size{\Tid}}{2}>0$ we get $\tshist'=\tshist(\msgmap[\msg\mapsto\emptyset])=\tshist(\msgmap)=\tshist$, and so \cref{abdaws:1b} still holds.
		
		Using \cref{abdaws:1c1} we get that after the update $\msgmap[\msg\mapsto\emptyset]$ \cref{abdaws:1c2} holds and specifically $\tup{\val_\tid,\ts_\tid}=\tup{\Tmaparg(\tid),\tup{\tm+1,\tid}}$ in the post-state $q'$.
		
		For \cref{abdaws:1d} the pair $\tup{\Tmaparg(\tid),\tup{\tm+1,\tid}}$ is possibly %
		added to the set mentioned in the condition, and satisfies the condition since $\Tmappc[\tid \mapsto \wupdate](\tid)=\wupdate$ and $\tup{\val_\tid,\ts_\tid}=\tup{\Tmaparg(\tid),\tup{\tm+1,\tid}}$.
		
		From the preconditions we get that there exists some $\val'',\tid''$ such that
		$\tup{\val'',\tup{\tm,\tid''}}=\max\set{\tup{\val',\ts'}\st\exists\tid'\in\quorum\ldotp\tup{\tid',\tup{\val',\ts'}}\in \msgmap(\query(\tid,\id))}$.
		
		As $\Tmappc(\tid)=\wquery$, using \cref{abdaws:3c4} taking $\quorum$ and $\quorum'=\emptyset$ we get\\ $\tup{\val'',\tup{\tm,\tid''}}=\max\set{\tup{\val',\ts'}\st\exists\tid'\in\quorum\ldotp\tup{\tid',\tup{\val',\ts'}}\in \msgmap(\query(\tid,\id))}\geq\Fver(\Tmapstart(\tid))$.
		Overall, $\tup{\val_\tid,\ts_\tid}=\tup{\Tmaparg(\tid),\tup{\tm+1,\tid}}>\tup{\val'',\tup{\tm,\tid''}}\geq\Fver(\Tmapstart(\tid))$.
		
		Also note that using \cref{abdaws:3c1}, $\Tmappc(\tid)=\wquery$ implies that $\abs\Tmappc(\tid)=\wwait$ thus from \cref{abdaws:3c2} $\Tmaparg(\tid)=\abs\Tmaparg(\tid)$.
		Overall we get $\val_\tid=\abs\Tmaparg(\tid)$.
		
		Now we consider three cases:
		\begin{itemize}
			\item If $\tup{\val_\tid,\ts_\tid}\notin\tshist$:
			
			We perform an empty sequence of transitions.
			
			Using what we proved above, we get that \cref{abdaws:3c5} and \cref{abdaws:3c2} hold.
			
			\item If $\tup{\val_{\tid},\ts_{\tid}}=\max\tshist$
			
			Since $\abs\Tmappc({\tid})=\wwait$ and $\lockLTS=0$, we can use the rule \writetop\ and get a transition\\
			$\tup{\val,\ver,\lockLTS,\abs\Tmapid,\abs\Tmappc,\Tmapvalset,\Tmapstart,\abs\Tmaparg} \asteplab{\objact\tup{\tid}}{}\\ \tup{\abs\Tmaparg(\tid),\ver+1,\lockLTS,\abs\Tmapid,\abs\Tmappc[\tid \mapsto \wend],\Tmapvalset,\Tmapstart,\abs\Tmaparg}$
			
			For \cref{abdaws:2b} and \cref{abdaws:2c} we can use $\ver<\ver+1$.
			
			To show the conditions on $\tup{q',{\abs q}'}$ hold we define $\Fver':\set{0\til \ver+1}\to\Val\times\TS$ by:
			\begin{equation*}
				\Fver'(\tm)=\begin{cases}
					\Fver(\tm) & \tm\leq\ver \\
					\tup{\val_{\tid},\ts_{\tid}} & \tm=\ver+1
				\end{cases}
			\end{equation*}
			This function is monotonic since $\Fver'(\ver+1)=\tup{\val_{\tid},\ts_{\tid}}=\max\tshist=\Fver(\ver)=\Fver'(\ver)$ and $\Fver$ is monotonic.
			
			For \cref{abdaws:3b} we have $\Fver'(\ver+1)=\max\tshist=\tup{\abs\Tmaparg(\tid),\_}$, where for the last equality we use $\val_{\tid}=\abs\Tmaparg(\tid)$, proved above.

			For conditions on each thread $\tid''\in\Tid$ (aside from \cref{abdaws:3c1} which is easily verified):
			Conditions dependent only on $\Fver'(\Tmapstart(\tid''))$ hold as they did before the transition since using \cref{abdaws:2c} we have $\Tmapstart(\tid'')\leq\ver<\ver+1$ and thus $\Fver'(\Tmapstart(\tid''))=\Fver(\Tmapstart(\tid''))$.
			For \cref{abdaws:3c3} we additionally use $\Fver'(\ver+1)=\Fver(\ver)$.
			\cref{abdaws:3c6} clearly holds for $\tid''=\tid$.
			All other values relevant for conditions are unchanged.
			
			\item If $\tup{\val_{\tid},\ts_{\tid}}\in\tshist$ and $\tup{\val_{\tid},\ts_{\tid}}<\max\tshist$
			
			As we have shown above $\Fver(\Tmapstart(\tid))\leq\tup{\val_{\tid},\ts_{\tid}}$, and using $\max\tshist=\Fver(\ver)$ we get $\Fver(\Tmapstart(\tid))<\Fver(\ver)$ and thus from monotonicity of $\Fver$ (in the contra-positive) we get $\Tmapstart(\tid)<\ver$.
			
			Together with the fact that $\abs\Tmappc({\tid})=\wwait$ and $\lockLTS=0$, we can use the rule \writerollback\ and get a transition\\
			$\tup{\val,\ver,\lockLTS,\abs\Tmapid,\abs\Tmappc,\Tmapvalset,\Tmapstart,\abs\Tmaparg} \asteplab{\objact\tup{\tid}}{}\\
			\tup{\abs\Tmaparg(\tid),\ver-1,1,\abs\Tmapid,\abs\Tmappc[\tid \mapsto \wrollback],\Tmapvalset,\Tmapstart,\abs\Tmaparg[\tid\mapsto\val]}$
			
			Then, Since $\abs\Tmappc[\tid \mapsto \wrollback](\tid)=\wrollback$ we can perform a transition using the rule \writerollforward\ and get a transition\\
			$\tup{\abs\Tmaparg(\tid),\ver-1,1,\abs\Tmapid,\abs\Tmappc[\tid \mapsto \wrollback],\Tmapvalset,\Tmapstart,\abs\Tmaparg[\tid\mapsto\val]}\\
			\asteplab{\objact\tup{\tid}}{}\\
			\tup{\abs\Tmaparg[\tid\mapsto\val](\tid),\ver-1+1,0,\abs\Tmapid,\abs\Tmappc[\tid \mapsto \wend],\Tmapvalset,\Tmapstart,\abs\Tmaparg[\tid\mapsto\val]}$.
			
			Simplifying this we get a state
			$\tup{\val,\ver,0,\abs\Tmapid,\abs\Tmappc[\tid \mapsto \wend],\Tmapvalset,\Tmapstart,\abs\Tmaparg[\tid\mapsto\val]}$
			
			For \cref{abdaws:2a} note that $\lockLTS=0$ in the post-state ${\abs q}'$.
			
			Additionally, \cref{abdaws:3c1} and \cref{abdaws:3c6} clearly hold for $\tid$.
		\end{itemize}

		\item \writeend:
		
		Here $\act=\objact\tup{\tid}$, $q'=\tup{\msgmap,\Tmapid,\Tmapval,\Tmappc[\tid \mapsto \wend],\Tmaparg}$,
		and we have preconditions
		$\Tmappc(\tid)=\wupdate$, $\id = \Tmapid(\tid)[\last]$,
		and $\size{\msgmap(\update(\tid,\id,\_)}>\frac{\size{\Tid}}{2}$.
		
		We perform an empty sequence of transitions.
		
		The precondition  $\size{\msgmap(\update(\tid,\id,\_)}>\frac{\size{\Tid}}{2}$ guarantees by definition that $\tup{\val_\tid,\ts_\tid}\in\tshist$, and together with $\Tmappc(\tid)=\wupdate$ we get that $\abs\Tmappc(\tid)=\wend$, and all conditions hold.
		
		\item \writeres: 
		
		Here $\act=\resact\tup{\regwrite,\bot,\tid,\id}$, $q'=\tup{\msgmap,\Tmapid,\Tmapval,\Tmappc[\tid \mapsto \main],\Tmaparg[\tid\mapsto 0]}$,
		and we have preconditions $\Tmappc(\tid)=\wend$ and $\id = \Tmapid(\tid)[\last]$.
		
		We have that $\abs\Tmapid=\Tmapid$ which implies $\id = \abs\Tmapid(\tid)[\last]$,
		$\Tmappc(\tid)=\wend$ implies $\abs\Tmappc(\tid)=\wend$,
		and thus using \writeres\ we get \\
		$\tup{\val,\ver,\lockLTS,\abs\Tmapid,\abs\Tmappc,\Tmapvalset,\Tmapstart,\abs\Tmaparg} \asteplab{\act}{}\\ \tup{\val,\ver,\lockLTS,\abs\Tmapid,\abs\Tmappc[\tid \mapsto \main],\Tmapvalset[\tid\mapsto\emptyset],\Tmapstart[\tid\mapsto 0],\abs\Tmaparg[\tid\mapsto 0]}$
		
		For \cref{abdaws:2c} we use \cref{abdaws:2b} and get $\Tmapstart[\tid\mapsto0](\tid)=0\leq\ver$.
		
		\item \serverqueryack: 
		
		Here $\act=\objact\tup{\tid}$ and $q'=\tup{\msgmap',\Tmapid,\Tmapval,\Tmappc,\Tmaparg}$ where
		$\msg = \query(\_,\_)\in\dom{\msgmap}$,\\
		$\forall \tup{\val',\ts'}. \tup{\tid,\tup{\val',\ts'}}\notin\msgmap(\msg)$, and
		$\msgmap'=\msgmap[\msg\mapsto\msgmap(\msg)\uplus\set{\tup{\tid,\Tmapval(\tid)}}]$.
		
		For \cref{abdaws:1d} we note that the set mentioned in the condition is unchanged as $\Tmapval(\tid)$ is already in it, and the condition for each element in the set still holds as $\Tmappc$ is unchanged.
		
		We perform an empty sequence of transitions
		
		Let $\tid'',\id''$ be the identifiers such that $\msg = \query(\tid'',\id'')$.
		All conditions on $\tup{q',{\abs q}'}$ for all threads except $\tid''$ are unaffected.
		We show that if $\id''=\Tmapid(\tid'')[\last]$ and $\Tmappc(\tid'')\in\set{\rquery,\wquery}$, then \cref{abdaws:3c4} still holds for $\tid''$.
		
		Indeed, Let $\quorum\subseteq\set{\tid'\st \tup{\tid',\_}\in\msgmap'(\msg)}$ and $\quorum'\subseteq\Tid\setminus\quorum$ be sets such that $\size{\quorum\uplus \quorum'}>\frac{\size{\Tid}}{2}$.
		If $\tid\notin\quorum$ we have $\quorum\subseteq\set{\tid'\st \tup{\tid',\_}\in\msgmap(\msg)}$ and from \cref{abdaws:3c4} holding in the pre-states $\tup{q,\abs q}$ we get
		$\max\set{\tup{\val',\ts'}\st\exists\tid'\in\quorum\ldotp\tup{\tid',\tup{\val',\ts'}}\in \msgmap(\msg)}\cup\set{\Tmapval(\tid')\st\tid'\in\quorum'}\geq\Fver(\Tmapstart(\tid))$.
		As $\msgmap'(\msg)=\tup{\tid,\Tmapval(\tid)}\uplus\msgmap(\msg)$ and $\tid\notin\quorum$ we get\\
		$\set{\tup{\val',\ts'}\st\exists\tid'\in\quorum\ldotp\tup{\tid',\tup{\val',\ts'}}\in \msgmap(\msg)}=\set{\tup{\val',\ts'}\st\exists\tid'\in\quorum\ldotp\tup{\tid',\tup{\val',\ts'}}\in \msgmap'(\msg)}$
		and overall $\max\set{\tup{\val',\ts'}\st\exists\tid'\in\quorum\ldotp\tup{\tid',\tup{\val',\ts'}}\in \msgmap'(\msg)}\cup\set{\Tmapval(\tid')\st\tid'\in\quorum'}\geq\Fver(\Tmapstart(\tid))$.
		
		Otherwise, consider $\hat{\quorum}=\quorum\setminus\set{\tid},\hat{\quorum'}=\quorum'\cup\set{\tid}$.
		As $\msgmap'=\msgmap[\msg\mapsto\msgmap(\msg)\uplus\set{\tup{\tid,\Tmapval(\tid)}}]$,
		we get that $\hat{\quorum}\subseteq\set{\tid'\st \tup{\tid',\_}\in\msgmap(\msg)}$.
		Additionally, $\hat{\quorum'}\subseteq\Tid\setminus\hat{\quorum}$ and $\size{\hat{\quorum}\uplus \hat{\quorum'}}=\size{\quorum\uplus \quorum'}>\frac{\size{\Tid}}{2}$.
		Thus from the assumption that \cref{abdaws:3c4} holds in the pre-states we get\\
		$\max\set{\tup{\val',\ts'}\st\exists\tid'\in\hat{\quorum}\ldotp\tup{\tid',\tup{\val',\ts'}}\in \msgmap(\msg)}\cup\set{\Tmapval(\tid')\st\tid'\in\hat{\quorum'}}\geq\Fver(\Tmapstart(\tid))$.
		Using the fact that $\msgmap'(\msg)=\tup{\tid,\Tmapval(\tid)}\uplus\msgmap(\msg)$, $\tid\notin\hat{\quorum}$, and $\tid\in\hat{\quorum'}$ we have:
		\begin{alignat*}{2}
			&\set{\tup{\val',\ts'}\st\exists\tid'\in\hat{\quorum}\ldotp\tup{\tid',\tup{\val',\ts'}}\in \msgmap(\msg)}&\cup\set{\Tmapval(\tid')\st\tid'\in\hat{\quorum'}} \\
			=&\set{\tup{\val',\ts'}\st\exists\tid'\in\hat{\quorum}\ldotp\tup{\tid',\tup{\val',\ts'}}\in \msgmap'(\msg)}\cup\set{\Tmapval(\tid)}&\cup\set{\Tmapval(\tid')\st\tid'\in\quorum'}\\
			=&\set{\tup{\val',\ts'}\st\exists\tid'\in\quorum\ldotp\tup{\tid',\tup{\val',\ts'}}\in \msgmap'(\msg)}&\cup\set{\Tmapval(\tid')\st\tid'\in\quorum'}
		\end{alignat*}
		And we get again the desired inequality:\\
		$\max\set{\tup{\val',\ts'}\st\exists\tid'\in\quorum\ldotp\tup{\tid',\tup{\val',\ts'}}\in \msgmap'(\msg)}\cup\set{\Tmapval(\tid')\st\tid'\in\quorum'}\geq\Fver(\Tmapstart(\tid))$.
		
		\item \serverupdateack: 
		
		Here $\act=\objact\tup{\tid}$ and $q'=\tup{\msgmap',\Tmapid,\Tmapval',\Tmappc,\Tmaparg}$ where
		$\msg = \update(\_,\_,\tup{\val',\ts'})\in\dom{\msgmap}$, $\msgmap'=\msgmap[\msg\mapsto\msgmap(\msg)\uplus\set{\tid}]$, and $\Tmapval' = \Tmapval[\tid\mapsto\max\set{\Tmapval(\tid),\tup{\val',\ts'}}]$.
		
		For \cref{abdaws:1a} we have that $\tid\in \msgmap'(\msg)$ and as required $\Tmapval'(\tid)\geq\tup{\val',\ts'}$, and other needed inequalities are preserved when increasing $\Tmapval'(\tid)$.
		
		For \cref{abdaws:1d} we note that the set mentioned in the condition is unchanged as both $\Tmapval(\tid)$ and $\tup{\val',\ts'}$ are already in it, and the condition for each element in the set still holds as $\Tmappc$ is unchanged.
		\cref{abdaws:1b} will be justified based on the cases considered below.
		
		Specifically, by definition of $\tshist$ we either have $\tshist'=\tshist$ or $\tshist'=\tshist\uplus\set{\tup{\val',\ts'}}$ . We consider both cases:
		\begin{itemize}
			\item If $\tshist'=\tshist$:
			
			For \cref{abdaws:1b},  $\Tmapval'(\tid)\geq\Tmapval(\tid)$ and so given any $\quorum$, $\max\set{\Tmapval'(\tid')\st\tid'\in\quorum}$ can only increase, while $\max\tshist'=\max\tshist$ is unchanged.
			
			We perform an empty sequence of transitions
			
			Note that the change to $\msgmap'(\msg)$ where $\msg$ is an update message can only possibly effect $\tshist'$, but here $\tshist'=\tshist$.
			
			For \cref{abdaws:3c4} we similarly have that a maximum involving $\Tmapval(\tid)$ is an upper bound on another value and this still holds when replacing it with $\Tmapval'(\tid)$.

			\item If $\tshist'=\tshist\uplus\set{\tup{\val',\ts'}}$:
			
			For \cref{abdaws:1b}, let $\quorum\subseteq\Tid$ be a set such that $\size{\quorum}>\frac{\size{\Tid}}{2}$.
			We have that $\tup{\val',\ts'}\in\tshist'$ and thus $\size{\msgmap'(\msg)}>\frac{\size{\Tid}}{2}$.
			Therefore, there must exist some $\tilde{\tid}\in\quorum\cap\msgmap'(\msg)$.
			We get that $\max\set{\Tmapval'(\tid')\st\tid'\in\quorum}\geq\Tmapval'(\tilde{\tid})\geq\tup{\val',\ts'}$, where the last inequality uses \cref{abdaws:1a}.
			Furthermore, similarly to the previous case as $\Tmapval$ only increases point-wise we have using \cref{abdaws:1b} that $\max\set{\Tmapval'(\tid')\st\tid'\in\quorum}\geq\max\set{\Tmapval(\tid')\st\tid'\in\quorum}\geq\max\tshist$.
			Combining these we get
			$\max\set{\Tmapval'(\tid')\st\tid'\in\quorum}\geq\max\set{\tup{\val',\ts'},\max\tshist}=\max\tshist'$.
			
			Now, since $\tup{\val',\ts'}\notin\tshist$ and $\msg = \update(\_,\_,\tup{\val',\ts'})\in\dom{\msgmap}$, using \cref{abdaws:1d} we get that for the thread $\tid'$ such that $\ts'=\tup{\_,\tid'}$, taking $\id'=\Tmapid(\tid')[\last]$ we have $\update(\tid',\id',\tup{\val',\ts'})\in\dom{\msgmap}$, $\Tmappc(\tid')=\wupdate$,
			and $\tup{\val',\ts'}=\tup{\val_{\tid'},\ts_{\tid'}}\notin\tshist$.
			Thus using \cref{abdaws:3c5} it must be the case that $\Fver(\Tmapstart({\tid'}))\leq\tup{\val_{\tid'},\ts_{\tid'}}$, $\abs\Tmappc({\tid'})=\wwait$, and thus from \cref{abdaws:3c2} $\Tmaparg({\tid'})=\abs\Tmaparg({\tid'})$.
			Using \cref{abdaws:1c2} we have that  $\val_{\tid'}=\Tmaparg({\tid'})$, and overall this gives $\val_{\tid'}=\abs\Tmaparg({\tid'})$.
			
			We consider two sub-cases:
			\begin{itemize}
				\item If $\tup{\val_{\tid'},\ts_{\tid'}}=\max\tshist'$
				
				Since $\abs\Tmappc({\tid'})=\wwait$ and $\lockLTS=0$, we can use the rule \writetop\ and get a transition\\
				$\tup{\val,\ver,\lockLTS,\abs\Tmapid,\abs\Tmappc,\Tmapvalset,\Tmapstart,\abs\Tmaparg} \asteplab{\objact\tup{\tid'}}{}\\ \tup{\abs\Tmaparg(\tid'),\ver+1,\lockLTS,\abs\Tmapid,\abs\Tmappc[\tid' \mapsto \wend],\Tmapvalset,\Tmapstart,\abs\Tmaparg}$
				
				Consider the set $\tidset=\set{\hat{\tid}\st\abs\Tmappc(\hat{\tid})=\rlisten}$.
				For every $\hat{\tid}\in\tidset$, $\abs\Tmappc(\hat{\tid})=\rlisten$ and \cref{abdaws:2c} gives that $\Tmapstart(\hat{\tid})\leq\ver<\ver+1$.
				These conditions together enable a transition using rule \readloop\ with label $\objact\tup{\hat{\tid}}$ for every $\hat{\tid}\in\tidset$ such that $\abs\Tmaparg(\tid')\notin\Tmapvalset(\hat{\tid})$, and this availability is unaffected by performing such a transition for other thread ids.
				Thus, we can perform a sequence of transitions using \readloop\ , one for each member of $\tidset$ such that $\abs\Tmaparg(\tid')\notin\Tmapvalset(\hat{\tid})$, and get:
				$\tup{\abs\Tmaparg(\tid'),\ver+1,\lockLTS,\abs\Tmapid,\abs\Tmappc[\tid' \mapsto \wend],\Tmapvalset,\Tmapstart,\abs\Tmaparg} \Longrightarrow^*\\ \tup{\abs\Tmaparg(\tid'),\ver+1,\lockLTS,\abs\Tmapid,\abs\Tmappc[\tid' \mapsto \wend],\Tmapvalset',\Tmapstart,\abs\Tmaparg}$
				
				where
				$\Tmapvalset'=\Tmapvalset[\forall\hat{\tid}\in\tidset\ldotp\hat{\tid} \mapsto \Tmapvalset(\hat{\tid})\cup\set{\abs\Tmaparg(\tid')}]$
				
				Now, consider the set $\readtidset=\set{\hat{\tid}\st\Tmappc(\hat{\tid})=\rupdate\land \tup{\val_{\hat{\tid}},\ts_{\hat{\tid}}}=\tup{\val_{\tid'},\ts_{\tid'}}}$.
				For every $\tid''\in\readtidset$, $\Tmappc(\tid'')=\rupdate$, $\tup{\val_{\tid''},\ts_{\tid''}}=\tup{\val_{\tid'},\ts_{\tid'}}\notin\tshist$ and thus from \cref{abdaws:3c5} we get $\abs\Tmappc(\tid'')=\rlisten$,
				which implies $\tid''\in\tidset$ and thus we have that $\abs\Tmaparg(\tid')\in\Tmapvalset'(\tid'')$.
				These conditions together enable a transition using rule \readpick\ with label $\objact\tup{\tid''}$, and they are unaffected by performing such a transition for other thread ids.
				Thus, we can perform a sequence of transitions using \readpick\ , one for each member of $\readtidset$:
				
				$\tup{\abs\Tmaparg(\tid'),\ver+1,\lockLTS,\abs\Tmapid,\abs\Tmappc[\tid' \mapsto \wend],\Tmapvalset',\Tmapstart,\abs\Tmaparg} \Longrightarrow^*\\ \tup{\abs\Tmaparg(\tid'),\ver+1,\lockLTS,\abs\Tmapid,{\abs\Tmappc}',\Tmapvalset',\Tmapstart,{\abs\Tmaparg}'}$
				
				where ${\abs\Tmappc}'=\abs\Tmappc[\tid' \mapsto \wend][\forall\hat{\tid}\in\readtidset\ldotp\hat{\tid}\mapsto\rend]$,  ${\abs\Tmaparg}'=\abs\Tmaparg[\forall\hat{\tid}\in\readtidset\ldotp\hat{\tid}\mapsto\abs\Tmaparg(\tid')]$
				
				For \cref{abdaws:2b} and \cref{abdaws:2c} we can use $\ver<\ver+1$.
				
				To show the conditions on $\tup{q',{\abs q}'}$ hold we define $\Fver:\set{0\til \ver+1}\to\Val\times\TS$ by:
				\begin{equation*}
					\Fver'(\tm)=\begin{cases}
						\Fver(\tm) & \tm\leq\ver \\
						\tup{\val_{\tid'},\ts_{\tid'}} & \tm=\ver+1
					\end{cases}
				\end{equation*}
				This function is monotonic since $\Fver$ is monotonic and $\tup{\val_{\tid'},\ts_{\tid'}}=\max\tshist'$ which implies that for all $\tup{\val'',\ts''}\in\tshist$ we have $\tup{\val'',\ts''}<\tup{\val_{\tid'},\ts_{\tid'}}$, and in particular $\Fver'(\ver)=\Fver(\ver)=\max\tshist$ and thus $\Fver'(\ver)<\Fver'(\ver+1)$.
				We also get for \cref{abdaws:3b} that $\Fver'(\ver+1)=\max\tshist'=\tup{\abs\Tmaparg(\tid'),\_}$, where for the last equality we use $\val_{\tid'}=\abs\Tmaparg(\tid')$, shown above.
				
				For the condition on each thread $\tid''\in\Tid$ (aside from \cref{abdaws:3c1} which is easily verified):
				
				To show \cref{abdaws:3c2} still holds for all $\tid''\in\readtidset$, we have $\Tmaparg(\tid'')=\val_{\tid''}=\val_{\tid'}=\abs\Tmaparg(\tid')={\abs\Tmaparg}'(\tid'')$,
				where we used \cref{abdaws:1c2}, the definition of $\readtidset$, an equality proven above, and the definition of ${\abs\Tmaparg}'$
				
				For \cref{abdaws:3c4}, using \cref{abdaws:2c} we have $\Tmapstart(\tid'')\leq\ver<\ver+1$ and thus $\Fver'(\Tmapstart(\tid''))=\Fver(\Tmapstart(\tid''))$.
				As other values are unchanged, the condition continues to holds.

				To show that \cref{abdaws:3c3} holds for all $\tid''\in\tidset$.
				Explicitly, we need to show that\\
				$\Tmapvalset[\forall\hat{\tid}\in\tidset\ldotp\hat{\tid} \mapsto \Tmapvalset(\hat{\tid})\cup\set{\abs\Tmaparg(\tid')}](\tid'') =\\ \set{\val'\st \exists \ts'\ldotp \Fver'(\Tmapstart(\tid''))\leq\tup{\val',\ts'}\leq\Fver'(\ver+1)\land \tup{\val',\ts'}\in \tshist'}$
				
				Simplifying both sides, and using $\Fver'(\ver+1)=\tup{\abs\Tmaparg(\tid'),\_}$, $\tshist'=\tshist\uplus\set{\tup{\abs\Tmaparg(\tid'),\_}}$, and the fact that$\Fver(\tm)=\Fver'(\tm)$ for $\tm\leq\ver$, we get\\
				$\Tmapvalset(\tid'')\cup\set{\abs\Tmaparg(\tid')} = \\
				\set{\val'\st \exists \ts'\ldotp \Fver(\Tmapstart(\tid''))\leq\tup{\val',\ts'}\leq\Fver(\ver)\land \tup{\val',\ts'}\in \tshist}\cup\set{\abs\Tmaparg(\tid')}$
				
				This equality follows from \cref{abdaws:3c3} holding for the pre-states $q,\abs q$ by adding $\set{\abs\Tmaparg(\tid')}$ to both sides.

				For \cref{abdaws:3c5} if we assume $\Tmappc(\tid'')\in\set{\rupdate,\wupdate}$ and $\tup{\val_{\tid''},\ts_{\tid''}}\notin\tshist'$,
				then as $\tshist\subseteq\tshist'$ we get $\tup{\val_{\tid''},\ts_{\tid''}}\notin\tshist$.
				From the assumption of \cref{abdaws:3c5} this means $\Fver(\Tmapstart(\tid''))\leq\tup{\val_\tid'',\ts_\tid''}$ and $\abs\Tmappc(\tid'')\in\set{\rlisten,\wwait}$.
				As $\tup{\val_{\tid''},\ts_{\tid''}}\notin\tshist'$ we have
				$\tup{\val_{\tid''},\ts_{\tid''}}\neq\tup{\val_{\tid'},\ts_{\tid'}}$ and thus $\tid''\neq\tid'$ and $\tid''\notin\readtidset$, which means ${\abs\Tmappc}'(\tid'')=\abs\Tmappc(\tid'')$, and as $\Fver'(\Tmapstart(\tid''))=\Fver(\Tmapstart(\tid''))$ the condition holds.
				
				For \cref{abdaws:3c6}, if $\Tmappc(\tid'')\in\set{\rupdate,\wupdate}$ and $\tup{\val_{\tid''},\ts_{\tid''}}\in \tshist'(\msgmap)$,
				Since $\tshist'=\tshist\uplus\set{\tup{\val_{\tid'},\ts_{\tid'}}}$, we need to consider two cases:
				The first case is $\tup{\val_{\tid''},\ts_{\tid''}}\in \tshist(\msgmap)$, where again we deduce $\tup{\val_{\tid''},\ts_{\tid''}}\neq\tup{\val_{\tid'},\ts_{\tid'}}$ and get that the condition still holds.
				The second case is that $\tup{\val_{\tid''},\ts_{\tid''}}=\tup{\val_{\tid'},\ts_{\tid'}}$.
				If $\Tmappc(\tid'')=\wupdate$ then using \cref{abdaws:1c2} and $\ts_{\tid''}=\ts_{\tid'}$ we get $\tid''=\tid'$,
				and we have that ${\abs\Tmappc}'(\tid')=\wend$, as needed.
				If $\Tmappc(\tid'')=\rupdate$ then we have $\tid''\in\readtidset$ and thus  ${\abs\Tmappc}'(\tid')=\rend$, again as needed.
				
				\item If $\tup{\val_{\tid'},\ts_{\tid'}}<\max\tshist'$
				
				Since $\tshist'=\tshist\uplus\set{\tup{\val',\ts'}}$ and $\tup{\val_{\tid'},\ts_{\tid'}}=\tup{\val',\ts'}$ is not the maximum of $\tshist'$,
				we have that $\max\tshist'=\max\tshist=\Fver(\ver)$.
				We get that $\Fver(\Tmapstart(\tid'))\leq\tup{\val_{\tid'},\ts_{\tid'}}<\Fver(\ver)$ and thus from monotonicity of $\Fver$ (in the contra-positive) we get $\Tmapstart(\tid')<\ver$.
				
				Together with the fact that $\abs\Tmappc({\tid'})=\wwait$ and $\lockLTS=0$, we can use the rule \writerollback\ and get a transition\\
				$\tup{\val,\ver,\lockLTS,\abs\Tmapid,\abs\Tmappc,\Tmapvalset,\Tmapstart,\abs\Tmaparg} \asteplab{\objact\tup{\tid'}}{}\\
				\tup{\abs\Tmaparg(\tid'),\ver-1,1,\abs\Tmapid,\abs\Tmappc[\tid' \mapsto \wrollback],\Tmapvalset,\Tmapstart,\abs\Tmaparg[\tid'\mapsto\val]}$
				
				Consider the set $\tidset'=\set{\hat{\tid}\st \abs\Tmappc(\hat{\tid})=\rlisten\land \Fver(\Tmapstart(\hat{\tid}))\leq\tup{\val_{\tid'},\ts_{\tid'}}}$.
				For every $\hat{\tid}\in\tidset'$, similar to the above using monotonicity of $\Fver$ we get $\Tmapstart(\hat{\tid})<\ver$ and thus $\Tmapstart(\hat{\tid})\leq\ver-1$.
				Therefore, a transition using rule \readloop\ with label $\objact\tup{\hat{\tid}}$ is available for every $\hat{\tid}\in\tidset'$ such that $\abs\Tmaparg(\tid')\notin\Tmapvalset(\hat{\tid})$, and remains available as such a transition is performed for other threads in $\tidset'$.
				Denoting $\Tmapvalset'=\Tmapvalset[\forall\hat{\tid}\in\tidset'\ldotp\hat{\tid} \mapsto \Tmapvalset(\hat{\tid})\cup\set{\abs\Tmaparg(\tid')}]$, we can perform a sequence of transitions using \readloop\ , one for each member of $\tidset'$ such that $\abs\Tmaparg(\tid')\notin\Tmapvalset(\hat{\tid})$:\\
				$\tup{\abs\Tmaparg(\tid'),\ver-1,1,\abs\Tmapid,\abs\Tmappc[\tid' \mapsto \wrollback],\Tmapvalset,\Tmapstart,\abs\Tmaparg[\tid'\mapsto\val]} \Longrightarrow^* \\ \tup{\abs\Tmaparg(\tid'),\ver-1,1,\abs\Tmapid,\abs\Tmappc[\tid' \mapsto \wrollback],\Tmapvalset',\Tmapstart,\abs\Tmaparg[\tid'\mapsto\val]}$
				
				now, Since $\abs\Tmappc[\tid' \mapsto \wrollback](\tid')=\wrollback$ we can perform a transition using the rule \writerollforward\ and get a transition\\
				$\tup{\abs\Tmaparg(\tid'),\ver-1,1,\abs\Tmapid,\abs\Tmappc[\tid' \mapsto \wrollback],\Tmapvalset',\Tmapstart,\abs\Tmaparg[\tid'\mapsto\val]}\\
				\asteplab{\objact\tup{\tid'}}{}\\
				\tup{\abs\Tmaparg[\tid'\mapsto\val](\tid'),\ver-1+1,0,\abs\Tmapid,\abs\Tmappc[\tid' \mapsto \wend],\Tmapvalset',\Tmapstart,\abs\Tmaparg[\tid'\mapsto\val]}$
				
				Simplifying this we get a state
				$\tup{\val,\ver,0,\abs\Tmapid,\abs\Tmappc[\tid' \mapsto \wend],\Tmapvalset',\Tmapstart,\abs\Tmaparg[\tid'\mapsto\val]}$
				
				Finally, consider the set $\readtidset=\set{\hat{\tid}\st\Tmappc(\hat{\tid})=\rupdate\land \tup{\val_{\hat{\tid}},\ts_{\hat{\tid}}}=\tup{\val_{\tid'},\ts_{\tid'}}}$.
				For every $\tid''\in\readtidset$, $\Tmappc(\tid'')=\rupdate$, $\tup{\val_{\tid''},\ts_{\tid''}}=\tup{\val_{\tid'},\ts_{\tid'}}\notin\tshist$ and thus from \cref{abdaws:3c5} we get $\Fver(\Tmapstart(\tid''))\leq\tup{\val_{\tid''},\ts_{\tid''}}$ and $\abs\Tmappc(\tid'')=\rlisten$,
				which implies $\tid''\in\tidset'$ and thus we have that $\abs\Tmaparg(\tid')\in\Tmapvalset'(\tid'')$.
				These conditions together enable a transition using rule \readpick\ with label $\objact\tup{\tid''}$, and they are unaffected by performing such a transition for other thread ids.
				Thus, we can perform a sequence of transitions using \readpick\ , one for each member of $\readtidset$:
				
				$\tup{\val,\ver,0,\abs\Tmapid,\abs\Tmappc[\tid' \mapsto \wend],\Tmapvalset',\Tmapstart,\abs\Tmaparg[\tid'\mapsto\val]} \Longrightarrow^*\\ \tup{\val,\ver,0,\abs\Tmapid,{\abs\Tmappc}',\Tmapvalset',\Tmapstart,{\abs\Tmaparg}'}$
				
				where ${\abs\Tmappc}'=\abs\Tmappc[\tid' \mapsto \wend][\forall\hat{\tid}\in\readtidset\ldotp\hat{\tid}\mapsto\rend]$ and\\
				${\abs\Tmaparg}'=\abs\Tmaparg[\tid'\mapsto\val][\forall\hat{\tid}\in\readtidset\ldotp\hat{\tid}\mapsto\abs\Tmaparg(\tid')]$

				For \cref{abdaws:2a} note that $\lockLTS=0$ in the post-state ${\abs q}'$.
				To show the conditions on $\tup{q',{\abs q}'}$:
				
				For \cref{abdaws:3b} we use the fact that $\max\tshist'=\max\tshist$.
				
				For the condition on each thread $\tid''\in\Tid$: \cref{abdaws:3c2}, \cref{abdaws:3c5} and \cref{abdaws:3c6} are justified exactly like in the previous case.
				For \cref{abdaws:3c4}, no relevant values are changed.
				
				It remains to show that \cref{abdaws:3c3} holds for every $\tid''\in\tidset=\set{\hat{\tid}\st \abs\Tmappc(\hat{\tid})=\rlisten}$.
				Explicitly, we need to show that
				$\Tmapvalset[\forall\hat{\tid}\in\tidset'\ldotp\hat{\tid} \mapsto \Tmapvalset(\hat{\tid})\cup\set{\abs\Tmaparg(\tid')}](\tid'') = \set{\val'\st \exists \ts'\ldotp \Fver(\Tmapstart(\tid''))\leq\tup{\val',\ts'}\leq\Fver(\ver)\land \tup{\val',\ts'}\in \tshist'}$.
				
				If $\tid''\in\tidset'$, this simplifies to:\\
				$\Tmapvalset(\tid'')\cup\set{\abs\Tmaparg(\tid')} = \set{\val'\st \exists \ts'\ldotp \Fver(\Tmapstart(\tid''))\leq\tup{\val',\ts'}\leq\Fver(\ver)\land \tup{\val',\ts'}\in \tshist'}$.
				From $\tid''\in\tidset'$ we get that $\Fver(\Tmapstart(\tid''))\leq\tup{\val_{\tid'},\ts_{\tid'}}$,
				and using the fact that
				$\tup{\abs\Tmaparg(\tid'),\_}=\tup{\val_{\tid'},\ts_{\tid'}}\leq\max\tshist'=\Fver(\ver)$
				and $\tup{\val_{\tid'},\ts_{\tid'}}\in\tshist'$
				we get\\
				$\set{\val'\st \exists \ts'\ldotp \Fver(\Tmapstart(\tid''))\leq\tup{\val',\ts'}\leq\Fver(\ver)\land \tup{\val',\ts'}\in \tshist'}=\\
				\set{\val'\st \exists \ts'\ldotp \Fver(\Tmapstart(\tid''))\leq\tup{\val',\ts'}\leq\Fver(\ver)\land \tup{\val',\ts'}\in \tshist}\cup\set{\abs\Tmaparg(\tid')}$,
				and then similarly to the previous case the desired equality follows from \cref{abdaws:3c3}.
				
				If $\tid''\in\tidset\setminus\tidset'$, this simplifies to:\\
				$\Tmapvalset(\tid'') = \set{\val'\st \exists \ts'\ldotp \Fver(\Tmapstart(\tid''))\leq\tup{\val',\ts'}\leq\Fver(\ver)\land \tup{\val',\ts'}\in \tshist'}$.
				
				Since $\abs\Tmappc(\tid'')=\rlisten$ and $\tid''\notin\tidset'$,
				$\Fver(\Tmapstart(\tid''))>\tup{\val_{\tid'},\ts_{\tid'}}$ and thus\\
				$\set{\val'\st \exists \ts'\ldotp \Fver(\Tmapstart(\tid''))\leq\tup{\val',\ts'}\leq\Fver(\ver)\land \tup{\val',\ts'}\in \tshist'}=\\
				\set{\val'\st \exists \ts'\ldotp \Fver(\Tmapstart(\tid''))\leq\tup{\val',\ts'}\leq\Fver(\ver)\land \tup{\val',\ts'}\in \tshist}$
				and it remains to verify the equality $\Tmapvalset(\tid'') = \set{\val'\st \exists \ts'\ldotp \Fver(\Tmapstart(\tid''))\leq\tup{\val',\ts'}\leq\Fver(\ver)\land \tup{\val',\ts'}\in \tshist}$, which is true using \cref{abdaws:3c3}. \qedhere
			\end{itemize}

		\end{itemize}

	\end{itemize}
	
\end{proof}

\end{document}